\documentclass[hidelinks,journal]{IEEEtran}

\usepackage{amsmath,amssymb,amsthm}
\usepackage{bm}
\usepackage{orcidlink}
\usepackage{enumerate}
\usepackage{graphicx}
\usepackage{algorithmic}
\usepackage{algorithm}
\usepackage{cite}
\usepackage{hyperref}
\usepackage{xcolor}
\usepackage{booktabs}
\usepackage{array}
\usepackage{multirow}
\usepackage[acronym,shortcuts]{glossaries}
\usepackage{threeparttable}
\usepackage{cleveref} 

\newcommand{\proofref}[1]{The proof is given in Appendix~\ref{#1}.}

\theoremstyle{definition}
\newtheorem{definition}{Definition}
\theoremstyle{remark}
\newtheorem{remark}{Remark}
\theoremstyle{plain}
\newtheorem{theorem}{Theorem}
\newtheorem{proposition}{Proposition}
\newtheorem{lemma}{Lemma}
\newtheorem{corollary}{Corollary}

\newcommand{\bmu}{\bm{\mu}}

\newcommand{\E}{\mathbb{E}}
\newcommand{\Var}{\mathrm{Var}}
\newcommand{\CN}{\mathcal{CN}}
\newcommand{\cR}{\mathcal{R}}

\newcommand{\Real}{\mathbb{R}}
\newcommand{\Complex}{\mathbb{C}}

\newcommand{\tr}{\mathrm{tr}}
\newcommand{\abs}[1]{\left|#1\right|}
\newcommand{\norm}[1]{\left\|#1\right\|}
\newcommand{\opnorm}[1]{\left\|#1\right\|_{\mathrm{op}}}

\newcommand{\trans}[0]{^{\mathsf{T}}}
\newcommand{\herm}[0]{^{\mathsf{H}}}

\newacronym{SOAV}{SOAV}{sum-of-absolute-values}
\newacronym{FP}{FP}{fractional programming}
\newacronym{MP}{MP}{message passing}
\newacronym{LM}{LM}{Lapidoth-Merhav}
\newacronym{GaBP}{GaBP}{Gaussian belief propagation}
\newacronym{GAMP}{GAMP}{generalized approximate message passing}
\newacronym{AMP}{AMP}{approximate message passing}
\newacronym{BOD}{BOD}{Bayes-optimal denoiser}
\newacronym{OBD}{OBD}{orbital Bessel denoiser}
\newacronym{LE}{LE}{linear estimator}
\newacronym{MAP}{MAP}{maximum a posteriori}
\newacronym{MMSE}{MMSE}{minimum mean square error}
\newacronym{LMMSE}{LMMSE}{linear minimum mean square error}
\newacronym{SPA}{SPA}{sum-product algorithm}
\newacronym{CLT}{CLT}{central limit theorem}
\newacronym{LLN}{LLN}{law of large numbers}
\newacronym{SE}{SE}{state evolution}
\newacronym{MSE}{MSE}{mean square error}
\newacronym{m-SE}{m-SE}{mismatched state evolution}
\newacronym{EM}{EM}{expectation maximization}
\newacronym{ARD}{ARD}{automatic relevance determination}
\newacronym{CPT}{CPT}{continuous phase trap}
\newacronym{APSK}{APSK}{amplitude and phase-shift keying}
\newacronym{QAM}{QAM}{quadrature amplitude modulation}
\newacronym{KL}{KL}{Kullback-Leibler}
\newacronym{MIMO}{MIMO}{multiple-input multiple-output}
\newacronym{SISO}{SISO}{single-input single-output}
\newacronym{SNR}{SNR}{signal-to-noise ratio}
\newacronym{BER}{BER}{bit error rate}
\newacronym{UBAE}{UBAE}{universal blind amplitude equalizer}
\newacronym{MOS}{MOS}{model order selection}
\newacronym{PMF}{PMF}{probability mass function}
\newacronym{SRAM}{SRAM}{static random-access memory}
\newacronym{ASIC}{ASIC}{application-specific integrated circuit}
\newacronym{FPGA}{FPGA}{field programmable gate array}
\newacronym{ALU}{ALU}{arithmetic logic unit}
\newacronym{vMBP}{vMBP}{von Mises belief propagation}
\newacronym{PSK}{PSK}{phase-shift keying}
\newacronym{SotA}{SotA}{state-of-the-art}
\newacronym{2D}{2D}{two dimensional}
\newacronym{1D}{1D}{one dimensional}
\newacronym{AWGN}{AWGN}{additive white Gaussian noise}
\newacronym{VN}{VN}{variable node}
\newacronym{FN}{FN}{factor node}
\newacronym{ASG}{ASG}{asymptotic SNR gap}
\newacronym{iid}{i.i.d.}{independent and identically distributed}
\newacronym{OD-GaBP}{OD-GaBP}{orbital detection Gaussian belief propagation}
\newacronym{OD-GAMP}{OD-GAMP}{orbital detection generalized approximate message passing}
\newacronym{OD-AMP}{OD-AMP}{orbital detection approximate message passing}
\newacronym{OD}{OD}{orbital detection}
\newacronym{DCT}{DCT}{dominated convergence theorem}
\newacronym{GMI}{GMI}{generalized mutual information}
\newacronym{OPD}{OPD}{orbital phase denoiser}
\newacronym{ML}{ML}{maximum likelihood}
\newacronym{OGD}{OGD}{orbital Gaussian denoiser}
\newacronym{I-MMSE}{I-MMSE}{information minimum mean square error}
\newacronym{ASK}{ASK}{amplitude-shift keying}
\newacronym{CRLB}{CRLB}{Cram\'{e}r-Rao lower bound}
\newacronym{MVUE}{MVUE}{minimum-variance unbiased estimator}
\newacronym{RS}{RS}{replica-symmetric}
\newacronym{EP}{EP}{expectation propagation}
\newacronym{EC}{EC}{expectation consistent}
\newacronym{GLM}{GLM}{generalized linear model}
\newacronym{SER}{SER}{symbol error rate}
\newacronym{LLR}{LLR}{log-likelihood ratio}
\newacronym{AEP}{AEP}{asymptotic equipartition property}
\newacronym{LDP}{LDP}{large deviation principle}
\newacronym{nu-OBD}{$U$-OBD}{$U$th-order orbital Bessel denoiser}
\newacronym{MUI}{MUI}{multi-user interference}
\newacronym{OAMP}{OAMP}{orthogonal approximate message passing}
\newacronym{MAMP}{MAMP}{memory approximate message passing}
\newacronym{VAMP}{VAMP}{vector approximate message passing}
\newacronym{PAM}{PAM}{pulse amplitude modulation}
\newacronym{CDMA}{CDMA}{code-division multiple access}

\begin{document}



\title{Orbital Detection}



\author{Kuranage Roche Rayan Ranasinghe,~\IEEEmembership{Graduate
Student Member,~IEEE,} \\ and Giuseppe Thadeu Freitas de Abreu,~\IEEEmembership{Senior Member,~IEEE}%
\thanks{The authors are with the School of Computer Science and Engineering, Constructor University, 28759 Bremen, Germany (e-mails: \{kranasinghe,gabreu\}@constructor.university).}
\vspace{-5ex}}

\maketitle

\begin{abstract}
We introduce \ac{OD}, a framework for designing asymptotically optimal, low-complexity \ac{MP} receivers for digitally modulated \ac{MIMO} systems, based on relaxing the discrete symbol prior into a mixed discrete-continuous density.
The resulting \emph{orbital prior} factors each symbol's distribution into a discrete radial component, supported on only the $L \ll M$ amplitude rings of an arbitrary constellation $\mathcal{M}$ of cardinality $M = |\mathcal{M}|$, and a continuous, maximum-entropy phase density on each ring.
This compresses the propagated posterior mean and variance losslessly into $3L$ real scalars, and collapses the optimal $\mathcal{O}(M)$-complexity denoiser into a closed-form hierarchy whose per-symbol cost falls to $\mathcal{O}(L)$ and ultimately $\mathcal{O}(1)$: the \ac{OBD}, its Bessel-free variant the \ac{OGD}, and the \ac{OPD}, proved irreducible on the ring manifold.
A Jacobi-Anger ladder recovers the exact detector with geometrically vanishing error.
Five information-theoretic results follow.
First, the \ac{OBD}, \ac{OGD}, and \ac{OPD} share an identical leading-order \ac{SE} fixed point.
Second, the sole price is a change in the high-\ac{SNR} error-decay law, from exponential to linear, which never hardens into an error floor.
Third, for any underloaded system the induced rate loss vanishes exponentially in \ac{SNR}, so every level is asymptotically capacity-achieving in the constellation-constrained sense, attaining $\log_2 M$.
Fourth, \ac{OD} attains an \ac{MMSE} dimension $d=1/2$, halfway between the $d=0$ \ac{BOD} and the $d=1$ linear receiver.
Fifth, a non-asymptotic optimal-transport bound in Wasserstein distance links constellation ring geometry directly to the achievable rate.
\end{abstract}


\glsresetall

\vspace{-2ex}
\section{Introduction}
\label{sec:intro}

\IEEEPARstart{R}{elaxing} a discrete constraint to a continuous one is among the most productive and recurring tools in information theory, with various well-known examples.
When the carrier phase of a received signal is unknown, averaging the likelihood over a uniform phase, that is, replacing a discrete reference by continuum, produces the modified Bessel function that has anchored noncoherent detection since Turin~\cite{Turin1956CommunicationTN} and Marcum~\cite{Marcum1948AST}.
Similarly, relaxing the points of a dense constellation into a continuous uniform density function, as per the continuous approximation of Forney and Wei \cite{ForneyWei1989}, reduces the constellation design problem to elementary geometry, exposing the ultimate shaping gain~\cite{ForneyUngerboeck1998}.
In turn, when the zero-error capacity of a graph resists combinatorial attack, Lov\'{a}sz's~\cite{Lovasz1979} relaxation of the discrete independence number to a continuous semidefinite program yields one of the tightest computable bounds known.

In each of these cases, a discrete object too costly to handle directly is embedded in a continuum in which its structure becomes analytically tractable, often at a cost that can be characterized approximately or exactly.
This paper brings that tool to the dominant complexity bottleneck of massive-\ac{MIMO} detection, focusing in particular on \ac{MP} architectures that underpin modern \ac{MIMO} systems~\cite{Donoho2009,Bayati2011}.


The gold standard of \ac{MP} receivers is the \ac{BOD}, which, for an $M$-point constellation, evaluates $M$ posterior weights per symbol, a kernel of complexity order $\mathcal{O}(M)$.
Executed once per user, per antenna, and per iteration, it is no longer an implementation nuisance but \emph{the} complexity constraint of receivers employed in modern systems.
For example, a single \ac{GAMP} iteration for $1024$-\ac{QAM} detection spends $1024$ Gaussian likelihood evaluations on every symbol evaluated~\cite{Rangan2011}.

While such costs were historically alleviated by transistor scaling, that subsidy will soon be unavailable, with the approaching end of Moore's Law~\cite{Moore2006,TheisWong2017}.
Arguably, a sustainable route to scalable high-order detection is to reduce the algorithmic order of the denoiser itself, and work exploiting this approach exists.

To cite a few, continuous relaxations of the search over discrete symbol sets in \ac{MIMO} detection  were proposed by Hayakawa and Hayashi, first in the form of a convex \ac{SOAV} optimization problem~\cite{HayakawaTWC2017,HayakawaACCESS2018},   and later reformulated as a discreteness-aware \ac{AMP} counterpart with a rigorous \ac{SE} characterization~\cite{HayakawaTSP2018}.
Building on those results towards larger systems, Iimori \emph{et al.}~\cite{iimoriOJCOMS2021}  proposed an improved discreteness-aware detection scheme for large-scale overloaded \ac{MIMO} detection where the search over a discrete space  is relaxed into an optimization problem with a continuous and differentiable regularized objective, convexized via \ac{FP}.
That scheme has also been recently shown to admit \ac{AMP}-based implementation~\cite{shrestha2026regularized}.

While the aforementioned methods explored different architectures and work for arbitrary constellations, considering the special case of \ac{MP}-based detection of $M$-\ac{PSK} modulation, Suresh et al. recently showed~\cite{Suresh2026} that relaxing only the discrete phase prior into a continuous von Mises distribution on the unit circle yields a denoiser that collapses to a closed-form $\mathcal{O}(1)$ update.
However, that construction is welded to the constant-modulus assumption, which breaks precisely where the complexity problem lives, namely, $M$-\ac{QAM} and \ac{APSK} constellations, whose multiple amplitudes are not nuisance parameters, but the primary information-bearing dimension.
Whether a principled continuous relaxation exists beyond $M$-\ac{PSK}, and what it costs in information-theoretic terms, remain open.

Motivated by the aforementioned ideas, the observation at the heart of this paper is that the $\mathcal{O}(M)$ complexity barrier in \ac{MIMO} detection is a fundamentally phase phenomenon.
To elaborate, any discrete digital constellation, including $M$-\ac{PSK}, $M$-\ac{QAM} and \ac{APSK}, can be decomposed into $L$ concentric amplitude rings, with $L$ small compared to $M$.
For square $M$-\ac{QAM}, for instance, $L$ is a vanishing fraction of $M$, and as small as $4$ for DVB-S2x \ac{APSK}.
In other words, while the ring identity is cheap, the $M_\ell$ discrete phases \emph{within} each ring carry the exponentially-scaling detection cost.

The key idea of \ac{OD} is therefore to relax only the phase distribution into the \emph{orbital prior}, keeping the radial marginal of the constellation intact while spreading the phase uniformly on each ring in a manner to provably preserve the maximum-entropy completion of the ring geometry.
Two consequences follow.
First, the posterior mean and variance of any multi-ring constellation can be shown to reside losslessly in a $3L$-dimensional subspace, such that the message dimension drops from $M$ to $3L$ real scalars before any approximation is made.
Second, under the orbital prior, the per-ring phase posterior is von Mises in closed form, collapsing the \ac{BOD} to the \ac{OD}-based $\mathcal{O}(L)$ \ac{OBD} and, by further degenerations, to a Bessel-free \ac{OGD} and an $\mathcal{O}(1)$ \ac{OPD} whose entire arithmetic is one ring lookup and one phase read-out.
The single approximation in this chain is quantified exactly, namely, the Wasserstein-1 distance between the orbital prior and the truth on ring $\ell$ is $\Theta(R_\ell/M_\ell)$, in closed form\footnotemark, and a Jacobi-Anger refinement drives it to zero geometrically in the retained order.

The central and most counter-intuitive result of the above is that this three-orders-of-magnitude complexity reduction is invisible at the macroscopic level.
In other words, all levels of the denoiser hierarchy share an \emph{identical} leading-order \ac{SE} fixed point $\mathrm{MSE}_\infty = \sigma_z^2/(2 - \alpha)$, where $\alpha = K/N < 1$ is the ratio of transmit to receive antennas and $\sigma_z^2$ the noise variance; the denoisers differ only in $\mathcal{O}(\sigma_z^4)$ corrections that we compute in closed form and rank by amplitude-shrinkage bias.
The significance of this implication can be captured by a thought experiment: an iterative detector monitoring its own error cannot tell whether it is running the $\mathcal{O}(M)$ exact denoiser or the $\mathcal{O}(1)$ projection.

We further show that the cost of the \ac{OD} relaxation is the \emph{decay law}, namely, the fact that an exponential error decay is traded for linear error decay, which manifests as a constant, bounded \ac{ASG}, which is a parallel shift of the error curve and not a floor.
The information-theoretic ledger, however, is shown to lean in favor of \ac{OD}, because: a) the \ac{GMI} of the mismatched orbital decoding metric is found to follow in closed-form from the decoupling principle, b) the \ac{I-MMSE} relation converts the \ac{MSE} gap into a rate gap that vanishes \emph{exponentially} in \ac{SNR}, and c) every level of the hierarchy is asymptotically constellation-constrained capacity-achieving (attaining $\log_2 M$) for underloaded systems.
In summary, the \ac{OD} relaxation costs nothing in the limit that matters and, more pointedly, nothing in the \emph{regime} that matters.

Instead, the price paid by \ac{OD} asymptotically detaches from the constellation order $M$, precisely where the cost of exact detection grows, such that the relaxation becomes free exactly for the high-order constellations that render exact detection prohibitive.
We shall return to and clarify this point further in Remark~\ref{rem:complexity_cost_inversion}.
Before we proceed with offering a detailed description of our contributions, which go beyond fundamental principles onto the design of concrete and feasible algorithms, enriched by several related information-theoretical results, we want to remark that we are well aware of the fact that the essence behind the \ac{OD} idea is not new, with parts of it having in fact has been utilized in various important contributions, some of which are listed and categorized below.

\vspace{-2ex}
\subsection{Selected Related Work}
\label{sec:related_work}

\noindent\emph{1) \Ac{EP}}: The closest algorithmic relative of \ac{OD} is \ac{EP}-based \ac{MIMO} detection~\cite{Minka2001,Opper2005,Cespedes2014}, which replaces the discrete posterior with a moment-matched Gaussian and attains near-optimal performance for high-order \ac{QAM}~\cite{Cespedes2014}.
A crucial distinction from \ac{OD} is, however, that \ac{EP} does not touch the $\mathcal{O}(M)$ bottleneck, since matching the tilted moments still requires evaluating all $M$ likelihoods.
In other words, \ac{EP} can be considered a competitor to the \emph{exact} \ac{BOD}-based \ac{AMP}, but not to the \ac{OD} hierarchy as a whole, which reaches $\mathcal{O}(1)$.
Notice also that \ac{EP} discards the ring structure entirely, thus forgoing both the Wasserstein-quantifiable mismatch and the possibility of a closed-form \ac{SE} fixed-point, both of which are preserved under \ac{OD}.

\noindent\emph{2) Mismatched inference in the large-system limit}: Our information-theoretic analysis, offered in Subsections~\ref{sec:gmi} through \ref{sec:ot_bound}, runs parallel to the exact free-energy approach of Barbier \emph{et al.}~\cite{Barbier2019}, which characterizes the \ac{MMSE} and capacity of \acp{GLM} under \emph{Bayes-optimal} priors.
However, the \ac{OD} prior is \emph{mismatched} by construction, such that both the information-theoretic limit and the \emph{cost} of the \ac{OD} relaxation are objects of study, with the \ac{I-MMSE} identity of Guo, Shamai and Verd\'{u}~\cite{Guo2005} employed to translate the \ac{SE} \ac{MSE} gap into the corresponding rate gap.

The large-system regime addressed in this paper builds on three key results that predate the \ac{AMP} algorithm~\cite{DonohoAMP2006}: the replica analysis of \ac{CDMA} multiuser detection by Tanaka~\cite{tanaka_tit_2002} and the random-spreading spectral-efficiency results of Verd\'{u} and Shamai~\cite{verdu_tit_1999}, which establish the fixed-point description; the effective-interference characterization of linear receivers by Tse and Hanly~\cite{tse_tit_1999}, which supplies the \ac{LMMSE} baseline for comparisons; and M\"uller's random-matrix model of antenna arrays~\cite{muller_tit_2002} which can be seen as its \ac{MIMO} counterpart.
Each of these characterizes, however, a \emph{matched} large-system limit, whereas we quantify the cost of a deliberate prior mismatch.

\footnotetext{$\Theta$ denotes Landau's asymptotic tightness from above and below~\cite{Knuth1976}.}

\noindent\emph{3) Group-structured AMP priors}: \ac{OD} exploits the fact that a prior factorization over disjoint groups, characterized by amplitude rings, induces a group-factorized \ac{AMP} denoiser with tractable per-group \ac{SE}~\cite{Eldar2009,Vila2013,DonohoJohnstoneMontanari2013}, offering the first construction for non-\ac{PSK} \ac{MIMO} detection that yields a closed-form denoiser hierarchy with an exact mismatched-\ac{SE} trajectory.

\noindent\emph{4) Ring-structured APSK detection}: Exploiting \ac{APSK} ring geometry for reduced-complexity detection is classical in satellite communications.

A key example is De~Gaudenzi et al.~\cite{deGaudenzi2006,deGaudenzi2006b}, where detection is split into a coarse ring decision followed by an intra-ring phase decision, which can be consider the closest geometric precursor of \ac{OD}'s orbital prior.
The method thereby is, however, single-user only, heuristic in its ring-conditional likelihoods, and disconnected from the \ac{MP}-based detection framework.
In contrast, we provide the first large-system treatment with  exact $3L$-dimensional moment-sufficient statistics (Theorem~\ref{thm:ring_decomp}), Wasserstein-bounded mismatch (Proposition~\ref{prop:cbm_error}), and shared \ac{SE} fixed point across the hierarchy.

\noindent\emph{5) Other complexity-reduction routes}: Three further lines of work found in current literature are worth of mention, which however are complementary, rather than competing.
First, deep-unfolded detectors~\cite{samuel2019learning,he2020model,shlezinger2023model} circumvent the $\mathcal{O}(M)$ cost by learning a global update, but do so implicitly and without an analytic guarantee, such that \ac{OD} can serve as an interpretable prior block within them.
Second, \ac{OAMP}~\cite{ma2017orthogonal} and \ac{MAMP}~\cite{liu2022memory} address the \emph{linear} stage of \ac{AMP} and compose directly with our \emph{denoiser} stage, paving the way for possible extensions of \ac{OD} to correlated channels, which we leave to future work.
Third, Cartesian per-axis \ac{PAM}~\cite{proakis2007digital,SimonAlouini2005} is near-Bayes-optimal for square $M$-\ac{QAM} at $\mathcal{O}(\sqrt{M})$, but does not exist for \ac{APSK} or non-square constellations, does not enjoy the circularity of the discrete-to-continuous phase relaxation, and offers none of the mismatched-\ac{SE}, Wasserstein, or \ac{GMI} machinery, while \ac{OPD} matches hard-rounding complexity and applies universally.

 \vspace{-2ex}
\subsection{Summary of Contributions}
\label{sec:contributions}
 
In view of all the above, we finally summarize ou main results organized in three parts as follows.

\vspace{1ex}
{\noindent\bf Part I: Fundamentals} \emph{(Section~\ref{sec:spatial_compression})}

\vspace{1ex}
\noindent\emph{1) The orbital ring decomposition (Subsection~\ref{sec:compression})}:
For any multi-ring constellation, the exact posterior mean and variance, which are all that a moment-based message-passing update consumes, reside in a $3L$-dimensional subspace, compressing the message dimension from $M$ to $3L$ real scalars (Theorem~\ref{thm:ring_decomp}). 
For $M$-\ac{PSK}, the state collapses to a single complex scalar, an exact one-third routing reduction (Proposition~\ref{prop:psk_sufficiency}). 
For square $M$-\ac{QAM} the ring count grows only as $L = \Theta(M/\sqrt{\ln M})$ (Proposition~\ref{prop:ring_scaling}), so that the routing-payload reduction $M/3L$ is itself $\Theta(\sqrt{\ln M})$, unbounded in $M$, and $\mathcal{O}(1)$ for fixed-$L$ \ac{APSK} (Corollary~\ref{cor:compression_scaling}).

\vspace{0.5ex}
\noindent\emph{2) The orbital prior and its geometric cost (Subsection~\ref{sec:orbital})}:
We introduce the orbital prior, the maximum-entropy relaxation that preserves the radial marginal exactly and relaxes only the intra-ring phase to uniform (Definition~\ref{def:orbital}, Proposition~\ref{prop:maxent}). 
Its induced mismatch is quantified exactly, giving the first quantitative bridge from constellation geometry to inference error: the Wasserstein-1 distance between the orbital prior and the discrete truth is $\Theta(R_\ell/M_\ell)$ per ring (Proposition~\ref{prop:cbm_error}), with a closed-form two-term refinement for non-equidistributed \ac{QAM} rings (Corollary~\ref{cor:w1_qam}).

\vspace{1ex}
{\noindent\bf Part II: Algorithms} \emph{(Section~\ref{sec:cbm})}
\vspace{1ex}

\noindent\emph{3) Asymptotic exactness beyond \ac{OBD} (Subsections~\ref{sec:nle}--{\it B})}:
The orbital prior induces, in closed form, the $\mathcal{O}(L)$ \ac{OBD}, which is the exact orbital posterior mean, a Bessel-weighted average over the rings (Proposition~\ref{prop:cbm_denoiser}). 
Expanding its partition function by the Jacobi-Anger identity (Proposition~\ref{prop:ja_partition}) and retaining $U$ Fourier harmonics yields the \ac{nu-OBD} of order $U$ (Definition~\ref{def:bcbm}), a ladder from the \ac{OBD} ($U = 0$) to the exact \ac{BOD} ($U \to \infty$) whose residual mismatch contracts \emph{geometrically}, as $\mathcal{O}(R_\ell/M_\ell^{U+1})$ per ring (Proposition~\ref{prop:bcbm_wasserstein}).

\vspace{0.5ex}
\noindent\emph{4) The \ac{OGD}, and the irreducible \ac{OPD} (Subsections~\ref{sec:cgr}--{\it D})}:
Two relaxations below the \ac{OBD} are derived: the \ac{OGD} removes every Bessel evaluation via von Mises--Gaussian convergence (Propositions~\ref{prop:vm_gaussian},~\ref{prop:cgr_denoiser}), and its high-\ac{SNR} limit, the \ac{OPD}, reduces detection to one ring lookup and one phase read-out at $\mathcal{O}(1)$ cost. The \ac{OPD} is proved irreducible on the ring manifold: no operation of lower complexity attains vanishing \ac{MSE} under the ring constraint (Proposition~\ref{prop:lpd_irreducible}). 
It is moreover exactly the Euclidean projection onto the ring set (Proposition~\ref{prop:lpd_projection}), and at high \ac{SNR} its ring detection and phase estimation separate into two individually optimal sub-problems -- a minimum-distance ring decision and a Cram\'er--Rao-achieving phase read-out (Proposition~\ref{prop:det_est}).

\vspace{1ex}
{\noindent\bf Part III: Information-Theoretic Properties} \emph{(Sections~\ref{sec:se}-\ref{sec:inf_theory})}

\noindent\emph{5) Mismatched \ac{SE} and fixed-points (Subsections~\ref{sec:se_fixedpoint}--{\it D})}:
The mismatched \ac{SE} recursion is fully characterized: fixed-point existence by Brouwer's theorem (Theorem~\ref{thm:se_convergence}), monotone convergence from the monotonicity of the \ac{MMSE} in the noise level (Proposition~\ref{prop:se_monotone}), and uniqueness for underloaded systems by the Banach contraction principle (Corollary~\ref{cor:se_uniqueness}). Culminating, the \ac{OBD}, \ac{OGD}, and \ac{OPD} are proved to share the identical leading-order fixed point $\mathrm{MSE}_\infty = \sigma_z^2/(2-\alpha)$, yielding a complete complexity-performance equivalence from $\mathcal{O}(L)$ to $\mathcal{O}(1)$ (Corollary~\ref{cor:lpd_se}, Proposition~\ref{prop:cgr_se}, Corollary~\ref{cor:snr_gap}).

\vspace{0.5ex}
\noindent\emph{6) Loading gain growth over linear receivers (Subsection~\ref{sec:load_sweep})}:
The ratio of the \ac{LMMSE} fixed point to that of any orbital level is $(2-\alpha)/(1-\alpha)$ to leading order (Corollary~\ref{cor:load_gap}), a function of the load alone: identical for every constellation order and family and for all three orbital levels, bounded below by $3$~dB, the ratio of \ac{MMSE} dimensions, and divergent as $\alpha \to 1^-$. The operating points at which a linear receiver is most tempting are precisely those at which it is most costly.

\vspace{0.5ex}
\noindent\emph{7) Asymptotic \ac{SNR} gap and \ac{MMSE} dimension (Subsections~\ref{sec:cross_level} and~\ref{sec:MMSE_dim})}:
\ac{OD} converts error decay from exponential to linear, a strictly constant, bounded \ac{ASG} (Proposition~\ref{prop:linear_decay}), read through two lenses: the \ac{MMSE} dimension, $d_{\mathrm{B}} = d_{\mathrm{G}} = d_{\mathrm{P}} = \tfrac{1}{2}$ against $d_{\mathrm{D}} = 0$ for the \ac{BOD} (Proposition~\ref{prop:mmse_dimension}), and the \ac{I-MMSE} relation~\cite{Guo2005}, under which the per-ring information loss is bounded and convergent (Corollary~\ref{cor:immse}).

\noindent\emph{8) Achievable rate w. orbital prior (Subsections~\ref{sec:decoupling}--{\it B})}:
The \ac{GMI} of the mismatched orbital metric is obtained in closed form and decouples into a single-letter achievable rate (Theorem~\ref{thm:decoupling}, Corollaries~\ref{thm:gmi} and~\ref{cor:single_letter_rate}). The \ac{I-MMSE} derivative ties the operational \ac{MSE} to this rate, whose loss relative to the \ac{BOD} vanishes exponentially in \ac{SNR} (Proposition~\ref{prop:immse_rate_gap}). The same orbital metric doubles, without modification, as the soft-output (bit-\ac{LLR}) rule for a coded receiver, extending the $\mathcal{O}(M)\!\to\!\mathcal{O}(1)$ saving from hard detection to the soft interface of forward-error correction.

\begin{table}[H]
  \centering
  \caption{Frequently used acronyms.}
  \vspace{-2ex}
  \label{tab:acronyms}
  \begin{tabular}{l|l}
    \hline
    \textbf{Acronym} & \textbf{Description} \\
    \hline
    OD       & orbital detection \\
    MP       & message passing \\
    MIMO     & multiple-input multiple-output \\
    AMP      & approximate message passing \\
    AWGN     & additive white Gaussian noise \\
    SNR      & signal-to-noise ratio \\
    i.i.d.   & independent and identically distributed \\
    PSK      & phase-shift keying \\
    QAM      & quadrature amplitude modulation \\
    APSK     & amplitude and phase-shift keying \\
    BOD      & Bayes-optimal denoiser \\
    OBD      & orbital Bessel denoiser \\
    $U$-OBD  & $U$th-order orbital Bessel denoiser \\
    OGD      & orbital Gaussian denoiser \\
    OPD      & orbital phase denoiser \\
    MSE      & mean square error \\
    MMSE     & minimum mean square error \\
    LMMSE    & linear minimum mean square error \\
    SE       & state evolution \\
    ASG      & asymptotic SNR gap \\
    GMI      & generalized mutual information \\
    I-MMSE   & information minimum mean square error \\
    RS       & replica-symmetric \\
    KL       & Kullback-Leibler \\
    CRLB     & Cram\'{e}r-Rao lower bound \\
    \hline
  \end{tabular}
\end{table}

\vspace{-2ex}
\noindent\emph{9) Optimal-transport bound (Subsection~\ref{sec:ot_bound})}: A non-asymptotic bound on the \ac{SE} fixed-point gap in terms of the Wasserstein distance closes the causal chain from constellation geometry, through macroscopic \ac{MSE} degradation, to rate loss (Theorem~\ref{thm:ot_bound}, Remark~\ref{rem:closing_loop}).
The per-ring transport distance is obtained in closed form for an \emph{arbitrary} phase set (Proposition~\ref{prop:w2_general}), which both identifies equidistributed rings as the exactly-solvable case and separates the constellation families: the certificate is $\Theta(1/M^2)$ for $M$-\ac{PSK}, $\Theta(L^2/M^2)$ for \ac{APSK}, and $\Theta(1)$ for square \ac{QAM}, whose rings stay sparsely populated at every order (Remark~\ref{rem:w2_families}).

\vspace{-2ex}
\subsection{Organization, Notation, and List of Acronyms}
The remainder of the paper is organized as follows.
Section~\ref{sec:system} presents the system model, the \ac{AMP} decoupling into a scalar denoising problem, and the ring geometry.
Section~\ref{sec:spatial_compression} establishes the $3L$-dimensional posterior compression and introduces the orbital prior, its maximum-entropy characterization, and the exact Wasserstein-1 mismatch.
Section~\ref{sec:cbm} derives the denoiser hierarchy -- the \ac{OBD} (Section~\ref{sec:nle}), the \ac{nu-OBD} (Section~\ref{sec:ja_cbm}), the \ac{OGD} (Section~\ref{sec:cgr}), and the \ac{OPD} (Section~\ref{sec:lpd}).
Section~\ref{sec:se} presents the mismatched \ac{SE} analysis and the cross-level equivalence theorem.
Section~\ref{sec:inf_theory} develops the information-theoretic consequences -- the \ac{MMSE} dimension, the \ac{GMI} under mismatched decoding (Section~\ref{sec:gmi}), the decoupling-based single-letter rate (Section~\ref{sec:decoupling}), and the optimal-transport bound (Section~\ref{sec:ot_bound}).
Numerical validations accompany each result.
Section~\ref{sec:conclusion} concludes.

\noindent \emph{Notation:} Boldface lower- and uppercase letters denote column vectors and matrices (e.g., $\mathbf{x}$ and $\mathbf{H}$), respectively, while calligraphic uppercase letters (e.g., $\mathcal{M}$, $\cR$) denote sets, with $|\cdot|$ the cardinality.
Lowercase italics are used for the transmitted symbol and every scalar derived from it and, in keeping with the algorithmic development, \emph{no} typographic distinction is drawn between a random quantity and its realization (thus $\E[\cdot]$ and mutual information $I(\cdot\,;\cdot)$ act on the lowercase symbols directly).
Uppercase italics are reserved for \emph{structural}, non-random quantities.
The sets of real and complex numbers are $\Real$ and $\Complex$; for $x \in \Complex$, $\norm{x}$, $\angle x$, and $x^{*}$ denote its modulus, phase, and conjugate, with $j \triangleq \sqrt{-1}$.
The operators $(\cdot)\trans$, $(\cdot)\herm$, and $\|\cdot\|_F$ denote transpose, Hermitian (conjugate) transpose, and Frobenius norm, while $\E[\cdot]$ and $\Var[\cdot]$ denote expectation and variance.
We write $\Pr(\cdot)$ for the probability of an event, and reserve the lowercase symbols $p$ and $q$ for probability densities and mass functions: $p$ always denotes the \emph{true} (matched) data-generating distribution, whereas $q$ denotes the \emph{orbital} (mismatched) relaxation on which the proposed denoisers are built.
$\E_p$, $\E_q$, $\Var_q$, and $\mathrm{mmse}_q$ denote expectation, variance, and \ac{MMSE} taken under the indicated distribution.
Finally, $\mathcal{CN}(\mu,\sigma^{2})$ -- and, for vectors, $\CN(\bmu, \mathbf{\Sigma})$ with mean $\bmu$ and covariance $\mathbf{\Sigma}$ -- is the circularly symmetric complex Gaussian distribution. Table~\ref{tab:acronyms} collects the frequently used acronyms used throughout, and Table~\ref{tab:symbols_system} the recurring symbols.
\vspace{-4ex}
\begin{table}[H]
  \centering
  \caption{Frequently used symbols.}
  \vspace{-2ex}
  \label{tab:symbols_system}
  \begin{tabular}{c|l@{\,}|l@{\,}}
    \hline
    & \textbf{Symbol} & \textbf{Description} \\
    \hline
    \multirow{6}{*}{\rotatebox[origin=c]{90}{\it System model}}
    & $K$, $N$ & numbers of transmit and receive antennas \\
    & $\alpha = K/N$ & system load \\
    & $\mathbf{H}$, $H_{nk}$ & channel matrix and its entries \\
    & $\mathbf{y}$, $\mathbf{x}$ & received and transmitted vectors \\
    & $\mathbf{z}$, $\sigma_z^2$ & physical noise vector and its per-entry variance \\
    & $\hat{\mathbf{x}}$ & symbol-vector estimate \\
    \hline
    \multirow{11}{*}{\rotatebox[origin=c]{90}{\it \;\;\; Constellation \& geometry}}
    & $\mathcal{M}$, $M$ & constellation and its cardinality \\
    & $s_m$, $p_m$ & $m$-th constellation point and its prior probability \\
    & $\phi_m$ & phase of $s_m$ \\
    & $E_d$ & average symbol energy \\
    & $L$ & number of amplitude rings \\
    & $R_\ell$, $M_\ell$, $r_\ell$ & radius, points and prior probability of ring $\ell$ \\
    & $\ell^*$ & index of the detected (dominant) ring \\
    & $d_R$ & minimum inter-ring distance \\
    & $d_{\min}$ & minimum Euclidean distance \\
    & $\gamma$ & inverse-energy coefficient $\sum_\ell r_\ell R_\ell^{-2}$ \\
    \hline
    \multirow{4}{*}{\rotatebox[origin=c]{90}{\it Distr.}}
    & $p$ & true (matched) distribution \\
    & $q$ & orbital (mismatched) distribution \\
    & $\Pr(\cdot)$ & probability of an event \\
    & $\phi_{\bar{\sigma}^2}$ & complex Gaussian kernel of variance $\bar{\sigma}^2$ \\
    \hline
    \multirow{12}{*}{\rotatebox[origin=c]{90}{\it \;\; Denoisers \& their statistics}}
    & $\eta_{\mathrm{D}}$, $\eta_{\mathrm{B}}$ & BOD (exact) and OBD (orbital Bessel) denoisers \\
    & $\eta_{\mathrm{G}}$, $\eta_{\mathrm{P}}$, $\eta_{\mathrm{L}}$ & OGD, OPD and LMMSE denoisers \\
    & $\hat{\sigma}^2_{\mathrm{B}}$ & OBD posterior variance \\
    & $\bm{\Sigma}$ & $2\!\times\!2$ real conditional covariance of $x$ given $\bar{x}$ \\
    & $w_m$, $w_\ell$ & symbol and ring posterior weights \\
    & $\Lambda_\ell$ & ring log-metric \\
    & $\kappa_\ell$ & von Mises concentration \\
    & $A(\kappa)$ & Bessel ratio $I_1(\kappa)/I_0(\kappa)$ (mean resultant length) \\
    & $I_0$, $I_1$ & modified Bessel functions of the first kind \\
    & $U$ & retained harmonic order of the $U$-OBD \\
    & $\beta$ & sharpness of the Gaussian mollification \\
    \hline
    \multirow{9}{*}{\rotatebox[origin=c]{90}{\it Effective ch. \& SE}}
    & $\bar{x}$ & scalar cavity statistic (denoiser input) \\
    & $\bar{z}$, $\bar{\sigma}^2$ & effective (cavity) noise and its variance \\
    & $\tilde{z}$ & standardized noise, $\tilde{z}\sim\CN(0,1)$ \\
    & $\bar{\sigma}_\infty^2$ & effective variance at the SE fixed point \\
    & $\mathrm{MSE}_t$, $\mathrm{MSE}_\infty$ & per-symbol MSE at iteration $t$; its fixed point \\
    & $\mathcal{F}$ & state-evolution map \\
    & $\langle\eta'\rangle$ & average denoiser divergence (Onsager term) \\
    & $c_{\mathrm{D}}$ & BOD contraction modulus $\sup\mathcal{F}_{\mathrm{D}}'$ \\
    & $\delta^{\eta}$ & sub-leading fixed-point correction of level $\eta$ \\
    \hline
    \multirow{10}{*}{\rotatebox[origin=c]{90}{\it Info. meas. \& transport}}
    & $I(\cdot\,;\cdot)$, $H(\cdot)$ & mutual information and entropy \\
    & $\mathrm{mmse}_p$, $\mathrm{mmse}_q$ & MMSE under the indicated distribution \\
    & $d_\eta$ & MMSE dimension of level $\eta$ \\
    & $R_\eta$, $\Delta R_\eta$ & single-letter rate and its gap to the BOD \\
    & $I_{\mathrm{GMI}}$ & generalized mutual information \\
    & $\varsigma$ & GMI tilt parameter \\
    & $\zeta$ & effective SNR $1/\bar{\sigma}^2$ (I-MMSE variable) \\
    & $\tau$ & Gaussian smoothing level in the transport proof \\
    & $W_1$, $W_2$, $\overline{W}_2^{\,2}$ & Wasserstein distances; aggregate squared $W_2$ \\
    & $D(\cdot\|\cdot)$, $J(\cdot\|\cdot)$ & relative entropy and relative Fisher information \\
    \hline
  \end{tabular}
\end{table}

\section{Fundamentals of Orbital Detection}
\label{sec:system}


Consider an uplink \ac{MIMO} system with $K$ transmit antennas and $N$ receive antennas.\footnote{Due to the leveraging of \ac{AMP}, the developed \ac{OD} framework applies to arbitrary linear systems in under- and fully-loaded conditions and \ac{iid} Gaussian channel matrices, given a discrete prior constellation.
Examples include downlink \ac{MIMO} and multi-user uplink/downlink systems.
Extension to other channel models, e.g. with correlation or noise color, however, is left to future work.} The received baseband signal $\mathbf{y} \in \Complex^{N}$
is modeled as
\vspace{-1ex}
\begin{equation}
  \mathbf{y} = \mathbf{H}\mathbf{x} + \mathbf{z},
  \label{eq:system}
\vspace{-1ex}
\end{equation}
where $\mathbf{H} \in \Complex^{N \times K}$ with\footnote{Note that the normalization factor $N$ in the channel power guarantees convergence of the asymptotic \ac{SE} \cite{Javanmard2013}.} $H_{nk} \sim \CN(0, 1/N)$ is the channel matrix, $\mathbf{x} \triangleq [x_1, x_2, \ldots, x_K]\trans \in \mathcal{M}^K$ is the transmitted symbol vector whose entries $x_k \in \mathcal{M}$ are drawn from an arbitrary discrete constellation $\mathcal{M} \triangleq \{s_1, s_2, \ldots, s_M\} \subset \Complex$ with cardinality $M \triangleq |\mathcal{M}|$ and prior probabilities $p_m$, such that the average symbol power of the constellation can be explicitly defined as $E_d \triangleq \sum_{m=1}^M p_m \norm{s_m}^2$, and $\mathbf{z} \sim \CN(\mathbf{0}, \sigma_z^2 \mathbf{I}_N)$ is the \ac{AWGN} vector.

\vspace{-2ex}
\subsection{Fundamentals of Approximate Message Passing}
\label{sec:amp_decoupling}

Direct Bayesian estimation of $\mathbf{x}$ from~\eqref{eq:system} is intractable in general, since the joint posterior couples all $K$ symbols through the Gram matrix $\mathbf{H}\herm\mathbf{H}$, and its marginals require a sum over the $M^K$ codewords of $\mathcal{M}^K$.
Partial remedy is provided by \Ac{AMP}~\cite{Donoho2009,Bayati2011}, which circumvents this coupling by alternating a linear (matched-filter) step with a component-wise denoiser and an Onsager correction, where the denoiser $\eta_t(\cdot\,;\bar{\sigma}_t^2):\Complex \times \Real_{>0} \to\Complex$ is any separable map, applied entrywise to the cavity statistic and parametrized by the effective noise variance $\bar{\sigma}_t^2$ (the variance of the residual interference at iteration $t$, fixed by the recursion~\eqref{eq:se_intro} below). 

With $\mathbf{x}^0 = \mathbf{r}^{0} = \mathbf{0}$, the $t$-th iteration of the algorithm is
\begin{subequations}
\label{eq:amp}
\begin{align}
  \mathbf{r}^{t}     &= \mathbf{y} - \mathbf{H}\mathbf{x}^{t-1} + \tfrac{K}{N}\,\bigl\langle \eta_{t-1}'(\bar{\mathbf{x}}^{t-1})\bigr\rangle\,\mathbf{r}^{t-1}, \label{eq:amp_onsager} \\
  \bar{\mathbf{x}}^t   &= \mathbf{x}^{t-1} + \mathbf{H}\herm \mathbf{r}^t, \label{eq:amp_cavity}\\
  \mathbf{x}^{t}     &= \eta_t\bigl(\bar{\mathbf{x}}^t; \bar{\sigma}_t^2\bigr), \label{eq:amp_denoise}
\end{align}
\end{subequations}
where $\langle \eta_{t-1}'(\bar{\mathbf{x}}^{t-1})\rangle \triangleq \tfrac{1}{K}\sum_{k=1}^K \partial\eta_{t-1}/\partial\bar{x}_k$ is the average denoiser divergence.

The Onsager term $\tfrac{K}{N}\,\bigl\langle \eta_{t-1}'(\bar{\mathbf{x}}^{t-1})\bigr\rangle\,\mathbf{r}^{t-1}$ in~\eqref{eq:amp_onsager} cancels, to leading order, the self-feedback that $\mathbf{H}\herm\mathbf{H}$ injects into the cavity statistic~\eqref{eq:amp_cavity}, which is precisely what decouples the vector problem into scalar ones.
This decoupling is made precise by the \ac{AMP} \ac{SE} theorem~\cite{Bayati2011,Javanmard2013}, which we recall here as it underlies the scalar denoiser developed throughout the paper. 
For \ac{iid} Gaussian $\mathbf{H}$ and any pseudo-Lipschitz denoiser, in the large-system limit ($N,K\to\infty$, $K/N\to\alpha\in(0,1)$) each coordinate of the cavity statistic~\eqref{eq:amp_cavity} produced by~\eqref{eq:amp_onsager} converges to a scalar \ac{AWGN} observation, given by
\vspace{-1ex}
\begin{equation}
  \label{eq:effective_channel}
  \vspace{-1ex}
  \bar{x}_k^t \;\stackrel{d}{=}\; x_k + \bar{z}_k^t,
\end{equation}
where $\bar{z}_k^t$ converges in distribution to $\CN\!\bigl(0,\bar{\sigma}_t^2\bigr)$ and is asymptotically independent of $x_k$, with effective variance fixed self-consistently by the \ac{SE} recursion
\vspace{-1ex}
\begin{subequations}
\begin{equation}
  \label{eq:se_intro}
  \bar{\sigma}_t^2 = \sigma_z^2 + \tfrac{K}{N}\,\mathrm{MSE}_t,
\vspace{-1ex}
\end{equation}
\begin{equation}
  \mathrm{MSE}_t = \E_{x,\tilde{z}}\!\bigl[\,\norm{x-\eta_{t-1}(x+\bar{\sigma}_{t-1}\tilde{z};\,\bar{\sigma}_{t-1}^2)}^2\bigr],
\vspace{-1ex}
\end{equation}
\end{subequations}
%
where $\tilde{z}\sim\CN(0,1)$ and $x$ is drawn from $\mathcal{M}$. 
\pagebreak

The full recursion, its fixed points, and its validity for the (mismatched) orbital denoisers are developed in Section~\ref{sec:se}.
This decoupling is the bridge from the $K$-dimensional model in~\eqref{eq:system} to inference on a \emph{single scalar channel}: receiver design reduces to constructing one per-symbol denoiser $\eta(\bar{x};\bar{\sigma}^2)$ for the effective channel $\bar{x}=x+\bar{z}$, $\bar{z}\sim\CN(0,\bar{\sigma}^2)$, $x\in\mathcal{M}$ (where we drop the iterative index $t$ and the symbol index $k$ for brevity) -- formalized as~\eqref{eq:scalar_awgn} below -- whose accuracy (i.e., \ac{MSE}), fed back through~\eqref{eq:se_intro}, sets the fixed-point error of the full system. 
Finally, we use the terms ``denoiser'' and ``detector'' interchangeably, since the scalar channel is equivalent to \ac{SISO} detection.

\begin{remark}
  The \ac{OD} scheme is in general applicable to any message passing algorithm that decouples the vector problem into scalar ones, including \ac{GaBP}~\cite{ShentaiISIT2008}, \ac{GAMP}~\cite{Rangan2011}, \ac{OAMP}~\cite{ma2017orthogonal}, \ac{MAMP}~\cite{liu2022memory}, \ac{EP}~\cite{TakeuchiTIT2020}, and \ac{VAMP}~\cite{ranganVAMP2019}.
  The choice of \ac{AMP} in this paper is for concreteness and use of the \ac{SE} framework, and the results can be extended to other algorithms with minor modifications.
\end{remark}

\vspace{-2ex}
\subsection{Ring Description}

The geometry of $\mathcal{M}$ is carried by its amplitudes. 
Let $R_1 < R_2 < \cdots < R_L$ be the $L$ distinct values of $\norm{s_m}$ over $s_m \in \mathcal{M}$; they partition the constellation into $L$ rings as
\begin{equation}
  \label{eq:ring_partition}
  \mathcal{M}_\ell \triangleq \{\, s_m \in \mathcal{M} \mid \norm{s_m} = R_\ell \,\},
  \qquad \ell = 1, \ldots, L,
\end{equation}
which are disjoint and exhaustive. 

Ring $\ell$ holds $M_\ell \triangleq \abs{\mathcal{M}_\ell}$ symbols, with $\sum_{\ell} M_\ell = M$, and carries the \emph{ring prior}
\vspace{-1ex}
\begin{equation}
  \label{eq:ring_prior}
  r_\ell = \sum_{m \,\mid\, s_m \in \mathcal{M}_\ell} p_m,
\end{equation}
which reduces to $r_\ell = M_\ell/M$ for equiprobable signaling.
\vspace{-1ex}
\begin{figure}[H]
  \centering
  \includegraphics[width=\columnwidth]{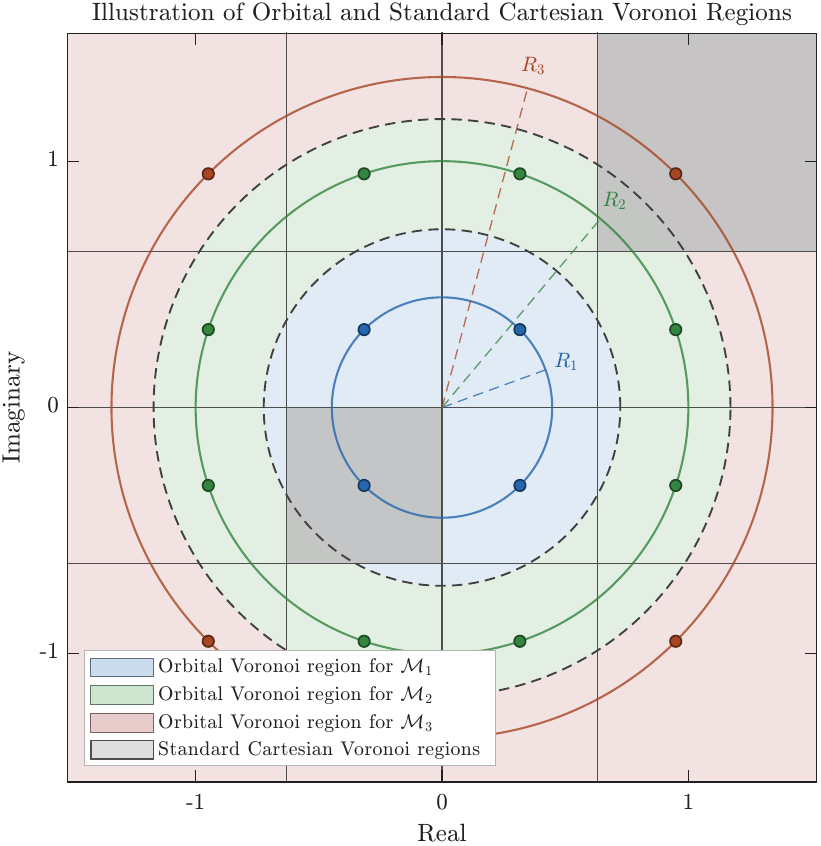}
  \vspace{-4ex}
  \caption{Orbital (ring-based) Voronoi regions of $16$-\ac{QAM} (colored) compared with representative Cartesian Voronoi regions (shaded in grey). The $M=16$ points collapse to $L=3$ amplitude rings.}
  \label{fig:constellation_rings}
  \vspace{-2ex}
\end{figure}

We refer to the pair $\{(R_\ell, r_\ell)\}_{\ell=1}^{L}$, with $\sum_{\ell=1}^L r_\ell = 1$, as the \emph{radial marginal} of the constellation: it is everything the geometry of $\mathcal{M}$ records once the phase is discarded. 
Thus, $\mathcal{M}$ supports two descriptions: the full discrete probability space $\{(s_m, p_m)\}_{m=1}^{M}$, which carries both amplitude and phase, and its radial marginal $\{(R_\ell, r_\ell)\}_{\ell=1}^{L}$, which carries amplitude alone.
The orbital prior of Section~\ref{sec:orbital} (Definition~\ref{def:orbital}) is built entirely on the latter: it retains the radial marginal exactly and relaxes everything else, which is the precise sense in which it is the least-committal, modulation-agnostic prior consistent with the ring geometry.
A visual comparison of the two is given in Fig.~\ref{fig:constellation_rings} for $16$-\ac{QAM}, where the $M=16$ points collapse to $L=3$ amplitude rings.

\vspace{-2ex}
\section{Posterior Compression and the Orbital Prior}
\label{sec:spatial_compression}

This section develops the two complementary pillars of the orbital framework, each removing a distinct bottleneck of the exact Bayesian denoiser for the effective channel~\eqref{eq:scalar_awgn}. 
First, we show that the exact discrete posterior state consumed by moment-based message passing (i.e., the posterior mean and variance), though nominally $M$-dimensional, resides \emph{losslessly} in a $3L$-dimensional subspace fixed by the ring geometry (Theorem~\ref{thm:ring_decomp}); this compresses the \emph{spatial memory} routed across the factor graph from $M$ to $3L$ real scalars, with no approximation whatsoever. 
Second, because computing this exact state still costs $\mathcal{O}(M)$ arithmetic, we introduce the \emph{orbital prior} (Definition~\ref{def:orbital}), the maximum-entropy relaxation that collapses the per-symbol \emph{computation} to $\mathcal{O}(L)$ in closed form, at a quantifiable approximation cost. The two pillars are orthogonal: exact compression governs what is \emph{stored and routed}, while the relaxation governs what is \emph{computed}; together they take the denoiser from $\mathcal{O}(M)$ memory and arithmetic to $3L$ memory and $\mathcal{O}(L)$ computation.

\vspace{-2ex}
\subsection{Exact Posterior Compression}
\label{sec:compression}

Consider the \ac{AWGN} observation of a transmitted symbol $x \in \mathcal{M}$, modeled as 
\vspace{-1ex}
\begin{equation}
  \vspace{-1ex}
  \bar{x} = x + \bar{z},
  \label{eq:scalar_awgn}
\end{equation}
where $\bar{z} \sim \CN(0, \bar{\sigma}^2)$ is a circularly symmetric complex Gaussian random variable with variance
\begin{equation}
  \bar{\sigma}^2 \triangleq \mathbb{E}\big[\norm{\bar{z}}^2\big].
\end{equation}

The channel~\eqref{eq:scalar_awgn} is precisely the effective per-symbol channel established in Section~\ref{sec:amp_decoupling} by \ac{AMP} decoupling~\eqref{eq:effective_channel}: $\bar{z}$ is the asymptotically Gaussian cavity noise, independent of the symbol $x$, and the effective variance $\bar{\sigma}^2 \ge \sigma_z^2$ -- the physical noise of~\eqref{eq:system} inflated by residual multi-user interference -- is fixed self-consistently by the \ac{SE} recursion of Section~\ref{sec:se}. 
The self-consistency is a feedback loop: a better denoiser lowers the \ac{MSE}, which by~\eqref{eq:se_intro} lowers $\bar{\sigma}^2$, which in turn sharpens the denoiser; the fixed point is where this loop closes, and it is there that denoiser quality is ultimately measured. Designing the receiver thus reduces to constructing a denoiser for~\eqref{eq:scalar_awgn}, which we now compress and then relax.

Assuming a discrete prior over $\mathcal{M}$, consistent with the uniform symbol distribution implied in Section~\ref{sec:system}, i.e.,
\begin{equation}
  \Pr(x = s_m) = 1/M, \quad s_m \in \mathcal{M},
\end{equation}
the likelihood function induced by \eqref{eq:scalar_awgn} is given by
\vspace{-1ex}
\begin{equation}
  \vspace{-1ex}
  p(\bar{x} \mid x = s_m) 
  = \frac{1}{\pi \bar{\sigma}^2} \exp\!\left(-\tfrac{\norm{\bar{x} - s_m}^2}{\bar{\sigma}^2}\right).
\end{equation}

Applying Bayes' rule, the posterior distribution becomes
\begin{equation}
  \label{eq:posterior_BOD}
  \Pr(x = s_m \mid \bar{x}) 
  = \frac{\exp\!\left(-\frac{\norm{\bar{x} - s_m}^2}{\bar{\sigma}^2}\right)}
         {\sum\limits_{m'=1}^M \exp\!\left(-\frac{\norm{\bar{x} - s_{m'}}^2}{\bar{\sigma}^2}\right)}.
\end{equation}

Defining the posterior probability weights
\vspace{-1ex}
\begin{equation}
  \label{eq:weights_standard}
  w_m \triangleq \Pr(x = s_m \mid \bar{x}), \quad \sum_{m=1}^M w_m = 1,
\vspace{-1ex}
\end{equation}
the exact Bayesian posterior mean (i.e., the \ac{MMSE} estimate) and variance are given by
\vspace{-1ex}
\begin{eqnarray}
  \label{eq:post_MMSE_mean}
  &\hat{x} = \mathbb{E}[x \mid \bar{x}] 
  = \sum\limits_{m=1}^M w_m \, s_m,& \\
  &\hat{\sigma}^2 = \Var[x \mid \bar{x}] 
  = \sum\limits_{m=1}^M w_m \, \norm{s_m - \hat{x}}^2.&
\end{eqnarray}

Equivalently, the posterior variance admits the second-moment representation
\vspace{-1ex}
\begin{equation}
  \label{eq:post_MMSE_var}
  \hat{\sigma}^2 = \sum_{m=1}^M w_m \norm{s_m}^2 - \norm{\hat{x}}^2.
\end{equation}

The computation of the Cartesian state pair $(\hat{x}, \hat{\sigma}^2)$ thus requires evaluating the full posterior weight vector $[w_1, \ldots, w_M]\trans \in [0,1]^M$. 
In dense inference graphs, propagating this full $M$-dimensional vector across edges induces the aforementioned \emph{spatial memory bottleneck}. 
Even when compressed to the Cartesian pair $(\hat{x}, \hat{\sigma}^2)$, the repeated evaluation of all $M$ components remains computationally prohibitive.

However, when the constellation possesses geometric structure, this dense representation exhibits significant redundancy due to shared radial symmetries. 
We now prove that the exact discrete posterior state inherently resides in a strictly lower-dimensional space, governed entirely by the radial geometry (the $L$ rings) of the constellation $\mathcal{M}$.

\subsubsection{\texorpdfstring{$M$}{M}-PSK: The Single Complex Sufficient Statistic}

For constant-modulus constellations, the dimensionality of the exact discrete belief state collapses intrinsically.

\begin{proposition}[Amplitude-Variance Identity for $M$-\ac{PSK}]
  \label{prop:psk_identity}
  For any constellation where $\norm{s_m}^2 = E_d, \forall m$, the exact posterior mean $\hat{x}$ and variance $\hat{\sigma}^2$ satisfy the constraint
  \begin{equation}
    \norm{\hat{x}}^2 + \hat{\sigma}^2 = E_d,
    \qquad \forall\, \bar{x} \in \Complex,\quad \bar{\sigma}^2 > 0.
    \label{eq:psk_identity}
  \end{equation}
\end{proposition}
\begin{proof}
  The Bayesian posterior variance can be expressed as $\hat{\sigma}^2 = \E[\norm{x}^2 \mid \bar{x}] - \norm{\hat{x}}^2$ from \eqref{eq:post_MMSE_var}. 
  Since all symbols satisfy $\norm{s_m}^2 = E_d$, we have $\E[\norm{x}^2 \mid \bar{x}] = \sum_m w_m E_d = E_d$ regardless of the weights, and the identity follows.
\end{proof}

\begin{proposition}[Sufficiency of the Posterior Mean for $M$-\ac{PSK}]
  \label{prop:psk_sufficiency}
  For $M$-\ac{PSK} with all symbols satisfying $\norm{s_m}^2 = E_d$, the complex posterior mean $\hat{x} \in \Complex$ serves as a strictly lossless encoding of the complete spatial state $(\hat{x}, \hat{\sigma}^2)$.
  Consequently, for any inference algorithm whose updates depend on the posterior weights $\{w_m\}_{m=1}^M$ exclusively through the first two moments, routing only the complex scalar $\hat{x}$ across the factor graph is mathematically equivalent to propagating the full discrete weight vector.
\end{proposition}

\begin{table}[H]
  \renewcommand{\arraystretch}{1.2}
  \caption{Exact ring count $L$ and spatial compression ratio $M/(3L)$ for representative constellations (normalized to $E_d=1$).}
  \vspace{-2ex}
  \label{tab:complexity}
  \centering
  \begin{tabular}{@{}lcccc@{}}
    \toprule
    Constellation & $M$ & $L$ & $3L$ & $M/(3L)$ \\
    \midrule
    $M$-PSK (any order)  & $M$   & 1   & 3   & $M/3$  \\
    16-QAM               & 16    & 3   & 9   & 1.78   \\
    64-QAM               & 64    & 9   & 27  & 2.37   \\
    256-QAM              & 256   & 32  & 96  & 2.67   \\
    1024-QAM             & 1024  & 109 & 327 & 3.13   \\
    \midrule
    16-APSK (DVB-S2)     & 16    & 2   & 6   & 2.67   \\
    32-APSK (DVB-S2)     & 32    & 3   & 9   & 3.56   \\
    64-APSK (DVB-S2x)    & 64    & 4   & 12  & 5.33   \\
    \bottomrule
  \end{tabular}\\
\flushleft
Note:  Values of $L$ for $M$-\ac{QAM} were obtained by exhaustive enumeration of $\{a^2+b^2 : a,b \in \{1,3,\ldots,2^n-1\}\}$, consistent with the Landau-Ramanujan scaling $L = \Theta(M/\sqrt{\ln M})$ derived in Proposition~\ref{prop:ring_scaling}.
\vspace{-2ex}
\end{table}

\begin{proof}
  By Proposition~\ref{prop:psk_identity}, the posterior variance is uniquely and deterministically recovered from the magnitude of the posterior mean via $\hat{\sigma}^2 = E_d - \norm{\hat{x}}^2$. Because the Cartesian state pair is entirely defined by the $2$ real degrees of freedom inherent in the complex scalar $\hat{x}$, the encoding is lossless.
\end{proof}

Propositions~\ref{prop:psk_identity} and \ref{prop:psk_sufficiency} establish that the exact discrete belief state for $M$-\ac{PSK} is fully characterized by the single complex variable $\hat{x}$. 
Compared to standard \ac{GaBP} and \ac{GAMP} implementations that route the Cartesian pair $(\hat{x}, \hat{\sigma}^2)$, thereby consuming $3$ real scalars per edge, routing only $\hat{x}$ requires exactly $2$ real scalars, corresponding to the two real degrees of freedom of $\hat{x} \in \Complex$.
This constitutes a strict one-third reduction in the spatial memory and routing payload required to evaluate the graph, achieved entirely without continuous-phase approximations.
The variance is simply reconstructed locally at any node requiring it via \eqref{eq:psk_identity}.

\subsubsection{\texorpdfstring{$M$}{M}-QAM/APSK -- Exact Multi-Ring Decomposition}

For $M$-\ac{QAM} and general \ac{APSK} constellations ($L \geq 2$), the conditional second moment can be partitioned across the distinct amplitude rings as
\vspace{-1ex}
\begin{equation}
  \E[\norm{x}^2 \mid \bar{x}] = \sum_{\ell=1}^L w_\ell R_\ell^2,
\vspace{-1ex}
\end{equation}
where $w_\ell$ denotes the posterior probability mass residing on the $\ell$-th ring.

Because symbols on different rings possess varying magnitudes and dynamic posterior probabilities, the second moment is no longer a deterministic function of the first moment. 
This dynamic fluctuation causes the single-message encoding of Proposition~\ref{prop:psk_sufficiency} to fail, as $w_\ell$ carries the requisite spatial information.

\begin{theorem}[Orbital Ring Decomposition]
  \label{thm:ring_decomp}
  The posterior belief state for any discrete multi-ring constellation can be exactly encoded by $L$ complex per-ring posterior-mean contributions $\hat{\mathbf{x}} \in \Complex^L$ and $L$ real ring probabilities $\mathbf{w} \in [0,1]^L$, defined as
  \vspace{-1ex}
  \begin{equation}
    \hat{x}_\ell \triangleq \sum_{m \in \mathcal{M}_\ell} w_m\,s_m,
    \label{eq:ring_mean_cbms}
  \vspace{-1ex}
  \end{equation}
  with
  \vspace{-1ex}
  \begin{equation}
    w_\ell \triangleq \sum_{m \in \mathcal{M}_\ell} w_m.
     \label{eq:ring_prob_MR}
  \vspace{-1ex}
  \end{equation}
  The exact Bayesian posterior mean (i.e., MMSE estimate) and variance are then
  \vspace{-1ex}
  \begin{eqnarray}
    &\hat{x} = \sum\limits_{\ell=1}^L \hat{x}_\ell,
    \label{eq:decode_x},& \\[-1ex]
    &\hat{\sigma}^2 = \sum\limits_{\ell=1}^L w_\ell R_\ell^2 - \norm{\hat{x}}^2.&
    \label{eq:decode_v}
    \vspace{-1ex}      
  \end{eqnarray}
\end{theorem}

\vspace{-2ex}
\begin{proof}
  The rings $\{\mathcal{M}_\ell\}$ form a disjoint partition of $\mathcal{M}$. 
  Therefore, we have that
  \vspace{-1ex}
  \begin{equation}
    \hat{x} = \sum_{m=1}^M w_m s_m = \sum_{\ell=1}^L \sum_{m \in
    \mathcal{M}_\ell} w_m s_m = \sum_{\ell=1}^L \hat{x}_\ell.
  \vspace{-1ex}
  \end{equation}

  For the second moment, since all $s_m \in \mathcal{M}_\ell$ satisfy the condition $\norm{s_m}^2 = R_\ell^2$, we have that
  \vspace{-1ex}
  \begin{equation}
  \hspace{-2ex}
    \E[\norm{x}^2 \mid \bar{x}] \!=\!\! \sum_{m=1}^M\!\! w_m \norm{s_m}^2 \!=\!\!
    \sum_{\ell=1}^L \!R_\ell^2\!\!\! \sum_{m \in \mathcal{M}_\ell} \!\!\!\!w_m \!\!=\!\! \sum_{\ell=1}^L w_\ell R_\ell^2,
  \end{equation}
  and \eqref{eq:decode_v} follows from $\hat{\sigma}^2 = \E[\norm{x}^2 \mid
  \bar{x}] - \norm{\hat{x}}^2$.
\end{proof}
\vspace{-1ex}

The compressed state $(\hat{\mathbf{x}}, \mathbf{w})$
requires exactly $3L$ real scalars ($2L$ from the $L$ complex per-ring means $\hat{x}_\ell$, one each from the $L$ ring masses $w_\ell$), and is \emph{the} canonical compressed state for any inference algorithm whose updates depend on the posterior weights $\{w_m\}$ only through the first two moments (as in \ac{AMP}): it reconstructs the \ac{MMSE} mean and variance exactly via~\eqref{eq:decode_x}--\eqref{eq:decode_v}, though not the full discrete posterior when $\abs{\mathcal{M}_\ell} > 1$.
Table~\ref{tab:complexity} lists
exact ring counts and compression ratios for standard constellations.

\begin{remark}[Degenerate Limit of Canonical Compressed State]
\label{rem:canonical_degenerate}
The canonical compressed state $(\hat{\mathbf{x}}, \mathbf{w})$ of Theorem~\ref{thm:ring_decomp} represents the most general $3L$-dimensional encoding of the exact posterior. 
At high \ac{SNR}, this state undergoes a complete collapse that connects the spatial compression result of Section~\ref{sec:spatial_compression} to the $\mathcal{O}(1)$ \ac{OPD} derived later in Section~\ref{sec:lpd}: the soft ring probabilities $w_\ell \to \mathbf{1}[\ell = \ell^*]$ (a single hard ring selection), so the dominant per-ring contribution collapses to the nearest constellation symbol, $\hat{x}_{\ell^*} \to s_{m^*}$, whose phase $\angle s_{m^*} \to \angle\bar{x}$ as $\bar{\sigma}^2 \to 0$; the orbital denoisers realize the same limit through their amplitude-shrinkage factor $A(\kappa_{\ell^*}) \to 1$, giving $R_{\ell^*} e^{j\angle\bar{x}}$.
In this limit the full $3L$-dimensional state degenerates to the ring index $\ell^*$ (equivalently its known radius $R_{\ell^*}$) and the phase $\angle\bar{x}$, two per-observation scalars, which is the irreducible minimum for any denoiser satisfying both a ring-constrained output and non-trivial estimation (Proposition~\ref{prop:lpd_irreducible}).
The \ac{OPD} therefore represents not merely a computational shortcut but the geometrically inevitable endpoint of the orbital decomposition: the point at which no denoiser meeting the ring-output and non-trivial-estimation conditions of Proposition~\ref{prop:lpd_irreducible} can compress the state further.
\end{remark}

While Table~\ref{tab:complexity} shows the compression ratio $M/(3L)$ increasing across constellation orders, the finite examples alone cannot reveal whether this gain saturates or grows without bound as $M \to \infty$; since the $3L$-scalar encoding is worthwhile only insofar as $L$ stays far below $M$, the asymptotic value of the ring decomposition is governed entirely by the growth rate of $L$.
The following proposition makes the scaling of $L$ for square $M$-QAM precise via the Landau--Ramanujan theorem.

\vspace{-1ex}
\begin{proposition}[Ring Count Scaling for Square $M$-\ac{QAM}]
  \label{prop:ring_scaling}
  For square $M$-\ac{QAM} with $M = 4^n$ ($n \geq 1$), the number of distinct amplitude rings satisfies
  \begin{equation}
    L \;=\; \Theta\!\left(\tfrac{M}{\sqrt{\ln M}}\right), 
    \qquad M \to \infty.
    \label{eq:ring_scaling}
  \end{equation}

\begin{figure}[H]
  \centering
  \includegraphics[width=\columnwidth]{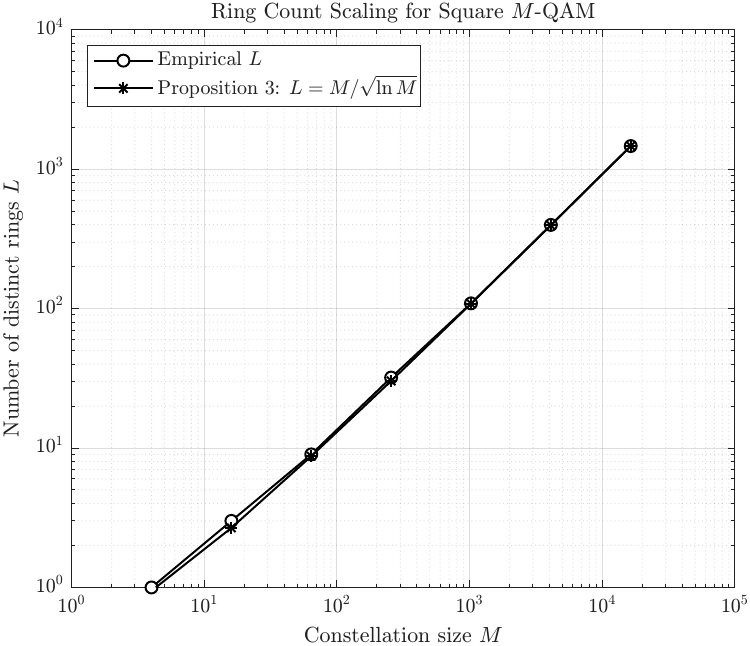}
  \vspace{-4ex}
  \caption{Validation of Proposition~\ref{prop:ring_scaling} via direct exhaustive enumeration of $\{a^2+b^2: a,b\in\{1,3,\ldots,2^n-1\}\}$ for square $M$-\ac{QAM} up to $M = 16384$.}
  \label{fig:Prop_3_L}
  \vspace{1ex}
  \includegraphics[width=\columnwidth]{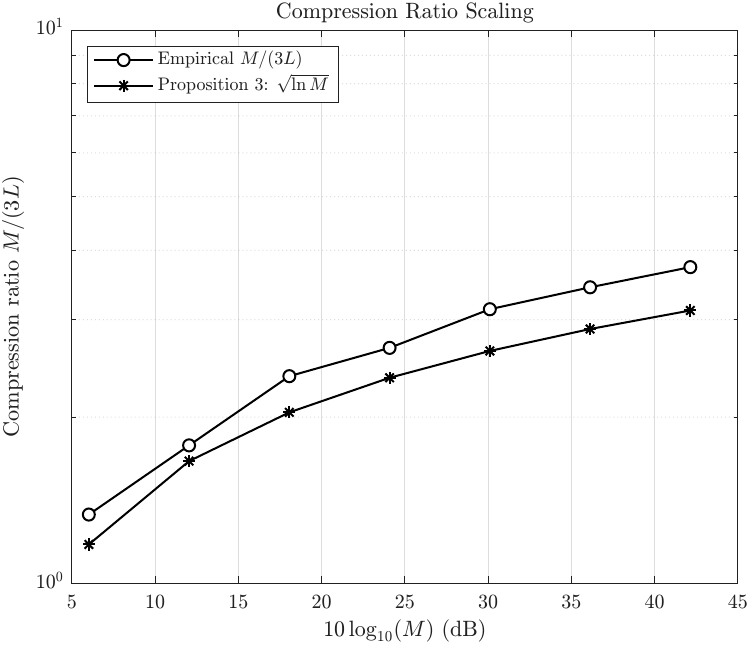}
  \vspace{-4ex}
  \caption{Validation of the asymptotic compression ratio $M/(3L) = \Theta\!\left(\sqrt{\ln M}\right)$ for square $M$-\ac{QAM} up to $M = 16384$.}
  \label{fig:Prop_3_M3L}
  \vspace{-2ex}
\end{figure}
  
  Consequently, the spatial compression ratio satisfies $M/(3L) = \Theta\bigl(\sqrt{\ln M}\bigr)$, confirming (slowly) unbounded growth with $M$ for square $M$-\ac{QAM}.
\end{proposition}

\vspace{-1.5ex}
\begin{proof}
\proofref{app:ring_scaling}
\end{proof}
\vspace{-1.5ex}

\begin{corollary}[Compression Ratio Scaling]
\label{cor:compression_scaling}
Under the conditions of Proposition~\ref{prop:ring_scaling}, we have
\begin{enumerate}[(i)]
\item \emph{Square $M$-QAM:}
  $M/(3L) = \Theta(\sqrt{\ln M})$, growing without bound but
  logarithmically slowly.
\item \emph{APSK with fixed $L$:}
  When $L = \mathcal{O}(1)$ independently of $M$ (e.g., $L\le4$ for
  DVB-S2x), the compression ratio is $\Theta(M)$, linear in constellation
  cardinality.
\end{enumerate}
In both cases the $(\hat{\mathbf{x}},\mathbf{w})$ encoding reduces the $\mathcal{O}(2^b)$ routing payload (for $b$-bit spectral efficiency $M=2^b$) -- by a $\sqrt{\ln M}$ factor for square \ac{QAM}, and to $\mathcal{O}(1)$ per symbol for fixed-$L$ \ac{APSK} -- while preserving the exact \ac{MMSE} estimate and its variance (the first two posterior moments).
\vspace{-1ex}
\end{corollary}

\pagebreak
\begin{proof}
Part~(i) restates the compression ratio recorded in Proposition~\ref{prop:ring_scaling}: dividing $M$ by $3L$ with $L = \Theta\bigl(M/\sqrt{\ln M}\bigr)$ from~\eqref{eq:ring_scaling} gives $M/(3L) = \Theta\bigl(\sqrt{\ln M}\bigr)$. For part~(ii), $L = \mathcal{O}(1)$ independently of $M$ makes the denominator $3L$ a constant, whence $M/(3L) = \Theta(M)$.
\vspace{-1ex}
\end{proof}

Figures \ref{fig:Prop_3_L} and \ref{fig:Prop_3_M3L} empirically validate the scaling of $L$ and $M/(3L)$ for square $M$-\ac{QAM} up to $M = 16384$, confirming the theoretical predictions of Proposition~\ref{prop:ring_scaling}.

By mathematically decomposing the posterior state into $L$ complex per-ring contributions $\hat{\mathbf{x}}$ and $L$ real ring probabilities $\mathbf{w}$, the exact spatial state is compressed to exactly $3L$ real scalars.
Having established that the true Bayesian posterior is fundamentally decoupled from the $\mathcal{O}(M)$ dimensional space, we can now formulate the iterative algorithms natively over this compressed $3L$-dimensional state.

\vspace{-2ex}
\subsection{The Orbital Prior: Definition and Wasserstein Mismatch}
\label{sec:orbital}

Theorem~\ref{thm:ring_decomp} establishes that the exact Bayesian posterior is fully characterized by the $3L$-dimensional state $(\hat{\mathbf{x}}, \mathbf{w})$.
However, \emph{computing} this state via~\eqref{eq:ring_mean_cbms} and~\eqref{eq:ring_prob_MR} still requires an $\mathcal{O}(M)$ iteration over all discrete constellation points, since the ring sums $w_\ell = \sum_{m \in \mathcal{M}_\ell} w_m$ themselves involve the full $M$-dimensional weight vector $[w_1, \ldots, w_M]$ of~\eqref{eq:weights_standard}.
The orbital prior introduced below resolves this by relaxing the phase on each ring to a continuous distribution, collapsing the per-ring computation to a single Bessel projection.

\vspace{-1ex}
\begin{definition}[Orbital Prior]
  \label{def:orbital}
  The orbital prior relaxes the discrete $L$-ring constellation $\mathcal{M}$ into a continuous mixture of uniform circular shells.
  The radial structure is preserved exactly while the phase on each ring is treated as uniformly distributed, given by
  \vspace{-2ex}
  \begin{equation}
    q(x) \triangleq \sum_{\ell=1}^L r_\ell\,
    \overbrace{\frac{\delta\!\left(\norm{x} - R_\ell\right)}{2\pi R_\ell}}^{\displaystyle\triangleq\; q_{\ell}(x)},
    \label{eq:cbm_prior}
  \vspace{-1ex}
  \end{equation}
  where $\delta(\cdot)$ is the one-dimensional Dirac delta and $r_\ell=M_\ell/M$ with $\sum_\ell r_\ell=1$. Densities are taken in the distributional sense w.r.t. the polar measure $\mathrm dx=r \mathrm dr\mathrm d\theta$, under which $1/(2\pi R_\ell)$ normalizes each shell, $\int_{\mathbb C}q_\ell\mathrm dx=1$. Hence $q(x)$ is a normalized mixture (not a sum) of the shells, $\int_{\mathbb{C}} q(x)\,\mathrm{d}x = \sum_{\ell} r_\ell \int_0^\infty \int_{-\pi}^{\pi} \frac{\delta(r - R_\ell)}{2\pi R_\ell}\,r\,\mathrm{d}r\,\mathrm{d}\theta = 1$.
\end{definition}

We hereafter write $p$ for the true discrete constellation prior of the transmitted symbol $x$, and $q$ for the orbital prior of Definition~\ref{def:orbital} (with density $q(x)$ given by~\eqref{eq:cbm_prior}). Paralleling~\eqref{eq:cbm_prior}, the true prior is the discrete distribution with density
\vspace{-1ex}
\begin{equation}
  \label{eq:true_prior}
  p(x) \triangleq \sum_{m=1}^{M} p_m\, \delta_2(x - s_m),
\vspace{-1ex}
\end{equation}
where $\delta_2(\cdot)$ is here the two-dimensional Dirac delta on $\Complex$, taken with respect to the polar area measure $\mathrm{d}x = r\,\mathrm{d}r\,\mathrm{d}\theta$.

By construction, $q$ preserves the radial marginal $\{(R_\ell, r_\ell)\}$ of $p$, and with it the symbol's first two moments (zero mean and average energy $\sum_{\ell} r_\ell R_\ell^2 = E_d$), while relaxing the discrete phase to a continuum. 
An estimator built on $q$ is therefore \emph{matched} to the orbital model but \emph{mismatched} to the data-generating prior $p$; this $p$-versus-$q$ mismatch, a change of prior, not of the symbol $x$, is the framework's sole approximation, whose cost is quantified in the sections that follow.
The uniform-phase choice in~\eqref{eq:cbm_prior} is not arbitrary: it is the unique \emph{maximum-entropy} completion of the radial marginal $\{(R_\ell, r_\ell)\}_{\ell=1}^{L}$, which is precisely what makes the orbital prior the least-committal continuous relaxation consistent with the constellation geometry.

\begin{proposition}[Maximum-Entropy Characterization of the Orbital Prior]
\label{prop:maxent}
Among all distributions on $\Complex$ whose radial marginal is $\{(R_\ell, r_\ell)\}_{\ell=1}^{L}$, the orbital prior~\eqref{eq:cbm_prior} uniquely maximizes the mixed discrete-continuous entropy
\vspace{-1ex}
\begin{equation}
\label{eq:mixed_entropy}
  J \;\triangleq\; H(\norm{x}) \;+\; h(\angle x \mid \norm{x}),
\vspace{-1ex}
\end{equation}
where $H(\norm{x})$ is the discrete radial entropy and $h(\angle x \mid \norm{x})$ is the conditional differential entropy of the phase on $[-\pi,\pi)$.
\end{proposition}

\vspace{-3ex}
\begin{proof}
Fixing the radial marginal fixes the discrete term $H(\norm{x}) = -\sum_\ell r_\ell\ln r_\ell$, so that maximizing $J$ reduces to maximizing the phase term $h(\angle x \mid \norm{x}) = \sum_\ell r_\ell\, h(\angle x \mid \norm{x}=R_\ell)$ ring by ring. 
On each ring the phase is supported on the bounded set $[-\pi,\pi)$ under the sole constraint of normalization, $\int_{-\pi}^\pi q_\ell(\theta)\,\mathrm d\theta = 1$. 
By the maximum-entropy principle on a bounded support~\cite{Jaynes1957}, the differential entropy $-\int q_\ell\ln q_\ell\,\mathrm d\theta$ is maximized (uniquely, by strict concavity of $u\mapsto-u\ln u$) by the uniform density $q_\ell(\theta)=1/(2\pi)$~\cite[Ch.~{12}]{CoverThomas2006}. Transcribing to the polar area element $\mathrm dx = r\,\mathrm dr\,\mathrm d\theta$, the mass $r_\ell$ on the shell $\norm{x}=R_\ell$ carries density $r_\ell/(2\pi R_\ell)$ relative to $\mathrm dx$, which is exactly~\eqref{eq:cbm_prior}.
\vspace{-1ex}
\end{proof}

The concentric-ring geometry of the orbital prior is not an artifact of our construction but a recurring feature of capacity-achieving inputs for amplitude-constrained and phase-uncertain Gaussian channels. For the quadrature Gaussian channel under simultaneous peak- and average-power constraints, the capacity-achieving input is discrete and supported on concentric circles~\cite{shamai_tit_1995}; for the noncoherent and partially coherent \ac{AWGN} channels the optimal input has uniform phase on a discrete set of radii~\cite{katz_tit_2004}; the same structure arises for discrete-time Rayleigh fading~\cite{abou_faycal_tit_2001} and in the duality bounds of Lapidoth and Moser~\cite{lapidoth_tit_2003}. The orbital prior may therefore be read as the relaxation that retains exactly the radial coordinate that these results identify as information-bearing, while discarding the coordinate they show carries a uniform distribution.
Proposition~\ref{prop:maxent} justifies the choice, by clarifying that the orbital prior imports no phase information beyond what the rings already fix, so any residual mismatch, quantified in Proposition~\ref{prop:cbm_error} below, is attributable solely to the geometry discarded by the relaxation, not to an incidental modeling choice. 
In turn, definition~\ref{def:orbital} immediately extends \ac{vMBP}~\cite{Suresh2026} beyond constant-modulus signals.
For $M$-\ac{PSK} ($L = 1$), the orbital prior preserves the amplitude exactly and relaxes only the phase, replacing the $M$ discrete angles by a continuous uniform circle; the resulting error is $\Theta(R_1/M)$ (Proposition~\ref{prop:cbm_error}), vanishing as $M \to \infty$.

Let us first motivate the metric in which this error is measured.
The Wasserstein-1 distance between two probability measures $\mu$ and $\nu$ on $\Complex$ is
\begin{equation}
\label{eq:w1_def}
W_1(\mu, \nu) \triangleq \inf_{\pi \in \Pi(\mu,\nu)} \E_{(x,y) \sim \pi}\bigl[\norm{x - y}\bigr],
\end{equation}
where $\Pi(\mu,\nu)$ is the set of all couplings (joint distributions with marginals $\mu$ and $\nu$)~\cite{villani2003}.

The intuition is that of an \emph{earth mover}: picture $\mu$ as a pile of sand and $\nu$ as a target arrangement of the same total mass; $W_1$ is the minimum total cost of rearranging one into the other, where moving a grain of mass costs its mass times the distance traveled.
This is precisely the right notion of discrepancy for our problem, for two reasons.
First, the familiar alternatives are blind here: the discrete prior $p_{\ell}$ and its continuous relaxation $q_{\ell}$ have \emph{disjoint supports} ($M_\ell$ points versus the full circle), so their \ac{KL} divergence is infinite and their total variation is maximal ($=1$) \emph{regardless} of how many points sit on the ring -- both metrics would incorrectly rate $1024$-\ac{PSK} as no better approximated by the uniform circle than $2$-\ac{PSK}.
The Wasserstein distance, on the other hand, sees the \emph{geometry}, as it charges only for how far the mass must move, such that densely populated rings are certified as nearly indistinguishable from their continuous relaxation.

Second, by Kantorovich--Rubinstein duality, $W_1(\mu,\nu) = \sup_{\mathrm{Lip}(g) \leq 1} |\E_\mu[g] - \E_\nu[g]|$, so that the distance is exactly the worst-case shift in the expectation of any Lipschitz observable, and estimation errors, posterior means, and the \ac{MSE} functionals of the \ac{SE} analysis are precisely such observables (Lipschitz on the bounded constellation support, with constants proportional to $R_L$).
This is the sense in which an optimal-transport distance is the right tool for the prior mismatch, as opposed to the \ac{KL} divergence or total variation, since both are blind to how far mass must move.
Theorem~\ref{thm:ot_bound} fixes the connection, bounding the fixed-point \ac{MSE} gap by the (Wasserstein-2) transport distance of the two priors.
For a single ring, the optimal rearrangement is clear: each symbol's mass $1/M_\ell$ must be smeared along its Voronoi arc of length $2\pi R_\ell / M_\ell$, so the typical grain travels a distance of order $R_\ell/M_\ell$, the scaling that the following proposition makes exact.

\begin{proposition}[Orbital Prior Approximation Quality for Uniform-Phase Rings]
  \label{prop:cbm_error}
  Let $p_{\ell}$ denote the discrete uniform distribution over $M_\ell$ symbols \emph{equally spaced in angle} on ring $\ell$ (e.g., $M$-\ac{PSK} or the DVB-S2 \ac{APSK} rings), and let $q_{\ell}$ be the continuous orbital prior on that ring.
  Under the Euclidean cost $c(x,y) = \norm{x-y}$, the exact Wasserstein-1 distance is
  \begin{equation}
    W_1\!\bigl(p_{\ell},\, q_{\ell}\bigr) =
    \frac{4 M_\ell R_\ell}{\pi} \!\left( 1 - \cos\frac{\pi}{2 M_\ell}
    \right) = \Theta\!\left(\frac{R_\ell}{M_\ell}\right).
    \label{eq:w1_cbm}
  \end{equation}
\end{proposition}

\begin{proof}
\proofref{app:cbm_error}
\end{proof}

\begin{remark}[Intrinsic Geodesic Cost]
\label{rem:geodesic_w1}

Since both $p_{\ell}$ and $q_{\ell}$ are supported on the same circle $\mathcal{S}^1(R_\ell)$, the discrepancy is purely angular, and one may replace the ambient Euclidean chord $2R_\ell\sin(\norm{\Delta\theta}/2)$ by the \emph{intrinsic} arc-length (geodesic) cost $c_{\mathrm{geo}} = R_\ell\norm{\Delta\theta}$.
The optimal coupling is unchanged (each Voronoi arc maps to its center), and the cost integral collapses to the elementary closed form
\begin{equation}
\label{eq:w1_geodesic}
W_1^{\mathrm{geo}}\bigl(p_{\ell},\, q_{\ell}\bigr) = \frac{\pi R_\ell}{2 M_\ell},
\end{equation}
free of any transcendental term and \emph{exactly} equal to the leading-order asymptotic of~\eqref{eq:w1_cbm}.

The $\Theta(R_\ell/M_\ell)$ scaling is thus self-evident (indeed an exact equality), and the pure angular metric $\norm{\Delta\theta}$ yields the radius-free $\pi/(2M_\ell)$, isolating the phase relaxation alone.

\begin{figure}[H]
  \centering
  \includegraphics[width=\columnwidth]{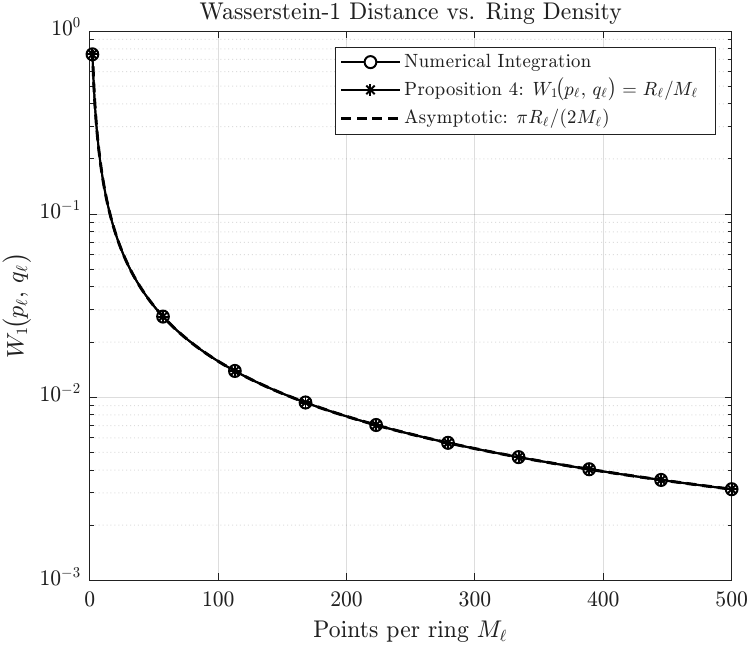}
  \vspace{-3ex}
  \caption{Validation of the exact Wasserstein-1 distance $W_1\!\bigl(p_{\ell},\, q_{\ell}\bigr)$ for uniform-phase rings as a function of the number of points $M_\ell$ on the ring.}
  \label{fig:Prop_4}
  \vspace{-2ex}
\end{figure}

As $\sin x \leq x$, the chord never exceeds the arc, so $W_1 \leq W_1^{\mathrm{geo}}$ with leading-order agreement.
We nonetheless retain the Euclidean cost throughout, since the physical estimation error $\norm{x - \hat{x}}$, and hence the optimal-transport-to-\ac{MMSE} and \ac{GMI} bounds built upon it, is Euclidean rather than geodesic.
\end{remark}

Figure~\ref{fig:Prop_4} empirically validates the exact Wasserstein-1 distance for uniform-phase rings as a function of $M_\ell$, confirming the $\Theta(R_\ell/M_\ell)$ scaling predicted by Proposition~\ref{prop:cbm_error}.

\begin{remark}[Validity for Standard Constellations]
  The exact distance~\eqref{eq:w1_cbm} characterizes the relaxation precisely: on a uniform-phase ring the geometric distortion grows linearly in the radius $R_\ell$ and inversely in the \emph{number} of points $M_\ell$, i.e.\ $W_1 = \Theta(R_\ell/M_\ell)$.
  It is therefore essentially exact for densely populated rings, notably high-order $M$-\ac{PSK}, whose single ring carries all $M_\ell = M$ points, so that $W_1 = \Theta(R/M) \to 0$.
  For rings with few points, such as the inner rings of DVB-S2 \ac{APSK} (where $M_\ell$ is as small as $4$), the per-ring distortion is bounded but \emph{not} negligible, of order $R_\ell/M_\ell$; there the orbital encoding is attractive primarily through its small ring count $L$ (Table~\ref{tab:complexity}).
\end{remark}

\begin{corollary}[Wasserstein Bound for Non-Uniform-Phase Rings]
\label{cor:w1_qam}
For a ring $\ell$ whose $M_\ell$ symbols are arranged with maximum angular half-gap $\vartheta_\ell \triangleq \tfrac{1}{2}\max_m\norm{\angle s_{m+1} - \angle s_m}$ and Voronoi-arc masses $\mu_m \triangleq v_m/(2\pi)$, where $v_m$ is the angular width of the nearest-symbol arc of $s_m$, the Wasserstein-1 distance satisfies
\begin{equation}
  W_1\!\bigl(p_{\ell},\, q_{\ell}\bigr)
  \leq 2R_\ell \sin({\vartheta_\ell/2}) {}+ R_\ell \sum_{m=1}^{M_\ell}\Bigl|\tfrac{1}{M_\ell} - \mu_m\Bigr|.
  \label{eq:w1_qam_bound}
\end{equation}

In particular, for an equidistributed ring ($v_m = 2\pi/M_\ell$ for all $m$) the correction vanishes, leaving the clean $2R_\ell\sin(\vartheta_\ell/2) = \mathcal{O}(R_\ell\vartheta_\ell)$, which is an upper bound of the same $\Theta(R_\ell/M_\ell)$ order as the exact value of Proposition~\ref{prop:cbm_error}.
\end{corollary}
\begin{proof}
Route the transport through the intermediate discrete distribution $\tilde{p}_\ell$ that places mass $\mu_m$ on each symbol $s_m$, and apply the triangle inequality $W_1(p_\ell, q_\ell) \le W_1(q_\ell, \tilde{p}_\ell) + W_1(\tilde{p}_\ell, p_\ell)$.
The nearest-symbol map collapses each Voronoi arc of $q_\ell$ onto its center symbol, delivering exactly the masses $\mu_m$, i.e., the distribution $\tilde{p}_\ell$; since every arc has angular half-width at most $\vartheta_\ell$, no mass travels farther than the chord subtending $\vartheta_\ell$, so $W_1(q_\ell, \tilde{p}_\ell) \le 2R_\ell\sin(\vartheta_\ell/2)$. (This map is a valid coupling only because its target is $\tilde{p}_\ell$; its masses $\mu_m$ equal the uniform $1/M_\ell$ of $p_\ell$ precisely when the ring is equidistributed, which is why the clean chord bound alone does not control $W_1(p_\ell,q_\ell)$ in general.)
Since $\tilde{p}_\ell$ and $p_\ell$ share the same support, transporting between them moves no mass farther than the ring diameter $2R_\ell$, whence $W_1(\tilde{p}_\ell, p_\ell) \le 2R_\ell\,\mathrm{TV}(\tilde{p}_\ell, p_\ell) = R_\ell\sum_{m}\lvert 1/M_\ell - \mu_m\rvert$. Adding the two bounds gives~\eqref{eq:w1_qam_bound}.
\end{proof}

\vspace{-2ex}
\section{The Orbital Detection Denoiser Hierarchy}
\label{sec:cbm}

In this section, we introduce the \ac{OBD}, its Jacobi-Anger refinement the \ac{nu-OBD}, a consequent \ac{OGD}, and an \ac{OPD}: a hierarchy of approximations to the exact discrete prior. The \ac{OBD}, \ac{OGD}, and \ac{OPD} progressively reduce the per-iteration arithmetic cost to $\mathcal{O}(L)$ and then $\mathcal{O}(1)$ at the price of a bounded, quantifiable approximation error, while the \ac{nu-OBD} interpolates in the reverse direction, closing that error and recovering the exact \ac{BOD} at $\mathcal{O}(M)$.

\vspace{-2ex}
\subsection{The \texorpdfstring{$\mathcal{O}(L)$}{O(L)} Orbital Bessel Denoiser}
\label{sec:nle}

We now derive the closed-form \ac{OBD} that replaces the $\mathcal{O}(M)$ \ac{BOD} of \eqref{eq:weights_standard}. 
To explicitly demonstrate the isolation of the angular component and its direct emergence from the Bayesian update, recall that the exact posterior weights $w_m$ in \eqref{eq:weights_standard} are driven by the \ac{AWGN} likelihood $p(\bar{x} \mid x = s_m)$. 
Let us express a discrete symbol on the $\ell$-th ring in polar coordinates as $s_m = R_\ell e^{j\phi_m}$, and the cavity observation as $\bar{x} = \norm{\bar{x}} e^{j\angle\bar{x}}$. 
Then, in the log-domain, the likelihood expands step-by-step into a radial bias and an angular projection as
\vspace{-1ex}
\begin{align}
  \vspace{-1ex}
  \label{eq:loglik_decomp}
  &\ln p(\bar{x} \mid x = s_m)  = -\ln(\pi \bar{\sigma}^2) - \frac{\norm{\bar{x} - s_m}^2}{\bar{\sigma}^2} \\
  &= -\ln(\pi \bar{\sigma}^2) - \frac{1}{\bar{\sigma}^2} \Big( \norm{\bar{x}}^2 + \norm{s_m}^2 - 2\operatorname{Re}(\bar{x}^* s_m) \Big) \nonumber \\
  &= -\ln(\pi \bar{\sigma}^2) - \frac{1}{\bar{\sigma}^2} \Big( \norm{\bar{x}}^2 + R_\ell^2 - 2\operatorname{Re}\bigl(\norm{\bar{x}} R_\ell e^{j(\phi_m - \angle\bar{x})}\bigr) \Big) \nonumber \\
  &= \underbrace{-\ln(\pi \bar{\sigma}^2) - \frac{\norm{\bar{x}}^2 + R_\ell^2}{\bar{\sigma}^2}}_{\text{Radial Bias}} + \underbrace{\frac{2 R_\ell \norm{\bar{x}}}{\bar{\sigma}^2} \cos(\phi_m - \angle\bar{x})}_{\text{Phase-matched Projection}}. \nonumber
\end{align}

Equation \eqref{eq:loglik_decomp} isolates the dependence on the symbol index $m$: all symbols on ring $\ell$ share the common radius $R_\ell$, so the radial bias is identical across them and the \emph{entire} $m$-dependence is carried by the phase $\phi_m$ through the phase-matched projection.

The likelihood thus depends on $s_m = R_\ell e^{j\phi_m}$ solely through its polar coordinates, so that conditioning on $s_m$ is equivalent to conditioning on its magnitude and phase, i.e.,
\vspace{-0.5ex}
\begin{align}
\label{eq:loglik_polar}
\vspace{-1ex}
\ln p(\bar{x} \mid x = s_m) &= \ln p(\bar{x} \mid \norm{x} = R_\ell,\, \angle x = \phi_m) \nonumber \\
&=: \ln p(\bar{x} \mid R_\ell, \phi_m),
\end{align}
where the final form abbreviates conditioning on the symbol's polar coordinates. This radial/angular separation underlies the ring decomposition.

Exponentiating \eqref{eq:loglik_decomp} returns the likelihood to product form
\vspace{-1ex}
\begin{align}
\label{eq:loglik_product}
\vspace{-1ex}
  &p(\bar{x} \mid R_\ell, \phi_m) \\
  &=\! \underbrace{\frac{1}{\pi \bar{\sigma}^2} \exp\!\left(\!\! -\tfrac{\norm{\bar{x}}^2 \!+\! R_\ell^2}{\bar{\sigma}^2}\! \right)}_{\text{radial factor}} \; \underbrace{\exp\!\left( \tfrac{2 R_\ell \norm{\bar{x}}}{\bar{\sigma}^2}\, \cos(\phi_m \!-\! \angle\bar{x}) \right)}_{\text{von Mises kernel}}. \nonumber
\end{align}

This proves that the directional concentration parameter emerges naturally from the physics of the \ac{AWGN} channel as
 \vspace{-1ex}
\begin{equation}
  \label{eq:kappa}
  \kappa_\ell \triangleq \frac{2 R_\ell \norm{\bar{x}}}{\bar{\sigma}^2},
\vspace{-1ex}
\end{equation}
entirely independent of any empirical heuristics, so that the kernel is $\exp\bigl(\kappa_\ell \cos(\phi_m - \angle\bar{x})\bigr)$, which is the un-normalized core of the von Mises distribution on ring $\ell$.

\begin{remark}[Circular Log-Likelihood Ratio Interpretation of $\kappa_\ell$]
\label{rem:circular_llr}
The concentration parameter $\kappa_\ell$ in \eqref{eq:kappa} admits a precise \emph{\ac{LLR}} interpretation that directly parallels the classical real-valued \ac{LLR} of \ac{AWGN} detection.
In the scalar real channel $y = s + n$, $n \sim \mathcal{N}(0,\sigma^2)$, the pairwise \ac{LLR} between two amplitude hypotheses $s_1$ and $s_2$ evaluates to
\vspace{-0.5ex}
\begin{equation}
\Lambda_{s_1,s_2} = \ln\frac{p(y \mid s = s_1)}{p(y \mid s = s_2)}
  = \frac{(s_1 - s_2)(2y - s_1 - s_2)}{2\sigma^2},
\label{eq:real_llr}
\vspace{-0.5ex}
\end{equation}
whose \emph{scale} is governed by the distance $(s_1-s_2)/(2\sigma^2)$.

On ring $\ell$, the analogous \emph{circular \ac{LLR}} between two phase hypotheses
$\phi_1$ and $\phi_2$ -- holding the radius $R_\ell$ fixed -- is
\vspace{-1ex}
\begin{equation}
\Lambda_{\phi_1,\phi_2}^{(\ell)}
  \!\triangleq\! \ln\frac{p(\bar{x} \mid R_\ell, \phi_1)}{p(\bar{x} \mid R_\ell, \phi_2)}
  \!=\! \kappa_\ell \bigl[\cos(\phi_1 - \angle\bar{x}) - \cos(\phi_2 -
  \angle\bar{x})\bigr].
\label{eq:circular_llr}
\end{equation}

Thus, $\kappa_\ell$ plays the role of a \emph{circular LLR gain}, scaling the maximum pairwise log-likelihood difference over all angular hypotheses on ring $\ell$.
This maximum is $2\kappa_\ell$, attained for antipodal hypotheses $\norm{\phi_1 - \phi_2} = \pi$ aligned with $\angle\bar{x}$, in exact analogy with the real case where the maximum pairwise \ac{LLR} scales as $\norm{s_1-s_2}\norm{y}/\sigma^2$.
Crucially, while a heuristic concentration parameter would be a fixed constant, $\kappa_\ell = 2R_\ell\norm{\bar{x}}/\bar{\sigma}^2$ is \emph{instance-adaptive} and proportional to the observation amplitude $\norm{\bar{x}}$, automatically increasing phase resolution when the received signal is strong and shrinking it under deep fading.
This instance-adaptivity is why $\kappa_\ell$ is canonical, as it is the unique scaling that renders the von Mises kernel $\exp(\kappa_\ell \cos(\phi_m - \angle\bar{x}))$ equal to the \ac{AWGN} likelihood on ring $\ell$, up to a radially-absorbed constant.
\end{remark}

Under the continuous orbital prior, the discrete phase $\phi_m$ relaxes to a continuous variable $\theta \in [-\pi, \pi)$. 
Conditioned on the cavity observation $\bar{x}$ and on the symbol residing on the $\ell$-th ring, the posterior phase distribution is the von Mises density
\begin{equation}
  p(\theta \mid \bar{x}, R_\ell) = \frac{\exp\bigl(\kappa_\ell \cos(\theta - \angle\bar{x})\bigr)}{2\pi I_0(\kappa_\ell)},
  \label{eq:vm_like}
\end{equation}
the continuous-phase analog of the discrete posterior~\eqref{eq:posterior_BOD}, with $2\pi I_0(\kappa_\ell)$ the von Mises normalizer.

By the chain rule for probabilities, the relaxed posterior factors as $p(R_\ell, \theta \mid \bar{x}) = w_\ell^{\mathrm{B}}\, p(\theta \mid \bar{x}, R_\ell)$, the product of the posterior ring probability $w_\ell^{\mathrm{B}} = p(R_\ell \mid \bar{x})$, and the von Mises phase factor~\eqref{eq:vm_like}.
In the sequel, we evaluate $w_\ell^{\mathrm{B}}$ and the per-ring conditional mean under~\eqref{eq:vm_like}, which together these yield the moments, just as the discrete posterior~\eqref{eq:posterior_BOD} supplies the weights $w_m$.
Replacing the discrete summations over the $M_\ell$ points with continuous integrals over $\theta$ yields the following closed-form $\mathcal{O}(L)$ denoiser equations.

\subsubsection{Posterior Ring Probabilities}
To compute the posterior probability $w_\ell$ that the transmitted symbol originated from the $\ell$-th ring, we must evaluate the marginal likelihood of the ring by integrating out the continuous phase $\theta$. 
Substituting the decomposed log-likelihood \eqref{eq:loglik_decomp} and the uniform phase prior $p(\theta) = 1/(2\pi)$, the marginal likelihood evaluates to
\begin{align}
  &p(\bar{x} \mid R_\ell) 
  = \int_{-\pi}^{\pi} p(\bar{x} \mid R_\ell, \theta) \, p(\theta) \, d\theta \nonumber \\
  &= \frac{1}{2\pi^2 \bar{\sigma}^2} \exp\!\left( -\tfrac{\norm{\bar{x}}^2 + R_\ell^2}{\bar{\sigma}^2} \right) \int_{-\pi}^{\pi} \exp\bigl(\kappa_\ell \cos(\theta - \angle\bar{x})\bigr) d\theta \nonumber \\
  &= \frac{1}{\pi \bar{\sigma}^2} \exp\!\left( -\tfrac{\norm{\bar{x}}^2 + R_\ell^2}{\bar{\sigma}^2} \right) I_0(\kappa_\ell),
  \label{eq:marginal_likelihood}
\end{align}
where the final step uses $\beta=\theta-\angle\bar{x}$ and the standard Bessel integral identity
$\int_{-\pi}^{\pi} e^{\kappa\cos\beta} d\beta = 2\pi I_0(\kappa)$.

By Bayes' rule, the exact posterior ring probability is proportional to the product of the prior probability $r_\ell$ and this marginal likelihood. 
Taking the natural logarithm and discarding the terms $-\ln(\pi \bar{\sigma}^2)$ and $-\norm{\bar{x}}^2 / \bar{\sigma}^2$ which are common to all rings, we define the unnormalized ring log-metric $\Lambda_\ell^{\mathrm{B}}$ as
\vspace{-2ex}
\begin{equation}
  \Lambda_\ell^{\mathrm{B}} \triangleq \ln r_\ell - \frac{R_\ell^2}{\bar{\sigma}^2} + \ln I_0(\kappa_\ell).
  \label{eq:log_metric}
\vspace{-1ex}
\end{equation}

The true posterior ring probabilities $w_\ell^{\mathrm{B}}$ are then obtained via the softmax normalization over the $L$ log-metrics as
\begin{equation}
  w_\ell^{\mathrm{B}} = \frac{\exp(\Lambda_\ell^{\mathrm{B}})}{\sum_{\ell'=1}^L \exp(\Lambda_{\ell'}^{\mathrm{B}})}.
  \label{eq:ring_prob}
\end{equation}

\subsubsection{Per-Ring Contributions}
The conditional center of mass for the $\ell$-th ring is the posterior mean of $R_\ell e^{j\theta}$ under the von Mises phase posterior~\eqref{eq:vm_like}, given by
\begin{align}
  \E\!\left[R_\ell e^{j\theta} \,\middle|\, \bar{x}, R_\ell\right]
  &= \frac{\int_{-\pi}^{\pi} R_\ell e^{j\theta} \exp\bigl(\kappa_\ell \cos(\theta - \angle\bar{x})\bigr)\,d\theta}
  {\int_{-\pi}^{\pi} \exp\bigl(\kappa_\ell\cos(\theta - \angle\bar{x})\bigr)\,d\theta} \nonumber \\
  &= R_\ell A(\kappa_\ell)\, e^{j\angle\bar{x}},
  \label{eq:ring_mean_cbm}
\end{align}
where the numerator evaluates exactly to $2\pi R_\ell I_1(\kappa_\ell) e^{j\angle\bar{x}}$ and the denominator to $2\pi I_0(\kappa_\ell)$ via the standard moment formula for the von~Mises distribution~\cite[Sec.~3.5]{mardia2009directional}, and
\begin{equation}
    A(\kappa) \triangleq \frac{I_1(\kappa)}{I_0(\kappa)},
    \label{eq:bessel_ratio}
\end{equation}
denotes the \emph{Bessel ratio} (also called the mean resultant length in directional statistics~\cite{mardia2009directional}), so that the weighted complex contribution of the $\ell$-th ring is given by
\begin{equation}
  \hat{x}_\ell^{\mathrm{B}} = w_\ell^{\mathrm{B}}\, R_\ell A(\kappa_\ell)\, e^{j\angle\bar{x}}.
  \label{eq:xi_cbm}
\end{equation}

%
%

\subsubsection{Posterior Moments}

Applying the reconstruction formulas of Theorem~\ref{thm:ring_decomp} to the \ac{OBD} compressed state $(\hat{\mathbf{x}}^{\mathrm{B}}, \mathbf{w}^{\mathrm{B}})$ yields the approximate posterior mean and variance
\begin{align}
  \hat{x}^{\mathrm{B}} &= \sum_{\ell=1}^L \hat{x}_\ell^{\mathrm{B}}, \label{eq:cbm_mean} \\
  \hat{\sigma}^2_{\mathrm{B}} &= \sum_{\ell=1}^L w_\ell^{\mathrm{B}} R_\ell^2 - \norm{\hat{x}^{\mathrm{B}}}^2.
  \label{eq:cbm_var}
\end{align}

Non-negativity $\hat{\sigma}^2_{\mathrm{B}} \geq 0$ holds because $\hat{\sigma}^2_{\mathrm{B}} = \mathbb{E}_{\mathrm{B}}[\norm{x}^2\mid\bar{x}] - \norm{\mathbb{E}_{\mathrm{B}}[x\mid\bar{x}]}^2 \geq 0$ by Jensen.

\begin{figure}[H]
  \centering
  \includegraphics[width=\columnwidth]{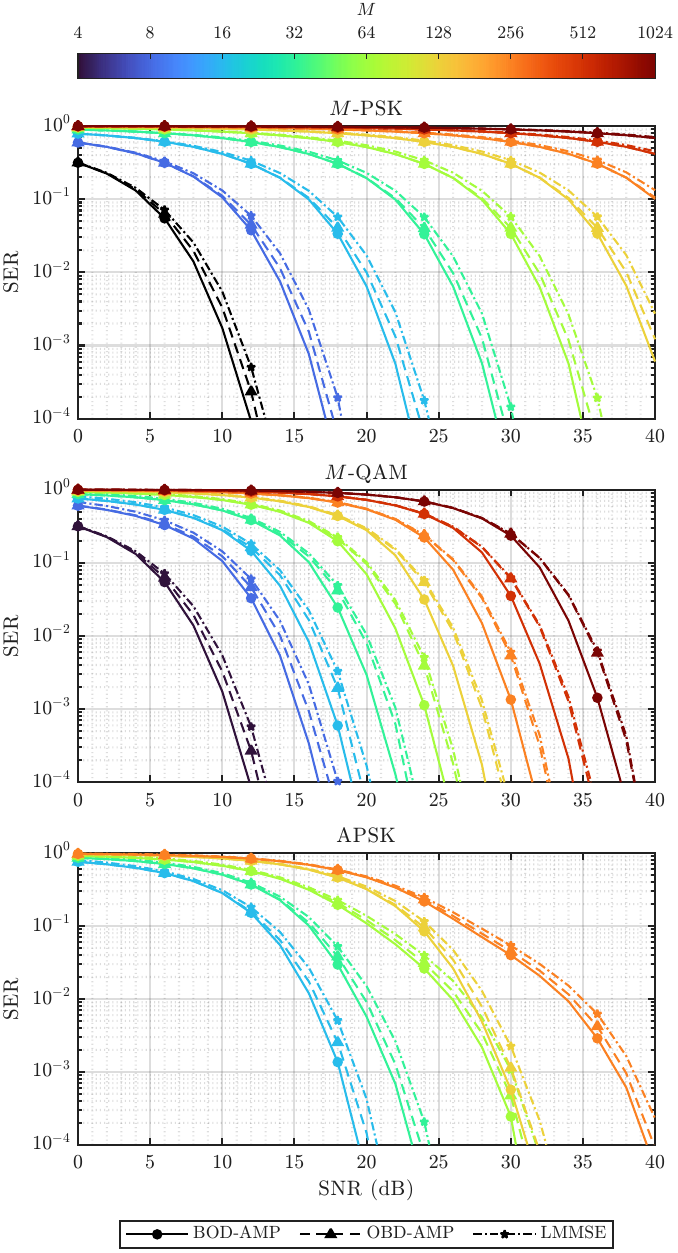}
  \vspace{-3ex}
  \caption{\ac{SER} versus \ac{SNR} of \ac{AMP} detection ($K = 64$, $N = 256$, $\alpha = 0.25$) under the Bayes-optimal \ac{BOD}, the proposed $\mathcal{O}(L)$ \ac{OBD}, and the \ac{LMMSE} baseline, for all $M$-\ac{PSK} and $M$-\ac{QAM} orders $M \in \{4, 8, \ldots, 1024\}$ and the DVB-S2/S2x \ac{APSK} constellations $M \in \{16, 32, 64, 128, 256\}$. Line style encodes the detector (shared legend); color encodes the modulation order $M$ (shared colorbar).}
  \label{fig:OBD_SER}
  \vspace{-2ex}
\end{figure}

\begin{remark}[Relationship to Orbital Ring Decomposition]
  The \ac{OBD} outputs $(\hat{\mathbf{x}}^{\mathrm{B}}, \mathbf{w}^{\mathrm{B}})$, which approximate the exact compressed state $(\hat{\mathbf{x}}, \mathbf{w})$ of Theorem~\ref{thm:ring_decomp} via the continuous phase relaxation of Definition~\ref{def:orbital}.
  The final moments $(\hat{x}^{\mathrm{B}}, \hat{\sigma}^2_{\mathrm{B}})$ computed via \eqref{eq:cbm_mean} and \eqref{eq:cbm_var} are therefore approximate Bayesian posteriors, with approximation error bounded by Proposition~\ref{prop:cbm_error}.
  More precisely, the \ac{OBD} computes the \emph{exact} posterior mean under the orbital prior $q$ (Definition~\ref{def:orbital}); applied to the true ($p$-distributed) observation it is a \emph{mismatched} estimator, and this single substitution of $q$ for the discrete prior $p$ is the framework's sole approximation, exact under $q$, and bounded under $p$ by Proposition~\ref{prop:cbm_error}.
\end{remark}

\vspace{-1ex}
\subsubsection{Complexity Analysis}
Because the Bessel ratio $A(\kappa)$ and the log-Bessel function $\ln I_0(\kappa)$ admit highly efficient per-ring $\mathcal{O}(1)$ piecewise polynomial approximations\footnote{Two-region approximations achieve maximum fractional errors below $2\%$ for $A(\kappa)$ and $10\%$ for $\ln I_0(\kappa)$ near the crossover $\kappa\approx1.2$; three-region or Chebyshev minimax fits tighten these bounds arbitrarily~\cite{DLMF}.}~\cite{Suresh2026}, computing the full posterior state $(\hat{\mathbf{x}}, \mathbf{w})$ requires exactly $L$ parallel Bessel projections. 
The arithmetic complexity of the denoiser is therefore strictly $\mathcal{O}(L)$, entirely bypassing the fundamental $\mathcal{O}(M)$ barrier.

\vspace{-1ex}
\begin{proposition}[Orbital Bessel Denoiser]
\label{prop:cbm_denoiser}
For each cavity observation $\bar{x} \in \mathbb{C}$ and noise variance $\bar{\sigma}^2 > 0$, the $\mathcal{O}(L)$ \ac{OBD} $\eta_B(\bar{x};\bar{\sigma}^2)$ approximating the Bayesian posterior mean and variance of any multi-ring constellation evaluates, for each ring $\ell = 1, \ldots, L$, the concentration, log-evidence, ring probability, and conditional contribution, and accumulates the moments:
\begin{subequations}
\label{eq:obd_denoiser}
\begin{align}
  \kappa_\ell &= \frac{2 R_\ell \norm{\bar{x}}}{\bar{\sigma}^2}, \label{eq:obd_kappa} \\
  \Lambda_\ell^{\mathrm{B}} &= \ln r_\ell - \frac{R_\ell^2}{\bar{\sigma}^2} + \ln I_0(\kappa_\ell), \label{eq:obd_metric} \\
  w_\ell^{\mathrm{B}} &= \frac{\exp(\Lambda_\ell^{\mathrm{B}})}{\sum_{\ell'=1}^{L} \exp(\Lambda_{\ell'}^{\mathrm{B}})}, \label{eq:obd_ring_prob} \\
  \hat{x}^{\mathrm{B}} &= \sum_{\ell=1}^L w_\ell^{\mathrm{B}}\, R_\ell\, A(\kappa_\ell)\, e^{j\angle\bar{x}}, \label{eq:obd_mean} \\
  \hat{\sigma}^2_{\mathrm{B}} &= \sum_{\ell=1}^{L} w_\ell^{\mathrm{B}}\, R_\ell^2 - \norm{\hat{x}^{\mathrm{B}}}^2, \label{eq:obd_var}
\end{align}
\end{subequations}
where $A(\kappa) = I_1(\kappa)/I_0(\kappa)$ is the Bessel ratio of \eqref{eq:bessel_ratio}. The denoiser map $\eta_B(\bar{x};\bar{\sigma}^2)$ returns the posterior mean~\eqref{eq:obd_mean}, with $\hat{\sigma}^2_{\mathrm{B}}$ of \eqref{eq:obd_var} its companion posterior variance. Equations~\eqref{eq:obd_kappa}--\eqref{eq:obd_var} collect the derivation \eqref{eq:kappa}--\eqref{eq:cbm_var}; their per-symbol arithmetic cost is $\mathcal{O}(L)$, reducing the $\mathcal{O}(M)$ discrete sum to $L \ll M$ Bessel projections.
\end{proposition}
\begin{proof}
Equations~\eqref{eq:obd_kappa}-\eqref{eq:obd_var} are the collected form of the ring-wise derivation~\eqref{eq:kappa}-\eqref{eq:cbm_var}, which evaluates the orbital posterior mean and its companion variance term by term; the statement adds no step beyond that derivation, and the $\mathcal{O}(L)$ cost is read off from the single pass over $\ell = 1,\ldots,L$ it requires.
\end{proof}

Figure~\ref{fig:OBD_SER} shows what the relaxation delivers in operation: the \ac{SER} of \ac{AMP} detection at load $\alpha = 0.25$ under the exact $\mathcal{O}(M)$ \ac{BOD}, the proposed $\mathcal{O}(L)$ \ac{OBD}, and the \ac{LMMSE} baseline, swept over the complete modulation-order range of each family: $M = 4$ through $1024$ for $M$-\ac{PSK} and $M$-\ac{QAM}, and the true DVB-S2/S2x constellations with $M = 16$ through $256$ for \ac{APSK}\footnotemark.

For $M$-\ac{PSK} the \ac{OBD} curve is indistinguishable from the Bayes-optimal one {at every order, and increasingly so as $M$ grows, the single ring densifying at rate $\Theta(R_1/M)$ (Proposition~\ref{prop:cbm_error})}, while for $M$-\ac{QAM} and \ac{APSK} it remains within a fraction of a dB of the \ac{BOD} {across the entire sweep}; the \ac{LMMSE} receiver, by contrast, cedes several dB {at every order}.
This is the operational face of the fixed-point ordering $\mathrm{MSE}_\infty^{\mathrm{D}} \leq \mathrm{MSE}_\infty^{\mathrm{B}} \leq \mathrm{MSE}_\infty^{\mathrm{L}}$ proved in Proposition~\ref{prop:fp_ordering}, holding uniformly over the modulation order.

\footnotetext{Throughout the paper, all numerical experiments use constellations normalized to unit average energy ($E_d = 1$), physical noise $\sigma_z^2 = 10^{-\mathrm{SNR}/10}$, and load $\alpha = K/N = 0.25$ unless stated otherwise. Scalar-channel and \ac{SE} expectations are evaluated by Monte Carlo, with sample sizes chosen so that the statistical fluctuation of every displayed point is negligible at the plotted scale; the \ac{SE} recursion is initialized at $\mathrm{MSE}_0 = E_d$ and iterated to its fixed point (geometric convergence, Remark~\ref{rem:geometric_rate}).}

\subsection{The \texorpdfstring{$U$}{U}-OBD: A Phase-Dependent Hierarchy}
\label{sec:ja_cbm}

The \ac{OBD} log-evidence~\eqref{eq:log_metric} uses only $\ln I_0(\kappa_\ell)$, the logarithm of the leading ($u = 0$) term of the exact discrete partition function.
The following proposition shows that the full partition function is a convergent series in modified Bessel functions, and that retaining $U \geq 1$ terms defines a systematic hierarchy that closes the Wasserstein gap of
Proposition~\ref{prop:cbm_error} geometrically in $U$.

\begin{proposition}[Jacobi-Anger Discrete Partition Function]
\label{prop:ja_partition}
For a ring $\ell$ with $M_\ell$ equally spaced symbols at phases $\phi_m = 2\pi m / M_\ell$, the exact discrete log-partition function satisfies
\begin{align}
\label{eq:ja_partition}
&\ln \sum_{m=0}^{M_\ell-1}
    e^{\kappa_\ell \cos(\phi_m - \angle\bar{x})}\\[-1ex]
&= \ln M_\ell
   + \ln\!\bigg[
       I_0(\kappa_\ell)
       + 2\sum_{u=1}^{\infty}
         I_{uM_\ell}(\kappa_\ell)\cos(uM_\ell\,\angle\bar{x})
     \bigg],\nonumber
\end{align}
which follows from the modified Bessel generating function~\cite{Watson1944} and the orthogonality identity
$\sum_{m=0}^{M_\ell-1} e^{in \cdot 2\pi m/M_\ell} = M_\ell\,\mathbf{1}[M_\ell \mid n]$.
The series converges absolutely for all $\kappa_\ell \geq 0$ and $M_\ell \geq 2$, since $I_{uM_\ell}(\kappa)/I_0(\kappa) \leq (\kappa/2)^{uM_\ell}/(uM_\ell)! \to 0$ superexponentially in $u$ for fixed $\kappa$ and $M_\ell$.
The \ac{OBD} log-evidence~\eqref{eq:log_metric} corresponds to truncating \eqref{eq:ja_partition} at $u = 0$, retaining only $\ln I_0(\kappa_\ell)$.
\end{proposition}

\begin{remark}[Phase-Offset and Non-Uniform Rings]
\label{rem:bcbm_phase}
Proposition~\ref{prop:ja_partition} assumes a ring of $M_\ell$ \emph{equally spaced} symbols starting at angle zero, $\phi_m = 2\pi m/M_\ell$. For a ring with arbitrary symbol phases $\{\phi_m^{(\ell)}\}$, expanding $e^{\kappa_\ell \cos(\phi_m^{(\ell)} - \angle\bar{x})}$ by Jacobi-Anger and summing over $m$ -- using $\sum_m \cos\bigl(n(\phi_m^{(\ell)} - \angle\bar{x})\bigr) = \operatorname{Re}\{e^{-in\angle\bar{x}} D_n^{(\ell)}\}$ -- gives the exact identity
\begin{align}
\label{eq:ja_general}
&\ln Z_\ell = \ln M_\ell \\
&+ \ln\!\left[
  I_0(\kappa_\ell) + 2\sum_{n=1}^{\infty}
  I_{n}(\kappa_\ell)\,\frac{\lvert D_{n}^{(\ell)}\rvert}{M_\ell}\,
  \cos\bigl(n\,\angle\bar{x} - \arg D_{n}^{(\ell)}\bigr)
\right], \nonumber
\end{align}
where $Z_\ell \triangleq \sum_m e^{\kappa_\ell \cos(\phi_m^{(\ell)} - \angle\bar{x})}$ and $D_{n}^{(\ell)} \triangleq \sum_m e^{i n \phi_m^{(\ell)}}$ is the (unnormalised) $n$-th Fourier coefficient of the ring's symbol phases, with $D_0^{(\ell)} = M_\ell$.

\emph{Equally-spaced rings.} If the $\ell$-th ring is equally spaced with angular offset $\psi_{0,\ell}$, that is, a rotated regular $M_\ell$-gon, as for $M$-\ac{PSK}/\ac{APSK} rings and the inner and outer rings of standard $M$-\ac{QAM} ($\psi_{0,\ell}=\pi/4$), then $D_{n}^{(\ell)} = M_\ell\, e^{i n \psi_{0,\ell}}\,\mathbf{1}[M_\ell \mid n]$. Only the harmonics $n = u M_\ell$ survive, each with $\lvert D_{uM_\ell}^{(\ell)}\rvert = M_\ell$ and $\arg D_{uM_\ell}^{(\ell)} = u M_\ell \psi_{0,\ell}$, and \eqref{eq:ja_general} collapses to
\begin{align}
\label{eq:ja_phase_corrected}
&\ln Z_\ell = \ln M_\ell \\
&+ \ln\!\left[
  I_0(\kappa_\ell) + 2\sum_{u=1}^{\infty}
  I_{uM_\ell}(\kappa_\ell)\cos\bigl(uM_\ell\,\angle\bar{x} - uM_\ell\,\psi_{0,\ell}\bigr)
\right]. \nonumber
\end{align}

For $\psi_{0,\ell}=0$ this is exactly~\eqref{eq:ja_partition}. For the inner and outer $M$-\ac{QAM} rings ($M_\ell=4$, $\psi_{0,\ell}=\pi/4$), $D_4^{(\ell)}=-4$, so $\arg D_4^{(\ell)}=\pi$ and the $u=1$ correction carries $\cos(4\,\angle\bar{x}-\pi)=-\cos(4\,\angle\bar{x})$, which is the opposite sign from the naive application of~\eqref{eq:ja_partition}.

\emph{Non-uniform rings.} When a ring is not equally spaced, Fourier coefficients $D_n^{(\ell)}$ at orders $n$ that are not multiples of $M_\ell$ are generally nonzero, and the surviving amplitudes need not equal $M_\ell$; the multiples-of-$M_\ell$ truncation~\eqref{eq:ja_phase_corrected} then fails to represent the partition function, and the general form~\eqref{eq:ja_general} must be used. This already arises within standard $16$-\ac{QAM}: the radius $\sqrt{10}$ ring carries $M_\ell=8$ symbols whose four-fold symmetry yields a nonzero sub-harmonic $D_4^{(\ell)} {\approx} 2.24$ at order $n=4$, not a multiple of $M_\ell=8$, which~\eqref{eq:ja_phase_corrected} omits entirely. In practice the coefficients $\{D_n^{(\ell)}\}$ are precomputed once from the actual symbol phases.
\end{remark}

\begin{definition}[$U$-OBD of Order $U$]
\label{def:bcbm}
For any integer $U \geq 0$, the \ac{nu-OBD} of order $U$ retains the first $U$ correction terms in Proposition~\ref{prop:ja_partition}, replacing the \ac{OBD} log-evidence~\eqref{eq:log_metric} with
\begin{align}
\label{eq:bcbm_log_metric}
\Lambda_\ell^{U} &\triangleq
\ln r_\ell - \frac{R_\ell^2}{\bar{\sigma}^2} \\
&+ \ln\!\left[
    I_0(\kappa_\ell)
    + 2\!\sum_{u=1}^{U}
      I_{uM_\ell}(\kappa_\ell)\cos\bigl(uM_\ell(\angle\bar{x} - \psi_{0,\ell})\bigr)
  \right]\!\!, \nonumber 
\end{align}
and per-ring contribution
\begin{align}
\label{eq:bcbm_ring_mean}
&\hat{x}_\ell^{U} \!=\! w_\ell^{U} R_\ell e^{j\angle\bar{x}} \\
&\cdot
\frac{I_1(\kappa_\ell)
      + \sum_{u=1}^{U}
        \bigl[I_{uM_\ell - 1}(\kappa_\ell)\,e^{-j\tilde{u}}
              + I_{uM_\ell + 1}(\kappa_\ell)\,e^{+j\tilde{u}}\bigr]}
     {I_0(\kappa_\ell)
      + 2\sum_{u=1}^{U}
        I_{uM_\ell}(\kappa_\ell)\cos(\tilde{u})}, \nonumber
\end{align}
with ring probabilities $w_\ell^{U} = \exp(\Lambda_\ell^{U})/\sum_{\ell'=1}^{L}\exp(\Lambda_{\ell'}^{U})$ formed from the order-$U$ log-evidences by the same softmax as the \ac{OBD}~\eqref{eq:ring_prob} and $\tilde{u} \triangleq uM_\ell(\angle\bar{x} - \psi_{0,\ell})$, where $\psi_{0,\ell}$ is the angular offset of the (equally spaced) $\ell$-th ring per Remark~\ref{rem:bcbm_phase}; for $\psi_{0,\ell} = 0$ this recovers the zero-offset form of Proposition~\ref{prop:ja_partition}, while for the offset rings of standard $M$-\ac{QAM} ($\psi_{0,\ell} = \pi/4$) the correction terms carry the sign flip derived in Remark~\ref{rem:bcbm_phase}. Non-equally spaced rings require the general Fourier form~\eqref{eq:ja_general}.

At $U = 0$, Definition~\ref{def:bcbm} reduces identically to the \ac{OBD} of Proposition~\ref{prop:cbm_denoiser}.
As $U \to \infty$, the \ac{nu-OBD} converges to the \ac{BOD}: the first omitted harmonic $I_{(U+1)M_\ell}(\kappa_\ell)$ decays factorially in its index $(U+1)M_\ell$ (Proposition~\ref{prop:ja_partition}), so the truncated partition function approaches the exact discrete one~\eqref{eq:ja_partition} rapidly and a modest order $U$ suffices in practice.
\end{definition}

\begin{proposition}[Geometric Convergence of the $U$-OBD to the \ac{BOD}]
\label{prop:bcbm_wasserstein}
For bounded concentration $\kappa_\ell$, the order-$U$ \ac{nu-OBD} retains every Jacobi-Anger harmonic up to index $U M_\ell$ and omits only those of index $\geq (U{+}1)M_\ell$; the relative weight of the first omitted harmonic in the partition function decays geometrically, indeed factorially, in the retained order,
\begin{equation}
\label{eq:bcbm_w1}
\frac{2\,I_{(U+1)M_\ell}(\kappa_\ell)}{I_0(\kappa_\ell)}
= \mathcal{O}\!\left(\tfrac{1}{M_\ell^{U+1}}\right).
\end{equation}

Consequently the \ac{nu-OBD} log-evidence~\eqref{eq:bcbm_log_metric}, ring probabilities, and per-ring mean~\eqref{eq:bcbm_ring_mean} converge to those of the exact \ac{BOD} at this rate. The induced Wasserstein-1 distance between the order-$U$ prior and the discrete truth on ring $\ell$ correspondingly scales as $\mathcal{O}(R_\ell/M_\ell^{U+1})$, the radius $R_\ell$ setting the transport length; a rigorous transport-metric characterization is deferred to a followup work.
\end{proposition}

\begin{proof}
\proofref{app:prop_bcbm_wasserstein}
\end{proof}

\begin{remark}[Relevance Regime of the $U$-OBD]
\label{rem:bcbm_regime}
The \ac{nu-OBD} corrections matter only in a moderate-\ac{SNR} window $\kappa_0 \leq \kappa_\ell < \kappa_1$ (with $\kappa_0 \approx 3$ and $\kappa_1 \approx 20$): below $\kappa_0$ the phase posterior is broad and the \ac{OBD} ($U=0$) is already accurate, while a low order $U$ closes the residual geometric gap inside the window.
As $\kappa_\ell \to \infty$, the correction terms $I_{uM_\ell}(\kappa_\ell)/I_0(\kappa_\ell) \to 1$ simultaneously, so the full Fourier series would be required, but this regime is never entered in practice: at $\kappa_\ell \geq \kappa_1$ the three-regime adaptive denoiser~\eqref{eq:three_regime} switches to the \ac{OPD}, which bypasses the soft posterior entirely via hard ring-phase projection at $\mathcal{O}(1)$ complexity.
The \ac{OGD} and \ac{OPD} thus do not arise from the Jacobi-Anger series converging with few terms; they arise from the posterior concentrating so sharply that the precise form of the partition function is irrelevant to the detection decision, which is a qualitatively distinct high-\ac{SNR} phenomenon characterized by Propositions~\ref{prop:ring_discrimination} and~\ref{prop:lpd_limit}.
\end{remark}

\begin{remark}[Phase-Sensitive Ring Detection and \ac{SE} Implications]
\label{rem:bcbm_se}
Unlike the \ac{OBD} and \ac{OGD} discussed in the next subsections, the \ac{nu-OBD} log-evidence~\eqref{eq:bcbm_log_metric} depends on $\angle\bar{x}$ for $U \geq 1$, breaking the phase-preserving property of Lemma~\ref{lem:phase_preserving} and therefore the \ac{SE} contraction proof of Corollary~\ref{cor:se_uniqueness}.
At the leading order $\mathrm{MSE}_\infty = \sigma_z^2/(2-\alpha)$ the same fixed point is expected for all $U$, since the phase-dependent corrections $2I_{uM_\ell}(\kappa_\ell)\cos(uM_\ell(\angle\bar{x}-\psi_{0,\ell}))/I_0(\kappa_\ell)$ modify only the sub-leading term of the log-evidence and average out over $\angle\bar{x}$, leaving the leading-order \ac{MSE} unchanged;
{the sub-leading correction is strictly smaller than $\delta^{\mathrm{B}} = \mathcal{O}(\sigma_z^4)$ of~\eqref{eq:delta_cbm} and vanishes as $U \to \infty$ (its precise order deferred to the companion paper), interpolating toward $\delta^{\mathrm{D}} = 0$}.
The full \ac{SE} analysis of \ac{nu-OBD}($U$) for $U \geq 1$, including the modified fixed-point contraction argument for phase-dependent denoisers, is reserved for a companion paper.
\end{remark}

\vspace{-2ex}
\subsection{The \texorpdfstring{$\mathcal{O}(L)$}{O(L)} Orbital Gaussian Denoiser (Bessel-Free)}
\label{sec:cgr}

The \ac{OBD} of Section~\ref{sec:nle} employs the von Mises distribution as the continuous phase model on each ring. 
While this is exact under the orbital prior, the von Mises phase posterior~\eqref{eq:vm_like} and its moments still require evaluation of the modified Bessel functions $I_0(\kappa_\ell)$ and $I_1(\kappa_\ell)$ via piecewise polynomial approximations.

\subsubsection{The Gaussian Phase Approximation}
Next, we introduce a further relaxation that eliminates all special-function evaluations by exploiting the Gaussian convergence of the von Mises distribution at moderate-to-high concentration.

\begin{definition}[Gaussian Phase Posterior]
\label{def:gaussian_phase}
The \ac{OGD} approximates the von Mises phase posterior~\eqref{eq:vm_like} on each ring by its Gaussian small-angle form about the observed phase $\angle\bar{x}$. Substituting $\cos(\Delta\theta_\ell) \approx 1 - \Delta\theta_\ell^2/2$ into the von Mises exponent, with angular deviation $\Delta\theta_\ell \triangleq \theta - \angle\bar{x}$, yields the \emph{a-posteriori} phase distribution
\begin{equation}
\label{eq:gaussian_phase}
p_{\mathrm{G}}(\theta \mid \bar{x},\, R_\ell) \propto \exp\!\left(-\frac{\kappa_\ell\, \Delta\theta_\ell^2}{2}\right), \quad \kappa_\ell = \frac{2R_\ell \norm{\bar{x}}}{\bar{\sigma}^2},
\end{equation}
which is a Gaussian in $\Delta\theta_\ell$ centered at $\angle\bar{x}$ with precision (inverse variance) $\kappa_\ell$.
\end{definition}

\begin{proposition}[Bessel Function Asymptotics for the \ac{OGD} Approximation]
\label{prop:vm_gaussian}
Let $A(\kappa) \triangleq I_1(\kappa)/I_0(\kappa)$ denote the mean resultant length of the
$\mathrm{vM}(0,\kappa)$ distribution.
The following two approximations hold for all $\kappa \geq 1$:
\begin{enumerate}[(i)]
\item \emph{Mean amplitude:}
\begin{align}
\label{eq:A_approx}
A(\kappa) &= 1 - \frac{1}{2\kappa} - \frac{1}{8\kappa^2} + \mathcal{O}(\kappa^{-3}), \\
\quad\text{with}\quad
&\norm{A(\kappa) - \left(1 - \frac{1}{2\kappa}\right)} \leq \frac{1}{4\kappa^2}.
\end{align}

\item \emph{Log-normalizer:}
\begin{align}
\label{eq:logI0_approx}
\ln I_0(\kappa) &= \kappa - \tfrac{1}{2}\ln(2\pi\kappa) + \frac{1}{8\kappa} + \mathcal{O}(\kappa^{-2}), \\
\quad\text{with}\quad
&\norm{\ln I_0(\kappa) - \kappa + \tfrac{1}{2}\ln(2\pi\kappa)} \leq \frac{1}{4\kappa}.
\end{align}
\end{enumerate}

In particular, the relative error in the mean amplitude satisfies $\norm{A(\kappa) - (1 - 1/(2\kappa))} / A(\kappa) \leq 2\%$ for all $\kappa \geq 4$.
\end{proposition}

\begin{proof}
\proofref{app:prop_vm_gaussian}
\end{proof}

\subsubsection{The \texorpdfstring{$\mathcal{O}(L)$}{O(L)} Gaussian Phase Denoiser}

Under the Gaussian phase approximation, the \ac{OBD} equations~\eqref{eq:kappa}-\eqref{eq:cbm_var} simplify to closed forms that require \emph{no} Bessel function evaluations.

\begin{proposition}[Orbital Gaussian Denoiser]
\label{prop:cgr_denoiser}
Under the Gaussian phase approximation, the \ac{OBD} equations reduce to the following Bessel-free closed forms, defining the \ac{OGD} map $\eta_G(\bar{x};\bar{\sigma}^2)$:

\noindent\emph{(i) Ring conditional mean:} the conditional center of mass on ring $\ell$ is, {to first order in $1/\kappa_\ell$,} given by
\begin{equation}
\label{eq:cgr_mean}
R_\ell\!\left(1 - \frac{1}{2\kappa_\ell}\right) e^{j\angle\bar{x}} = R_\ell\!\left(1 - \frac{\bar{\sigma}^2}{4R_\ell\norm{\bar{x}}}\right) e^{j\angle\bar{x}}.
\end{equation}

\noindent\emph{(ii) Ring log-evidence:}
\begin{align}
\Lambda_\ell^{\mathrm{G}} &= \ln r_\ell - \frac{R_\ell^2}{\bar{\sigma}^2} + \kappa_\ell - \frac{1}{2}\ln(2\pi\kappa_\ell) \label{eq:cgr_evidence_1} \\
&= \ln r_\ell - \frac{R_\ell^2}{\bar{\sigma}^2} + \frac{2R_\ell\norm{\bar{x}}}{\bar{\sigma}^2} - \frac{1}{2}\ln\!\left(\frac{4\pi R_\ell \norm{\bar{x}}}{\bar{\sigma}^2}\right). \label{eq:cgr_evidence_2}
\end{align}

\noindent\emph{(iii) Posterior ring probabilities:}
\begin{equation}
\label{eq:cgr_ring_prob}
w_\ell^{\mathrm{G}} = \frac{\exp(\Lambda_\ell^{\mathrm{G}})}{\sum_{\ell'=1}^{L} \exp(\Lambda_{\ell'}^{\mathrm{G}})}.
\end{equation}

\noindent\emph{(iv) Posterior mean and variance:}
\begin{align}
\label{eq:cgr_posterior}
\hat{x}_\ell^{\mathrm{G}} &= w_\ell^{\mathrm{G}}\, R_\ell\!\left(1 - \tfrac{1}{2\kappa_\ell}\right) e^{j\angle\bar{x}}, \\
\hat{x}^{\mathrm{G}} &= \sum_{\ell=1}^{L} \hat{x}_\ell^{\mathrm{G}}, \\
\hat{\sigma}^2_{\mathrm{G}} &= \sum_{\ell=1}^{L} w_\ell^{\mathrm{G}}\, R_\ell^2 - \norm{\hat{x}^{\mathrm{G}}}^2.
\end{align}
The \ac{OGD} map $\eta_G(\bar{x};\bar{\sigma}^2)$ returns $\hat{x}^{\mathrm{G}}$, the mean of the Gaussian-phase posterior approximation~\eqref{eq:gaussian_phase}, with $\hat{\sigma}^2_{\mathrm{G}}$ its companion variance.
\end{proposition}

\begin{proof}
\proofref{app:prop_cgr_denoiser}
\end{proof}

\begin{remark}[Elimination of Special Functions]
\label{rem:no_bessel}
The \ac{OGD} requires only multiplication, division, subtraction, natural logarithm, and exponential operations, all of which are standard \ac{ALU} operations.
No Bessel function lookup tables or piecewise polynomial approximations are needed. 
This makes the \ac{OGD} directly synthesizable in fixed-point \ac{ASIC} logic without function approximation circuits.
\end{remark}

\subsubsection{Adaptive OBD/OGD Switching}

The \ac{OGD} approximation is accurate for $\kappa_\ell \geq \kappa_0 \approx 3$ {(this threshold trades a few-percent amplitude bias for earlier Bessel elimination; a larger $\kappa_0$ tightens the $\mathcal{O}(\kappa^{-2})$ approximation bound of Proposition~\ref{prop:vm_gaussian})}, but breaks down for small~$\kappa_\ell$ (low \ac{SNR} or small rings): the shrinkage factor $1 - 1/(2\kappa_\ell)$ turns negative for $\kappa_\ell < 1/2$, whereas the exact mean resultant length satisfies $A(\kappa_\ell) \in [0,1)$ with $A(\kappa_\ell) \sim \kappa_\ell/2 \to 0^+$, and the log-evidence $\kappa_\ell - \tfrac{1}{2}\ln(2\pi\kappa_\ell)$ diverges as $\kappa_\ell \to 0$ while the exact $\ln I_0(\kappa_\ell) \to 0$. A pure \ac{OGD} is therefore unusable at low concentration, and we propose the adaptive strategy of Definition~\ref{def:adaptive_cgr} that selects per-ring between the \ac{OBD} and \ac{OGD} equations based on the value of $\kappa_\ell$, retaining the exact \ac{OBD} below $\kappa_0$.

\begin{definition}[Adaptive OGD-OBD Denoiser]
\label{def:adaptive_cgr}
The adaptive denoiser selects per-ring the contribution
\begin{equation}
\label{eq:adaptive_switch}
\hat{x}_\ell = w_\ell \begin{cases}
R_\ell\, A(\kappa_\ell)\, e^{j\angle\bar{x}} & \text{if } \kappa_\ell < \kappa_0 \quad \text{(OBD)}, \\[4pt]
R_\ell\bigl(1 - 1/(2\kappa_\ell)\bigr)\, e^{j\angle\bar{x}} & \text{if } \kappa_\ell \geq \kappa_0 \quad \text{(OGD)},
\end{cases}
\end{equation}
with $w_\ell$ the ring probability of the active regime.
At the \ac{SE} fixed point (see Section~\ref{sec:se}), as iterations progress and $\bar{\sigma}^2$ decreases, all rings eventually transition to the \ac{OGD} regime, yielding a fully Bessel-free denoiser at convergence.
\end{definition}

\subsection{The \texorpdfstring{$\mathcal{O}(1)$}{O(1)} Orbital Phase Denoiser and its Optimality}
\label{sec:lpd}

The \ac{OGD} of Section~\ref{sec:cgr} eliminates Bessel function evaluations by exploiting the Gaussian convergence of the von Mises distribution for $\kappa_\ell \geq 3$. 
We now take the final limiting step: in the extreme high-\ac{SNR} regime where $\kappa_\ell \gg 1$ for \emph{all} rings, the \ac{OGD} equations admit a further collapse that eliminates both the amplitude-shrinkage factor $1 - 1/(2\kappa_\ell)$ \emph{and} the softmax ring selection, yielding a denoiser of $\mathcal{O}(1)$ arithmetic complexity per symbol.

\subsubsection{Motivation: The Three Asymptotic Regimes}

As the effective noise $\bar{\sigma}^2 \to 0$, the \ac{OBD}/\ac{OGD} undergoes three successive simplifications:
\begin{enumerate}[(i)]
  \item \textbf{Ring selection becomes deterministic.} The softmax probabilities $w_\ell^{\mathrm{G}}$ concentrate on a single dominant ring $\ell^*$ as the log-evidence gap $\Lambda_{\ell^*}^{\mathrm{G}} - \Lambda_\ell^{\mathrm{G}} \to \infty$ for $\ell \neq \ell^*$. 
  Formally, $w_{\ell^*}^{\mathrm{G}} = 1 - \mathcal{O}(e^{-c/\bar{\sigma}^2})$ for some $c > 0$.
  \item \textbf{Amplitude shrinkage vanishes.} The \ac{OGD} conditional center of mass on the dominant ring, $R_{\ell^*}(1 - 1/(2\kappa_{\ell^*}))e^{j\angle\bar{x}}$, satisfies $1 - 1/(2\kappa_{\ell^*}) \to 1$ and hence tends to $R_{\ell^*} e^{j\angle\bar{x}}$.
  \item \textbf{Phase estimation becomes a pure projection.} The denoiser output converges to $\hat{x} \to R_{\ell^*} e^{j\angle\bar{x}}$: the symbol is projected onto the nearest ring at the observed phase angle.
\end{enumerate}

\subsubsection{The \texorpdfstring{$\mathcal{O}(1)$}{O(1)} Orbital Phase Denoiser}

\begin{definition}[Orbital Phase Denoiser]
\label{def:lpd}
The \ac{OPD} is the high-\ac{SNR} limit of the \ac{OGD}, defined as
\begin{equation}
\label{eq:lpd}
\eta_P(\bar{x};\, \bar{\sigma}^2) \triangleq R_{\ell^*(\bar{x})} \, e^{j\angle\bar{x}},
\end{equation}
where the dominant ring is selected via hard nearest-radius detection, given by
\begin{equation}
\label{eq:lpd_ring_select}
\ell^*(\bar{x}) \triangleq \arg\min_{\ell \in \{1,\ldots,L\}} \norm{ \norm{\bar{x}} - R_\ell }.
\end{equation}
\end{definition}

The \ac{OPD} requires exactly one magnitude computation $\norm{\bar{x}}$, one nearest-radius lookup (a sorted binary search over $L$ values, costing $\mathcal{O}(\log L)$, or $\mathcal{O}(1)$ with precomputed decision boundaries/fixed constellations), and one complex multiplication $R_{\ell^*} e^{j\angle\bar{x}}$. 
No logarithms, exponentials, divisions, or softmax normalizations are needed. Intuitively, once the phase posterior is sharp the denoiser has nothing left to average: it simply snaps the observation to the nearest legal radius and keeps its angle.

\subsubsection{Optimality of the \texorpdfstring{$\mathcal{O}(1)$}{O(1)} Bound}

\begin{proposition}[Irreducibility of the \ac{OPD}]
\label{prop:lpd_irreducible}
Let $\eta: \mathbb{C} \times \mathbb{R}_{>0} \to \mathbb{C}$ be any denoiser satisfying
\begin{enumerate}[(i)]
  \item $\norm{\eta(\bar{x}; \bar{\sigma}^2)} \in \{R_1, \ldots, R_L\}$ for all $\bar{x}, \bar{\sigma}^2$ (ring-constrained output),
  \item $\E\bigl[\,\norm{x - \eta(\bar{x};\bar{\sigma}^2)}^2\,\bigr] \to 0$ as $\bar{\sigma}^2 \to 0$ (consistency at high \ac{SNR}), where $\bar{x} = x + \bar{z}$ with $x$ drawn from the constellation prior and $\bar{z} \sim \CN(0,\bar{\sigma}^2)$.
\end{enumerate}

Then, $\eta$ requires at minimum
\begin{itemize}
  \item One amplitude comparison (to select $\ell^*$), and
  \item One phase extraction (to determine $\angle\eta$).
\end{itemize}

The \ac{OPD} achieves both with exactly one magnitude computation, one sorted lookup, and one complex multiplication, {so no denoiser meeting~(i) and (ii) can dispense with either the amplitude selection or the phase read-out: the \ac{OPD} is minimal in this operation-count sense} on the ring-orbital manifold $\bigcup_{\ell=1}^L \mathcal{S}^1(R_\ell)$.
Minimality is meant at the level of operation \emph{classes} (one amplitude selection, one phase extraction), not gate counts; the matching lower bound, that any $\eta$ omitting either class violates~(ii), is established in the proof.
\end{proposition}

\begin{proof}
\proofref{app:prop_lpd_irreducible}
\end{proof}

\begin{table*}[t]
\centering
\caption{The Complete Denoiser Hierarchy: Complexity, Operations, and Validity}
\vspace{-2ex}
\label{tab:full_hierarchy}
\renewcommand{\arraystretch}{1.1}
\begin{tabular}{@{}l|c|c|c|c|c|c@{}}
\toprule
\textbf{Denoiser} 
  & \textbf{Cost} 
  & \textbf{Per-Ring $\hat{x}_\ell$}
  & \textbf{Log-Evidence $\Lambda_\ell$} 
  & \textbf{Special Functions} 
  & \textbf{Hardware LUTs} 
  & \textbf{Regime} \\
\midrule
Exact Discrete 
  & $\mathcal{O}(M)$ 
  & $\displaystyle\sum_{m \in \mathcal{M}_\ell} w_m s_m$
  & Full Posterior
  & None
  & Not applicable 
  & All $\bar{\sigma}^2$ \\[8pt]
OBD (Sec.~\ref{sec:nle}) 
  & $\mathcal{O}(L)$ 
  & $w_\ell^{\mathrm{B}}\, R_\ell A(\kappa_\ell)\, e^{j\angle\bar{x}}$
  & $\ln r_\ell - \dfrac{R_\ell^2}{\bar{\sigma}^2} + \ln I_0(\kappa_\ell)$
  & $A(\kappa)$, $\ln I_0$
  & Required 
  & All $\kappa \geq 0$ \\[8pt]
$U$-OBD (Sec.~\ref{sec:ja_cbm})\!\!\!
  & $\!\!\!\!\mathcal{O}(U L)\!\!\!\!$
  & Eq.~\eqref{eq:bcbm_ring_mean}
  & $\Lambda_\ell^{U}$, Eq.~\eqref{eq:bcbm_log_metric}
  & $I_0, I_1, \{I_{uM_\ell}, I_{uM_\ell \pm 1}\}_{u=1}^U$
  & Required
  & All $\kappa \geq 0$ \\[8pt]
OGD (Sec.~\ref{sec:cgr}) 
  & $\mathcal{O}(L)$ 
  & \!\!\!$w_\ell^{\mathrm{G}}\, R_\ell\!\left(1 \!-\! \tfrac{1}{2\kappa_\ell}\right) e^{j\angle\bar{x}}$\!\!\!
  & \!\!$\ln r_\ell \!-\! \dfrac{R_\ell^2}{\bar{\sigma}^2} \!+\! \kappa_\ell \!-\! \dfrac{1}{2}\ln(2\pi\kappa_\ell)$\!\!
  & None (only $\ln$, $/$) 
  & Not required 
  & $\kappa \geq 3$ \\[8pt]
OPD (Sec.~\ref{sec:lpd}) 
  & $\mathcal{O}(1)$ 
  & $R_{\ell^*} \cdot e^{j\angle\bar{x}}$ 
  & $\ell^* \!=\! \arg\min_\ell \norm{\norm{{\bar{x}}} - R_\ell}$ 
  & None 
  & Not required
  & $\kappa \gg 1$ \\
\bottomrule
\end{tabular}
\vspace{1ex}

\begin{minipage}{\textwidth}
\footnotesize
\textit{Notes:} $U=0$ recovers the \ac{OBD} exactly and $U\to\infty$ recovers the \ac{BOD}.
All four denoisers share the same concentration parameter $\kappa_\ell = 2R_\ell\norm{\bar{x}}/\bar{\sigma}^2$ (Eq.~\eqref{eq:kappa}). 
The adaptive three-regime denoiser (Eq.~\eqref{eq:three_regime}) selects per-ring among \ac{OBD} ($\kappa_\ell < \kappa_0$), \ac{OGD} ($\kappa_0 \leq \kappa_\ell < \kappa_1$), and \ac{OPD} ($\kappa_\ell \geq \kappa_1$), with $\kappa_0 \approx 3$ and $\kappa_1 \approx 20$. The cost is specified per symbol.
\end{minipage}
\vspace{-2ex}
\end{table*}

\subsubsection{Information-Theoretic Interpretation via Large Deviations}

The convergence of the softmax ring probabilities to a hard decision admits a precise large-deviations characterization.

\begin{proposition}[Exponential Ring Discrimination]
\label{prop:ring_discrimination}
The probability of incorrect ring selection under the \ac{OPD} decays exponentially with the effective \ac{SNR} ($\bar{\sigma}^2 \to 0$) as
\begin{equation}
\label{eq:ring_error_exp}
\Pr(\ell^* \neq \ell_{\mathrm{true}}) \leq (L-1)\,C_R\, \exp\!\left(-\frac{d_R^2}{4\bar{\sigma}^2}\right),
\end{equation}
where $d_R \triangleq \min_{\ell \neq \ell'} \norm{R_\ell - R_{\ell'}}$ is the minimum inter-ring distance, and $C_R$ is a bounded, \ac{SNR}-independent constant absorbing the Rician correction quantified in the proof.
\end{proposition}

\begin{proof}
The ring selection error occurs when $\norm{\norm{\bar{x}} - R_\ell} < \norm{\norm{\bar{x}} - R_{\ell_{\mathrm{true}}}}$ for some $\ell \neq \ell_{\mathrm{true}}$. 
Since $\norm{\bar{x}} = R_{\ell_{\mathrm{true}}} + n_r$ where $n_r \sim \mathcal{N}(0, \bar{\sigma}^2/2)$ is the radial noise component (the real projection of $\mathcal{CN}(0, \bar{\sigma}^2)$ onto the radial direction), the nearest-radius decision boundary between ring~$\ell_{\mathrm{true}}$ and any competing ring~$\ell$ lies at the midpoint $(R_{\ell_{\mathrm{true}}} + R_\ell)/2$.

For a \emph{specific} competing ring~$\ell$ with $R_\ell > R_{\ell_{\mathrm{true}}}$, an error occurs only when $n_r > (R_\ell - R_{\ell_{\mathrm{true}}})/2 \geq d_R/2$.
Symmetrically, for $R_\ell < R_{\ell_{\mathrm{true}}}$, an error occurs only when $n_r < -(R_{\ell_{\mathrm{true}}} - R_\ell)/2 \leq -d_R/2$.
In either case, the error is a \emph{one-sided} tail event.
The one-sided Gaussian tail bound (Chernoff bound) for $n_r \sim \mathcal{N}(0, \bar{\sigma}^2/2)$ gives
\begin{equation*}
\Pr\!\left(n_r > \frac{d_R}{2}\right)
\leq \exp\!\left(-\frac{(d_R/2)^2}{2 \cdot \bar{\sigma}^2/2}\right)
= \exp\!\left(-\frac{d_R^2}{4\bar{\sigma}^2}\right).
\end{equation*}

A union bound over the at most $L-1$ competing rings yields~\eqref{eq:ring_error_exp}.
Note that no factor of~$2$ appears because each competing ring contributes only a \emph{single} tail direction; the two-sided bound $\Pr(\norm{n_r} > d_R/2) \leq 2\exp(-d_R^2/(4\bar{\sigma}^2))$ would be needed only if a single ring could induce errors in both directions, which is geometrically impossible for the nearest-radius rule.
Strictly, the observed magnitude is Rician rather than Gaussian: $\norm{\bar{x}} = \sqrt{(R_{\ell_{\mathrm{true}}} + n_r)^2 + n_\perp^2}$, with $n_\perp \sim \mathcal{N}(0, \bar{\sigma}^2/2)$ the tangential component. Since $\sqrt{a^2 + b^2} \leq a + b^2/(2a)$ for $a > 0$, the tangential term inflates the radial statistic by at most $n_\perp^2/(2(R_{\ell_{\mathrm{true}}} + n_r)) = \mathcal{O}_p(\bar{\sigma}^2)$. Integrating this random $\mathcal{O}(\bar{\sigma}^2)$ shift of a decision threshold at distance $\Theta(1)$ over the tangential component perturbs the tail exponent by only $\mathcal{O}(1)$, hence multiplies the bound by a \emph{bounded constant} factor $C_R = \exp(\mathcal{O}(1))$, not a vanishing one, while leaving the exponent $d_R^2/(4\bar{\sigma}^2)$ unchanged; this $C_R$ is the constant prefactor in~\eqref{eq:ring_error_exp}.
\end{proof}

\begin{remark}[Gallager's Error Exponent Connection]
\label{rem:gallager}
The exponential decay in~\eqref{eq:ring_error_exp} {coincides with the pairwise (minimum-distance) error exponent for uncoded $L$-ary amplitude modulation with minimum distance $d_R$ and noise variance $\bar{\sigma}^2/2$; cf.\ the union-bound analysis of Gallager~\cite{Gallager1968}.}
{This shows the ring-detection component of the \ac{OPD} attains the minimum-distance exponent governing uncoded \ac{ML} amplitude detection at high \ac{SNR}}, while the phase estimation component achieves the \ac{CRLB} $\mathrm{Var}[\hat{\theta}] = 1/\kappa_{\ell^*}$ (with $\kappa_{\ell^*} \to 2R_{\ell^*}^2/\bar{\sigma}^2$ as $\bar{\sigma}^2 \to 0$) for phase estimation in circular Gaussian noise.
\end{remark}

\begin{proposition}[\ac{OPD} as the Limit of \ac{OGD}]
\label{prop:lpd_limit}
For any fixed $\bar{x} \neq 0$, the \ac{OGD} converges pointwise to the \ac{OPD} as $\bar{\sigma}^2 \to 0^+$ such that
\begin{equation}
\label{eq:lpd_convergence}
\lim_{\bar{\sigma}^2 \to 0^+} \eta_G(\bar{x};\, \bar{\sigma}^2) = \eta_P(\bar{x};\, \bar{\sigma}^2) = R_{\ell^*} \, e^{j\angle\bar{x}},
\end{equation}
provided $\norm{\bar{x}}$ does not lie exactly on a ring decision boundary {(equidistance from two radii, a Lebesgue-null set of $\bar{x}$)}.
\end{proposition}

\begin{proof}
As $\bar{\sigma}^2 \to 0^+$, the \ac{OGD} log-evidence~\eqref{eq:cgr_evidence_2} is dominated by the term $(2R_\ell\norm{\bar{x}} - R_\ell^2)/\bar{\sigma}^2 = R_\ell(2\norm{\bar{x}} - R_\ell)/\bar{\sigma}^2$. 
For the ring $\ell^*$ closest to $\norm{\bar{x}}$ in radius, this term is maximized: {since $R_\ell(2\norm{\bar{x}} - R_\ell) = \norm{\bar{x}}^2 - (\norm{\bar{x}} - R_\ell)^2$, maximizing over the discrete radii is exactly the nearest-radius rule $\ell^* = \arg\min_\ell \norm{\,\norm{\bar{x}} - R_\ell}$.}
The softmax gap $\Lambda_{\ell^*}^{\mathrm{G}} - \Lambda_\ell^{\mathrm{G}} = (R_{\ell^*} - R_\ell)(2\norm{\bar{x}} - R_{\ell^*} - R_\ell)/\bar{\sigma}^2 + \mathcal{O}(\ln\bar{\sigma}^2)$ diverges as $1/\bar{\sigma}^2$ for $\ell \neq \ell^*$, driving $w_{\ell^*}^{\mathrm{G}} \to 1$ exponentially. 
Simultaneously, $1 - 1/(2\kappa_{\ell^*}) = 1 - \bar{\sigma}^2/(4R_{\ell^*}\norm{\bar{x}}) \to 1$. Therefore $\eta_G(\bar{x};\bar{\sigma}^2) = w_{\ell^*}^{\mathrm{G}} R_{\ell^*}(1 - 1/(2\kappa_{\ell^*}))e^{j\angle\bar{x}} + \sum_{\ell \neq \ell^*} \hat{x}_\ell^{\mathrm{G}} \to R_{\ell^*} e^{j\angle\bar{x}}$.
\end{proof}

\subsubsection{Connection to Classical Detection Theory}

The \ac{OPD} also admits a direct interpretation in terms of classical detection-estimation separation.

\begin{proposition}[Detection-Estimation Separation of the \ac{OPD}]
\label{prop:det_est}
The \ac{OPD} $\eta_P(\bar{x}) = R_{\ell^*} e^{j\angle\bar{x}}$ admits an {asymptotic} factorization {(exact as $\bar{\sigma}^2 \to 0$)} into two {asymptotically independent} operations on the sufficient statistics $(\norm{\bar{x}}, \angle\bar{x})$:

\begin{figure}[H]
  \centering
  \includegraphics[width=\columnwidth]{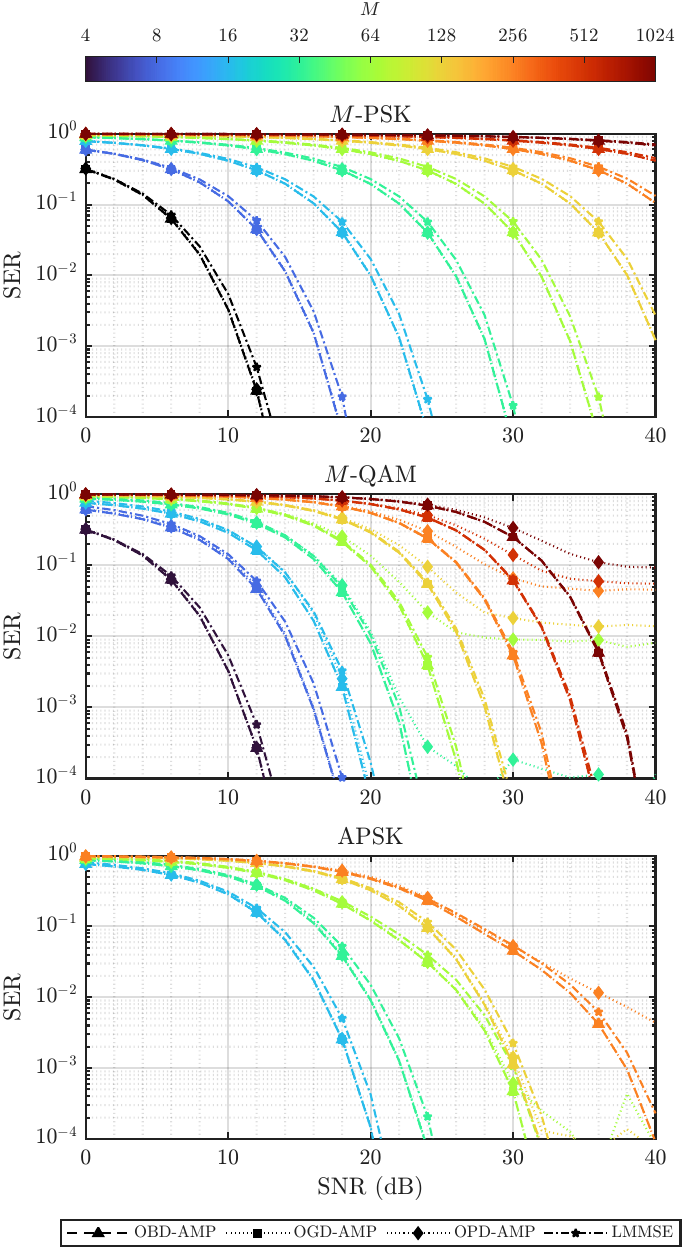}
  \vspace{-3ex}
  \caption{\ac{SER} versus \ac{SNR} of full \ac{AMP} detection ($K = 64$, $N = 256$, $\alpha = 0.25$, as in Fig.~\ref{fig:OBD_SER}) under the orbital denoiser hierarchy -- \ac{OBD}, \ac{OGD}, \ac{OPD} -- and the \ac{LMMSE} baseline, for all $M$-\ac{PSK} and $M$-\ac{QAM} orders $M \in \{4, 8, \ldots, 1024\}$ and the DVB-S2/S2x \ac{APSK} constellations $M \in \{16, 32, 64, 128, 256\}$. Line style encodes the detector (shared legend); color encodes the modulation order $M$ (shared colorbar). The \ac{OBD}, \ac{OGD}, and \ac{LMMSE} track their fixed points at every order; the hard, non-Lipschitz \ac{OPD} coincides with the \ac{OBD}/\ac{OGD} except for high-order \ac{QAM} ($M \gtrsim 64$), where the undamped \ac{AMP} recursion floors (Lemma~\ref{lem:pseudo_lip}); damping, or the adaptive fallback of Definition~\ref{def:adaptive_cgr}, restores convergence. The scalar-\ac{SE} evaluation of Figs.~\ref{fig:MSE_MPSK_MQAM} and~\ref{fig:SE_MPSK_MQAM} confirms the fixed-point equivalence directly.}
  \label{fig:Hierarchy_SER}
  \vspace{-2ex}
\end{figure}

\begin{enumerate}[(i)]
  \item \textbf{Amplitude detection:} The ring index $\ell^*$ is determined solely from the radial observation $\norm{\bar{x}}$ via the \ac{ML} rule for an $L$-ary \ac{ASK} signal in Gaussian noise, given by
  \begin{align}
  \label{eq:amplitude_ml}
  \ell^* &= \arg\min_{\ell \in \{1,\ldots,L\}} \norm{ \norm{\bar{x}} - R_\ell } \nonumber \\
  &= \arg\max_{\ell} \; p\bigl(\norm{\bar{x}} \;\big|\; \norm{x} = R_\ell\bigr),
  \end{align}
  where the radial observation satisfies $\norm{\bar{x}} = R_{\ell_{\mathrm{true}}} + n_r$ with $n_r \sim \mathcal{N}(0, \bar{\sigma}^2/2)$. {The \ac{ML} rule shown coincides with the nearest-radius (\ac{MAP}) rule under equal ring priors; unequal $r_\ell$ shift the decision boundaries by $\mathcal{O}(\bar{\sigma}^2 \ln r_\ell)$, vanishing as $\bar{\sigma}^2 \to 0$.}
  \item \textbf{Phase estimation:} The phase is estimated as $\hat{\theta} = \angle\bar{x}$, which is the \ac{ML} estimate of a deterministic phase parameter $\theta_0$ observed through the circular Gaussian channel $\bar{x} = R_{\ell^*} e^{j\theta_0} + \bar{z}$, $\bar{z} \sim \CN(0, \bar{\sigma}^2)$.
\end{enumerate}

These two operations are statistically independent in the high-\ac{SNR} regime, and each individually achieves the fundamental performance limit for its respective sub-problem:
\begin{itemize}
  \item The amplitude detector achieves the minimum-distance error exponent $\exp(-d_R^2/(4\bar{\sigma}^2))$ (Proposition~\ref{prop:ring_discrimination}), which is optimal for uncoded \ac{ML} amplitude detection (Remark~\ref{rem:gallager}).
  \item The phase estimator achieves the \ac{CRLB} asymptotically: $\Var[\hat{\theta}] = 1/\kappa_{\ell^*} \to \bar{\sigma}^2/(2R_{\ell^*}^2)$ as $\bar{\sigma}^2 \to 0$.
\end{itemize}
\end{proposition}

\begin{proof}
\proofref{app:prop_det_est}
\end{proof}

\begin{corollary}[\ac{OPD} as the Degenerate Limit of the Orbital Decomposition]
\label{cor:lpd_as_limit}
As $\bar{\sigma}^2 \to 0^+$, the canonical $3L$-dimensional state
$(\hat{\mathbf{x}}, \mathbf{w})$ of Theorem~\ref{thm:ring_decomp} degenerates
continuously to the two-scalar representation $(\ell^*, \angle\bar{x})$ of the \ac{OPD} (Definition~\ref{def:lpd}), equivalently $(R_{\ell^*}, \angle\bar{x})$ since the radius is determined by the ring index, via two simultaneous collapses:
\begin{enumerate}[(i)]
\item $w_\ell \to \mathbf{1}[\ell = \ell^*]$ (soft $\to$ hard ring selection),
\item $A(\kappa_\ell) \to 1$, so the conditional mean $\hat{x}_\ell/w_\ell \to R_\ell e^{j\angle\bar{x}}$
      (shrinkage $\to$ unit-length projection).
\end{enumerate}
The complexity hierarchy $\mathcal{O}(M)\to\mathcal{O}(L)\to\mathcal{O}(1)$
is therefore a single continuous geometric limit, not three separate
approximations.
\end{corollary}
\begin{proof}
As $\bar{\sigma}^2 \to 0^+$ every concentration $\kappa_\ell = 2R_\ell\norm{\bar{x}}/\bar{\sigma}^2$ of~\eqref{eq:obd_kappa} diverges for $\bar{x} \neq 0$. Collapse~(i) is Proposition~\ref{prop:lpd_limit}: the log-evidence gaps $\Lambda_{\ell^*}^{\mathrm{B}} - \Lambda_{\ell}^{\mathrm{B}}$ in~\eqref{eq:obd_metric} grow without bound, so the softmax~\eqref{eq:obd_ring_prob} tends to the indicator of the dominant ring. Collapse~(ii) is the Bessel-ratio expansion~\eqref{eq:A_approx}, $A(\kappa) = 1 - 1/(2\kappa) + \mathcal{O}(\kappa^{-2}) \to 1$, so the conditional contribution~\eqref{eq:obd_mean} tends to $R_\ell e^{j\angle\bar{x}}$. Both limits are continuous in $\bar{\sigma}^2$, and together they reduce $(\hat{\mathbf{x}},\mathbf{w})$ to $(\ell^*, \angle\bar{x})$, which is Definition~\ref{def:lpd}.
\end{proof}

\begin{remark}[Information-Theoretic Optimality of Separation]
\label{rem:separation_optimality}
Proposition~\ref{prop:det_est} establishes that the \ac{OPD}'s detection-estimation separation is not merely a computational convenience but is \emph{information-theoretically optimal} in the high-\ac{SNR} regime. 
The polar sufficient statistics $(\norm{\bar{x}}, \angle\bar{x})$ asymptotically decouple into a real Gaussian channel for amplitude and an independent circular Gaussian channel for phase, with the \ac{OPD} applying the individually optimal processor to each. 
{Since $(\norm{\bar{x}}, \angle\bar{x})$ is an invertible reparametrization of $\bar{x}$ (hence a sufficient statistic) and the two coordinates carry disjoint parameters ($R_\ell$ versus $\theta_0$) through asymptotically independent noises, applying the individually optimal processor to each is jointly optimal in the high-\ac{SNR} limit.}
This connects the \ac{OPD} to classical sufficiency-based detection theory (see, e.g., the treatment in Verd\'{u}~\cite{Verdu1998}): the polar decomposition is a sufficient statistic for the ring-orbital estimation problem, and the \ac{OPD} {is the asymptotically optimal detector-estimator built on it}.
\end{remark}

\begin{proposition}[\ac{OPD} as a Metric Projection]
\label{prop:lpd_projection}
The \ac{OPD} is the metric projection of $\bar{x}$ onto
$\mathcal{C} \triangleq \bigcup_{\ell=1}^{L} \{z\in\mathbb{C}: \norm{z}=R_\ell\}$:
\begin{equation}
  \eta_P(\bar{x})
  \!=\! P_{\mathcal{C}}(\bar{x})
  \!\triangleq\! \arg\min_{c\in\mathcal{C}} \norm{c - \bar{x}}.
\end{equation}

The nearest point of each circle $\{z: \norm{z} = R_\ell\}$ to $\bar{x} \neq 0$ is $R_\ell e^{j\angle\bar{x}}$ (the point on that circle along the ray $\angle\bar{x}$), so the projection onto the union $\mathcal{C}$ reduces to the nearest-radius selection $\ell^* = \arg\min_\ell \norm{\,\norm{\bar{x}} - R_\ell}$ followed by the phase read-out $\angle\bar{x}$; the constraint $\angle c = \angle\bar{x}$ is thus automatic, not imposed. Since $\mathcal{C}$ is a finite union of circles, closed but neither convex nor connected, the Hilbert projection theorem does not apply, but the projection is nonetheless single-valued except on the (Lebesgue-null) set of $\bar{x}$ equidistant from two radii. In the sense of the generalized (nonconvex) proximal map, cf.\ Moreau's convex proximity operator~\cite{Moreau1962}, $\eta_P$ is the Euclidean projection $P_{\mathcal{C}}$ onto the nonconvex set $\mathcal{C}$ (equivalently $\mathrm{prox}_{\iota_{\mathcal{C}}}$), single-valued off that null set.
\end{proposition}

\subsubsection{Complexity Comparison}

Table~\ref{tab:full_hierarchy} summarizes the complete hierarchy of denoisers derived in this paper, from exact discrete to the \ac{OPD} limit, along with their computational complexity, mathematical operations, and validity regimes.

\subsubsection{The Complete Denoiser Hierarchy}

\begin{remark}[Adaptive Multi-Regime Denoiser]
\label{rem:adaptive_full}
Combining Definition~\ref{def:adaptive_cgr} with the \ac{OPD}, a three-regime adaptive denoiser can be constructed with per-ring contribution
\begin{equation*}
\hat{x}_\ell = w_\ell \begin{cases}
R_\ell\, A(\kappa_\ell)\, e^{j\angle\bar{x}} & \kappa_\ell < \kappa_0 \;\;\text{(OBD)}, \\
R_\ell(1 - 1/(2\kappa_\ell))\, e^{j\angle\bar{x}} & \kappa_0 \leq \kappa_\ell < \kappa_1 \;\;\text{(OGD)}, \\
R_\ell\, e^{j\angle\bar{x}} & \kappa_\ell \geq \kappa_1 \;\;\text{(OPD)},
\end{cases}
\vspace{-3.5ex}
\end{equation*}
\begin{equation}
\label{eq:three_regime}
\vspace{-0.5ex}
\end{equation}
with $\kappa_0 \approx 3$ and $\kappa_1 \approx 20$ (where the neglected shrinkage $1/(2\kappa_\ell) \leq 2.5\%$ for all $\kappa_\ell \geq \kappa_1$).

At the \ac{SE} fixed point of a well-designed receiver operating at sufficiently high \ac{SNR}, which is above roughly $10$~dB for constant-modulus constellations, and correspondingly higher for the innermost rings of dense multi-ring constellations (since $\kappa_\ell \propto R_\ell^2$ on the correct ring), all rings satisfy $\kappa_\ell \geq \kappa_1$, and the denoiser collapses entirely to the $\mathcal{O}(1)$ \ac{OPD}.
\end{remark}

The hierarchy is now complete: from $\mathcal{O}(M)$ to $\mathcal{O}(1)$ in three steps, with each approximation error Wasserstein-bounded.
The question that remains is whether these three denoisers, despite their vast complexity difference, ultimately deliver the same detection performance.
Figure~\ref{fig:Hierarchy_SER} supplies the empirical answer with full \ac{AMP} detection, under the identical system configuration as Fig.~\ref{fig:OBD_SER} ($K = 64$, $N = 256$, $\alpha = 0.25$) and swept over the complete modulation-order range of each family ($M = 4$ through $1024$ for $M$-\ac{PSK} and $M$-\ac{QAM}, the DVB-S2/S2x constellations for \ac{APSK}): across \ac{PSK}, \ac{APSK}, and \ac{QAM} up to moderate order the \ac{SER} curves of \ac{AMP} running the \ac{OBD}, \ac{OGD}, and \ac{OPD} denoisers are visually coincident, three orders of arithmetic complexity with no observable performance separation, while the \ac{LMMSE} baseline is uniformly and visibly worse; the exact \ac{BOD}, whose $\mathcal{O}(M)$ cost is precisely what the hierarchy removes, was already matched by the \ac{OBD} in Fig.~\ref{fig:OBD_SER}.
One exception is visible and instructive: for high-order $M$-\ac{QAM} ($M \gtrsim 64$) the \ac{OPD} develops an error floor that rises with $M$, while the \ac{OBD} and \ac{OGD} continue to track the shared fixed point down.
\pagebreak

This is a failure of the \emph{iteration}, not of the fixed-point equivalence.
The \ac{OPD} is a hard, non-Lipschitz ring projection (Lemma~\ref{lem:pseudo_lip}); on the closely spaced rings of dense \ac{QAM}, its \ac{AMP} Onsager correction, built on the projection's almost-everywhere divergence, misses the contribution of its jump discontinuities, so the self-interference it is meant to cancel leaks back and the recursion stalls short of its fixed point. Light damping of the \ac{AMP} updates, or the adaptive \ac{OBD}/\ac{OGD} fallback of Definition~\ref{def:adaptive_cgr} on the near-boundary symbols, removes the floor and restores the \ac{OBD}/\ac{OGD} curve. That the equivalence itself is intact is confirmed by the scalar \ac{SE} recursion (which is immune to \ac{AMP} iteration dynamics) under which all three denoisers coincide at every order (Figs.~\ref{fig:MSE_MPSK_MQAM} and~\ref{fig:SE_MPSK_MQAM}).
Section \ref{sec:se} proves that this equivalence is not an artifact of the plot but a theorem, through the state evolution analysis.

\vspace{-2ex}
\section{State Evolution Analysis}
\label{sec:se}

We characterize the orbital denoisers through their \ac{SE}, the low-dimensional recursion that tracks the error of \ac{AMP} in the large-system limit, which renders the error of the full $K$-dimensional system \emph{exactly} tractable: \ac{SE} collapses the coupled recovery into a single deterministic scalar map, placing every level of the denoiser hierarchy on a common, analytically comparable footing on which the cost of the orbital relaxation can be quantified, rather than merely simulated.
The \ac{OBD}, \ac{OGD}, and \ac{OPD} act on the same statistic, the concentration $\kappa_\ell$, and differ only in how accurately they approximate $A(\kappa_\ell)$ and $\ln I_0(\kappa_\ell)$, such that their recursions collapse onto a single trajectory, with Theorem~\ref{thm:cross_level_fp} showing that all three reach the \emph{identical} leading-order fixed point $\sigma_z^2/(2-\alpha)$, differing only in higher-order corrections.

\vspace{-1.5ex}
\subsection{SE Under Mismatched Priors: Setup and Background}
\label{sec:three_se}

{Recall from Section~\ref{sec:amp_decoupling} that \ac{AMP}~\cite{Donoho2009,Bayati2011} estimates $\mathbf{x} \in \mathcal{M}^K$ from the linear model~\eqref{eq:system} by alternating a linear matched-filter step with a separable denoiser, and that its defining ingredient is the \emph{Onsager correction}, a scalar feedback term that cancels the self-interference fed back through $\mathbf{H}$.
In the large-system limit ($N, K \to \infty$, $K/N \to \alpha \in (0,1)$) this correction makes the per-symbol residual asymptotically Gaussian and symbol-independent, so the vector problem decouples into $K$ identical scalar channels, each exactly the cavity model~\eqref{eq:scalar_awgn}, $\bar{x} = x + \bar{z}$ with $\bar{z} \sim \CN(0, \bar{\sigma}^2)$, on which a denoiser $\eta(\bar{x};\bar{\sigma}^2)$ acts.
The effective variance $\bar{\sigma}^2$ is not the physical noise $\sigma_z^2$ but $\sigma_z^2$ inflated by residual interference; it changes from iteration to iteration, and \ac{SE} is the scalar recursion that tracks it exactly.

This recursion is \emph{mismatched} because the orbital denoisers apply the continuous-phase orbital prior of Definition~\ref{def:orbital} (the orbital prior $q$ of density~\eqref{eq:cbm_prior}), rather than the true discrete constellation prior $p$ that generates $\mathbf{x}$; only the \ac{BOD}, built on $p$, is \emph{matched} and hence Bayes-optimal.
Since $q$ preserves the ring radii and probabilities of $\mathcal{M}$, it has zero mean and the same average energy as the constellation, $\mathrm{Var}_q(x) = \sum_\ell r_\ell R_\ell^2 = E_d$; throughout, $\E_q$ and $\mathrm{Var}_q$ denote expectation and variance under $q$ and $\E_p$ expectation under $p$, while
\vspace{-0.5ex}
\begin{equation}
  \label{eq:mmse_Q_def}
  \mathrm{mmse}_q(\bar{\sigma}^2) \;\triangleq\; \E_q\!\bigl[\,\norm{x - \E_q[x\mid\bar{x}]}^2\,\bigr],
\end{equation}
is the \ac{MMSE} of estimating $x \sim q$ from the cavity observation $\bar{x} = x + \bar{z}$, $\bar{z}\sim\CN(0,\bar{\sigma}^2)$.
\vspace{-0.5ex}
\pagebreak

Crucially, \ac{SE} stays \emph{exact} in the large-system limit for any pseudo-Lipschitz separable denoiser, matched or not~\cite{Bayati2011,Javanmard2013}. 

This regularity, and \emph{not} Bayes-optimality, is precisely what licenses the \ac{SE} invoked throughout this section; that the orbital hierarchy satisfies it (the \ac{OBD} pointwise, and the \ac{OGD}/\ac{OPD} through a fixed-$\bar{\sigma}^2$ Lipschitz regularization) is established in Lemma~\ref{lem:pseudo_lip} below, which makes the mismatched orbital denoisers amenable to the rigorous analysis that follows.
Let $\mathrm{MSE}_t$ denote the per-symbol \ac{MSE} at iteration $t$: a single scalar that summarizes the error of the full $K$-dimensional estimate and is the \ac{SE} state.
After Onsager decoupling, the denoiser at iteration $t$ acts on the cavity channel~\eqref{eq:scalar_awgn} itself, now carrying an iteration index: $\bar{x}_t = x + \bar{z}_t$ with effective noise $\bar{z}_t \sim \CN(0,\bar{\sigma}_t^2)$, where $x$ is the transmitted symbol drawn from the true discrete constellation prior.
For the \ac{SE} recursion we write this noise in standardized form, $\bar{z}_t \triangleq \bar{\sigma}_t \tilde{z}$ with $\tilde{z} \sim \CN(0,1)$, so that the iteration dependence is carried entirely by the deterministic scale $\bar{\sigma}_t$ while the error is averaged over the fixed pair $(x,\tilde{z})$; this is the standard \ac{SE} normalization~\cite{Bayati2011} and yields the $x + \bar{\sigma}_t \tilde{z}$ form appearing in the recursions below.
This effective variance $\bar{\sigma}_t$ is not a free parameter: the Onsager correction ties it self-consistently to the current error as
\begin{equation}
  \bar{\sigma}_t^2 = \sigma_z^2 + \frac{K}{N}\,\mathrm{MSE}_t,
  \label{eq:se_noise}
\end{equation}
where $\sigma_z^2$ is the physical per-component \ac{AWGN} variance of~\eqref{eq:system}, fixing the operating point $\mathrm{SNR} = E_d/\sigma_z^2$, and $\tfrac{K}{N}\,\mathrm{MSE}_t$ is the residual interference from the remaining symbols; hence $\bar{\sigma}_t^2 \geq \sigma_z^2$, with equality only at zero error.
Denoising $\bar{x}_t$ and remeasuring the error gives the per-denoiser update 
\begin{equation}
  \mathrm{MSE}_{t+1} = \mathbb{E}_{x,\tilde{z}}\bigl[\,\norm{x - \eta(\bar{x}_t;\bar{\sigma}_t^2)}^2\,\bigr],
\end{equation}
specified for each denoiser below; composed with~\eqref{eq:se_noise} it becomes a deterministic scalar map $\mathrm{MSE}_{t+1} = \mathcal{F}(\mathrm{MSE}_t)$ that, in the large-system limit ($N, K \to \infty$, $K/N \to \alpha$), predicts the per-symbol error of the full $N \times K$ system exactly~\cite{Bayati2011}.
Two features of this recursion matter for what follows.  First, $\mathrm{MSE}_t$ is the \emph{actual} per-symbol squared error of the denoised estimate (the right-hand side above), well-defined for \emph{every} separable (matched or mismatched) $\eta$.  Second, it is \emph{not}, in general, the denoiser's posterior variance: by the \ac{MMSE}/Nishimori identity it equals the expected posterior variance $\mathbb{E}[\operatorname{Var}(x\mid\bar{x})]$ \emph{only} when $\eta$ is a Bayes posterior mean, namely the matched \ac{BOD}, and the \ac{OBD} (whose mismatched error reduces to the matched \ac{MMSE} under $q$ by circular symmetry{, Lemma~\ref{lem:phase_preserving}(iii)}), whereas the non-Bayes \ac{OGD} and \ac{OPD} incur a strictly larger error, exceeding it by the orthogonality excess of Lemma~\ref{lem:excess}.

\subsection{Fixed-Point Existence, Monotonicity, and Uniqueness}
\label{sec:se_fixedpoint}

\begin{definition}[SE Fixed Point]
  \label{def:fixed_point}
  Fix any denoiser $\eta \in \{\eta_D, \eta_B, \eta_G, \eta_P, \eta_{\mathrm{L}}\}$ of the hierarchy. Its one-step \ac{SE} map sends the current error $\mathrm{MSE}$ to the next as
  \begin{equation}
    \mathcal{F}(\mathrm{MSE}) \triangleq \mathbb{E}_{x,\tilde{z}}\!\Bigl[
    \norm{x - \eta\bigl(x + \bar{\sigma}\,\tilde{z};\;\bar{\sigma}^2\bigr)}^2\Bigr],
    \label{eq:fixed_point}
  \end{equation}
  where the effective variance $\bar{\sigma}^2$ is given by~\eqref{eq:se_noise}.
  A value $\mathrm{MSE}^* \ge 0$ is a \emph{fixed point} if $\mathrm{MSE}^* =\mathcal{F}(\mathrm{MSE}^*)$.
\end{definition}

\begin{lemma}[Phase-Preserving Structure and Stein Identity for the Orbital Denoisers]
\label{lem:phase_preserving}
Every orbital denoiser $\eta \in \{\eta_B, \eta_G, \eta_P\}$ is phase-preserving: $\eta(\bar{x}; \bar{\sigma}^2) = \rho_\eta(\norm{\bar{x}}, \bar{\sigma}^2)\, e^{j\angle\bar{x}}$ with a real radial profile $\rho_\eta$. For the \ac{OBD} and \ac{OPD} this profile is non-negative and bounded {($\rho_{\eta_B} \in [0, R_L)$ for the \ac{OBD}; $\rho_{\eta_P} \in [0, R_L]$ for the \ac{OPD}, which attains $R_L$ on the outer ring)}; for the \ac{OGD} the same holds at high concentration $\kappa_\ell \geq \tfrac12$, while below it the Gaussian-phase factor $1 - 1/(2\kappa_\ell)$ turns negative (handled by the adaptive fallback of Definition~\ref{def:adaptive_cgr}). The identities below are stated for the representative \ac{OBD} $\eta_B$; the \ac{OPD} and the high-concentration \ac{OGD} satisfy~(i) and~(iii) but, not being exact posterior means, not the Stein identity~(ii) (see the proof). 
Consequently:
\begin{enumerate}[(i)]
\item The Wirtinger derivative is real-valued and non-negative:
\begin{equation}
\label{eq:wirtinger_real}
\frac{\partial \eta_B}{\partial \bar{x}} = \frac{\rho(\norm{\bar{x}}, \bar{\sigma}^2)}{2\norm{\bar{x}}} + \frac{\rho'(\norm{\bar{x}}, \bar{\sigma}^2)}{2} \;\geq\; 0,
\end{equation}
where $\rho' \triangleq \partial\rho/\partial\norm{\bar{x}}$.

\item (Bayesian Stein identity.) Since $\eta_B = \E_q[x \mid \bar{x}]$ is the Bayes posterior mean under the orbital prior $q$ (Definition~\ref{def:orbital}), its Wirtinger derivative equals the normalized posterior variance:
\begin{equation}
\label{eq:stein_identity}
\frac{\partial \eta_B}{\partial \bar{x}} = \frac{\hat{\sigma}^2_{\mathrm{B}}(\bar{x})}{\bar{\sigma}^2},
\end{equation}
where $\hat{\sigma}^2_{\mathrm{B}}(\bar{x}) = \sum_{\ell=1}^L w_\ell^{\mathrm{B}} R_\ell^2 - \norm{\eta_B(\bar{x})}^2$ is the \ac{OBD} posterior variance from~\eqref{eq:cbm_var}.

\item (Circular-symmetry equivalence.) Because $\eta$ is phase-preserving and $\tilde{z}$ is rotationally invariant, its \ac{MSE} depends only on the radial distribution $\{R_\ell, r_\ell\}$ of the input, not on the input-phase distribution. Hence the mismatched \ac{MSE} under the true discrete prior $p$ equals the matched \ac{MSE} under the orbital prior $q$ (Definition~\ref{def:orbital}):
\begin{equation}
\label{eq:circular_symmetry_equiv}
\E_p\bigl[\norm{x - \eta(\bar{x})}^2\bigr] = \E_q\bigl[\norm{x - \eta(\bar{x})}^2\bigr].
\end{equation}
\end{enumerate}
\end{lemma}

\begin{proof}
\proofref{app:lem_phase_preserving}
\end{proof}

The rigorous state evolution used throughout this section rests on the pseudo-Lipschitz continuity property of the denoiser, which, with separability, is the hypothesis under which the \ac{SE} recursion of~\cite{Bayati2011,Javanmard2013} tracks the true per-symbol \ac{MSE} exactly, and which later licenses the decoupling of Theorem~\ref{thm:decoupling}. \emph{Mismatch} (using the orbital prior $q$ in place of the true $p$) is immaterial to this hypothesis: the \ac{SE} theorems require separability and regularity, not Bayes-optimality. \emph{Regularity}, however, is not optional, so we establish it now for the whole hierarchy, before invoking \ac{SE} below.

\begin{lemma}[Regularity of the Orbital Denoisers]
\label{lem:pseudo_lip}
At each fixed $\bar{\sigma}^2 > 0$, the \ac{OBD} is globally Lipschitz, with constant $\Theta(R_L^2/\bar{\sigma}^2)$, hence separable and pseudo-Lipschitz of order $2$, the regularity required by the \ac{SE} theorems of~\cite{Bayati2011,Javanmard2013}. 
The adaptive \ac{OGD} and the \ac{OPD} are uniformly bounded but discontinuous (regime-switch jump and a hard ring slice, respectively, each confined to a Lebesgue-null set), so they are \emph{not} pseudo-Lipschitz; their \ac{SE} is instead inherited as the fixed-$\bar{\sigma}^2$ limit of a pseudo-Lipschitz regularization. State evolution consequently governs all three levels of the hierarchy.
\end{lemma}
\begin{proof}
  \proofref{app:lem_denoiser_regularity}
\end{proof}

\begin{theorem}[Existence of the \ac{SE} Fixed Point]
\label{thm:se_convergence}
Fix $\sigma_z^2 > 0$ and let $\mathcal{I}_0 \triangleq [0, 4R_L^2] \supseteq \mathcal{I} \triangleq [0, E_d]$, with $x$ drawn from the true discrete constellation prior and $\tilde{z} \sim \mathcal{CN}(0,1)$ independently. For every denoiser of the hierarchy (\ac{BOD}, \ac{OBD}, \ac{OGD}, \ac{OPD}) and the \ac{LMMSE} baseline, the \ac{SE} map $\mathcal{F}$ of~\eqref{eq:fixed_point} is non-negative, continuous, and a self-map of a compact interval; by Brouwer's fixed-point theorem it admits at least one fixed point $\mathrm{MSE}^*$, localized as follows.
\begin{enumerate}[(i)]
  \item \textbf{\ac{BOD}, \ac{OBD}, \ac{LMMSE}.} The \ac{BOD} and \ac{OBD} are posterior means and the \ac{LMMSE} is the linear \ac{MMSE} ($\mathcal{F}_{\mathrm{L}} = E_d\bar{\sigma}^2/(E_d+\bar{\sigma}^2)$), so in each case $\mathcal{F}(\mathrm{MSE}) < E_d$ at finite $\bar{\sigma}^2$; thus $\mathcal{F}(\mathcal{I}) \subseteq \mathcal{I}$ and $\mathrm{MSE}^* \in [0, E_d]$.
  \item \textbf{\ac{OGD}, \ac{OPD}.} These do \emph{not} realize the Bayes posterior mean (the \ac{OGD} approximates it via the Gaussian phase~\eqref{eq:gaussian_phase}, the \ac{OPD} replaces it by a hard ring projection), so $\mathcal{F} < E_d$ need not hold; only the boundedness estimate $\mathcal{F}(\mathrm{MSE}) \leq 4R_L^2$ is available, giving $\mathcal{F}(\mathcal{I}_0) \subseteq \mathcal{I}_0$ and $\mathrm{MSE}^* \in [0, 4R_L^2]$. This uses $\norm{\eta} \leq R_L$, which holds for the \ac{OPD} ($\norm{\eta_P} = R_{\ell^*}$) and for the \ac{OGD} at high concentration $\kappa_\ell \geq \tfrac12$ or under the adaptive fallback (Definition~\ref{def:adaptive_cgr}); the \ac{OGD}/\ac{OPD} fixed points are in any case computed explicitly at high \ac{SNR} in Section~\ref{sec:asymptotic_gap}.
\end{enumerate}
\end{theorem}

\begin{proof}
\proofref{app:thm_se_convergence}
\end{proof}

\begin{proposition}[Monotonicity of the SE Map]
\label{prop:se_monotone}
Let $\mathcal{F}$ be the \ac{SE} map~\eqref{eq:fixed_point} of any denoiser that realizes the Bayes posterior mean for the effective channel $\bar{x} = x + \bar{\sigma}\tilde{z}$ under some prior -- the matched \ac{BOD} (true prior $p$), the \ac{OBD} (orbital prior $q$), or the \ac{LMMSE} (Gaussian prior of second moment $E_d$). Then $\mathcal{F}$ is non-decreasing on $[0, E_d]$. 
Consequently, the \ac{SE} iterates initialized at $\mathrm{MSE}_0 = E_d$ converge monotonically to the largest fixed point, which we denote $\mathrm{MSE}_\infty$ (with $\mathrm{MSE}_\infty^{\eta}$ for a specific denoiser $\eta$).
\end{proposition}

\begin{proof}
\proofref{app:prop_se_monotone}
\end{proof}

\begin{proposition}[Ordering of SE Fixed Points]
\label{prop:fp_ordering}
Let $\mathcal{F}_{\mathrm{D}}$, $\mathcal{F}_{\mathrm{B}}$, and $\mathcal{F}_{\mathrm{L}}$ denote the \ac{SE} maps of the \ac{BOD} (exact discrete, Bayes-optimal), \ac{OBD}, and \ac{LMMSE} estimators, and write $\mathrm{MSE}_\infty^{\eta} = \lim_{t\to\infty}\mathrm{MSE}_t$ for the corresponding largest fixed point of denoiser $\eta$ (iterates initialized at $\mathrm{MSE}_0 = E_d$). For any underloaded system with $\alpha = K/N < 1$ and every $\sigma_z^2 > 0$,
\begin{equation}
\label{eq:fp_ordering}
\mathrm{MSE}_\infty^{\mathrm{D}} \;\leq\; \mathrm{MSE}_\infty^{\mathrm{B}} \;\leq\; \mathrm{MSE}_\infty^{\mathrm{L}}.
\end{equation}
The \ac{OGD} and \ac{OPD} are deferred: once their fixed points are characterized in Section~\ref{sec:asymptotic_gap}, this ordering extends \emph{at high \ac{SNR}} to the full hierarchy $\mathrm{MSE}_\infty^{\mathrm{D}} \le \mathrm{MSE}_\infty^{\mathrm{B}} \le \mathrm{MSE}_\infty^{\mathrm{G}} \le \mathrm{MSE}_\infty^{\mathrm{P}} \le \mathrm{MSE}_\infty^{\mathrm{L}}$ (Corollary~\ref{cor:fp_ordering_hierarchy}); the extension is necessarily restricted to low noise, since the \ac{OGD}/\ac{OPD} fixed points are only {\emph{localized and ordered} at high \ac{SNR} (their existence on $[0, 4R_L^2]$ holds at all \ac{SNR} by Theorem~\ref{thm:se_convergence})}.
\end{proposition}

\begin{proof}
\proofref{app:prop_fp_ordering}
\end{proof}

\begin{remark}[Role of Monotonicity in the Comparison]
\label{rem:comparison_role}
The comparison lemma used in Proposition~\ref{prop:fp_ordering} critically requires
monotonicity of the \ac{SE} maps (Proposition~\ref{prop:se_monotone}).
Without monotonicity, the per-step ordering $\mathcal{F}_{\mathrm{D}}(\mathrm{MSE}) \leq \mathcal{F}_{\mathrm{B}}(\mathrm{MSE})$ would \emph{not} imply fixed-point ordering: a map with smaller values everywhere can have a larger fixed point if it oscillates or is non-monotone.
The Brouwer theorem (Theorem~\ref{thm:se_convergence}) guarantees existence but says nothing about ordering; the comparison lemma provides this ordering by exploiting the monotone structure specific to \ac{SE} maps arising from separable denoisers in \ac{AMP}.
This subtlety is precisely why the ordering claim requires a formal proof rather than the informal argument ``the optimal denoiser minimizes \ac{MSE} by definition.''
\end{remark}

\begin{corollary}[Uniqueness of the SE Fixed Point]
\label{cor:se_uniqueness}
The uniqueness claims below concern the posterior-mean denoisers of Proposition~\ref{prop:se_monotone}, namely the \ac{BOD}, \ac{OBD}, and \ac{LMMSE}; the \ac{OGD} and \ac{OPD}, not being posterior means, are excluded.
\begin{enumerate}[(i)]
  \item \emph{\ac{LMMSE} (all \ac{SNR}).}
    For every $\sigma_z^2 > 0$ and $\alpha < 1$, $\mathcal{F}_{\mathrm{L}}$ is a global contraction on $[0,E_d]$ and has a \emph{unique} fixed point.
  \item \emph{\ac{BOD} and \ac{OBD} (two regimes).}
    $\mathcal{F}_{\mathrm{D}}$ and $\mathcal{F}_{\mathrm{B}}$ admit a unique fixed point in $[0,E_d]$ in either of two regimes: (a) \emph{low \ac{SNR}}, $\sigma_z^2 \geq R_L^2$, where the map is a global Banach contraction; and (b) \emph{high \ac{SNR}}, $\sigma_z^2 < \sigma_{\mathrm{th}}^2(\alpha)$ for a geometry-dependent threshold $\sigma_{\mathrm{th}}^2 > 0$, where the map is locally contractive near $\mathrm{MSE}^* = \mathcal{O}(\sigma_z^2)$ and global uniqueness follows by monotone convergence (Proposition~\ref{prop:se_monotone}). {In the intervening band $\sigma_{\mathrm{th}}^2(\alpha) \le \sigma_z^2 < {R_L^2}$ global uniqueness is not asserted; the largest fixed point $\mathrm{MSE}_\infty$ (Proposition~\ref{prop:se_monotone}) nonetheless remains well-defined and is the one used throughout.}
\end{enumerate}
\end{corollary}

\begin{proof}
\proofref{app:cor_se_uniqueness}
\end{proof}

\begin{remark}[Scope of the Fixed-Point Theory Across the Hierarchy]
\label{rem:fp_scope}
Theorem~\ref{thm:se_convergence} settles existence for every denoiser, on $[0,E_d]$ for the posterior-mean denoisers (\ac{BOD}, \ac{OBD}) and the \ac{LMMSE}, on $[0,4R_L^2]$ for the \ac{OGD} and \ac{OPD}.
Monotonicity (Proposition~\ref{prop:se_monotone}) and uniqueness (Corollary~\ref{cor:se_uniqueness}) require only that the denoiser be a Bayes posterior mean, so that its \ac{SE} map is an \ac{MMSE}: this holds for the \ac{BOD} (true prior $p$), the \ac{OBD} (orbital prior $q$), and the \ac{LMMSE} (Gaussian prior).
For all three the map is non-decreasing by the \ac{I-MMSE} relation, so the iterates converge monotonically to the largest fixed point $\mathrm{MSE}_\infty$.
Uniqueness is regime-dependent: global for the \ac{LMMSE}, two-regime for the \ac{BOD} and \ac{OBD} (Corollary~\ref{cor:se_uniqueness} and Remark~\ref{rem:overloaded}).
The \ac{OGD} and \ac{OPD} are \emph{not} posterior means (the \ac{OGD} approximates the \ac{OBD}'s von Mises phase posterior by a Gaussian (Definition~\ref{def:gaussian_phase}), the \ac{OPD} is a hard ring projection) so neither the \ac{MMSE} monotonicity nor the contraction applies; their fixed points are instead characterized explicitly in Section~\ref{sec:asymptotic_gap} (Proposition~\ref{prop:cgr_se}, Corollary~\ref{cor:lpd_se}).
\end{remark}

\begin{remark}[Geometric Convergence Rate]
\label{rem:geometric_rate}
In any regime where $\mathcal{F}$ is a {global} contraction {on $[0,E_d]$} (Corollary~\ref{cor:se_uniqueness}), with modulus $c \triangleq \sup_{[0,E_d]}\norm{\mathcal{F}'} < 1$, the monotone convergence of Proposition~\ref{prop:se_monotone} is in fact geometric: since $\mathrm{MSE}_0 = E_d$, we have
\begin{equation*}
\norm{\mathrm{MSE}_t - \mathrm{MSE}_\infty} \;\leq\; c^t\,\norm{\mathrm{MSE}_0 - \mathrm{MSE}_\infty} \;\leq\; c^t E_d.
\end{equation*}
The rate is the modulus $c$, not $\alpha$: for the \ac{LMMSE} at all \ac{SNR} ({$c < \alpha$}) and for the \ac{BOD}/\ac{OBD} at high noise ({$c \leq \alpha$}) the map is globally contractive on $[0,E_d]$ (Corollary~\ref{cor:se_uniqueness}), giving $\norm{\mathrm{MSE}_t - \mathrm{MSE}_\infty} \leq \alpha^t E_d$.
For example, $\norm{\mathrm{MSE}_t - \mathrm{MSE}_\infty} \leq 0.25^t E_d$ at $\alpha = 0.25$. At high \ac{SNR}, the \ac{OBD} contraction is only local; once the iterates enter the basin, the asymptotic rate is the fixed-point slope $\mathcal{F}'(\mathrm{MSE}^*) \to \alpha/2$, smaller still.  In the intervening band, where contraction is not guaranteed, convergence remains monotone (Proposition~\ref{prop:se_monotone}) though possibly sub-geometric.
\end{remark}

\begin{remark}[Extension to $\alpha \geq 1$]
\label{rem:overloaded}
The contraction modulus of Remark~\ref{rem:geometric_rate} is $c = \sup_{[0,E_d]}\alpha\,\mathrm{d}\,\mathrm{mmse}_\varpi/\mathrm{d}\bar{\sigma}^2 \leq \sup_{[0,E_d]} 2\alpha\,\E[D^2]$ (with $D = \hat{\sigma}^2/\bar{\sigma}^2$ the normalized posterior variance; see the proof of Corollary~\ref{cor:se_uniqueness}), so the load $\alpha$ enters only through the explicit factor in this product.
At the \emph{critical} load $\alpha = 1$ the fixed-point theory still goes through: in the high-noise regime $D \leq R_L^2/\bar{\sigma}^2 \leq 1$ pointwise and $2\,\E[D^2] \leq 2\,\E[D] = 2\,\mathrm{mmse}_\varpi/\bar{\sigma}^2 < 2E_d/(E_d + \bar{\sigma}^2) \leq 1$ strictly (the strict Gaussian maximum-\ac{MMSE} bound together with $\bar{\sigma}^2 \geq R_L^2 \geq E_d$), while at high \ac{SNR} the slope $\mathrm{d}\,\mathrm{mmse}_\varpi/\mathrm{d}\bar{\sigma}^2 \to \tfrac12$ (\ac{OBD}) or $0$ (\ac{BOD}), whence $c < 1$ and $\mathcal{F}$ remains a strict contraction with existence, monotone convergence, and a unique fixed point intact.  What degrades at $\alpha = 1$ is the \ac{LMMSE} \emph{baseline}, whose fixed point $\sigma_z^2/(1-\alpha)$ diverges; the hierarchy fixed point $\sigma_z^2/(2-\alpha) = \sigma_z^2$ stays finite.
The obstruction is the \emph{strictly overloaded} regime $\alpha > 1$, where $c < 1$ no longer holds under $c \leq \alpha$, even at low \ac{SNR}, causing contraction to fail, so that multiple \ac{SE} fixed points may coexist (``good'' and ``bad'' convergence basins), a well-documented fact in the \ac{AMP} literature~\cite{Bayati2011}.
In any case, the underloaded assumption $\alpha < 1$ is adopted throughout, since the \ac{LMMSE} baseline and the capacity-achieving conclusions of Section~\ref{sec:inf_theory} require it.
\end{remark}

\vspace{-2ex}
\subsection{State-Evolution Fixed Points of the Hierarchy}
\label{sec:asymptotic_gap}

We now evaluate the \ac{SE} fixed point of each denoiser, descending the hierarchy from the matched \ac{BOD} to the \ac{LMMSE}.
The general theory of Section~\ref{sec:se_fixedpoint} guarantees that each fixed point exists; here, we compute its \emph{value}.
The recipe is uniform: equate $\mathrm{MSE}_\infty$ to the denoiser's per-symbol conditional error $\mathbb{E}_{x,\tilde{z}}\bigl[\norm{x-\eta(\bar{x};\bar{\sigma}^2)}^2\bigr]$ at effective noise $\bar{\sigma}^2$ via~\eqref{eq:se_noise}; i.e., $\bar{\sigma}^2 = \sigma_z^2 + \alpha\,\mathrm{MSE}_\infty$, and solve the resulting scalar equation.  For the posterior-mean \ac{BOD} and \ac{OBD} this error coincides with the expected posterior variance (the Nishimori identity); the \ac{OGD} and \ac{OPD} are not posterior means and additionally carry the orthogonality excess of Lemma~\ref{lem:excess}.

\subsubsection{Bayes-Optimal Denoiser}
The matched recursion follows from the Nishimori identity (see, e.g.,~\cite{Macris2007}), i.e.
\vspace{-0.5ex}
\begin{equation}
  \mathrm{MSE}_{t+1}^{\mathrm{D}} = \mathbb{E}_{x,\tilde{z}}\!\bigl[
  \hat{\sigma}^2(x + \bar{\sigma}_t \tilde{z};\;\bar{\sigma}_t^2)\bigr],
  \label{eq:se_optimal}
\vspace{-0.5ex}
\end{equation}
where $\hat{\sigma}^2(\bar{x};\bar{\sigma}^2)$ is the exact posterior variance~\eqref{eq:post_MMSE_var}; the identity holds because the matched \ac{BOD} returns the posterior mean $\eta_D(\bar{x};\bar{\sigma}^2) = \E[x \mid \bar{x}]$, so its conditional error coincides with the posterior variance.

For a discrete constellation, this variance decays \emph{exponentially} in the inverse noise, i.e., $\hat{\sigma}^2 = \mathcal{O}(e^{-c/\bar{\sigma}^2})$, with $c > 0$ governed by the minimum distance, so that the fixed point
\vspace{-0.5ex}
\begin{equation}
  \label{eq:bod_fixed_point}
  \mathrm{MSE}_\infty^{\mathrm{D}} = \mathcal{O}\bigl(e^{-c/\sigma_z^2}\bigr),
\vspace{-0.5ex}
\end{equation}
is exponentially small, vanishing faster than any power of $\sigma_z^2$.

This is the baseline against which the penalty of \ac{OD} is measured.

\textit{Justification:} The fixed point~\eqref{eq:bod_fixed_point} is the \emph{averaged} posterior variance $\mathbb{E}_{x,\tilde{z}}[\hat{\sigma}^2(\bar{x})]$ of~\eqref{eq:se_optimal}, so its decay rate is that of this average, not of the typical conditional variance, which decays faster.
Write $d_{\min} = \min_{s \neq s'}\norm{s-s'}$ for the constellation's minimum distance.
As $\bar{\sigma}^2 \to 0^+$ the posterior concentrates on the nearest point $s^*$ to $\bar{x}$, with only the nearest competitor $s'$ contributing; the two-point posterior has conditional variance
\begin{subequations}
\begin{equation*}
  \hat{\sigma}^2(\bar{x}) = w_{s^*} w_{s'}\, \norm{s^*-s'}^2,
\end{equation*}
\begin{equation*}
  \frac{w_{s'}}{w_{s^*}} = \exp\!\Bigl(-\tfrac{\norm{\bar{x}-s'}^2 - \norm{\bar{x}-s^*}^2}{\bar{\sigma}^2}\Bigr).
\end{equation*}
\end{subequations}
For $\bar{x} = s^* + \bar{\sigma}\tilde{z}$, $\tilde{z}\sim\mathcal{CN}(0,1)$, the log-likelihood gap is $\norm{\bar{x}-s'}^2 - \norm{\bar{x}-s^*}^2 = \norm{s^*-s'}^2 + 2\bar{\sigma}\,\mathrm{Re}((s^*-s')^*\tilde{z})$, so $\hat{\sigma}^2(\bar{x}) = \Theta(\norm{s^*-s'}^2)$ \emph{only} when $\bar{x}$ lies within $\mathcal{O}(\bar{\sigma})$ of the $s^*$-$s'$ boundary ($w_{s^*}\!\approx\!w_{s'}$) and is exponentially smaller otherwise.
The average is therefore dominated by these rare near-boundary observations: reaching the boundary requires the noise projection $\mathrm{Re}((s^*-s')^*\tilde{z})/\norm{s^*-s'} \sim \mathcal{N}(0,\tfrac12)$ to reach $\norm{s^*-s'}/(2\bar{\sigma})$, an event of probability $\Theta(e^{-\norm{s^*-s'}^2/(4\bar{\sigma}^2)})$, which is largest for the closest pair, at distance $d_{\min}$.
Hence $\mathbb{E}[\hat{\sigma}^2]$ decays as $e^{-d_{\min}^2/(4\bar{\sigma}^2)}$ (up to polynomial prefactors), establishing~\eqref{eq:bod_fixed_point} with $c = d_{\min}^2/4$, an exponent of the same minimum-distance form $d^2/(4\bar{\sigma}^2)$ as the ring-discrimination bound (Proposition~\ref{prop:ring_discrimination}), there with the inter-ring separation $d_R$ in place of $d_{\min}$.

\subsubsection{Orbital Bessel Denoiser}
Under the continuous orbital relaxation, the per-iteration \ac{MSE} obeys the mismatched recursion
\begin{equation}
  \mathrm{MSE}_{t+1}^{\mathrm{B}} = \mathbb{E}_{x,\tilde{z}}\!\Bigl[
  \norm{x - \eta_B\bigl(x + \bar{\sigma}_t \tilde{z};\;\bar{\sigma}_t^2\bigr)
  }^2\Bigr],
  \label{eq:se_mismatch}
\end{equation}
with $\eta_B$ computed via~\eqref{eq:obd_denoiser} and the expectation evaluated by Monte Carlo over the discrete prior $x$ and Gaussian noise $\tilde{z}$.
The \ac{OGD} and \ac{OPD} obey the \emph{same} recursion with $\eta_B$ replaced by $\eta_G$ or $\eta_P$, as all three act on the identical cavity statistic through the common concentration $\kappa_\ell$ and differ only in how they approximate $A(\kappa_\ell)$ and $\ln I_0(\kappa_\ell)$. 

Setting $\mathrm{MSE}_{t+1}^{\mathrm{B}} = \mathrm{MSE}_t^{\mathrm{B}}$ in~\eqref{eq:se_mismatch} characterizes the \ac{OBD} fixed point \emph{exactly}, as the solution of $\mathrm{MSE}_\infty^{\mathrm{B}} = \mathrm{mmse}_q(\sigma_z^2 + \alpha\,\mathrm{MSE}_\infty^{\mathrm{B}})$.  This equation is transcendental since $\mathrm{mmse}_q$ involves the Bessel ratio $A(\kappa) = I_1(\kappa)/I_0(\kappa)$ averaged against the Rician distribution of $\norm{\bar{x}}$ and has no elementary closed form, so it admits no closed-form root and is solved numerically (Fig.~\ref{fig:MSE_MPSK_MQAM} below).  A closed-form value instead follows from the high-\ac{SNR} decay of the posterior variance, which we now establish.

\begin{proposition}[Linear Posterior Variance Decay Under OBD]
  \label{prop:linear_decay}
  As the cavity noise variance vanishes ($\bar{\sigma}^2 \to 0^+$), the \ac{OBD} posterior variance conditioned on the dominant ring $\ell^*$ decays linearly with respect to $\bar{\sigma}^2$, satisfying
  \begin{equation}
    \hat{\sigma}_{\mathrm{B},\ell^*}^2(\bar{\sigma}^2) \triangleq \mathrm{Var}[x \mid \bar{x},\, R_{\ell^*}] \sim \frac{R_{\ell^*}\,\bar{\sigma}^2}{2\norm{\bar{x}}}.
    \label{eq:cbm_var_decay}
  \end{equation}

\end{proposition}

\begin{proof}
\proofref{app:prop_linear_decay}
\end{proof}

At the \ac{SE} fixed point, we have that $\norm{\bar{x}} \to R_{\ell^*}$, so the conditional decay of Proposition~\ref{prop:linear_decay} reduces to $\hat{\sigma}^2_{\mathrm{B}} \approx \bar{\sigma}^2/2$ for the total \ac{OBD} posterior variance. The factor $\tfrac{1}{2}$ has a transparent origin: once the radius is pinned to a ring, only the phase (one of the two real degrees of freedom of the complex noise) remains uncertain, so the denoiser removes exactly half the effective noise. Substituting this into the \ac{SE} fixed-point relation pins down the \ac{OBD} fixed point derived next, the leading-order anchor that the \ac{OGD} and \ac{OPD} are later shown to share.

\begin{proposition}[\ac{OBD} State Evolution Fixed Point]
\label{prop:obd_se}
For any underloaded system ($\alpha = K/N < 1$, $\sigma_z^2 > 0$), let $\gamma \triangleq \sum_{\ell=1}^{L} r_\ell\, R_\ell^{-2}$ be the constellation's inverse-energy coefficient. Then, the mismatched \ac{SE}~\eqref{eq:se_mismatch} under the \ac{OBD} has fixed point
\begin{subequations}
\begin{equation}
  \label{eq:obd_fixed_point}
  \mathrm{MSE}_\infty^{\mathrm{B}} = \frac{\sigma_z^2}{2 - \alpha} + \delta^{\mathrm{B}} + \mathcal{O}(\sigma_z^6),
\end{equation}
\begin{equation}
  \delta^{\mathrm{B}} \triangleq \frac{\gamma\,\sigma_z^4}{(2-\alpha)^3}.
\end{equation}
\end{subequations}
\end{proposition}

\begin{proof}
By the Bayes-optimality of $\eta_B$ under $q$, the \ac{OBD} posterior variance $\hat{\sigma}^2_{\mathrm{B}}$ equals the actual per-iteration error, so the fixed point solves $\mathrm{MSE}_\infty^{\mathrm{B}} = \E[\hat{\sigma}^2_{\mathrm{B}}]$ with $\bar{\sigma}^2 = \sigma_z^2 + \alpha\,\mathrm{MSE}_\infty^{\mathrm{B}}$.
We evaluate $\E[\hat{\sigma}^2_{\mathrm{B}}]$ from the per-ring posterior variance and solve this equation order by order in $\bar{\sigma}^2$, the leading-order fixed point arising as its first term.
Conditioned on any ring $\ell$, the von Mises posterior variance is $\hat{\sigma}^2_{\mathrm{B},\ell} = R_\ell^2\bigl(1 - A(\kappa_\ell)^2\bigr)$ from $\E[\norm{x}^2 \mid R_\ell] = R_\ell^2$ and $\norm{\E[x \mid \bar{x}, R_\ell]}^2 = R_\ell^2 A(\kappa_\ell)^2$ via~\eqref{eq:ring_mean_cbm}, of which Proposition~\ref{prop:linear_decay} is the dominant-ring ($\ell = \ell^*$) instance. 
The second-order term of $A(\kappa) = 1 - \tfrac{1}{2\kappa} - \tfrac{1}{8\kappa^2} + \mathcal{O}(\kappa^{-3})$ (Proposition~\ref{prop:vm_gaussian}) cancels in the $\kappa_\ell^{-2}$ term, so $1 - A(\kappa_\ell)^2 = 1/\kappa_\ell + \mathcal{O}(\kappa_\ell^{-3})$ and $\hat{\sigma}^2_{\mathrm{B},\ell} = R_\ell\bar{\sigma}^2/(2\norm{\bar{x}}) + \mathcal{O}((\bar{\sigma}^2)^3)$. 
At high \ac{SNR}, the posterior concentrates on the correct ring (misdetection exponentially rare, Proposition~\ref{prop:ring_discrimination}), so the dominant ring is the transmitted one, distributed with prior $r_\ell$; averaging over it and over the Gaussian fluctuation\footnote{Write $\bar{x} = x + \bar{z}$ with $\norm{x} = R_\ell$ and resolve $\bar{z}$ into radial and tangential components $n_\parallel, n_\perp \sim \mathcal{N}(0, \bar{\sigma}^2/2)$ relative to $x$, so that $\norm{\bar{x}} = \sqrt{(R_\ell + n_\parallel)^2 + n_\perp^2}$. The second-order Taylor expansion $R_\ell/\norm{\bar{x}} = 1 - n_\parallel/R_\ell + (n_\parallel^2 - \tfrac{1}{2}n_\perp^2)/R_\ell^2 + \mathcal{O}(\bar{\sigma}^3)$ has zero-mean odd-order terms, and $\E[n_\parallel^2] = \E[n_\perp^2] = \bar{\sigma}^2/2$ then give $\E[R_\ell/\norm{\bar{x}}] = 1 + \bar{\sigma}^2/(4R_\ell^2) + \mathcal{O}(\bar{\sigma}^4)$.} of $\norm{\bar{x}}$ about $R_\ell$ with $\E[R_\ell/\norm{\bar{x}}] = 1 + \bar{\sigma}^2/(4R_\ell^2) + \mathcal{O}(\bar{\sigma}^4)$ gives the ring average
\begin{equation}
  \E[\hat{\sigma}^2_{\mathrm{B}}]
  \!=\! \sum_{\ell=1}^{L} r_\ell\,\frac{\bar{\sigma}^2}{2}\,\E\!\Bigl[\frac{R_\ell}{\norm{\bar{x}}}\Bigr]
  \!=\! \frac{\bar{\sigma}^2}{2}\sum_{\ell} r_\ell\Bigl(1 + \frac{\bar{\sigma}^2}{4R_\ell^2}\Bigr) + \mathcal{O}\bigl((\bar{\sigma}^2)^3\bigr),
\end{equation}
such that with $\sum_\ell r_\ell = 1$ and $\gamma = \sum_\ell r_\ell R_\ell^{-2}$, we have
\begin{equation}
  \label{eq:obd_Fmap}
  \E[\hat{\sigma}^2_{\mathrm{B}}] = \frac{\bar{\sigma}^2}{2} + \frac{\gamma\,(\bar{\sigma}^2)^2}{8} + \mathcal{O}\!\bigl((\bar{\sigma}^2)^3\bigr).
\end{equation}

Inserting~\eqref{eq:obd_Fmap} into this fixed-point equation and solving order by order, the leading term retains only $\E[\hat\sigma^2_{\mathrm{B}}] = \bar{\sigma}^2/2$, so $\mathrm{MSE}_\infty^{\mathrm{B}} = \bar{\sigma}^2/2$ combined with $\bar{\sigma}^2 = \sigma_z^2 + \alpha\,\mathrm{MSE}_\infty^{\mathrm{B}}$ gives $\bar{\sigma}^2(1-\alpha/2) = \sigma_z^2$, i.e., $\bar{\sigma}^2 = 2\sigma_z^2/(2-\alpha)$ and $\mathrm{MSE}_\infty^{\mathrm{B}} = \sigma_z^2/(2-\alpha)$; writing $\mathrm{MSE}_\infty^{\mathrm{B}} = \sigma_z^2/(2-\alpha) + \delta$ and linearizing the self-consistency, the perturbation solves $\delta(1 - \alpha/2) = \gamma(\bar{\sigma}^2)^2/8$ evaluated at the leading value (the factor $1 - \alpha/2 = 1 - \mathcal{F}_{\mathrm{B}}'$ supplies the fixed-point amplification $2/(2-\alpha)$) so that, at $\bar{\sigma}^2 = 2\sigma_z^2/(2-\alpha)$, we have
\vspace{-1ex}
\begin{equation}
  \delta = \frac{1}{1\!-\!\alpha/2}\cdot\frac{\gamma}{8}\Bigl(\frac{2\sigma_z^2}{2\!-\!\alpha}\Bigr)^{\!2}\!\!
  = \! \frac{2}{2\!-\!\alpha}\cdot\frac{\gamma\,\sigma_z^4}{2(2\!-\!\alpha)^2}
  = \frac{\gamma\,\sigma_z^4}{(2\!-\!\alpha)^3},
\vspace{-1ex}
\end{equation}
which yields
\vspace{-1ex}
\begin{equation}
  \mathrm{MSE}_\infty^{\mathrm{B}} = \frac{\sigma_z^2}{2-\alpha} + \frac{\gamma\,\sigma_z^4}{(2-\alpha)^3} + \mathcal{O}(\sigma_z^6),
 \vspace{-1ex} 
\end{equation}
that is, $\delta^{\mathrm{B}} = \gamma\sigma_z^4/(2-\alpha)^3$, {the $\mathcal{O}(\sigma_z^6)$ being the remainder displayed in~\eqref{eq:obd_fixed_point}}.
\end{proof}
\vspace{-1ex}

The \ac{OGD} and \ac{OPD} differ from the \ac{OBD} only in their amplitude-shrinkage factor; the following identity, a direct consequence of the \ac{OBD}'s (under $q$) Bayes-optimality, quantifies the resulting excess error.

\begin{lemma}[Excess Error over the \ac{OBD}]
\label{lem:excess}
Let $\eta$ be any phase-preserving denoiser. Since the \ac{OBD} $\eta_B = \E_q[x\mid\bar{x}]$ is the Bayes posterior mean under the orbital prior $q$ and $\eta - \eta_B$ is a function of $\bar{x}$, the orthogonality principle makes the cross term vanish, giving
\begin{equation}
\label{eq:excess}
\hspace{-1ex}\mathcal{F}_\eta(\mathrm{MSE}) = \mathcal{F}_{\mathrm{B}}(\mathrm{MSE}) + \E\bigl[\,\norm{\eta_B - \eta}^2\,\bigr] \geq \mathcal{F}_{\mathrm{B}}(\mathrm{MSE}),
\end{equation}
for every \ac{SNR}. 

The cross-term cancellation holds under $q$ and uses only that $\eta_B$ is the $q$-posterior mean, not phase-preservation. Phase-preservation enters solely in transferring the identity to the true data-generating prior $p$: by Lemma~\ref{lem:phase_preserving}(iii) the \ac{MSE} of a phase-preserving denoiser depends on the input only through the radial marginal $\{R_\ell, r_\ell\}$ shared by $p$ and $q$, so $\mathcal{F}_\eta$, $\mathcal{F}_{\mathrm{B}}$, and the radial excess $\E\norm{\eta_B - \eta}^2$ each take the same value under both priors, and~\eqref{eq:excess} holds verbatim under $p$.

At high \ac{SNR}, the posterior concentrates on the dominant ring $\ell^*$ (the softmax weights satisfy $w_{\ell^*}\to1$ as the log-evidence gap diverges, cf.\ Proposition~\ref{prop:lpd_limit}), where $\eta = R_{\ell^*}\,m_\eta(\kappa_{\ell^*})\,e^{j\angle\bar{x}}$ and $\eta_B = R_{\ell^*}\,A(\kappa_{\ell^*})\,e^{j\angle\bar{x}}$, so the excess is the squared magnitude-approximation bias given by
\begin{equation}
\label{eq:excess_ring}
\E\bigl[\,\norm{\eta_B - \eta}^2\,\bigr] = \E\bigl[\,R_{\ell^*}^2\,(A(\kappa_{\ell^*}) - m_\eta)^2\,\bigr],
\end{equation}
with $m_{\mathrm{B}} = A(\kappa)$, $m_{\mathrm{G}} = 1 - 1/(2\kappa)$, and $m_{\mathrm{P}} = 1$. 
\end{lemma}

\subsubsection{Orbital Gaussian Denoiser}
The Gaussian phase approximation yields an exact closed form for the per-ring variance (Proposition~\ref{prop:cgr_variance}); the \ac{OGD} fixed point itself follows from the excess identity of Lemma~\ref{lem:excess}.

\begin{proposition}[Exact \ac{OGD} Per-Ring Variance]
\label{prop:cgr_variance}
Under the Gaussian phase model, the conditional variance on ring~$\ell$ is
\begin{equation}
\label{eq:cgr_ring_var}
\hat{\sigma}_{\mathrm{G},\ell}^2 = R_\ell^2\!\left(\frac{1}{\kappa_\ell} - \frac{1}{4\kappa_\ell^2}\right) = \frac{R_\ell\, \bar{\sigma}^2}{2\norm{\bar{x}}} - \frac{(\bar{\sigma}^2)^2}{16\norm{\bar{x}}^2}.
\end{equation}

In the high-\ac{SNR} limit ($\kappa_\ell \to \infty$), this recovers the \ac{OBD} linear decay $\hat{\sigma}_{\mathrm{B},\ell}^2 \sim R_\ell\bar{\sigma}^2/(2\norm{\bar{x}})$ of Proposition~\ref{prop:linear_decay}.
\end{proposition}

\begin{proof}
The conditional second moment on ring~$\ell$ is $\E[\norm{x}^2 \mid R_\ell] = R_\ell^2$ (the radius is deterministic).
The \ac{OGD} conditional mean carries the magnitude factor $m_{\mathrm{G}} = 1 - 1/(2\kappa_\ell)$ (the leading two terms of $A(\kappa_\ell)$ (Proposition~\ref{prop:vm_gaussian}), i.e.~\eqref{eq:cgr_mean}, and \emph{not} the exact Gaussian resultant $e^{-1/(2\kappa_\ell)}$) so its squared magnitude is

\begin{align}
\norm{\E_{\mathrm{G}}[x \mid \bar{x},\, R_\ell]}^2 &= R_\ell^2\,m_{\mathrm{G}}^2 = R_\ell^2\bigl(1 - 1/(2\kappa_\ell)\bigr)^2\nonumber \\
&= R_\ell^2\bigl(1 - 1/\kappa_\ell + 1/(4\kappa_\ell^2)\bigr).
\end{align}

Therefore $\hat{\sigma}_{\mathrm{G},\ell}^2 = R_\ell^2(1 - m_{\mathrm{G}}^2) = R_\ell^2\bigl(1/\kappa_\ell - 1/(4\kappa_\ell^2)\bigr)$. Substituting $\kappa_\ell = 2R_\ell\norm{\bar{x}}/\bar{\sigma}^2$ yields~\eqref{eq:cgr_ring_var}. The result is thus exact for the \ac{OGD} (built on $m_{\mathrm{G}}$) and agrees with the true von Mises variance $R_\ell^2(1 - A(\kappa_\ell)^2)$ to leading order $R_\ell^2/\kappa_\ell$, differing only at $\mathcal{O}(\kappa_\ell^{-2})$.
\end{proof}

\begin{remark}[Nominal Variance versus Per-Symbol Error]
\label{rem:ogd_nominal_var}
Since $\eta_G$ is not a posterior mean, the quantity~\eqref{eq:cgr_ring_var} is the \emph{nominal} variance reported by the \ac{OGD} (the one feeding the Onsager divergence), not the per-symbol error that drives the \ac{SE}.
The latter is $\E\norm{x - \eta_G}^2$, whose per-ring value $R_\ell^2(1 - 2m_{\mathrm{G}}A(\kappa_\ell) + m_{\mathrm{G}}^2) = R_\ell^2/\kappa_\ell + \mathcal{O}(\kappa_\ell^{-3})$ has its $\kappa_\ell^{-2}$ term cancel exactly; the \ac{OGD}-\ac{OBD} error gap is the squared bias $R_\ell^2(A - m_{\mathrm{G}})^2 = \mathcal{O}(\kappa_\ell^{-4})$ of Lemma~\ref{lem:excess}, two orders smaller than the $\mathcal{O}(\kappa_\ell^{-2})$ discrepancy in~\eqref{eq:cgr_ring_var}, which must therefore not be read as a fixed-point gap.
\end{remark}

\begin{proposition}[\ac{OGD} State Evolution]
\label{prop:cgr_se}
The mismatched \ac{SE} under the \ac{OGD} follows the recursion
\begin{equation}
\label{eq:cgr_se}
\mathrm{MSE}_{t+1}^{\mathrm{G}} = \E_{x,\tilde{z}}\Bigl[\norm{x - \eta_G\bigl(x + \bar{\sigma}_t \tilde{z};\, \bar{\sigma}_t^2\bigr)}^2\Bigr],
\end{equation}
where $\eta_G(\bar{x};\, \bar{\sigma}^2) \triangleq \sum_{\ell=1}^{L} \hat{x}_\ell^{\mathrm{G}}$ is the \ac{OGD}. 

Its fixed point satisfies
\begin{equation}
\label{eq:cgr_fixed_point}
\mathrm{MSE}_\infty^{\mathrm{G}} = \frac{\sigma_z^2}{2 - \alpha} + \delta^{\mathrm{B}} + \mathcal{O}(\sigma_z^6),
\end{equation}
coinciding with the \ac{OBD} fixed point not merely at leading order but through its $\mathcal{O}(\sigma_z^4)$ correction $\delta^{\mathrm{B}}$ of Proposition~\ref{prop:obd_se} (the $\mathcal{O}(\sigma_z^6)$ remainder is inherited from the \ac{OBD} expansion~\eqref{eq:obd_fixed_point} itself); the \ac{OGD} departs from the \ac{OBD} only at $\mathcal{O}(\sigma_z^8)$: $\mathrm{MSE}_\infty^{\mathrm{G}} - \mathrm{MSE}_\infty^{\mathrm{B}} = \mathcal{O}(\sigma_z^8)$.
\end{proposition}

\begin{proof}
The \ac{OGD} uses the magnitude factor $m_{\mathrm{G}} = 1 - 1/(2\kappa)$ in place of the \ac{OBD}'s optimal $m_{\mathrm{B}} = A(\kappa)$. By Lemma~\ref{lem:excess}, the \ac{OGD} and \ac{OBD} \ac{SE} \emph{maps} differ by the squared magnitude bias
\begin{equation*}
\mathcal{F}_{\mathrm{G}}(\mathrm{MSE}) = \mathcal{F}_{\mathrm{B}}(\mathrm{MSE}) + \E\bigl[R_{\ell^*}^2\,(A(\kappa) - m_{\mathrm{G}})^2\bigr].
\end{equation*}

Since $A(\kappa) - m_{\mathrm{G}} = -1/(8\kappa^2) + \mathcal{O}(\kappa^{-3})$ and $\kappa = \Theta(1/\sigma_z^2)$ at the fixed point, this map gap is $\mathcal{O}(\kappa^{-4}) = \mathcal{O}(\sigma_z^8)$. Propagated through the fixed-point first-order perturbation, whose contraction factor satisfies $1 - \mathcal{F}_{\mathrm{B}}' = \Theta(1)$ (Corollary~\ref{cor:se_uniqueness}), the two fixed points differ by the same order, $\delta^{\mathrm{G}} - \delta^{\mathrm{B}} = \E[R_{\ell^*}^2(A-m_{\mathrm{G}})^2]/(1-\mathcal{F}_{\mathrm{B}}') = \mathcal{O}(\sigma_z^8)/\Theta(1) = \mathcal{O}(\sigma_z^8)$, $\mathrm{MSE}_\infty^{\mathrm{G}} = \mathrm{MSE}_\infty^{\mathrm{B}} + \mathcal{O}(\sigma_z^8)$. With $\mathrm{MSE}_\infty^{\mathrm{B}} = \sigma_z^2/(2-\alpha) + \delta^{\mathrm{B}}$ (Proposition~\ref{prop:obd_se}), this gives~\eqref{eq:cgr_fixed_point}.
\end{proof}

\begin{remark}[Hardware Implications]
\label{rem:hardware}
For \ac{ASIC} implementations targeting high-throughput 5G-NR or 6G receivers, the \ac{OGD} eliminates the most complex hardware block in the \ac{OBD} pipeline: the Bessel ratio lookup table. 
Since massive \ac{MIMO} systems typically operate at \ac{SNR}~$\geq 5$~dB (where $\kappa_\ell \geq 3$ for normalized constellations), the \ac{OGD} is the preferred implementation for production silicon. 
The adaptive strategy of Definition~\ref{def:adaptive_cgr} provides a graceful fallback for the rare low-$\kappa$ events during early iterations or at cell-edge \ac{SNR} conditions.
\end{remark}

\subsubsection{Orbital Phase Denoiser}
The single-ring projection retains only the phase residual on the detected ring; conditioned on correct ring detection it attains the same leading-order variance.

\begin{proposition}[\ac{OPD} Variance at the SE Fixed Point]
\label{prop:lpd_mse}
At the \ac{SE} fixed point with effective noise $\bar{\sigma}_\infty^2$ (that of the \ac{OPD} recursion; the per-level convention is fixed in Section~\ref{sec:inf_theory}), the \ac{MSE} of the \ac{OPD} conditioned on correct ring detection ($\ell^* = \ell_{\mathrm{true}}$) is
\begin{align}
\label{eq:lpd_mse}
\mathrm{MSE}_{\mathrm{P}} &= \E\bigl[\norm{x - R_{\ell^*} e^{j\angle\bar{x}}}^2 \;\big|\; \ell^* = \ell_{\mathrm{true}}\bigr] \nonumber \\
&= R_{\ell^*}^2 \cdot \E\bigl[\norm{e^{j\Delta\theta} - 1}^2\bigr],
\end{align}
where $\Delta\theta = \angle\bar{x} - \angle x$ is the phase estimation error. 

Under the orbital (von Mises) posterior the phase residual has mean resultant $\E[\cos\Delta\theta] = \E[A(\kappa_{\ell^*})]$, so $\mathrm{MSE}_{\mathrm{P}} = 2R_{\ell^*}^2\,\E[1 - A(\kappa_{\ell^*})]$; via $1-A^2 = (1-A)(1+A)$ this equals exactly the \ac{OBD} variance plus the squared shrinkage bias of Lemma~\ref{lem:excess}, consistent with Corollary~\ref{cor:lpd_se}. Since $1 - A(\kappa) = 1/(2\kappa) + \mathcal{O}(\kappa^{-2})$ (Proposition~\ref{prop:vm_gaussian}) and $\kappa_{\ell^*} = 2R_{\ell^*}\norm{\bar{x}}/\bar{\sigma}^2$ with $\norm{\bar{x}} \approx R_{\ell^*}$, we have
\begin{equation}
\label{eq:lpd_mse_gaussian}
\mathrm{MSE}_{\mathrm{P}} = \frac{R_{\ell^*}^2}{\kappa_{\ell^*}} + \mathcal{O}\bigl((\bar{\sigma}^2)^2\bigr) = \frac{R_{\ell^*} \bar{\sigma}^2}{2\norm{\bar{x}}} + \mathcal{O}\bigl((\bar{\sigma}^2)^2\bigr) \approx \frac{\bar{\sigma}^2}{2},
\end{equation}
matching the leading-order \ac{OGD}/\ac{OBD} variance $\bar{\sigma}^2/2$ (Propositions~\ref{prop:linear_decay},~\ref{prop:cgr_variance}); the \ac{OPD}'s strictly larger sub-leading correction is quantified in Corollary~\ref{cor:lpd_se}.
\end{proposition}

\begin{proof}
Conditioned on correct ring detection, $x = R_{\ell^*} e^{j\angle x}$ and the \ac{OPD} output is $\hat{x} = R_{\ell^*} e^{j\angle\bar{x}}$. Therefore, we have
\begin{align}
\norm{x - \hat{x}}^2 &= R_{\ell^*}^2 \norm{e^{j\angle x} - e^{j\angle\bar{x}}}^2 = R_{\ell^*}^2 \norm{1 - e^{j\Delta\theta}}^2\nonumber \\
&= 2R_{\ell^*}^2 (1 - \cos\Delta\theta).
\end{align}

By the circular-symmetry transfer (Lemma~\ref{lem:phase_preserving}(iii)) this error equals its value under the orbital prior $q$, whose phase posterior is von Mises with mean resultant $A(\kappa_{\ell^*})$; hence $\E[1 - \cos\Delta\theta] = \E[1 - A(\kappa_{\ell^*})] = 1/(2\kappa_{\ell^*}) + \mathcal{O}(\kappa_{\ell^*}^{-2})$ (Proposition~\ref{prop:vm_gaussian}), and not the Gaussian value $1 - e^{-1/(2\kappa_{\ell^*})}$, whose $\mathcal{O}(\kappa^{-2})$ term carries the opposite sign.
Substituting $\kappa_{\ell^*} = 2R_{\ell^*}\norm{\bar{x}}/\bar{\sigma}^2$ and using $\norm{\bar{x}} \approx R_{\ell^*}$ at high \ac{SNR} gives $\mathrm{MSE}_{\mathrm{P}} \approx \bar{\sigma}^2/2$.
\end{proof}

\begin{corollary}[\ac{OPD} State Evolution Fixed Point]
\label{cor:lpd_se}
The \ac{SE} fixed point under the \ac{OPD} satisfies
\begin{equation}
\label{eq:lpd_fixed_point}
\mathrm{MSE}_\infty^{\mathrm{P}} = \frac{\sigma_z^2}{2 - \alpha} + \tfrac{3}{2}\,\delta^{\mathrm{B}} + \mathcal{O}(\sigma_z^6) + \mathcal{O}\!\bigl(e^{-c/\sigma_z^2}\bigr),
\end{equation}
sharing the \ac{OBD}/\ac{OGD} leading order $\sigma_z^2/(2-\alpha)$ but with a \emph{strictly larger} correction: using $m_{\mathrm{P}} = 1$ (no amplitude shrinkage) rather than the optimal $A(\kappa)$, the \ac{OPD} incurs the squared-bias excess $\E[R_{\ell^*}^2(1-A(\kappa))^2] = \Theta(\sigma_z^4)$ (Lemma~\ref{lem:excess}), which through the fixed-point first-order perturbation inflates the correction by one half, $\delta^{\mathrm{P}} = \tfrac{3}{2}\,\delta^{\mathrm{B}}$, plus an exponentially small ring-misdetection term.
\end{corollary}

\begin{proof}
On correct ring detection the \ac{OPD} output $\eta_P = R_{\ell^*}e^{j\angle\bar{x}}$ has magnitude factor $m_{\mathrm{P}} = 1$, so Lemma~\ref{lem:excess} gives $\mathcal{F}_{\mathrm{P}}(\mathrm{MSE}) = \mathcal{F}_{\mathrm{B}}(\mathrm{MSE}) + \E[R_{\ell^*}^2(1-A(\kappa))^2]$.

{Since $1 - A(\kappa) = 1/(2\kappa) + \mathcal{O}(\kappa^{-2})$ with $\kappa_{\ell^*} = 2R_{\ell^*}\norm{\bar{x}}/\bar{\sigma}^2$ and $\norm{\bar{x}} \approx R_{\ell^*}$, the per-ring excess averages to $\E[R_{\ell^*}^2(1-A(\kappa))^2] = \gamma(\bar{\sigma}^2)^2/16 + \mathcal{O}((\bar{\sigma}^2)^3)$, which is exactly \emph{half} the \ac{OBD} curvature $\gamma(\bar{\sigma}^2)^2/8$ of~\eqref{eq:obd_Fmap}. The \ac{OPD} error map $\mathcal{F}_{\mathrm{P}}$ therefore has quadratic coefficient $\gamma/8 + \gamma/16 = 3\gamma/16$, and since the fixed-point first-order perturbation is linear in this coefficient, $\delta^{\mathrm{P}}/\delta^{\mathrm{B}} = (3\gamma/16)/(\gamma/8) = 3/2$, i.e. $\delta^{\mathrm{P}} = \tfrac{3}{2}\delta^{\mathrm{B}} = \Theta(\sigma_z^4)$.} Ring misdetection ($\ell^* \neq \ell_{\mathrm{true}}$) has probability $\mathcal{O}(e^{-c/\sigma_z^2})$ (Proposition~\ref{prop:ring_discrimination}), contributing the exponentially small term (explicitly, $c = (2-\alpha)d_R^2/8$, from the bound $\Pr(\ell^*\neq\ell_{\mathrm{true}}) \le (L-1)e^{-d_R^2/(4\bar{\sigma}^2)}$ evaluated at the leading value $\bar{\sigma}^2 = 2\sigma_z^2/(2-\alpha)$). With $\mathrm{MSE}_\infty^{\mathrm{B}} = \sigma_z^2/(2-\alpha) + \delta^{\mathrm{B}}$ {(Proposition~\ref{prop:obd_se})} this yields~\eqref{eq:lpd_fixed_point}.
\end{proof}
\vspace{-3ex}
\begin{figure}[H]
    \centering
    \includegraphics[width=\columnwidth]{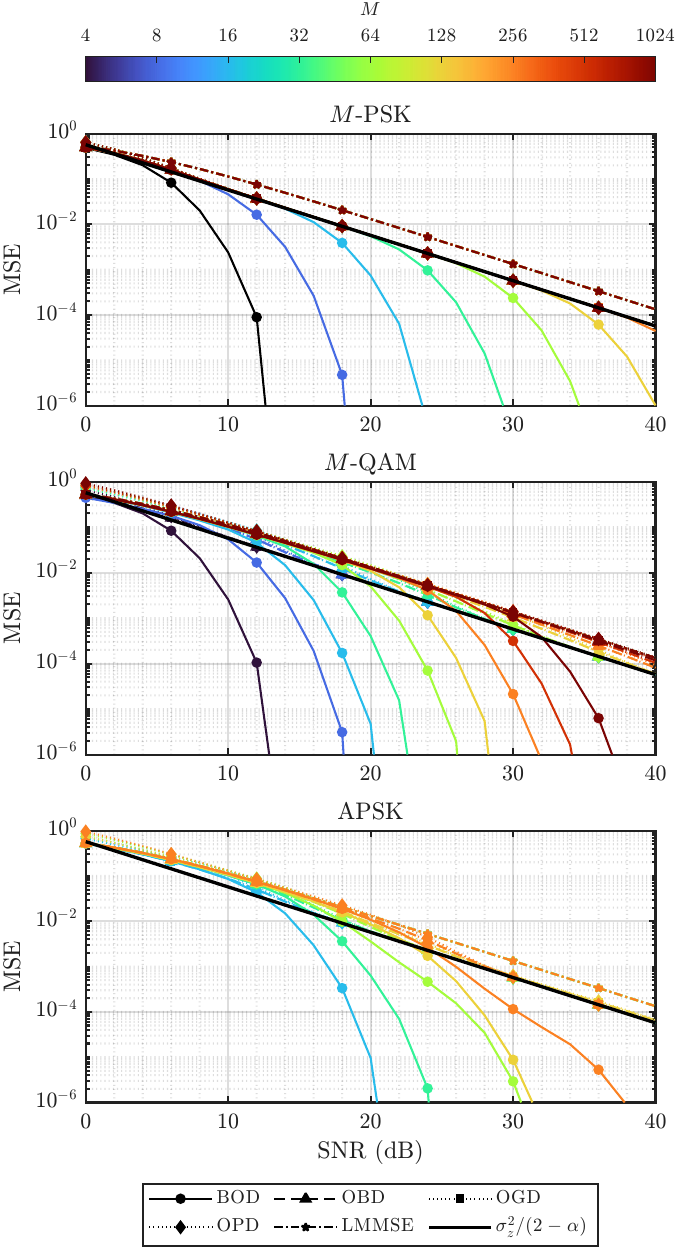}
    \vspace{-3ex}
    \caption{Fixed-point $\mathrm{MSE}_\infty$ versus \ac{SNR} at $\alpha = 0.25$ for all $M$-\ac{PSK} and $M$-\ac{QAM} orders $M \in \{4,\ldots,1024\}$ and the DVB-S2/S2x \ac{APSK} constellations $M \in \{16,\ldots,256\}$; line style encodes the detector (shared legend), color the order $M$ (shared colorbar). The \ac{OBD}/\ac{OGD}/\ac{OPD} relaxations coincide with the shared law $\sigma_z^2/(2-\alpha)$ (black dotted); the \ac{BOD} decays exponentially; the \ac{LMMSE} saturates at $\sigma_z^2/(1-\alpha)$.}
    \label{fig:MSE_MPSK_MQAM}
    \vspace{-2ex}
\end{figure}

\subsubsection{LMMSE Baseline}
The linear estimator converges to
\begin{equation}
  \mathrm{MSE}_{t+1}^{\mathrm{L}} =
  \frac{E_d\,\bar{\sigma}_t^2}{E_d + \bar{\sigma}_t^2},
  \label{eq:se_lmmse}
\end{equation}
whose fixed point couples to~\eqref{eq:se_noise} through $\bar{\sigma}^2 = \sigma_z^2 + \alpha\,\mathrm{MSE}_\infty^{\mathrm{L}}$.
At high \ac{SNR} ($\bar{\sigma}^2 \ll E_d$) the gain saturates, $\mathrm{MSE}_\infty^{\mathrm{L}} \approx \bar{\sigma}^2$, giving
\begin{equation}
  \label{eq:lmmse_fixed_point}
  \mathrm{MSE}_\infty^{\mathrm{L}} \approx \frac{\sigma_z^2}{1 - \alpha},
\end{equation}
strictly larger than the orbital $\sigma_z^2/(2-\alpha)$ and consistent with the ordering of Proposition~\ref{prop:fp_ordering}.

\medskip
Because the continuous phase relaxation prevents the probability mass from collapsing into a discrete Dirac delta until $\bar{\sigma}^2 = 0$ exactly, the orbital residual variance decays linearly rather than exponentially.
As the next subsection shows, this linear decay manifests macroscopically as a constant asymptotic \ac{SNR} penalty, a parallel shift in the \ac{BER} curve relative to the \ac{BOD}, rather than an absolute asymptotic performance penalty.

Figure~\ref{fig:MSE_MPSK_MQAM} verifies each of these fixed-point computations by Monte Carlo evaluation of the mismatched \ac{SE} recursion~\eqref{eq:se_noise}-\eqref{eq:se_mismatch} at load $\alpha = 0.25$, {across all orders of the three constellation families}: the \ac{OBD}, \ac{OGD}, and \ac{OPD} trajectories collapse onto the shared law $\sigma_z^2/(2-\alpha)$ of Propositions~\ref{prop:obd_se} and~\ref{prop:cgr_se} and Corollary~\ref{cor:lpd_se}, the \ac{BOD} decays exponentially faster (Proposition~\ref{prop:linear_decay} versus~\eqref{eq:bod_fixed_point}), and the \ac{LMMSE} fixed point sits strictly above at $\sigma_z^2/(1-\alpha)$, the ordering of Proposition~\ref{prop:fp_ordering} at every \ac{SNR}, with no error floor for any level.

\vspace{-2ex}
\subsection{Cross-Level SE Fixed-Point Equivalence}
\label{sec:cross_level}

Having computed each orbital fixed point separately -- the \ac{OBD} in Proposition~\ref{prop:obd_se}, the \ac{OGD} in Proposition~\ref{prop:cgr_se}, and the \ac{OPD} in Corollary~\ref{cor:lpd_se}, each sharing the leading order $\sigma_z^2/(2-\alpha)$ -- we now place them side by side to compare their \emph{sub-leading} corrections.
The assembly is not merely a rephrasing: through the excess identity of Lemma~\ref{lem:excess}, the three fixed points are ordered by each denoiser's amplitude-shrinkage bias relative to the optimal $A(\kappa)$. The \ac{OGD} matches that shrinkage to second order in $\kappa^{-1}$, so its excess over the \ac{OBD} is only $\mathcal{O}(\sigma_z^8)$ -- the two are indistinguishable through $\mathcal{O}(\sigma_z^6)$. The \ac{OPD} applies \emph{no} shrinkage and therefore pays a strictly positive $\Theta(\sigma_z^4)$ penalty, of the same order as the \ac{OBD}'s own correction; ring misdetection adds only a separate, exponentially small term.
All three effects vanish at leading order, leaving a single shared fixed point.

\begin{theorem}[\ac{SE} Fixed-Point Equivalence Across the
               Denoiser Hierarchy]
\label{thm:cross_level_fp}
Let $\alpha \triangleq K/N \in (0,1)$ and $\sigma_z^2 > 0$, and recall from Proposition~\ref{prop:obd_se} the inverse-energy coefficient $\gamma \triangleq \sum_{\ell} r_\ell R_\ell^{-2}$; let $d_R \triangleq \min_{\ell \neq \ell'} \norm{R_\ell - R_{\ell'}}$ denote the minimum inter-ring amplitude separation. Throughout this theorem, $\delta^{\eta}(\sigma_z^2) \triangleq \mathrm{MSE}_\infty^{\eta} - \sigma_z^2/(2-\alpha)$ denotes the \emph{total} sub-leading correction of level $\eta$; the quantity denoted $\delta^{\mathrm{B}}$ in Proposition~\ref{prop:obd_se} is its quartic leading term, with which it agrees up to $\mathcal{O}(\sigma_z^6)$ by~\eqref{eq:delta_cbm} below.
The \ac{OBD} fixed point exists and is unique by Theorem~\ref{thm:se_convergence} and Corollary~\ref{cor:se_uniqueness}; the \ac{OGD} and \ac{OPD} fixed points are those constructed in Proposition~\ref{prop:cgr_se} and Corollary~\ref{cor:lpd_se}.

\begin{figure}[H]
    \centering
    \includegraphics[width=\columnwidth]{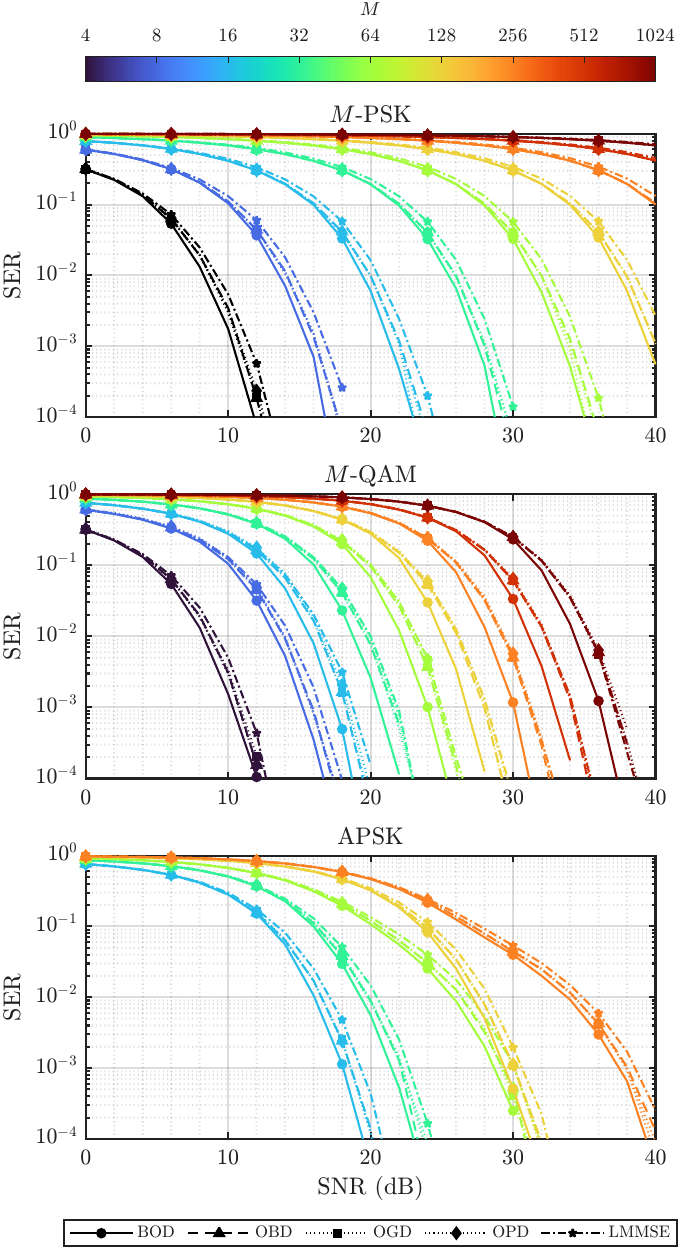}
    \vspace{-3ex}
    \caption{\ac{SE}-predicted \ac{SER} versus \ac{SNR} at $\alpha = 0.25$ of the \ac{BOD}, \ac{OBD}, \ac{OGD}, \ac{OPD}, and \ac{LMMSE}, for all $M$-\ac{PSK} and $M$-\ac{QAM} orders $M \in \{4,\ldots,1024\}$ and the DVB-S2/S2x \ac{APSK} constellations $M \in \{16,\ldots,256\}$; line style encodes the detector (shared legend), color the order $M$ (shared colorbar). The three orbital relaxations are indistinguishable and track the \ac{BOD} at a constant offset, confirming Theorem~\ref{thm:cross_level_fp}.}
    \label{fig:SE_MPSK_MQAM}
    \vspace{-2ex}
\end{figure}

{Assembling these three per-level results, all share the identical leading order as $\sigma_z^2 \to 0^+$ and admit the unified expansion}
\begin{equation}
  \label{eq:unified_fp}
  \mathrm{MSE}_\infty^{\eta}
  \;=\;
  \frac{\sigma_z^2}{2 - \alpha}
  \;+\;
  \delta^{\eta}(\sigma_z^2),
  \qquad
  \eta \in \bigl\{\mathrm{B},\, \mathrm{G},\, \mathrm{P}\bigr\},
\end{equation}
whose sub-leading corrections rank the denoisers by their amplitude-shrinkage bias relative to the optimal $A(\kappa)$:
\begin{align}
  \delta^{\mathrm{B}}
  &\;=\;
  \frac{\gamma\,\sigma_z^4}{(2-\alpha)^3}
  \;+\;
  \mathcal{O}\!\left(\sigma_z^6\right),
  \label{eq:delta_cbm}
  \\[4pt]
  \delta^{\mathrm{G}}
  &\;=\;
  \delta^{\mathrm{B}}
  \;+\;
  \mathcal{O}\!\left(\sigma_z^8\right),
  \label{eq:delta_cgr}
  \\[4pt]
  \delta^{\mathrm{P}}
  &\;=\;
  \tfrac{3}{2}\,\delta^{\mathrm{B}}
  \;+\;
  \mathcal{O}\!\left(\sigma_z^6\right)
  \;+\;
  \mathcal{O}\!\left(
    \exp\!\left(-\tfrac{(2-\alpha)\,d_R^2}{8\,\sigma_z^2}\right)
  \right).
  \label{eq:delta_lpd}
\end{align}
Consequently all pairwise fixed-point differences satisfy
\begin{align}
  \label{eq:fp_gaps}
  \norm{\mathrm{MSE}_\infty^{\mathrm{G}}
        - \mathrm{MSE}_\infty^{\mathrm{B}}}
  \;&=\;
  \mathcal{O}\!\left(\sigma_z^8\right), \\
  \mathrm{MSE}_\infty^{\mathrm{P}}
        - \mathrm{MSE}_\infty^{\mathrm{B}}
  \;&=\;
  \tfrac{1}{2}\,\delta^{\mathrm{B}}
  \;+\; \mathcal{O}\!\left(\sigma_z^6\right)
  \;=\;
  \Theta\!\left(\sigma_z^4\right) \;>\; 0 .
\end{align}
\end{theorem}

\begin{proof}
\proofref{app:thm_cross_level_fp}
\end{proof}

\begin{corollary}[Asymptotic \ac{SNR} Gap]
\label{cor:snr_gap}
For any underloaded system with $\alpha = K/N \in (0,1)$ and $\sigma_z^2 > 0$, the \ac{SE} fixed-point \ac{MSE} under all three levels of the denoiser hierarchy satisfies
\begin{equation}
  \label{eq:asymptotic_mse}
  \mathrm{MSE}_\infty
  \;=\;
  \frac{\sigma_z^2}{2 - \alpha}
  \;+\;
  \mathcal{O}\!\left(\sigma_z^4\right).
\end{equation}

As $\sigma_z^2 \to 0$, the \ac{MSE} vanishes, confirming that the continuous relaxation incurs no error floor: the \ac{SNR} gap is a finite, bounded constant that disappears in the high-\ac{SNR} limit.
\end{corollary}
\begin{proof}
The leading-order term is~\eqref{eq:unified_fp} of Theorem~\ref{thm:cross_level_fp}.
The $\mathcal{O}(\sigma_z^4)$ bound follows from the sub-leading corrections~\eqref{eq:delta_cbm}--\eqref{eq:delta_lpd}, noting $\delta^{\mathrm{B}}, \delta^{\mathrm{G}} = \mathcal{O}(\sigma_z^4)$ and {$\delta^{\mathrm{P}} = \tfrac{3}{2}\delta^{\mathrm{B}} = \mathcal{O}(\sigma_z^4)$ as well, its dominant term being the shrinkage penalty, with the ring-misdetection contribution $\mathcal{O}(e^{-c/\sigma_z^2})$ exponentially smaller}.
\end{proof}

\begin{remark}[Macroscopic Blindness to Denoiser Level]
\label{rem:macroscopic_blind}
Theorem~\ref{thm:cross_level_fp} establishes that the iterative detector is \emph{macroscopically blind} to which level of the denoiser hierarchy it is running.
The \ac{AMP} algorithm converges to the same leading-order \ac{MSE} whether it employs $\mathcal{O}(M)$ exact discrete posteriors, $\mathcal{O}(L)$ \ac{OBD} Bessel projections, $\mathcal{O}(L)$ Bessel-free \ac{OGD} projections, or a single $\mathcal{O}(1)$ ring-projection of the \ac{OPD}.
The three-orders-of-magnitude complexity reduction from $\mathcal{O}(M)$ to $\mathcal{O}(1)$ is ``free'' at the macroscopic level: the complexity difference is visible only in the sub-leading terms \eqref{eq:delta_cbm}--\eqref{eq:delta_lpd}, all of which vanish
strictly faster than $\sigma_z^2$ itself.
Equivalently, a decoder monitoring its own \ac{BER} curve at any finite operating \ac{SNR} would observe the three curves as parallel (shifted by at most $\mathcal{O}(\sigma_z^4)$ in \ac{MSE}), with no qualitative ``performance cliff'' separating any two levels.
\end{remark}

Figure~\ref{fig:SE_MPSK_MQAM} is the operational rendering of Theorem~\ref{thm:cross_level_fp}: {across all orders of the three families,} the \ac{SE}-predicted \ac{SER} of the $\mathcal{O}(L)$ \ac{OBD}, the Bessel-free \ac{OGD}, and the $\mathcal{O}(1)$ \ac{OPD} are visually indistinguishable over the entire \ac{SNR} range, tracking the exact \ac{BOD} up to the constant horizontal offset of Corollary~\ref{cor:snr_gap}.
Spanning three orders of arithmetic complexity costs nothing observable at the macroscopic level (Remark~\ref{rem:macroscopic_blind}).

\begin{remark}[Qualitative Distinction Among the Sub-Leading Corrections]
\label{rem:sublevel_corrections}
The three corrections in~\eqref{eq:delta_cbm}--\eqref{eq:delta_lpd} are governed by one mechanism: the squared bias of each denoiser's amplitude-shrinkage factor relative to the optimal $A(\kappa)$ (Lemma~\ref{lem:excess}).
\begin{enumerate}[(i)]
  \item \textbf{\ac{OGD} vs.\ \ac{OBD}.} The \ac{OGD} magnitude $m_{\mathrm{G}} = 1 - 1/(2\kappa)$ matches $A(\kappa) = 1 - 1/(2\kappa) - 1/(8\kappa^2) + \cdots$ through order $\kappa^{-1}$, so the bias $A - m_{\mathrm{G}} = \mathcal{O}(\kappa^{-2})$ enters~\eqref{eq:excess_ring} squared: the excess is $\mathcal{O}(\sigma_z^8)$, giving $\delta^{\mathrm{G}} = \delta^{\mathrm{B}} + \mathcal{O}(\sigma_z^8)$. The \ac{OGD} is thus asymptotically \emph{indistinguishable} from the \ac{OBD}, agreeing through $\mathcal{O}(\sigma_z^6)$; i.e., the Bessel-ratio lookup buys nothing at this order.
  \item \textbf{\ac{OPD}.} The \ac{OPD} applies no shrinkage ($m_{\mathrm{P}} = 1$), so its bias is the full $1 - A(\kappa) = \Theta(\kappa^{-1})$, giving a \emph{strictly positive} excess $\E[R_{\ell^*}^2(1-A)^2] = \Theta(\sigma_z^4)$ which, propagated to the fixed point, inflates the correction by one half, implicating in $\delta^{\mathrm{P}} = \tfrac{3}{2}\,\delta^{\mathrm{B}}$. This amplitude-shrinkage penalty, \emph{not} ring misdetection, is the dominant \ac{OPD} correction.
  As a result, the wrong-ring large deviations (Proposition~\ref{prop:ring_discrimination}) add only a separate, exponentially small $\mathcal{O}(e^{-c/\sigma_z^2})$ term. Hence $\mathrm{MSE}_\infty^{\mathrm{B}} \leq \mathrm{MSE}_\infty^{\mathrm{G}} \leq \mathrm{MSE}_\infty^{\mathrm{P}}$, and at high \ac{SNR}, $\sigma_z^2 < \sigma_\star^2$ (Corollary~\ref{cor:fp_ordering_hierarchy}), with \ac{OGD} coincident with \ac{OBD}, while \ac{OPD} is a clear $\Theta(\sigma_z^4)$ above.
\end{enumerate}
\end{remark}

\begin{corollary}[Ordering Across the Full Denoiser Hierarchy]
\label{cor:fp_ordering_hierarchy}
Let $\mathcal{F}_{\mathrm{G}}$ and $\mathcal{F}_{\mathrm{P}}$ be the \ac{OGD} and \ac{OPD} \ac{SE} maps, with fixed points $\mathrm{MSE}_\infty^{\mathrm{G}}$ and $\mathrm{MSE}_\infty^{\mathrm{P}}$ as constructed in Proposition~\ref{prop:cgr_se} and Corollary~\ref{cor:lpd_se} (the ``largest fixed point'' characterization via monotone iteration is reserved for the posterior-mean denoisers of Proposition~\ref{prop:se_monotone}).
There exists $\sigma_\star^2 > 0$ such that, for every $\sigma_z^2 < \sigma_\star^2$, the all-\ac{SNR} ordering of Proposition~\ref{prop:fp_ordering} refines to the full hierarchy
\begin{equation}
\label{eq:fp_ordering_full}
\mathrm{MSE}_\infty^{\mathrm{D}} \leq \mathrm{MSE}_\infty^{\mathrm{B}} \leq \mathrm{MSE}_\infty^{\mathrm{G}} \leq \mathrm{MSE}_\infty^{\mathrm{P}} \leq \mathrm{MSE}_\infty^{\mathrm{L}}.
\end{equation}
The restriction to $\sigma_z^2 < \sigma_\star^2$ is essential: above that threshold, $\mathrm{MSE}_\infty^{\mathrm{G}}$ and $\mathrm{MSE}_\infty^{\mathrm{P}}$ need not exist (the \ac{OGD} map leaves $[0,E_d]$ and grows without bound once its truncated gain $1-1/(2\kappa)$ turns negative ($\kappa<\tfrac12$), and the non-shrinking \ac{OPD} fixed point escapes above $E_d$) so at low \ac{SNR} the \ac{OGD} and \ac{OPD} fixed points escape above $\mathrm{MSE}_\infty^{\mathrm{L}}$, leaving only the all-\ac{SNR} ordering $\mathrm{MSE}_\infty^{\mathrm{D}} \le \mathrm{MSE}_\infty^{\mathrm{B}} \le \mathrm{MSE}_\infty^{\mathrm{L}}$ of Proposition~\ref{prop:fp_ordering}.
\end{corollary}

\begin{proof}
\proofref{app:cor_fp_ordering_hierarchy}
\end{proof}

The last inequality of~\eqref{eq:fp_ordering_full} admits a quantitative form, and it is the one of most operational interest: the margin over the linear baseline is governed entirely by the load.

\begin{corollary}[Load Dependence of the Linear-Baseline Gap]
\label{cor:load_gap}
For every $\alpha \in (0,1)$ and each orbital level $\eta \in \{\mathrm{B},\mathrm{G},\mathrm{P}\}$, as $\sigma_z^2 \to 0$, we have
\vspace{-1ex}
\begin{equation}
\label{eq:load_gap}
\frac{\mathrm{MSE}_\infty^{\mathrm{L}}}{\mathrm{MSE}_\infty^{\eta}}
\;=\; \frac{2-\alpha}{1-\alpha} \;+\; \mathcal{O}(\sigma_z^2).
\vspace{-1ex}
\end{equation}

The orbital advantage over the linear baseline therefore \emph{grows without bound} as the load approaches unity: $3.0$~dB as $\alpha \to 0$, $4.8$~dB at $\alpha = \tfrac12$, $7.0$~dB at $\alpha = \tfrac34$, and $13.2$~dB at $\alpha = 0.95$.
\end{corollary}
\begin{proof}
Divide the \ac{LMMSE} fixed point $\mathrm{MSE}_\infty^{\mathrm{L}} = \sigma_z^2/(1-\alpha) + \mathcal{O}(\sigma_z^4)$ by the shared orbital leading order $\mathrm{MSE}_\infty^{\eta} = \sigma_z^2/(2-\alpha) + \mathcal{O}(\sigma_z^4)$ of Corollary~\ref{cor:snr_gap}; the quartic remainders contribute $\mathcal{O}(\sigma_z^2)$ relative error. The $\alpha \to 0$ endpoint is the ratio of \ac{MMSE} dimensions, the \ac{LMMSE} obeying $\mathrm{MMSE}_{\mathrm{L}} \sim \bar{\sigma}^2$ against $\bar{\sigma}^2/2$ for every orbital level (Proposition~\ref{prop:mmse_dimension}): with no interference the fixed point is the single-shot \ac{MMSE}, and the entire advantage is the one freed real coordinate.
\end{proof}

Figure~\ref{fig:load_gap} plots~\eqref{eq:load_gap} and makes the two regimes visible at once. The gap never falls below $3$~dB, however lightly the system is loaded, because the relaxation frees one real coordinate rather than two, which is the horizontal asymptote of the figure, and the only part of the gain that survives at vanishing load.

\begin{figure}[H]
  \centering
  \includegraphics[width=\columnwidth]{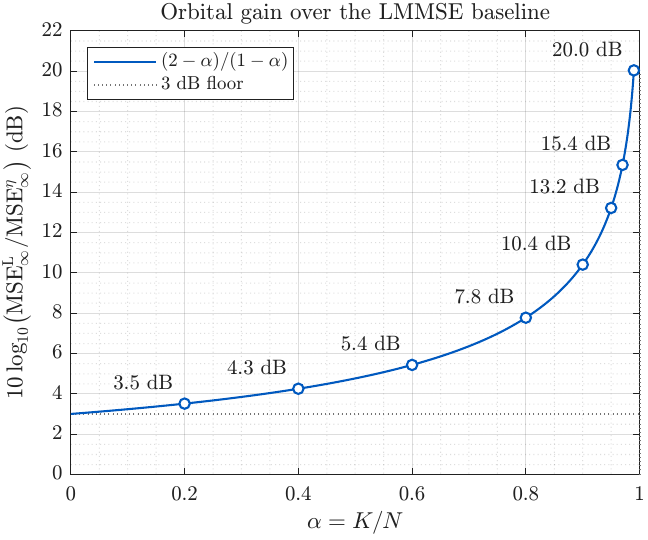}
  \vspace{-4ex}
  \caption{Gain of the orbital hierarchy over the \ac{LMMSE} baseline versus system load, $10\log_{10}\bigl(\mathrm{MSE}_\infty^{\mathrm{L}}/\mathrm{MSE}_\infty^{\eta}\bigr) = 10\log_{10}\frac{2-\alpha}{1-\alpha}$, from Corollary~\ref{cor:load_gap}. The curve is bounded below by the $10\log_{10}2 = 3.0$ dB floor (dotted) -- the ratio of \ac{MMSE} dimensions, $1$ for the \ac{LMMSE} against $\tfrac12$ for every orbital level (Proposition~\ref{prop:mmse_dimension}) -- and diverges at unit load, where the linear baseline loses its contraction while the orbital fixed point stays finite up to $\alpha \to 2$. The gain is independent of $M$, of the constellation family, and of which orbital denoiser is used.}
  \label{fig:load_gap}
  \vspace{-2ex}
\end{figure}

The gap diverges as $\alpha \to 1^-$, where the linear baseline loses its contraction while the orbital fixed point remains finite up to $\alpha \to 2$. The growth is slow over the lightly loaded range and then abrupt: $3.7$~dB at $\alpha = \tfrac14$ and $4.8$~dB at $\alpha = \tfrac12$, but $7.0$~dB by $\alpha = \tfrac34$ and $10.4$~dB by $\alpha = 0.9$. Two consequences are worth drawing. First, the advantage measured at the moderate loads of our numerical evaluations ($\alpha = 0.25$) is close to the \emph{smallest} the framework offers; the operating points where a linear receiver is most tempting -- heavily loaded ones -- are precisely those where it is most costly. Second, because the ratio is a function of $\alpha$ alone, the gain is the same for every constellation order and family and for all three orbital levels: even the $\mathcal{O}(1)$ \ac{OPD} inherits it in full. The restriction $\alpha < 1$ is exactly that of Proposition~\ref{prop:fp_ordering}, and is not conservative: at $\alpha = 1$ the law $\sigma_z^2/(1-\alpha)$ ceases to hold and $\mathrm{MSE}_\infty^{\mathrm{L}}$ saturates toward $E_d$ instead of diverging, since $\mathcal{F}_{\mathrm{L}}$ remains a self-map of $[0,E_d]$ at every load; the curve of Fig.~\ref{fig:load_gap} is therefore to be read as the small-$\sigma_z^2$ asymptote, which the measured ratio tracks closely while $\alpha$ stays away from unity (Fig.~\ref{fig:MSE_alpha}, Section~\ref{sec:load_sweep}).

\begin{remark}[Absence of an Error Floor and Operational Boundary]
\label{rem:no_error_floor_v2}
Equation~\eqref{eq:asymptotic_mse} confirms that as $\sigma_z^2 \to 0$, the shared \ac{MSE} vanishes for all three denoiser levels, with no irreducible residual.
The factor $1/(2-\alpha)$ is a finite, bounded amplification of the physical noise; it does not represent a minimum floor.
The condition $\alpha < 1$ (Corollary~\ref{cor:se_uniqueness}) is the true operational boundary for uniqueness and capacity achievement, not the apparent singularity at $\alpha = 2$.
Near $\alpha = 2$, the leading-order approximation $\hat{\sigma}^2 \approx \bar{\sigma}^2/2$ breaks down before the denominator collapses, and the fixed point still exists by Theorem~\ref{thm:se_convergence}; only uniqueness and the linear approximation~\eqref{eq:asymptotic_mse} are lost.
\end{remark}

\begin{corollary}[Full Asymptotic Capacity Across the Hierarchy]
  \label{cor:capacity}
  Let $C_{\mathrm{SE}}^{\eta}(\mathrm{SNR})$ denote the mutual information of the \ac{SE}-equivalent scalar channel at the fixed point of denoiser $\eta \in \{\mathrm{B}, \mathrm{G}, \mathrm{P}\}$.
  For each level, the effective noise variance at the \ac{SE} fixed point is $\bar{\sigma}_\infty^2(\mathrm{SNR}) = \sigma_z^2 + (K/N)\,\mathrm{MSE}_\infty^{\eta}(\mathrm{SNR})$.
  As the physical \ac{SNR} grows large ($\mathrm{SNR} \to \infty$, yielding $\sigma_z^2 \to 0$), Corollary~\ref{cor:snr_gap} guarantees that the tracking error strictly vanishes for every level, $\mathrm{MSE}_\infty^{\eta}(\mathrm{SNR}) \to 0$.
  Consequently, the effective cavity noise also identically vanishes ($\bar{\sigma}_\infty^2 \to 0$), and the \ac{SE}-equivalent mutual information (in bits/channel use) fully achieves the discrete constellation capacity for all three denoisers
  \begin{equation}
  \!\lim_{\mathrm{SNR} \to \infty} \!\!C_{\mathrm{SE}}^{\eta}(\mathrm{SNR}) = \!\!\lim_{\bar{\sigma}_\infty \to 0} I\!\left(x;\; x\! +\! \bar{\sigma}_\infty \tilde{z}\right) = \!\log_2 M,
    \label{eq:capacity_achieved}
  \end{equation}
  where $\tilde{z} \sim \mathcal{CN}(0,1)$ and $x \in \mathcal{M}$ is drawn from the uniform discrete prior.
  In particular, the $\mathcal{O}(1)$ \ac{OPD}, using a single ring-projection, is asymptotically capacity-achieving.
\end{corollary}
\begin{proof}
The limit rests solely on $\bar{\sigma}_\infty^2 \to 0$, which Corollary~\ref{cor:snr_gap} establishes uniformly over $\eta \in \{\mathrm{B}, \mathrm{G}, \mathrm{P}\}$ via the shared leading order $\mathrm{MSE}_\infty^{\eta} = \sigma_z^2/(2-\alpha) + \mathcal{O}(\sigma_z^4)$. Since the \ac{SE}-equivalent scalar channel $\bar{x} = x + \bar{\sigma}_\infty \tilde{z}$ depends on the denoiser only through $\bar{\sigma}_\infty^2$, its mutual information $I(x; x + \bar{\sigma}_\infty \tilde{z}) \to \log_2 M$ is common to all three levels. The identification of $C_{\mathrm{SE}}^{\eta}$ with this mutual information is exact for the Bayes-optimal (under $q$) \ac{OBD}; for the \ac{OGD} and \ac{OPD} the \ac{SE} still yields the equivalent \ac{AWGN} channel of variance $\bar{\sigma}_\infty^2$, and as $\sigma_z^2 \to 0$ the concentration $\kappa \to \infty$ places the \ac{OGD} in its valid regime ($\kappa \geq \tfrac12$), so the limit holds unconditionally.
\end{proof}

\begin{remark}[Capacity Penalty as a Parallel SNR Shift]
  Unlike severe mismatched decoding scenarios that induce strict, impassable capacity ceilings~\cite{Merhav1994TIT}, the orbital relaxation is fundamentally constellation-constrained capacity-achieving for any underloaded system ($K < N$), at every level of the denoiser hierarchy (\ac{OBD}, \ac{OGD}, and \ac{OPD}).
  The geometric prior mismatch manifests entirely in the rate of approach: to achieve a target mutual information $C$ arbitrarily close to $\log_2 M$, any orbital receiver simply requires a fixed, constant dB increase in transmit power relative to the \ac{BOD}, perfectly corresponding to the linear variance penalty derived in Section~\ref{sec:asymptotic_gap}; the choice of denoiser affects only the $\mathcal{O}(\sigma_z^4)$ sub-leading offset, not the capacity limit itself.
\end{remark}

\subsection{Load Dependence of the Hierarchy}
\label{sec:load_sweep}

The fixed-point laws of this section depend on the load as much as on the noise variance, yet every evaluation so far has swept the \ac{SNR} at the single load $\alpha = 0.25$. 
Figures~\ref{fig:MSE_alpha} and~\ref{fig:SER_alpha} complete the picture by fixing $\mathrm{SNR} = 30$~dB and sweeping $\alpha$ instead, over the same three constellation families and the same five detectors. 
Three predictions of this section become visible at once. 
First, the \ac{OBD}, \ac{OGD} and \ac{OPD} curves lie on the shared law $\sigma_z^2/(2-\alpha)$ throughout the sweep: the cross-level equivalence of Theorem~\ref{thm:cross_level_fp} holds uniformly in the load, and is not an artifact of the particular $\alpha$ chosen for Figs.~\ref{fig:MSE_MPSK_MQAM} and~\ref{fig:SE_MPSK_MQAM}. 
Second, throughout the underloaded range the \ac{BOD} lies strictly below that law, and falls further below it as the load lightens. 
This separation is Proposition~\ref{prop:mmse_dimension} in geometric form: the exact detector has \ac{MMSE} dimension $0$ and decays exponentially in \ac{SNR}, whereas every orbital level has dimension $\tfrac12$ and is pinned to $\sigma_z^2/(2-\alpha)$; the two can never meet, and the vertical offset between them is the price of the relaxation, displayed as a function of load. 
Third, the \ac{LMMSE} baseline follows $\sigma_z^2/(1-\alpha)$ while $\alpha$ stays away from unity and then saturates toward $E_d$ rather than diverging, exactly as anticipated in the discussion of Fig.~\ref{fig:load_gap}. The vertical distance between the two black references is the load gap $(2-\alpha)/(1-\alpha)$ of Corollary~\ref{cor:load_gap}: Fig.~\ref{fig:MSE_alpha} is its direct measurement, and the widening of that distance with $\alpha$ is the divergence the corollary predicts.
\vspace{-3ex}
\begin{figure}[H]
    \centering
    \includegraphics[width=\columnwidth]{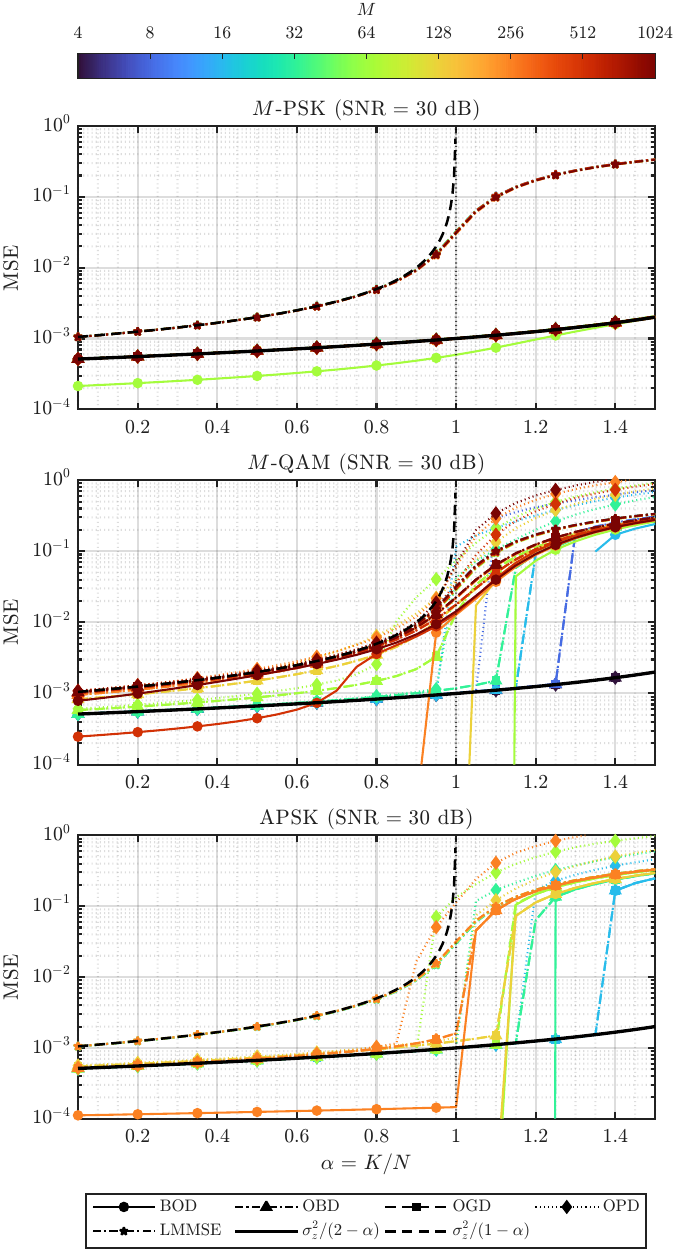}
    \vspace{-3ex}
    \caption{Fixed-point $\mathrm{MSE}_\infty$ versus load $\alpha = K/N$ at $\mathrm{SNR} = 30$~dB, for all $M$-\ac{PSK} and $M$-\ac{QAM} orders $M \in \{4,\ldots,1024\}$ and the DVB-S2/S2x \ac{APSK} constellations; line style encodes the detector (shared legend), color the order $M$ (shared colorbar), as in Fig.~\ref{fig:MSE_MPSK_MQAM}. The three orbital levels track the shared law $\sigma_z^2/(2-\alpha)$ (black solid) across the sweep; the \ac{BOD} lies strictly below it while $\alpha < 1$; the \ac{LMMSE} follows $\sigma_z^2/(1-\alpha)$ (black dashed) until it saturates toward $E_d$ near unit load. The vertical distance between the two black references is the load gap $(2-\alpha)/(1-\alpha)$ of Corollary~\ref{cor:load_gap}. The dotted vertical line marks $\alpha = 1$; the region to its right lies outside the hypotheses of Corollary~\ref{cor:se_uniqueness} and is shown for completeness only.}
    \label{fig:MSE_alpha}
    \vspace{-2ex}
\end{figure}

\begin{figure}[H]
    \centering
    \includegraphics[width=\columnwidth]{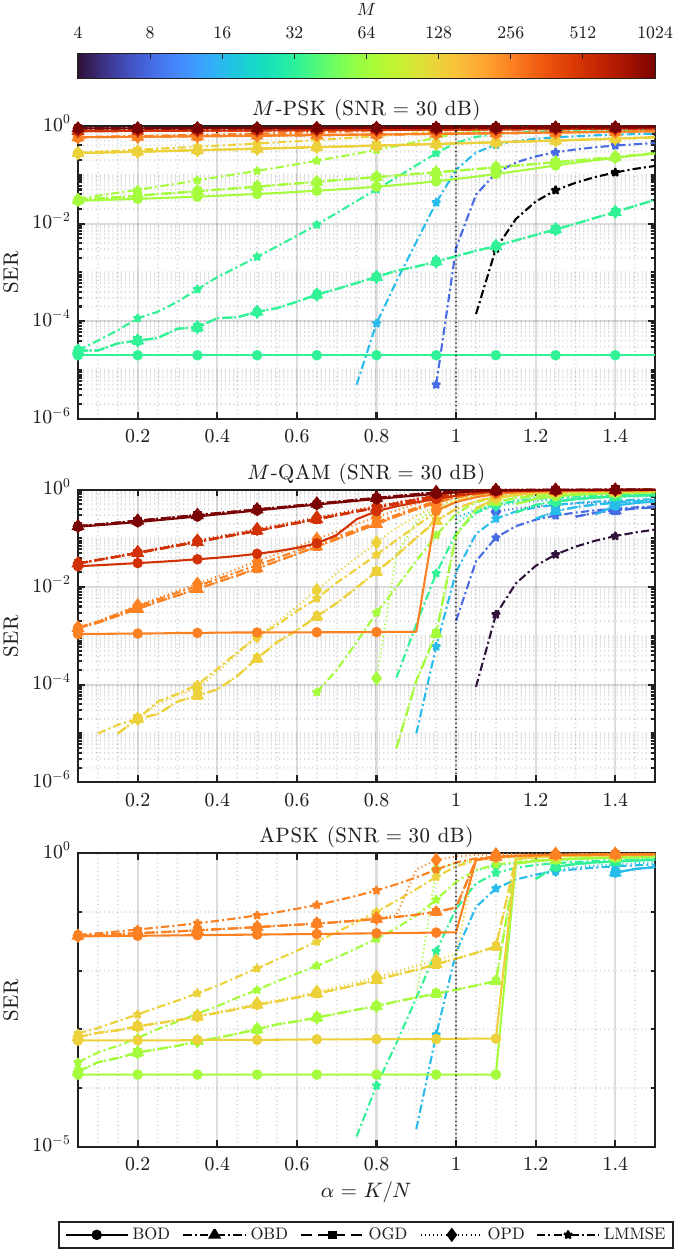}
    \vspace{-3ex}
    \caption{\ac{SE}-predicted \ac{SER} versus load $\alpha = K/N$ at $\mathrm{SNR} = 30$~dB, under the conventions and constellation set of Fig.~\ref{fig:MSE_alpha}. Within a single color -- at fixed $M$ -- the ordering $\mathrm{BOD} \le \mathrm{OBD} \le \mathrm{OGD} \le \mathrm{OPD} \le \mathrm{LMMSE}$ of Proposition~\ref{prop:fp_ordering} is preserved at every load. The load at which the \ac{SER} collapses decreases with $M$.}
    \label{fig:SER_alpha}
    \vspace{-2ex}
\end{figure}

Two features of the dense panels should be read with care. Above $\alpha \approx 0.9$ the \ac{QAM} and \ac{APSK} curves depart steeply from the shared law and cross one another. This is the same failure of the \emph{iteration} rather than of the fixed point that was identified at the opening of this section: on closely spaced rings the \ac{OPD}'s non-Lipschitz projection (Lemma~\ref{lem:pseudo_lip}) carries an Onsager correction that misses its jump discontinuities, so the recursion stalls before reaching the fixed point that Theorem~\ref{thm:cross_level_fp} certifies. The heavier the load, the more residual interference that correction must cancel, so the effect strengthens with $\alpha$ just as it strengthens with $M$; the scalar \ac{SE} recursion, which is immune to \ac{AMP} iteration dynamics, continues to place all three levels on the shared law. Separately, the region $\alpha \ge 1$ to the right of the dotted line lies outside the hypotheses of Corollary~\ref{cor:se_uniqueness}: uniqueness of the fixed point is no longer guaranteed there, and the law $\sigma_z^2/(1-\alpha)$ for the linear baseline is meaningless. That the orbital curves nonetheless continue to track $\sigma_z^2/(2-\alpha)$ is consistent with the fixed point remaining finite up to $\alpha \to 2$ (Remark~\ref{rem:no_error_floor_v2}), and suggests that orbital detection retains an advantage in the overloaded regime in which a linear receiver possesses no contraction at all. We record this as an empirical observation; it is not a claim of the present theory.

Figure~\ref{fig:SER_alpha} is the operational counterpart. Within a single color the ordering of Proposition~\ref{prop:fp_ordering} is preserved at every load, so the hierarchy is an ordering in the load as much as in the \ac{SNR}. The \ac{OPD} pays for its $\mathcal{O}(1)$ cost by a margin that stays $\mathcal{O}(\sigma_z^4)$ at every load, though not a load-independent one: by~\eqref{eq:lpd_fixed_point} its excess scales as $(2-\alpha)^{-3}$, and so grows by roughly a factor of four between $\alpha = 0.1$ and $\alpha = 0.8$ while remaining subleading. The load at which the \ac{SER} collapses decreases with $M$: denser constellations surrender to residual interference earlier, since a given effective noise variance covers proportionally more of the minimum distance. Unlike the \ac{MSE}, the \ac{SER} obeys no closed-form law in $\alpha$; the panel is included to confirm that the \ac{MSE} ordering survives the passage to a decision rule, which is the quantity a receiver ultimately reports.
The \ac{SE} analysis characterizes performance in terms of \ac{MSE}. 
The natural question is how this \ac{MSE} translates into information-theoretic rate. 
Section \ref{sec:inf_theory} makes this translation precise via the \ac{I-MMSE} identity, the \ac{GMI} framework, and an optimal transport bound that connects constellation geometry directly to rate loss.

\vspace{-2ex}
\section{Fundamental Geometry-Rate Tradeoffs}
\label{sec:inf_theory}

We now discuss the information-theoretic consequences that result from the relaxed geometry of the proposed architectures.
Throughout this section, quantities that depend on the choice of denoiser carry the index $\eta \in \{\mathrm{B}, \mathrm{G}, \mathrm{P}\}$: $\mathrm{MSE}_\infty^{\eta}$ is the \ac{SE} fixed-point error, $\bar{\sigma}_{\infty,\eta}^2 = \sigma_z^2 + \alpha\,\mathrm{MSE}_\infty^{\eta}$ (with $\alpha = K/N$) the corresponding effective noise, and the subscript $\mathrm{D}$ denotes the discrete \ac{BOD} reference. Where a statement specializes to the representative \ac{OBD} we abbreviate $\mathrm{MSE}_\infty \equiv \mathrm{MSE}_\infty^{\mathrm{B}}$ and $\bar{\sigma}_\infty^2 \equiv \bar{\sigma}_{\infty,\mathrm{B}}^2$. By Corollary~\ref{cor:snr_gap}, all three levels share the leading order $\mathrm{MSE}_\infty^{\eta} = \sigma_z^2/(2-\alpha) + \mathcal{O}(\sigma_z^4)$.

\vspace{-2ex}
\subsection{Decoupling Principle and Single-Letter Rate Formula}
\label{sec:decoupling}

The \ac{SE} framework of Section~\ref{sec:se} implicitly relies on the \emph{decoupling principle}: in the large-system limit, the vector \ac{MIMO} channel decouples into $K$ parallel scalar channels. 
We now make this explicit and derive a single-letter formula for the per-user achievable rate of the orbital receiver, connecting our framework to the random matrix theory foundations of Guo and Verd\'{u}~\cite{GuoVerdu2005CDMA}.
What Theorem~\ref{thm:decoupling} adds to Section~\ref{sec:se} is not the
value of the effective noise variance -- the \ac{SE} recursion already delivers that -- but a
\emph{distributional} statement. 
The \ac{SE} propagates only the scalar second moment
$\mathrm{MSE}_t$, whereas decoupling asserts that the joint empirical distribution of the
transmitted symbols and their cavity statistics converges to the distribution of a genuine scalar
\ac{AWGN} channel with \emph{Gaussian} effective noise -- each user behaving, in the
test-function sense made precise below, as if observed alone through~\eqref{eq:decoupled_channel}.

It is this distributional equivalence -- not the \ac{MSE} alone --
that licenses replacing the $K$-dimensional mutual information by $K$ copies of the scalar
$I(x; x + \bar{\sigma}_\infty \tilde{z})$ in Corollary~\ref{cor:single_letter_rate}.

\begin{theorem}[Decoupling Under the Orbital Denoisers]
\label{thm:decoupling}
Consider the system model~\eqref{eq:system} with $\mathbf{H}$ having \ac{iid} entries $H_{nk} \sim \CN(0, 1/N)$.
In the large-system limit ($N, K \to \infty$ with $K/N \to \alpha$), the \ac{AMP} iteration under each orbital denoiser $\eta \in \{\mathrm{B}, \mathrm{G}, \mathrm{P}\}$ decouples at its fixed point: the joint empirical distribution of the pairs $\{(x_k, \bar{x}_k)\}_{k=1}^K$ converges weakly, almost surely -- equivalently, averages of pseudo-Lipschitz test functions converge -- to the distribution of $(x, \bar{x})$ generated by the scalar channel
\begin{equation}
\label{eq:decoupled_channel}
\bar{x}_k = x_k + \bar{\sigma}_{\infty{,\eta}} \tilde{z}_k, \quad \tilde{z}_k \sim \CN(0, 1),
\end{equation}
where $\bar{\sigma}_{\infty{,\eta}}^2$ is the {denoiser's} \ac{SE} fixed point satisfying
\begin{equation}
\label{eq:fp_implicit}
\bar{\sigma}_{\infty{,\eta}}^2 = \sigma_z^2 + \alpha \cdot \mathcal{F}\!\left(\tfrac{\bar{\sigma}_{\infty{,\eta}}^2 - \sigma_z^2}{\alpha}\right),
\end{equation}
with $\mathcal{F}$ defined in~\eqref{eq:fixed_point} and $\alpha = K/N$ (equivalently $\mathrm{MSE}_\infty^{\eta} = \mathcal{F}(\mathrm{MSE}_\infty^{\eta})$ with $\mathrm{MSE}_\infty^{\eta} = (\bar{\sigma}_{\infty{,\eta}}^2 - \sigma_z^2)/\alpha$). {The three fixed-point variances coincide to leading order (Corollary~\ref{cor:snr_gap}), differing only at $\mathcal{O}(\sigma_z^4)$.}
\end{theorem}

\begin{proof}
The proof follows from the rigorous \ac{SE} analysis of Bayati and Montanari~\cite{Bayati2011} extended to mismatched denoisers by Javanmard and Montanari~\cite{Javanmard2013}, whose pseudo-Lipschitz hypothesis holds directly for the \ac{OBD} and, for the \ac{OGD} and \ac{OPD}, through the fixed-$\bar{\sigma}^2$ regularization of Lemma~\ref{lem:pseudo_lip}. The complex model is handled by the standard real $2\times2$ vectorization of the $\CN$ quantities, under which the \ac{iid} $\CN(0,1/N)$ ensemble and the (regularized) $\Real^2$-Lipschitz denoisers of Lemma~\ref{lem:pseudo_lip} meet the hypotheses of~\cite{Bayati2011,Javanmard2013}.
Two strengthenings of that machinery, not needed here but worth recording, bear on the scope of the statement: finite-sample concentration bounds quantify the approach to the large-system limit at finite $K$~\cite{rush_tit_2018}, and the rigorous theory extends beyond \ac{iid} Gaussian ensembles to unitarily invariant matrices through \ac{VAMP}~\cite{ranganVAMP2019} and its expectation-propagation analysis~\cite{TakeuchiTIT2020} -- the natural route to the correlated-channel extension we leave to future work.
Under the \ac{AMP} iteration with {denoiser $\eta$}, the effective observation (the denoiser input) obeys $\bar{x}_k^{(t)} = x_k + \bar{\sigma}_t^{\eta}\tilde{z}_k + o_P(1)$ with $\tilde{z}_k\sim\CN(0,1)$ asymptotically independent of $x_k$; equivalently, the joint empirical distribution of $\{(x_k, \bar{x}_k^{(t)})\}_{k=1}^K$ converges weakly to the distribution of $(x,\, x + \bar{\sigma}_t^{\eta}\tilde{z})$ as $N,K\to\infty$. The denoised error then has mean square $\mathrm{MSE}_t^{\eta} = \E\norm{x - \eta(x + \bar{\sigma}_t^{\eta}\tilde{z};\,(\bar{\sigma}_t^{\eta})^2)}^2$ -- note that it is the observation noise, not the (generally non-Gaussian) estimate error, that is asymptotically Gaussian.
At the fixed point, the effective per-user channel becomes~\eqref{eq:decoupled_channel} with noise variance $\bar{\sigma}_{\infty{,\eta}}^2$ determined self-consistently by~\eqref{eq:fp_implicit}.
The decoupling is exact in the sense that any separable test function $\frac{1}{K}\sum_k \psi(x_k, (\hat{x}^{\eta})_k)$ converges almost surely to $\E[\psi(x, \eta(x + \bar{\sigma}_{\infty{,\eta}} \tilde{z}; \bar{\sigma}_{\infty{,\eta}}^2))]$.
\end{proof}

The decoupling principle immediately yields a single-letter achievable-rate formula.

\begin{corollary}[Single-Letter Achievable Rate]
\label{cor:single_letter_rate}
For each denoiser $\eta \in \{\mathrm{B}, \mathrm{G}, \mathrm{P}\}$, the per-user achievable rate in the decoupled regime is the matched information rate of the decoupled \ac{AWGN} channel at effective noise $\bar{\sigma}_{\infty,\eta}$, i.e.
\begin{equation}
\label{eq:single_letter_rate}
R_{\eta}(\mathrm{SNR}) = I\!\left(x;\; x + \bar{\sigma}_{\infty,\eta}\, \tilde{z}\right) \quad \text{bits/symbol},
\end{equation}

The operational rate of the mismatched ($q$-based) receiver is the \ac{GMI} of Corollary~\ref{thm:gmi}, which lower-bounds~\eqref{eq:single_letter_rate} (cf.\ Fig.~\ref{fig:GMI_rate}),
where $\bar{\sigma}_{\infty,\eta}^2 = \sigma_z^2 + \alpha \cdot \mathrm{MSE}_\infty^{\eta}$ and $\mathrm{MSE}_\infty^{\eta}$ is the \ac{SE} fixed point~\eqref{eq:se_mismatch} of that denoiser.
At high \ac{SNR}, substituting Corollary~\ref{cor:snr_gap} (under which $\bar{\sigma}_{\infty,\eta}^2 = 2\sigma_z^2/(2-\alpha) + \mathcal{O}(\sigma_z^4)$ for all three levels) yields the common rate
\begin{equation}
\label{eq:rate_high_snr}
R_{\eta} = I\!\left(x;\, x + \sigma_z\sqrt{\tfrac{2}{2-\alpha}}\,\tilde{z}\right) + \mathcal{O}(\sigma_z^4),
\end{equation}
identical across the hierarchy up to the $\mathcal{O}(\sigma_z^4)$ remainder.
\end{corollary}

\begin{proof}
By Theorem~\ref{thm:decoupling} (which decouples every level of the hierarchy), the per-user channel is equivalent -- in the distributional, test-function sense of that theorem -- to a scalar \ac{AWGN} channel with noise variance $\bar{\sigma}_{\infty,\eta}^2$.
The mutual information of this channel with discrete input $x \in \mathcal{M}$ is the maximum achievable rate over that equivalent channel. What Theorem~\ref{thm:decoupling} rigorously supplies is the weak convergence of the per-user distribution $(x_k,\bar{x}_k)$ and of pseudo-Lipschitz test-function averages (hence the \ac{MSE}); the further decomposition of the $K$-dimensional mutual information into $K$ copies of this scalar term is the replica-symmetric single-letter characterization~\cite{GuoVerdu2005CDMA} (Remark~\ref{rem:replica}), expected but not proven exact for $\alpha < 1$. We therefore read~\eqref{eq:single_letter_rate} as the single-letter rate predicted by decoupling rather than a fully proven finite-$K$ converse.
Substituting $\mathrm{MSE}_\infty^{\eta} \approx \sigma_z^2/(2 - \alpha)$ from Corollary~\ref{cor:snr_gap} gives $\bar{\sigma}_{\infty,\eta}^2 = \sigma_z^2 + \alpha \sigma_z^2/(2-\alpha) = 2\sigma_z^2/(2-\alpha)$ {for each $\eta$}, yielding~\eqref{eq:rate_high_snr}.
\end{proof}

\begin{remark}[Connection to Replica Analysis]
\label{rem:replica}
The single-letter formula~\eqref{eq:single_letter_rate} is the \ac{RS} prediction for the mutual information of the massive \ac{MIMO} channel under mismatched decoding. 
Under the \ac{RS} assumption (expected to hold for $\alpha < 1$ with \ac{iid} Gaussian channels~\cite{Bayati2011}), the \ac{SE} fixed point is the stationary point, {in a trial \ac{MSE} $m$,} of a replica-symmetric potential of the {schematic} form
\begin{align}
\label{eq:free_energy}
&\mathcal{F}_{\mathrm{RS}}({m}) \\ 
&= I\!\left(x; x + {\sqrt{\sigma_z^2 + \alpha m}}\, \tilde{z}\right) - \frac{\alpha}{2}\ln\!\left(1 + {\frac{\alpha m}{\sigma_z^2}}\right) + \text{const.}, \nonumber 
\end{align}
whose saddle point -- taken jointly with the conjugate order parameter that the schematic single-parameter form suppresses, as in the full replica-symmetric potential -- reproduces~\eqref{eq:fp_implicit}. This provides a variational characterization of the achievable rate that is dual to the operational \ac{SE} characterization.
This variational duality is specific to the \ac{OBD}: the free energy~\eqref{eq:free_energy} is stationary at the \ac{SE} fixed point only because $\eta_B$ is the Bayes posterior mean under the orbital prior $q$ (its per-iteration error equals the posterior variance, the Nishimori property invoked in Section~\ref{sec:three_se}). The \ac{OGD} and \ac{OPD} are not posterior means, so they admit no such free-energy potential; they still reach a decoupled fixed point (Theorem~\ref{thm:decoupling}) and hence an operational rate~\eqref{eq:single_letter_rate}, but not the variational characterization of this remark.
\end{remark}

\begin{proposition}[\ac{I-MMSE} Rate Gap and Exponential Capacity Achievement]
\label{prop:immse_rate_gap}
The achievable rate gap between the optimal receiver and any orbital denoiser $\eta \in \{\mathrm{B}, \mathrm{G}, \mathrm{P}\}$ is
\begin{equation}
\label{eq:rate_gap_correct}
\Delta R_{\eta}(\mathrm{SNR})
\;=\;
{\frac{1}{\ln 2}}
\int_{1/\bar{\sigma}_{\infty,\eta}^2}^{1/(\bar{\sigma}_\infty^{\mathrm{D}})^2}
\!\!\mathrm{mmse}_p(\zeta)\,\mathrm{d}\zeta,
\end{equation}
where $\mathrm{mmse}_p(\zeta) \triangleq \mathbb{E}[\norm{x - \mathbb{E}[x\mid\sqrt{\zeta}\,x\!+\!\tilde{z}]}^2]$,
$\tilde{z}\sim\mathcal{CN}(0,1)$, is the \ac{MMSE} under the true prior $p$ -- the true-prior counterpart of $\mathrm{mmse}_q$ from~\eqref{eq:mmse_Q_def}, the two related by $\zeta = 1/\bar{\sigma}^2$ -- at effective \ac{SNR} $\zeta$, and the effective noise levels are $\bar{\sigma}_{\infty,\eta}^2 = \sigma_z^2 + \alpha\,\mathrm{MSE}_\infty^{\eta}$ and $(\bar{\sigma}_\infty^{\mathrm{D}})^2 = \sigma_z^2 + \alpha\,\mathrm{MSE}_\infty^{\mathrm{D}}$.
At high \ac{SNR}, $\Delta R_{\eta}$ decays exponentially, at the same rate for all three levels,
\begin{equation}
\label{eq:rate_gap_bound}
\Delta R_{\eta}(\mathrm{SNR})
\;=\;
\mathcal{O}\!\left(
  \exp\!\left(-\frac{d_{\min}^2(2-\alpha)}{8\sigma_z^2}\right)
\right),
\end{equation}
where $d_{\min}$ is the minimum Euclidean distance of $\mathcal{M}$.
\end{proposition}
\begin{proof}
Apply the \ac{I-MMSE} identity {$dI/d(1/\bar{\sigma}^2)=\mathrm{mmse}_p(1/\bar{\sigma}^2)$ (in nats; the \emph{complex} scalar channel carries no factor $\tfrac12$, unlike the real phase channel of Corollary~\ref{cor:immse})}
\cite{Guo2005} to integrate between effective \ac{SNR} levels
$1/\bar{\sigma}_{\infty,\eta}^2$ and $1/(\bar{\sigma}_\infty^{\mathrm{D}})^2$ {(the $1/\ln 2$ converting nats to bits)}.
At high \ac{SNR}, Corollary~\ref{cor:snr_gap} gives
$1/\bar{\sigma}_{\infty,\eta}^2 = (2-\alpha)/(2\sigma_z^2)+\mathcal{O}(1)$ {for every $\eta \in \{\mathrm{B}, \mathrm{G}, \mathrm{P}\}$ -- the $\mathcal{O}(\sigma_z^4)$ spread among the three fixed points enters only the $\mathcal{O}(1)$ term} and
$1/(\bar{\sigma}_\infty^{\mathrm{D}})^2 = (1+o(1))/\sigma_z^2$.
Since $\mathrm{mmse}_p(\zeta)\le C\exp(-d_{\min}^2 \zeta/4)$
\cite[Prop.~1]{Guo2005}, substitution gives~\eqref{eq:rate_gap_bound} {uniformly in $\eta$}.
\end{proof}

This confirms that the orbital relaxation is, at every level of the hierarchy, not merely capacity-approaching but exponentially capacity-achieving: the rate loss decays at the same exponential rate as the optimal decoder's \ac{MSE}.

\emph{Operational interpretation at finite \ac{SNR}.}
At any fixed operating point, $\Delta R_{\eta}$ is a finite, computable quantity given
by~\eqref{eq:rate_gap_correct}. Taking $\eta = \mathrm{B}$ for concreteness, the \ac{MSE} gap $\mathrm{MSE}_\infty^{\mathrm{B}} - \mathrm{MSE}_\infty^{\mathrm{D}}$
from Theorem~\ref{thm:ot_bound} provides the difference in effective noise levels
($\bar{\sigma}_{\infty,\mathrm{B}}^2 - (\bar{\sigma}_\infty^{\mathrm{D}})^2 = \alpha(\mathrm{MSE}_\infty^{\mathrm{B}} - \mathrm{MSE}_\infty^{\mathrm{D}})$),
which in turn determines the integration width in~\eqref{eq:rate_gap_correct}.
At moderate \ac{SNR} where the exponential bound is loose, the rate gap can be evaluated
numerically via Monte Carlo (Remark~\ref{rem:gmi_computation}).

\vspace{-2ex}
\subsection{GMI Under Mismatched Decoding}
\label{sec:gmi}

Theorem~\ref{thm:decoupling} and Corollary~\ref{cor:single_letter_rate} price one mismatch, and one only. They identify \emph{which channel} the orbital receiver faces -- a scalar \ac{AWGN} channel whose noise variance $\bar{\sigma}_{\infty,\eta}^2$ is inflated by the excess error of the mismatched \emph{estimator} -- and the rate $R_\eta = I(x; x + \bar{\sigma}_{\infty,\eta}\tilde{z})$ attached to it is the \emph{matched} mutual information of that channel: the rate of a decoder that knows the true prior $p$. Our receiver does not. It scores candidate codewords with the orbital metric $q$, so a second and logically independent mismatch is still unpriced -- that of the \emph{decoding rule} itself. This subsection supplies the missing quantity: the largest rate the actual $q$-based decoder extracts from the channel decoupling handed us. In short, decoupling settles the channel; the \ac{GMI} settles what our metric can extract from it, and the two penalties compose without overlap.
We now derive a single-letter \emph{achievable rate} under the mismatched orbital decoding metric via the \ac{GMI} framework of Merhav et al.~\cite{Merhav1994TIT}, evaluated with the \ac{I-MMSE} machinery of Guo et al.~\cite{Guo2005}.
The mismatched-decoding literature this draws on begins with the \ac{LM} rate of Csisz\'ar and K\"orner~\cite{csiszar_tit_1981}, subsequently extended to multiple-access channels by Lapidoth~\cite{lapidoth_tit_1996}, refined to error exponents and second-order rates~\cite{scarlett_tit_2014}, and completed by a general formula for the mismatch capacity~\cite{somekh_baruch_tit_2015}. We use only the \ac{iid}-ensemble \ac{GMI}, the weakest of these guarantees and the one matching an unconstrained random codebook.

\begin{definition}[Orbital Decoding Metric]
\label{def:cbm_metric}
The orbital decoding metric is the mismatched log-likelihood induced by the orbital prior from Definition~\ref{def:orbital}, given by
\begin{equation}
\label{eq:cbm_metric}
q(\bar{x}\! \mid\! x) \!\triangleq\!
\frac{r_\ell}{\pi\bar{\sigma}^2}
\exp\!\left(\!-\tfrac{\norm{\bar{x}}^2 \!+\! R_\ell^2}{\bar{\sigma}^2}\!\right)
\!\exp\!\left(\kappa_\ell\cos(\phi \!-\! \angle\bar{x})\right)
\bigg|_{\genfrac{}{}{0pt}{}{\ell = \ell(x)}{\phi = \angle x}},
\end{equation}
where $\ell(x) \triangleq \operatorname{argmin}_\ell\norm{\norm{x} - R_\ell}$ identifies the ring of $x$.

The metric is read off the orbital model directly. Under Definition~\ref{def:orbital} ring $\ell$ carries mass $r_\ell$ spread uniformly in phase, so the orbital joint density of a symbol and its observation is $r_\ell$ times the \ac{AWGN} kernel.

The factor $\exp(\kappa_\ell\cos(\phi - \angle\bar{x}))$ is the von Mises angular likelihood (cf.\ eq.~\eqref{eq:vm_like}), which resolves intra-ring phase at cost $O(1)$.
This single metric is the decoding rule of the orbital receiver, common to the whole hierarchy: the \ac{OGD} and \ac{OPD} are estimation shortcuts that replace its von Mises phase factor by a Gaussian surrogate $\exp(-\tfrac12\kappa_\ell(\phi-\angle\bar{x})^2)$ and a hard ring-projection, respectively. These agree with~\eqref{eq:cbm_metric} up to $\mathcal{O}(\kappa_\ell^{-2})$ and induce the same ring-decision regions, so they enter the achievable rate below only through the fixed-point noise $\bar{\sigma}_\infty^2$, at $\mathcal{O}(\sigma_z^4)$.
The same metric also serves, without modification, as the soft-output rule of a bit-interleaved coded receiver~\cite{Caire1998BICM}: bit \acp{LLR} for a soft-input forward-error-correction decoder are read directly from~\eqref{eq:cbm_metric}, and its radial/angular factorization evaluates them at the hierarchy's $\mathcal{O}(L)$-to-$\mathcal{O}(1)$ cost rather than the $\mathcal{O}(M)$ of a full max-log demapper -- with the \ac{GMI} of Corollary~\ref{thm:gmi} furnishing the exact achievable rate of that coded receiver under mismatched (bit-metric) decoding~\cite{Martinez2009BICM}.
\end{definition}

\begin{corollary}[\ac{GMI} of the Orbital Receiver]
\label{thm:gmi}
At the \ac{SE} fixed point with effective noise variance $\bar{\sigma}_\infty^2 = \sigma_z^2 + (K/N)\,\mathrm{MSE}_\infty$, the \ac{GMI} -- the largest rate achievable with an \ac{iid} random-coding ensemble under the mismatched orbital decoding metric (the mismatch capacity itself may be larger, e.g.\ via constant-composition ensembles~\cite{Merhav1994TIT}) -- is given in the standard single-letter form~\cite[eq.~(25)]{Martinez2009BICM} by
\begin{equation}
\label{eq:gmi}
I_{\mathrm{GMI}} = \sup_{\varsigma \geq 0} \left\{ \E_{x,\bar{x}}\!\left[\ln \frac{q(\bar{x} \mid x)^\varsigma}{\E_{x'}\bigl[q(\bar{x} \mid x')^\varsigma\bigr]}\right] \right\},
\end{equation}
where $\bar{x} = x + \bar{\sigma}_\infty \tilde{z}$, $\tilde{z} \sim \CN(0,1)$, and $x, x'$ are drawn independently from the uniform discrete prior over $\mathcal{M}$.
\end{corollary}

\begin{proof}
The expression~\eqref{eq:gmi} is the \ac{GMI} of a mismatched decoder in its standard single-letter form~\cite{Martinez2009BICM}, whose validity for the \emph{continuous} output alphabet at hand follows from the general-alphabet treatment of Ganti, Lapidoth, and Telatar~\cite{Ganti2000}.

\begin{figure}[H]
  \centering
  \includegraphics[width=\columnwidth]{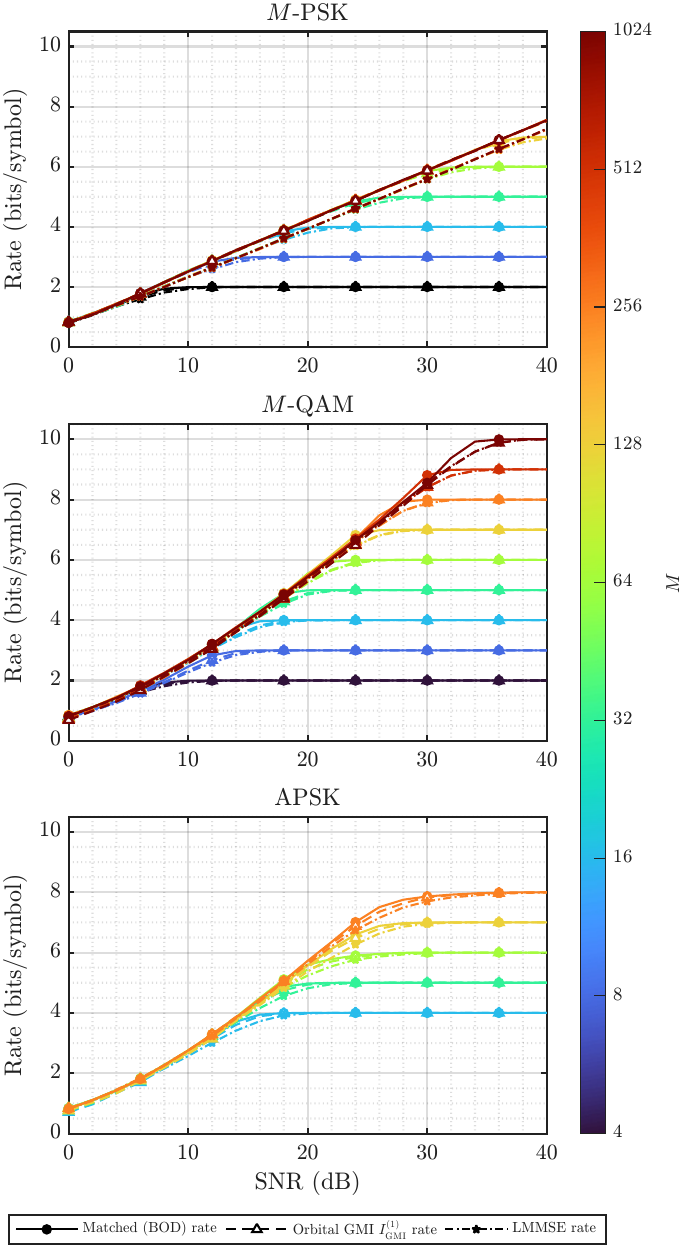}
  \vspace{-4ex}
  \caption{Achievable rate versus \ac{SNR} at $\alpha = 0.5$ for all $M$-\ac{PSK} and $M$-\ac{QAM} orders $M \in \{4,\ldots,1024\}$ and the DVB-S2/S2x \ac{APSK} constellations $M \in \{16,\ldots,256\}$: the $\varsigma = 1$ \ac{GMI} lower bound~\eqref{eq:gmi_lower} of the mismatched orbital receiver at its \ac{SE} fixed point against the matched \ac{BOD} and \ac{LMMSE} rates. Line style encodes the rate curve (shared legend), color the order $M$ (shared colorbar); each pair saturates at its ceiling $\log_2 M$. The mismatch loss vanishes exponentially in \ac{SNR} (Proposition~\ref{prop:immse_rate_gap}).}
  \label{fig:GMI_rate}
  \vspace{-2ex}
\end{figure}

The \ac{SE} framework (Section~\ref{sec:se}) establishes that in the large-system limit, the vector channel decouples into $K$ parallel scalar channels, each with the equivalent model $\bar{x} = x + \bar{\sigma}_\infty \tilde{z}$.
The mismatched decoder employs the metric $q(\bar{x} \mid x)$ from Definition~\ref{def:cbm_metric} rather than the true discrete likelihood.
By the underlying random-coding analysis~\cite{Ganti2000,Martinez2009BICM}, the supremum over the parameter $\varsigma \geq 0$ yields the largest rate achievable with \ac{iid} random-coding under this mismatched metric.
The \ac{LM} rate of~\cite{Merhav1994TIT}, achievable with constant-composition ensembles, dominates and is not pursued here. The reduction to the scalar channel is taken in the distributional sense of Theorem~\ref{thm:decoupling}.
\end{proof}
\vspace{-1ex}

At the \ac{SE} fixed point, the \ac{GMI} admits the following single-letter characterization that connects directly to the \ac{OD} outputs.

\begin{corollary}[Single-Letter GMI at the SE Fixed Point]
\label{cor:gmi_single_letter}
Evaluating~\eqref{eq:gmi} at $\varsigma = 1$ yields the  achievable lower bound
\begin{align}
I_{\mathrm{GMI}} &\geq I_{\mathrm{GMI}}^{(1)} \triangleq \E_{x,\tilde{z}}\!\left[\ln \frac{q(\bar{x} \mid x)}{\E_{x'}[q(\bar{x} \mid x')]}\right] \nonumber \\
&= H(x) - \E_{x,\tilde{z}}\!\left[\ln \frac{\sum_{m=1}^{M} q(\bar{x} \mid x = s_m)}{q(\bar{x} \mid x)}\right],
\label{eq:gmi_lower}
\end{align}
where $H(x) = \ln M$ nats under the (uniform) true prior $p$.

Furthermore, at high \ac{SNR}, we have
\begin{equation}
\label{eq:gmi_high_snr}
I_{\mathrm{GMI}} = \log_2 M - \mathcal{O}\!\left(e^{-c/\sigma_z^2}\right) \quad \text{bits/symbol},
\end{equation}
for a constant $c > 0$ governed by the minimum distance of $\mathcal{M}$: the rate loss vanishes \emph{exponentially} in the \ac{SNR}, consistent with Proposition~\ref{prop:immse_rate_gap} above. (The \ac{MSE} gap of Corollary~\ref{cor:snr_gap} is linear in $\sigma_z^2$; the corresponding \emph{rate} gap is not -- the \ac{I-MMSE} integral that converts one into the other is dominated by exponentially rare decision-boundary events.)
\end{corollary}

\begin{proof}
Since the supremum in~\eqref{eq:gmi} dominates any fixed $\varsigma$, the value at $\varsigma = 1$ is an achievable lower bound~\cite{Merhav1994TIT}; we note it is \emph{not} the \ac{LM} constant-composition rate, which is a distinct (generally larger) quantity.
For the high-\ac{SNR} expansion, split the loss as $\log_2 M - I_{\mathrm{GMI}} = \bigl(\log_2 M - I(x;\bar{x})\bigr) + \bigl(I(x;\bar{x}) - I_{\mathrm{GMI}}\bigr)$ at the fixed-point noise $\bar{\sigma}_\infty^2 = 2\sigma_z^2/(2-\alpha) + \mathcal{O}(\sigma_z^4)$.
The first term is the equivocation of the discrete-input Gaussian channel, which decays as $\mathcal{O}(e^{-d_{\min}^2/(4\bar{\sigma}_\infty^2)})$ since it is governed by the pairwise symbol-error probability.
The second term is the mismatch penalty of the orbital metric: as $\bar{\sigma}_\infty^2 \to 0$ the metric $q(\bar{x} \mid x)$ selects the correct ring up to the exponentially rare misdetections of~\eqref{eq:ring_prob}, and within the ring the von Mises factor separates the true phase from its nearest competitor (angular distance $2\pi/M_\ell$) by the exponentially large ratio $\exp\bigl(\kappa_\ell[1 - \cos(2\pi/M_\ell)]\bigr)$, so this penalty is likewise $\mathcal{O}(e^{-c/\bar{\sigma}_\infty^2})$.
Collecting constants yields~\eqref{eq:gmi_high_snr}, in agreement with the exponential rate-gap bound of Proposition~\ref{prop:immse_rate_gap}.
\end{proof}

\begin{remark}[Relationship to Corollary~\ref{cor:capacity}]
\label{rem:gmi_vs_capacity}
Corollary~\ref{cor:capacity} proves $\lim_{\mathrm{SNR} \to \infty} C_{\mathrm{SE}}(\mathrm{SNR}) = \log_2 M$ but does not quantify the finite-\ac{SNR} rate loss. 
Corollary~\ref{cor:gmi_single_letter} fills this gap: the \ac{GMI} provides the \emph{exact} achievable rate at every operating point, not just asymptotically. 
The rate penalty at finite \ac{SNR} is precisely the \ac{I-MMSE} integral of the excess \ac{MSE} induced by the continuous phase relaxation.
\end{remark}

\begin{remark}[Computational Evaluation]
\label{rem:gmi_computation}
The \ac{GMI} in~\eqref{eq:gmi} can be evaluated numerically via Monte Carlo using the same samples generated for the \ac{SE} recursion~\eqref{eq:se_mismatch}. 
The optimization over $\varsigma$ is one-dimensional and concave, admitting efficient bisection. 
This provides a complete rate-\ac{SNR} characterization of the orbital receiver at negligible additional computational cost beyond the \ac{SE} evaluation.
\end{remark}

The rate evaluations of this section adopt $\alpha = 0.5$, rather than the $\alpha = 0.25$ of the \ac{SE} figures of Section~\ref{sec:se}, for a reason supplied by Corollary~\ref{cor:load_gap}: the margin over the linear baseline is $3.7$~dB at $\alpha = \tfrac14$ but $4.8$~dB at $\alpha = \tfrac12$, so the heavier load separates the \ac{LMMSE} curve from the orbital ones visibly without approaching the unit-load regime in which the linear fixed point ceases to obey $\sigma_z^2/(1-\alpha)$. 

\begin{figure}[H]
  \centering
  \includegraphics[width=\columnwidth]{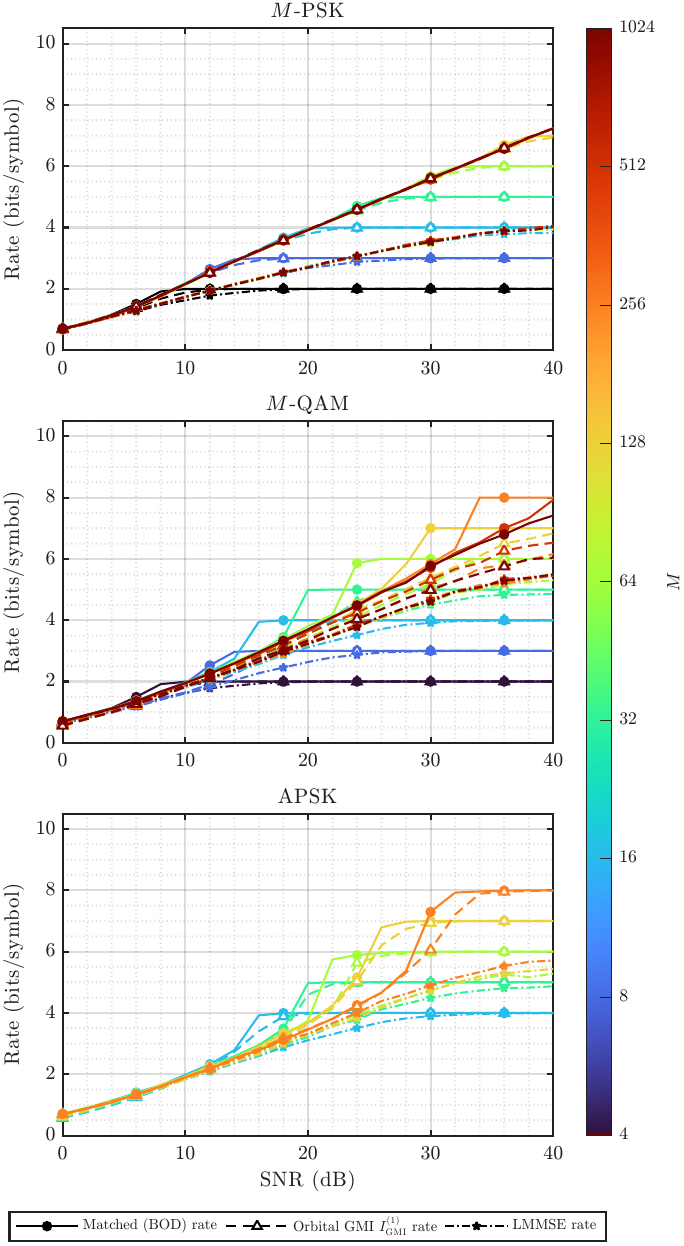}
  \vspace{-3ex}
  \caption{Achievable rate versus \ac{SNR} at $\alpha = 1$ for all $M$-\ac{PSK} and $M$-\ac{QAM} orders $M \in \{4,\ldots,1024\}$ and the DVB-S2/S2x \ac{APSK} constellations $M \in \{16,\ldots,256\}$: the $\varsigma = 1$ \ac{GMI} lower bound~\eqref{eq:gmi_lower} of the mismatched orbital receiver at its \ac{SE} fixed point against the matched \ac{BOD} and \ac{LMMSE} rates.}
  \label{fig:GMI_rate_a1}
  \vspace{-2ex}
\end{figure}

Figure~\ref{fig:GMI_rate_a1} then pushes to $\alpha = 1$ to exhibit that breakdown directly.
Figure~\ref{fig:GMI_rate} carries out this evaluation at $\alpha = 0.5$ across all orders of the three constellation families: the $\varsigma = 1$ \ac{GMI} lower bound $I_{\mathrm{GMI}}^{(1)}$ of~\eqref{eq:gmi_lower} for the orbital receiver, evaluated at its own \ac{SE} fixed point, is plotted against the matched rate $I(x; x + \bar{\sigma}_\infty^{\mathrm{D}} \tilde{z})$ of the \ac{BOD} and the \ac{LMMSE} rate $I(x; x + \bar{\sigma}_\infty^{\mathrm{L}} \tilde{z})$.
For every order the three rate curves are separated by a fraction of a bit at moderate \ac{SNR} and merge as the \ac{SNR} grows -- the rate loss vanishing exponentially, as Proposition~\ref{prop:immse_rate_gap} predicts -- with both attaining $\log_2 M$.

Finally, Fig~\ref{fig:GMI_rate_a1} carries out the same evaluation at $\alpha = 1$.
Unlike the $\alpha < 1$ case, the \ac{LMMSE} rate degrades significantly while the orbital rates maintain a better performance throughout.

\subsection{MMSE Dimension and the I-MMSE Bridge}
\label{sec:MMSE_dim}

The two preceding subsections leave an apparent contradiction on the table, and resolving it is the purpose of this one. 
Section~\ref{sec:asymptotic_gap} showed that the relaxation costs an \ac{MSE} penalty that is \emph{linear} in $\sigma_z^2$ -- a bounded but firmly non-vanishing \ac{ASG}. Corollary~\ref{thm:gmi} and Proposition~\ref{prop:immse_rate_gap} nonetheless report a \emph{rate} penalty that vanishes \emph{exponentially}.
Both are correct, and the reconciliation is the \ac{I-MMSE} identity: rate is the \emph{integral} of the \ac{MMSE} over \ac{SNR}, so a penalty in \ac{MSE} is charged against the rate only in proportion to how much \ac{MMSE} curve remains to be accumulated -- and at the high-\ac{SNR} operating point that curve has already collapsed to exponentially small values.
A linear estimation penalty levied where estimation no longer buys information is an exponentially small rate penalty.
What follows makes this bridge quantitative in both directions: the \ac{MMSE} dimension isolates the estimation-side invariant of the relaxation ($d = \tfrac12$, one freed real coordinate, hence linear decay and a \emph{finite} \ac{ASG}), and the \ac{I-MMSE} integral converts it into the information-side statement. 
The finiteness of $d$ makes the penalty a parallel \ac{SNR} shift rather than a rate ceiling.

\begin{proposition}[MMSE Dimensions of the Denoiser Hierarchy]
\label{prop:mmse_dimension}
For a denoiser $\eta \in \{\mathrm{D}, \mathrm{B}, \mathrm{G}, \mathrm{P}\}$ acting on the scalar channel $\bar{x} = x + \bar{\sigma}\tilde{z}$ (with $x \in \mathcal{M}$, $\tilde{z} \sim \CN(0,1)$) at input signal-to-noise ratio $\mathrm{SNR} = 1/\bar{\sigma}^2$, define its \emph{per-symbol mean-squared error}
\begin{equation}
\label{eq:mmse_eta_def}
\mathrm{MMSE}_\eta(\mathrm{SNR}) \;\triangleq\; \E_{x,\tilde{z}}\!\left[\,\norm{x - \eta(\bar{x};\, 1/\mathrm{SNR})}^2\,\right].
\end{equation}

This equals the true \ac{MMSE} for the Bayes-optimal \ac{BOD} (and for the \ac{OBD}, whose mismatched error reduces to the matched \ac{MMSE} under the orbital prior $q$ by circular symmetry), and the achieved error for the non-Bayes \ac{OGD} and \ac{OPD}. The \emph{\ac{MMSE} dimension}~\cite{WuTIT2011,Renyi1959} of denoiser $\eta$ is then (a quantity that, under mild regularity, equals \emph{half} the R\'enyi information dimension of the input measure on $\Real^2$, the factor $\tfrac12$ being the per-complex-dimension normalization adopted throughout; the operational meaning of the latter is fixed by the almost-lossless analog compression theorem of Wu and Verd\'u~\cite{WuTIT2011})
\begin{equation}
\label{eq:mmse_dim_def}
d_\eta \;\triangleq\; \lim_{\mathrm{SNR}\to\infty}\mathrm{SNR}\cdot\mathrm{MMSE}_\eta(\mathrm{SNR}).
\end{equation}

The discrete \ac{BOD} and the entire orbital hierarchy satisfy
\begin{subequations}
  \begin{eqnarray}
  &d_{\mathrm{D}} = 0,&\\ 
  &d_{\mathrm{B}} = d_{\mathrm{G}} = d_{\mathrm{P}} = \tfrac{1}{2}.&
  \end{eqnarray}
\end{subequations}

The dimensional gap $d_\eta - d_{\mathrm{D}} = 1/2$ per complex dimension is the information-theoretic signature of the continuous phase relaxation; its finiteness guarantees the \ac{ASG} remains bounded (Corollary~\ref{cor:snr_gap}). The value is $\tfrac12$ rather than $1$ because the relaxation frees exactly one of the two real degrees of freedom of a complex symbol -- the phase -- while the radius stays pinned to the discrete ring set; the estimator therefore pays the \ac{MMSE}-dimension price of a single continuous real coordinate.
The value $\tfrac12$ is shared by all three orbital denoisers because it is fixed by the common leading-order decay $\hat\sigma^2 \sim \bar\sigma^2/2$ (Corollary~\ref{cor:snr_gap}), independent of which denoiser realizes it.
\end{proposition}
\begin{proof}
{By~\eqref{eq:mmse_eta_def}, $\mathrm{MMSE}_\eta(\mathrm{SNR}) = \E[\hat{\sigma}^2_\eta]$ evaluated at $\bar\sigma^2 = 1/\mathrm{SNR}$, where $\hat\sigma^2_\eta$ is the per-symbol error of the corresponding level: {the exponential decay $\mathcal{O}(e^{-c/\bar\sigma^2})$ for the \ac{BOD} {(established with~\eqref{eq:bod_fixed_point})}, the linear decay of Proposition~\ref{prop:linear_decay} for the \ac{OBD}, and the shared leading order $\bar\sigma^2/2$ of Corollary~\ref{cor:snr_gap} for the \ac{OGD} and \ac{OPD}}.}
For the \ac{BOD}, exponential error decay
{$\mathrm{MMSE}_{\mathrm{D}} = \mathcal{O}(e^{-c/\bar\sigma^2}) = \mathcal{O}(e^{-c\,\mathrm{SNR}})$} {(the minimum-distance decay established with~\eqref{eq:bod_fixed_point})}
gives $\mathrm{SNR}\cdot\mathrm{MMSE}_{\mathrm{D}}\to0$, hence $d_{\mathrm{D}}=0$.
For the \ac{OBD}, linear decay {$\mathrm{MMSE}_{\mathrm{B}} \sim \bar\sigma^2/2 = 1/(2\,\mathrm{SNR})$
(Proposition~\ref{prop:linear_decay}) gives
$\mathrm{SNR}\cdot\mathrm{MMSE}_{\mathrm{B}} = \mathrm{SNR}\cdot\tfrac{1}{2\,\mathrm{SNR}} \to \tfrac{1}{2}$},
a finite positive constant, so $d_{\mathrm{B}}=1/2$.
{The \ac{OGD} and \ac{OPD} share the same leading-order decay $\mathrm{MMSE}_\eta \sim \bar\sigma^2/2 = 1/(2\,\mathrm{SNR})$ (Corollary~\ref{cor:snr_gap}); their sub-leading offsets are $\mathcal{O}(\bar\sigma^4) = \mathcal{O}(\mathrm{SNR}^{-2})$ and vanish in the product $\mathrm{SNR}\cdot\mathrm{MMSE}_\eta$, so $d_{\mathrm{G}} = d_{\mathrm{P}} = 1/2$ identically.}
\end{proof}

More generally, relaxing the phase on only a subset $\mathcal{A}$ of the rings -- naturally the dense, high-$M_\ell$ rings, where the mismatch is cheapest by Proposition~\ref{prop:cbm_error}, while keeping the sparse rings discrete -- realizes any intermediate \ac{MMSE} dimension $d = \tfrac{1}{2}\sum_{\ell \in \mathcal{A}} r_\ell \in [0,\tfrac{1}{2}]$ -- an adaptive complexity-accuracy dial whose selective per-ring preservation of the discrete amplitude/phase structure is useful for many applications.

The per-ring phase likelihood induced by the effective channel $\bar{x} = x + \bar{z}$ is von Mises, $\propto \exp(\kappa_\ell\cos(\angle\bar{x} - \theta))$ with concentration $\kappa_\ell = 2R_\ell\norm{\bar{x}}/\bar{\sigma}^2$; this is exact and independent of the denoiser. 
At high concentration, it is a wrapped Gaussian, so the per-ring angular estimation decouples into a scalar \ac{AWGN} channel $\Delta\theta_{\mathrm{obs}} = \Delta\theta_{\mathrm{true}} + n_\theta$, $n_\theta\sim\mathcal{N}(0,1/\kappa_\ell)$, with $\kappa_\ell$ playing the role of effective (phase) \ac{SNR}. 
The \ac{OGD} adopts this Gaussian channel exactly, while the \ac{OBD} uses the exact von Mises posterior and the \ac{OPD} the raw observed phase; but all three share the same effective \ac{SNR} $\kappa_\ell$. 
The \ac{I-MMSE} identity~\cite{Guo2005} on this phase channel then bounds the information loss that the continuous relaxation -- at any level of the hierarchy -- incurs relative to the discrete constellation.
Both applications of the identity in this section are \emph{matched} ones: each integrand is the \ac{MMSE} of the prior that defines the corresponding mutual information, with the mismatch confined to the limits of integration in~\eqref{eq:rate_gap_correct} and to the choice of prior in~\eqref{eq:immse_bound}. This matters because the genuinely mismatched counterpart is a different statement: the \ac{SNR}-integrated \emph{excess} error of an estimator built on the wrong prior equals twice the relative entropy between the two distributions~\cite{verdu_tit_2010}, a quantity that is \emph{infinite} here, since the discrete $p$ is not absolutely continuous with respect to its continuous relaxation $q$ -- which is precisely why Section~\ref{sec:ot_bound} bounds the excess pointwise in \ac{SNR}, through a transport distance that remains finite, rather than integrating it. Related representations of mutual information via input estimates~\cite{palomar_tit_2007} and functional properties of the \ac{MMSE}~\cite{wu_tit_2012} complete the Gaussian picture, and the identity itself extends beyond Gaussian observations to Poisson~\cite{guo_tit_2008} and, more generally, L\'evy channels~\cite{jiao_tit_2017}, where the squared error is replaced by the Bregman divergence generated by the channel's cumulant generating function.

\begin{corollary}[\ac{I-MMSE} Characterization of the Phase Loss]
\label{cor:immse}
For any orbital denoiser $\eta \in \{\mathrm{B}, \mathrm{G}, \mathrm{P}\}$, the per-ring mutual information loss of the continuous phase relaxation relative to the discrete $M_\ell$-ary phase input is
\begin{equation}
\label{eq:immse_loss}
\Delta I_\ell^{\eta} = \int_{\kappa_0}^{\kappa_{\max}} \frac{1}{2}\Bigl[\mathrm{MMSE}^{\mathrm{ph}}_{\eta}(\kappa) - \mathrm{MMSE}^{\mathrm{ph}}_{\mathrm{D}}(\kappa)\Bigr]\, d\kappa,
\end{equation}
where $\kappa_0$ and $\kappa_{\max}$ are the endpoints of the operating \ac{SNR} range, $\mathrm{MMSE}^{\mathrm{ph}}_\eta(\kappa)$ is the per-ring \emph{phase} \ac{MMSE} at angular concentration $\kappa$, and $\mathrm{MMSE}^{\mathrm{ph}}_{\mathrm{D}}(\kappa) \sim \exp\bigl(-2\kappa\sin^2(\pi/(2M_\ell))\bigr)$ is the discrete-input phase \ac{MMSE} (which coincides with the \ac{MAP} error up to $\mathcal{O}(e^{-c\kappa})$ at high concentration), governed by the von Mises boundary-crossing rate $\kappa[1 - \cos(\pi/M_\ell)]$ at the half-spacing decision boundary $\pi/M_\ell$. All three orbital denoisers share the leading value $\mathrm{MMSE}^{\mathrm{ph}}_\eta(\kappa) = 1/\kappa + \mathcal{O}(\kappa^{-2})$ (exactly $1/\kappa$ for the Gaussian-phase \ac{OGD}).

Since $\mathrm{MMSE}^{\mathrm{ph}}_{\eta} \geq \mathrm{MMSE}^{\mathrm{ph}}_{\mathrm{D}}$ pointwise, the integrand is nonnegative and, over any finite range $[\kappa_0,\kappa_{\max}]$, bounded above by $1/(2\kappa)$ up to the integrable $\mathcal{O}(\kappa^{-2})$ correction, so
\begin{equation}
\label{eq:immse_bound}
\Delta I_\ell^{\eta} \leq \frac{1}{2}\ln\!\left(\frac{\kappa_{\max}}{\kappa_0}\right) + \mathcal{O}(1),
\end{equation}
which is finite for any bounded operating \ac{SNR} range and identical across the hierarchy to leading (logarithmic) order -- with the $\mathcal{O}(1)$ term absent for the \ac{OGD}, whose phase \ac{MMSE} is exactly $1/\kappa$.
\end{corollary}

\begin{proof}
The \ac{I-MMSE} relationship~\cite{Guo2005} states that for a scalar Gaussian channel the derivative of mutual information with respect to the \ac{SNR} equals half the \ac{MMSE}: $dI/d\mathrm{SNR} = \frac{1}{2}\mathrm{MMSE}(\mathrm{SNR})$.
The per-ring angular channel is a scalar Gaussian channel with \ac{SNR} parameter $\kappa_\ell$ -- exactly for the \ac{OGD}, and asymptotically ($\kappa_\ell \gg 1$) for the von Mises \ac{OBD} and the \ac{OPD} -- so the mutual information difference between the continuous relaxation $\eta$ and the discrete input is the integral of the phase-\ac{MMSE} difference over $\kappa$.
With $\mathrm{MMSE}^{\mathrm{ph}}_\eta(\kappa) = 1/\kappa + \mathcal{O}(\kappa^{-2})$ (the shared leading value, $=1/\kappa$ exactly for the \ac{OGD}) and the discrete error decaying exponentially once $\kappa \gtrsim 1/(2\sin^2(\pi/(2M_\ell)))$, the excess integrand is bounded by $1/(2\kappa) + \mathcal{O}(\kappa^{-2})$, whose integral over $[\kappa_0,\kappa_{\max}]$ is $\tfrac{1}{2}\ln(\kappa_{\max}/\kappa_0) + \mathcal{O}(1)$, giving~\eqref{eq:immse_bound}. The $1/\kappa$ term is not integrable on $[\kappa_0,\infty)$, so the bound is finite only over a bounded \ac{SNR} range.
\end{proof}

\begin{remark}[Operational Interpretation]
\label{rem:immse_interpretation}
Corollary~\ref{cor:immse} explains why the orbital relaxation -- at any level of the hierarchy -- incurs only a controlled phase-information loss: the Gaussian phase maximizes the \ac{MMSE} at any given~$\kappa$ (by the maximum entropy property), so it upper-bounds the discrete \ac{MMSE}, {and every orbital denoiser shares the same leading phase $\mathrm{MMSE}^{\mathrm{ph}}_\eta(\kappa) = 1/\kappa$ (their $\mathcal{O}(\kappa^{-2})$ differences stemming from the common expansion $1-A(\kappa) = 1/(2\kappa) + \mathcal{O}(\kappa^{-2})$).}
The \ac{I-MMSE} integral of this excess \ac{MMSE} is finite over any bounded \ac{SNR} range, growing only logarithmically in the range width~\eqref{eq:immse_bound}, so the per-ring information loss stays controlled at every finite operating point.
\end{remark}

\subsection{Geometry-to-Rate Chain: An Optimal Transport Bound}
\label{sec:ot_bound}

Everything established so far shares a common limitation. The rate formula, the \ac{GMI}, the rate gap and the \ac{MMSE} dimension are all statements \emph{at the \ac{SE} fixed point}, and all but the first are \emph{asymptotic} in $\sigma_z^2$: they describe what the relaxation costs once the recursion has been run and the noise driven small. None of them can be evaluated from the constellation before a receiver is built, and none certifies anything at a moderate operating point. This final subsection removes both restrictions at once, bounding the cost by a quantity computed from the ring geometry $\{R_\ell, M_\ell\}$ alone, non-asymptotically and uniformly in \ac{SNR}. It is what turns the preceding analysis into a design rule.

Proposition~\ref{prop:cbm_error} quantifies the geometric quality of the \ac{OBD} approximation via the Wasserstein-1 distance $W_1(p_{\ell}, q_{\ell})$.
Corollary~\ref{cor:snr_gap} characterizes the macroscopic \ac{MSE} gap at high \ac{SNR}. 
We now close the loop between these two results by deriving a \emph{non-asymptotic} bound on the \ac{SE} fixed-point gap in terms of the Wasserstein distance, valid at all \ac{SNR}. 
The natural metric here is the Wasserstein-2 distance -- the companion of the squared-error \ac{MSE} -- rather than the $W_1$ of Proposition~\ref{prop:cbm_error}; the two share the same $\Theta(R_\ell/M_\ell)$ per-ring transport map and differ only in the cost exponent, so $W_2^2 = \Theta(R_\ell^2/M_\ell^2)$ is the exact squared-cost analogue of the $W_1$ result.

\begin{theorem}[Transport-Theoretic Bound on the \ac{ASG}]
\label{thm:ot_bound}
Let $\eta_D$ and $\eta_B$ denote the \ac{BOD} and \ac{OBD}, and define the aggregate squared Wasserstein-2 distance between the discrete prior and its orbital relaxation as
\begin{equation}
\label{eq:aggregate_w2}
\overline{W}_2^{\,2} \;\triangleq\; \sum_{\ell=1}^{L} r_\ell \, W_2\bigl(p_{\ell},\, q_{\ell}\bigr)^2,
\end{equation}
where, for uniform-phase rings, the per-ring distance is \emph{exact} (by the same monotone-rearrangement argument as Proposition~\ref{prop:cbm_error}):
\begin{equation}
\label{eq:w2_ring}
W_2\bigl(p_{\ell},\, q_{\ell}\bigr)^2
= 2R_\ell^2\left(1 - \frac{\sin(\pi/M_\ell)}{\pi/M_\ell}\right)
= \Theta\!\left(\frac{R_\ell^2}{M_\ell^2}\right).
\end{equation}

\noindent\emph{(i) Operating-point excess (unconditional).} At any common effective noise $\bar{\sigma}^2 > 0$, the excess \ac{MSE} of the orbital denoiser over the optimal one obeys the exact identity and bound
\begin{equation}
\label{eq:ot_mse_bound}
\mathcal{F}_{\mathrm{B}}(\mathrm{MSE}) - \mathcal{F}_{\mathrm{D}}(\mathrm{MSE})
= \E\bigl[\norm{\eta_D(\bar{x}) - \eta_B(\bar{x})}^2\bigr]
\;\leq\; 4\,\overline{W}_2^{\,2},
\end{equation}
uniformly in $\bar{\sigma}^2$ -- the identity by the orthogonality principle, the bound by Gaussian-smoothing arguments (Tweedie's identity, the relative de Bruijn identity, and the joint convexity of relative entropy; see Appendix~\ref{app:thm_ot_bound} for the precise statements and attributions).
Transport-type continuity of information measures has precedent in this journal~\cite{polyanskiy_tit_2016}, and the de Bruijn/Fisher chain we use is the same machinery that underlies information-theoretic proofs of the entropy power inequality~\cite{rioul_tit_2011}, itself obtainable directly from the \ac{I-MMSE} relation~\cite{verdu_tit_2006}.

\noindent\emph{(ii) Fixed-point transfer.} Whenever the \ac{BOD} \ac{SE} map is contractive with modulus $c_{\mathrm{D}} \triangleq \sup_{[\mathrm{MSE}_\infty^{\mathrm{D}},\, \mathrm{MSE}_\infty^{\mathrm{B}}]} \mathcal{F}_{\mathrm{D}}' < 1$ (the high-noise and high-\ac{SNR} regimes of Corollary~\ref{cor:se_uniqueness}, where $c_{\mathrm{D}} \leq \alpha$ and $c_{\mathrm{D}} \to 0$, respectively),
\begin{equation}
\label{eq:ot_fp_bound}
\mathrm{MSE}_\infty^{\mathrm{B}} - \mathrm{MSE}_\infty^{\mathrm{D}}
\;\leq\; \frac{4\,\overline{W}_2^{\,2}}{1 - c_{\mathrm{D}}}.
\end{equation}
\end{theorem}

\begin{proof}
\proofref{app:thm_ot_bound}
\end{proof}

\begin{remark}[Prior-Induced Gap Is Shared; Denoiser Adds $\mathcal{O}(\sigma_z^4)$]
\label{rem:ot_hierarchy}
The transport bound~\eqref{eq:ot_mse_bound} measures the mismatch of the orbital \emph{prior} $q$ against the discrete truth $p$, a geometric quantity common to the whole hierarchy since \ac{OBD}, \ac{OGD}, and \ac{OPD} all target the same prior $q$. For the \ac{OGD} and \ac{OPD} the gap to the optimal receiver splits as
\begin{equation*}
\mathrm{MSE}_\infty^{\eta} \!- \mathrm{MSE}_\infty^{\mathrm{D}}
= \underbrace{(\mathrm{MSE}_\infty^{\eta} \!-\! \mathrm{MSE}_\infty^{\mathrm{B}})}_{\mathcal{O}(\sigma_z^4)\ \text{(Thm.~\ref{thm:cross_level_fp})}}
\!+\! \underbrace{(\mathrm{MSE}_\infty^{\mathrm{B}} \!-\! \mathrm{MSE}_\infty^{\mathrm{D}})}_{\text{bounded by~\eqref{eq:ot_fp_bound}}},
\end{equation*}
for $\eta \in \{\mathrm{G},\mathrm{P}\}$; hence~\eqref{eq:ot_fp_bound} governs the prior-geometry contribution for every level, with only the small additive denoiser term distinguishing them.
\end{remark}

\begin{figure*}[t]
  \centering
  \includegraphics[width=2\columnwidth]{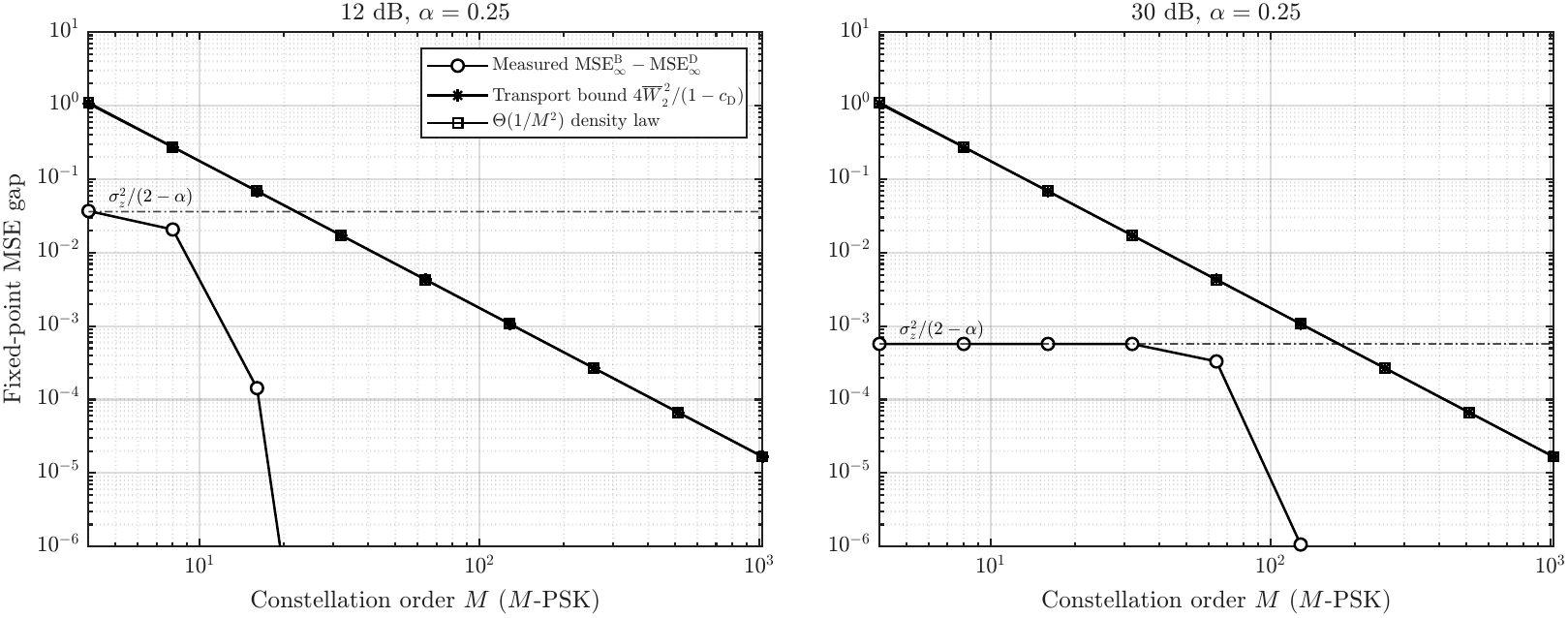}
  \caption{\ac{SE} fixed-point gap $\mathrm{MSE}_\infty^{\mathrm{B}} - \mathrm{MSE}_\infty^{\mathrm{D}}$ versus constellation order $M$ for $M$-\ac{PSK} at $12$~dB (left) and $30$~dB (right), $\alpha = 0.25$, computed by deterministic Gauss--Hermite quadrature, against the \ac{SNR}-uniform transport bound~\eqref{eq:ot_fp_bound} of Theorem~\ref{thm:ot_bound} -- the same $\Theta(1/M^2)$ density law of Corollary~\ref{cor:ot_scaling} in both panels. The dash-dotted horizontal line marks the \ac{ASG} value $\sigma_z^2/(2-\alpha)$, at which the measured gap plateaus before its superexponential collapse beyond $M \approx \pi/\bar{\sigma}_\infty$ (where the \ac{BOD}'s own resolution fails); gaps below the $10^{-6}$ numerical resolution are omitted.}
  \label{fig:OT_bound}
  \vspace{-2ex}
\end{figure*}

\begin{corollary}[Scaling of the \ac{ASG} with Ring Density]
\label{cor:ot_scaling}
Substituting the exact per-ring distance~\eqref{eq:w2_ring} into~\eqref{eq:ot_fp_bound} yields the single, \ac{SNR}-uniform density law
\begin{align}
\label{eq:ot_scaling}
\mathrm{MSE}_\infty^{\mathrm{B}} - \mathrm{MSE}_\infty^{\mathrm{D}}
  &\leq \frac{8}{1 - c_{\mathrm{D}}} \sum_{\ell=1}^{L} r_\ell R_\ell^2 \left(1 - \frac{\sin(\pi/M_\ell)}{\pi/M_\ell}\right) \nonumber \\
  &= \mathcal{O}\!\left(\max_\ell \frac{R_\ell^2}{M_\ell^2}\right).
\end{align}

For $M$-\ac{PSK} ($L = 1$, $M_1 = M$, $R_1 = 1$), the bound evaluates to $\tfrac{8}{1-c_{\mathrm{D}}}\bigl(1 - \tfrac{\sin(\pi/M)}{\pi/M}\bigr) \approx \tfrac{4\pi^2}{3(1-c_{\mathrm{D}})}\,\tfrac{1}{M^2}$: the geometric price of the relaxation decays \emph{quadratically} in the constellation density, uniformly over the operating \ac{SNR}.
\end{corollary}
\begin{proof}
Immediate from~\eqref{eq:ot_fp_bound},~\eqref{eq:aggregate_w2}, and~\eqref{eq:w2_ring}, using $1 - \mathrm{sinc}(x) = x^2/6 + \mathcal{O}(x^4)$.
\end{proof}

Equation~\eqref{eq:w2_ring} presumes uniform-phase rings. That holds for $M$-\ac{PSK} and for the \ac{APSK} families, whose rings carry equally spaced points, but \emph{not} for square \ac{QAM}. The following proposition removes the hypothesis, at the cost of one scalar per ring.

\begin{proposition}[Per-Ring Transport Distance for Arbitrary Phase Sets]
\label{prop:w2_general}
Let ring $\ell$ carry $M_\ell$ equal-mass symbols at arbitrary angles $\{\phi_m\}_{m=1}^{M_\ell}$, and let $q_\ell$ be its uniform-phase relaxation. Then
\begin{equation}
\label{eq:w2_general}
W_2\bigl(p_\ell,\, q_\ell\bigr)^2 = 2R_\ell^2\Bigl[\,1 - C_\ell\,\tfrac{\sin(\pi/M_\ell)}{\pi/M_\ell}\Bigr],
\end{equation}
where the \emph{phase-regularity factor}
\begin{equation}
\label{eq:w2_C}
C_\ell \triangleq \max_{\varphi \in [0,2\pi)}\ \frac{1}{M_\ell}\sum_{m=1}^{M_\ell}\cos\psi_m(\varphi) \in [0,1]
\end{equation}
is determined by the offsets $\psi_m(\varphi)$ of the symbols from the centres of the $M_\ell$ equal-length arcs of a cut at $\varphi$. One has $C_\ell = 1$ -- and hence~\eqref{eq:w2_ring} -- if and only if the ring is equidistributed. 

For a square-\ac{QAM} ring generated by a lattice point $(a,b)$ with $a > b > 0$, the dihedral orbit has $M_\ell = 8$ and
\begin{equation}
\label{eq:w2_qam_C}
C_\ell = \cos\!\left(\arctan\tfrac{b}{a} - \tfrac{\pi}{8}\right),
\end{equation}
so such a ring is equidistributed when $b/a = \tan(\pi/8)$, with rings with $a = b$ or $b = 0$ having $M_\ell = 4$ and $C_\ell = 1$.

Equation~\eqref{eq:w2_qam_C} presumes that the ring is a single dihedral orbit. A radius admitting several essentially distinct representations as a sum of two odd squares instead carries a \emph{union} of orbits, the first being $R_\ell^2 = 50$ in $64$-\ac{QAM}, where $(1,7)$ and $(5,5)$ together give $M_\ell = 12$; in $256$-\ac{QAM} the radii $R_\ell^2 \in \{130, 170, 250\}$ each carry $M_\ell = 16$. Such rings are covered by the general factor~\eqref{eq:w2_C}, which is what the values reported in Remark~\ref{rem:w2_families} evaluate.
\end{proposition}
\begin{proof}
The squared-chord cost $\norm{R_\ell e^{j\theta} - R_\ell e^{j\phi}}^2 = 2R_\ell^2(1-\cos(\theta-\phi))$ is increasing in angular distance on $[0,\pi]$, so the optimal circular coupling is monotone; matching the equal atom masses $1/M_\ell$ forces the $M_\ell$ arcs to have equal length $2\pi/M_\ell$, leaving only the cut position $\varphi$ free.

Writing $u$ for the displacement from the centre of the arc assigned to $\phi_m$, that arc contributes
\begin{equation}
\tfrac{M_\ell}{2\pi}\!\!\!\!\!\!\int\limits_{-\pi/M_\ell}^{\pi/M_\ell}\!\!\!\!\!\! 2R_\ell^2\bigl(1-\cos(u - \psi_m)\bigr)\,\mathrm{d}u = 2R_\ell^2\Bigl[1 - \tfrac{\sin(\pi/M_\ell)}{\pi/M_\ell}\cos\psi_m\Bigr],
\end{equation}
since the $\sin(u-\psi_m)$ term integrates to zero over the symmetric interval. Averaging over $m$ with weights $1/M_\ell$ and minimizing the result over $\varphi$ -- equivalently maximizing $\frac{1}{M_\ell}\sum_m\cos\psi_m$ -- gives~\eqref{eq:w2_general} and~\eqref{eq:w2_C}. Equality $C_\ell = 1$ requires every $\psi_m = 0$, i.e.\ the symbols sit at the arc centres, which is equidistribution. For the \ac{QAM} dihedral orbit the eight angles are $\pm\theta_\ell,\ \pm(\tfrac{\pi}{2}-\theta_\ell)$ and their $\pi$-shifts with $\theta_\ell = \arctan(b/a)$; the gaps alternate between $2\theta_\ell$ and $\tfrac{\pi}{2}-2\theta_\ell$, and by that symmetry the optimal cut leaves $\psi_m = \pm(\theta_\ell - \pi/8)$, whence~\eqref{eq:w2_qam_C}.
\end{proof}

Since $C_\ell \le 1$, the general value~\eqref{eq:w2_general} never falls below~\eqref{eq:w2_ring}: uniform-phase rings are the most favourable case, and using~\eqref{eq:w2_ring} outside that case would understate the transport distance.

\begin{remark}[Consequence for the Constellation Families]
\label{rem:w2_families}
Expanding $1-\mathrm{sinc}(x) = x^2/6 + \mathcal{O}(x^4)$ in~\eqref{eq:w2_general} gives the governing law
\begin{equation}
\label{eq:w2_law}
\overline{W}_2^{\,2} \;\approx\; \frac{\pi^2}{3}\sum_{\ell} r_\ell\,\frac{R_\ell^2}{M_\ell^{2}}
\qquad (\text{equidistributed rings}),
\end{equation}
the energy-weighted mean of $R_\ell^2/M_\ell^2$: \emph{the bound decays only if the number of symbols \emph{per ring} grows}. Three regimes follow.
\begin{itemize}
  \item \textbf{$M$-\ac{PSK}} ($L=1$, $M_1 = M$): $\overline{W}_2^{\,2} \to \pi^2/(3M^2)$, the $\Theta(1/M^2)$ law of Corollary~\ref{cor:ot_scaling}. Numerically $\overline{W}_2^{\,2}M^2 = 3.29$ for $M \ge 64$, against $\pi^2/3 = 3.290$.
  \item \textbf{\ac{APSK}}: rings are equidistributed, so~\eqref{eq:w2_ring} is exact and the closed form applies verbatim; with $M_\ell = \Theta(M/L)$ the bound is $\Theta(L^2/M^2)$, still vanishing in $M$ at fixed $L$.
  \item \textbf{$M$-\ac{QAM}}: here $M_\ell$ stays bounded ($4$ to $16$) however large $M$ becomes, because the constellation responds to larger $M$ by \emph{adding} rings rather than filling them ($L = \Theta(M/\sqrt{\ln M})$, Proposition~\ref{prop:ring_scaling}). Consequently $\overline{W}_2^{\,2} = \Theta(1)$: evaluated exactly via~\eqref{eq:w2_general}, it is $0.128,\ 0.089,\ 0.094,\ 0.083$ for $M = 16, 64, 256, 1024$ -- essentially flat over two decades.
\end{itemize}
The transport certificate is therefore informative precisely for the ring-structured constellations the orbital prior is designed for, and is not asymptotically informative for $M$-\ac{QAM}, where a uniform-phase relaxation of a ring holding eight symbols is a poor approximation at any $M$. This is the same asymmetry recorded in Corollary~\ref{cor:compression_scaling}, seen through the transport lens rather than the payload one.
\end{remark}

\begin{remark}[Non-Asymptotic Nature]
\label{rem:non_asymptotic}
Unlike Corollary~\ref{cor:snr_gap}, which characterizes the \ac{ASG} only in the high-\ac{SNR} limit, Theorem~\ref{thm:ot_bound} provides a bound valid at \emph{all} \ac{SNR} values. 
This is particularly useful for system design: given a target \ac{MSE} gap tolerance $\Delta_{\mathrm{max}}$, one can determine the minimum constellation density $M_\ell$ per ring required to guarantee $\mathrm{MSE}_\infty^{\mathrm{B}} - \mathrm{MSE}_\infty^{\mathrm{D}} \leq \Delta_{\mathrm{max}}$ at any operating point.
\end{remark}

\begin{remark}[Closing the Loop: Geometry $\to$ Performance]
\label{rem:closing_loop}
Theorem~\ref{thm:ot_bound} establishes the complete causal chain
\begin{equation*}
\begin{gathered}
\underbrace{{W_2(p_\ell, q_\ell)^2 = \mathcal{O}\!\left(\tfrac{R_\ell^2}{M_\ell^2}\right)}}_{\text{{Eq.~\eqref{eq:w2_ring} (cf.\ Prop.~\ref{prop:cbm_error}): Geometric mismatch}}} \\[0ex]
\Downarrow \\[0ex]
\underbrace{{\mathrm{MSE}_\infty^{\mathrm{B}} - \mathrm{MSE}_\infty^{\mathrm{D}}
\leq \tfrac{4\,\overline{W}_2^{\,2}}{1 - c_{\mathrm{D}}} = \mathcal{O}\!\left(\max_\ell \tfrac{R_\ell^2}{M_\ell^2}\right)}}_{\text{Theorem~\ref{thm:ot_bound}: MSE gap}} \\[0ex]
\Downarrow \\[-2ex]
\underbrace{\Delta R \!=\! {\frac{1}{\ln 2}}\int_{1/\bar{\sigma}_\infty^2}^{1/(\bar{\sigma}_\infty^{\mathrm{D}})^2} \!\!\!\!\!\!\!\mathrm{mmse}_p(\zeta)\, \mathrm{d}\zeta \!=\! \mathcal{O}\!\left(e^{-c\,\mathrm{SNR}}\right) \;\text{bits/symbol}}_{\text{Proposition~\ref{prop:immse_rate_gap}: Rate loss (exponentially vanishing)}}
\end{gathered}
\end{equation*}
connecting the constellation geometry (Section~\ref{sec:orbital}) to the macroscopic detection performance (Section~\ref{sec:se_fixedpoint}) to the information-theoretic rate (Section~\ref{sec:gmi}) in a single unified framework. Each arrow is quantified by an explicit, computable bound.

\vspace{-0.5ex}
\emph{Regime note.}
The transport bound~\eqref{eq:ot_fp_bound} is uniform over the operating \ac{SNR} and quantifies the $M$-dependence of the gap: $\mathcal{O}(\max_\ell R_\ell^2/M_\ell^2)$ (Corollary~\ref{cor:ot_scaling}).
The \ac{SNR}-dependence at fixed $M$ is instead governed by Theorem~\ref{thm:cross_level_fp} (Corollary~\ref{cor:snr_gap}): at high \ac{SNR} the true gap decays to zero with $\sigma_z^2$, so the \ac{SNR}-uniform transport bound is increasingly conservative there.
The two analyses are complementary: one resolves the density axis, the other the noise axis.
\end{remark}
\vspace{-1ex}

Figure~\ref{fig:OT_bound} closes the section with the transport bound in action: for $M$-\ac{PSK}, the \ac{SE} fixed-point gap $\mathrm{MSE}_\infty^{\mathrm{B}} - \mathrm{MSE}_\infty^{\mathrm{D}}$ -- computed by deterministic Gauss--Hermite quadrature of the \ac{SE} maps, so that gaps far below any Monte Carlo noise floor are resolved -- is swept over the constellation order $M$ at two operating points and compared against the \ac{SNR}-uniform bound~\eqref{eq:ot_fp_bound} evaluated with the exact per-ring distance~\eqref{eq:w2_ring}.
Because the bound does not depend on the operating point, the \emph{same} $\Theta(1/M^2)$ curve upper-bounds both panels -- the single density law of Corollary~\ref{cor:ot_scaling}.
The bound holds uniformly at every order and both operating points, while the measured gap traces exactly the physics the theory predicts: at high \ac{SNR} it first \emph{plateaus} at the \ac{ASG} value $\sigma_z^2/(2-\alpha)$ -- there $\mathrm{MSE}_\infty^{\mathrm{D}}$ is exponentially negligible, so the gap \emph{is} the orbital fixed point of Corollary~\ref{cor:snr_gap} -- and then collapses superexponentially once the \ac{BOD} itself loses its resolution advantage. This collapse is driven not by the orbital estimator (whose plateau error $\bar{\sigma}_\infty^2/2$ is essentially $M$-independent) but by $\mathrm{MSE}_\infty^{\mathrm{D}}$ \emph{rising} to meet it: the half-symbol spacing $\sin(\pi/M)$ -- the distance from a symbol to its decision boundary -- falls to the noise scale once $M \gtrsim \pi/\bar{\sigma}_\infty$ (numerically $M \approx 12$ at $12$~dB and $M \approx 93$ at $30$~dB, matching the plot), beyond which discrete and continuous detection are equally noise-limited (gaps below the numerical resolution of the quadrature are omitted from the plot).
The bound is conservative -- as any estimate built from the worst-case Lipschitz observable must be -- but correctly certifies that the geometric mismatch, not any property of the iterative algorithm, governs the cost of the relaxation.
This delimits the theorem's role precisely: it is a \emph{structural} certificate -- computable from the constellation geometry $\{R_\ell, M_\ell\}$ alone, before any receiver is built or any recursion run, and valid at every \ac{SNR} where the asymptotic characterizations of Sections~\ref{sec:se}--\ref{sec:decoupling} are silent -- and, although its proof exploits Gaussian smoothing (Tweedie, de Bruijn), the geometry enters only through the Wasserstein distance of the \emph{unsmoothed} priors, which charges in full for the high-frequency difference between $M_\ell$ spikes and a continuum that the observation kernel in fact annihilates.
Sharpening the bound to track the collapse would require replacing the unsmoothed $W_2(p_\ell, q_\ell)$ by its noise-convolved counterpart $W_2(p_\ell * \CN(0,\bar{\sigma}^2),\, q_\ell * \CN(0,\bar{\sigma}^2))$, which shrinks once the observation kernel blurs adjacent symbols ($M \gtrsim \pi/\bar{\sigma}$) -- precisely the regime in which the channel can no longer distinguish the discrete prior from its continuous relaxation; we leave this refinement to future work.

\begin{remark}[The Complexity-Cost Inversion]
\label{rem:complexity_cost_inversion}
Figure~\ref{fig:OT_bound} exposes the alignment on which the entire framework rests, and it deserves to be stated plainly.
The cost of the orbital relaxation is governed by whether the effective noise can resolve adjacent phases.
For $M \lesssim \pi/\bar{\sigma}_\infty$, the discrete detector can still ``snap'' to the correct symbol -- its decision-boundary distance $\sin(\pi/M)$ exceeds the noise scale -- while the orbital receiver cannot, and the relaxation costs the full -- but bounded -- \ac{ASG} $\sigma_z^2/(2-\alpha)$: a cap, never a floor.
For $M \gtrsim \pi/\bar{\sigma}_\infty$, adjacent symbols fall within the noise scale, the channel itself erases the phase discreteness before the receiver ever sees it, and discarding that discreteness costs nothing measurable.
The complexity of exact detection, by contrast, \emph{grows} linearly in $M$.
The two trends are inversely aligned: at large $M$ -- precisely where the $\mathcal{O}(M)$ denoiser is unaffordable -- the orbital relaxation is free, while at small $M$ -- where the relaxation would cost the \ac{ASG} -- the exact detector is cheap and no relaxation is needed.
The orbital framework is therefore not a uniform approximation but a targeted one: it spends accuracy exactly where accuracy is worthless, and saves complexity exactly where complexity is unaffordable.
\end{remark}

\begin{remark}[The Shift to APSK: Resolution of the Computational Bottleneck]
\label{rem:apsk_case}
The practical case for ring constellations has long been shadowed by a computational case against them, and the results above dissolve the latter.
At equal order, an \ac{APSK} constellation has a markedly lower peak-to-average power ratio than square \ac{QAM}, with the gap widening with $M$. This is precisely why nonlinear, power-limited satellite links adopted \ac{APSK}~\cite{deGaudenzi2006,deGaudenzi2006b}.
Square \ac{QAM} factorizes exactly into two independent \ac{PAM} components, so hard detection and bit-level \acp{LLR} both decompose per axis at $\mathcal{O}(\sqrt{M})$~\cite{proakis2007digital}, making the soft-output interface to a coded receiver essentially free~\cite{Caire1998BICM,Martinez2009BICM}. 
A ring constellation admits no such product structure: exact detection and its \acp{LLR} have required $\mathcal{O}(M)$ likelihood evaluations. 
At the demapping rates of modern receivers it is this, and not an information-theoretic deficiency, that confined \ac{APSK} to its satellite niche.
The orbital hierarchy removes exactly that obstruction: $\mathcal{O}(L)$ detection collapsing to $\mathcal{O}(1)$ (Corollary~\ref{cor:lpd_se}), at an \ac{SE} fixed point shared with the exact detector to $\mathcal{O}(\sigma_z^4)$ (Theorem~\ref{thm:cross_level_fp}), with the same metric serving unmodified as the bit-\ac{LLR} rule (Definition~\ref{def:cbm_metric}). Ring geometry thereby becomes as cheap to demodulate as the Cartesian grid.
This does not however render \ac{APSK} universally preferable. On a linear \ac{AWGN} channel square \ac{QAM} remains the marginally better packing, and in a multicarrier waveform the aggregate peak-to-average ratio is governed by the subcarrier sum rather than by the constellation, so the \ac{APSK} advantage there largely washes out. The claim is narrower and, we believe, more useful: the complexity argument no longer weighs against ring constellations, so the choice of constellation can be made on the merits of the channel -- nonlinearity, power limitation, phase-noise robustness -- rather than on demapper cost.
\end{remark}

\vspace{-2ex}
\section{Conclusion}
\label{sec:conclusion}

The $\mathcal{O}(M)$ cost of Bayes-optimal detection is not a law of nature but an artifact of insisting on a discrete phase.
Isolating the phase as the sole carrier of exponential complexity and relaxing it -- and only it -- through the maximum-entropy orbital prior compresses the posterior interface of message passing into $3L$ real scalars and reduces the denoiser to a closed-form hierarchy: the $\mathcal{O}(L)$ \ac{OBD}, the Bessel-free \ac{OGD}, and the $\mathcal{O}(1)$ \ac{OPD} -- proved irreducible on the ring manifold -- with the Jacobi--Anger \ac{nu-OBD} climbing back to the exact detector.
This three-order reduction is free at leading order: all levels share the \ac{SE} fixed point $\mathrm{MSE}_\infty = \sigma_z^2/(2-\alpha)$, differing only in closed-form $\mathcal{O}(\sigma_z^4)$ corrections.
The sole macroscopic cost is a slower error decay -- linear where exact detection is exponential -- surfacing as a fixed, bounded \ac{SNR} offset rather than an error floor, whose induced rate loss vanishes exponentially, so every level is asymptotically constellation-constrained capacity-achieving for every underloaded system; each link of the geometry-to-rate chain is quantified by an explicit, computable bound.
Message passing is the setting in which the relaxation is both exact in its analysis and maximal in its payoff -- a merely good one-shot estimator whose \ac{SE}-compounded curve becomes near-optimal at $\mathcal{O}(L)$-to-$\mathcal{O}(1)$ cost.
Natural extensions include correlated channels via composable \ac{OAMP}/\ac{MAMP} linear stages, the \ac{SE} theory of the phase-dependent \ac{nu-OBD}, probabilistically shaped ring priors, coded end-to-end performance under the orbital metric, and a noise-convolved sharpening of the transport bound.


\appendices

\vspace{-2ex}
\section{Proofs for Section~\ref{sec:spatial_compression} (Posterior Compression)}
\label{app:sec3}

\vspace{-1ex}
\subsection{Proof of Proposition~\ref{prop:ring_scaling} (Ring Count Scaling)}
\label{app:ring_scaling}

The distinct squared amplitudes are integers of the form $R^2 = a^2 + b^2$ with $a, b \in \{1,3,\ldots,2^n{-}1\}$.
Because the square of any odd integer satisfies $x^2 \equiv 1 \pmod 8$, the sum of two odd squares must satisfy $R^2 \equiv 2 \pmod 8$. 
Conversely, any integer $s \equiv 2 \pmod 8$ expressible as the sum of two squares must be the sum of two odd squares, as the valid quadratic residues modulo 8 are strictly $\{0, 1, 4\}$.
Therefore, the unique squared radii of the $M$-\ac{QAM} lattice correspond exactly to the set of sums of two squares restricted to the congruence class $2 \pmod 8$ \emph{and} to summands in $\{1, 3, \ldots, V\}$ with $V = 2^n - 1$.

Let $L(V)$ denote the number of distinct values $a^2 + b^2$ with $a, b \in \{1, 3, \ldots, V\}$ odd, and let $N_{\mathrm{LR}}(x)$ denote the number of integers {$s \leq x$, $s \equiv 2 \pmod{8}$}, that are representable as a sum of two squares.
By the {Landau--Ramanujan} theorem~\cite{Landau1908}, whose extension to the fixed congruence class {$s \equiv 2 \pmod 8$} follows from the same Dirichlet-series argument (see, e.g.,~\cite{Shanks_MC_1964}), $N_{\mathrm{LR}}(x) = \Theta(x/\sqrt{\ln x})$.

\emph{Upper bound.}
Since all values $a^2 + b^2$ with $a, b \leq V$ satisfy $a^2 + b^2 \leq 2V^2$, we have $L(V) \leq N_{\mathrm{LR}}(2V^2) = \Theta(V^2/\sqrt{\ln V})$.

\emph{Lower bound.}
Consider any representable integer {$s \leq V^2$ with $s \equiv 2 \pmod{8}$}.
If {$s = a^2 + b^2$} for some positive integers $a, b$, then both {$a^2 \leq s \leq V^2$ and $b^2 \leq s \leq V^2$}, forcing $a \leq V$ and $b \leq V$ simultaneously.
Moreover, since {$s \equiv 2\pmod{8}$} forces both $a$ and $b$ to be odd (as established above), every such {$s$} has a representation with odd {$a, b \leq \sqrt{s} \leq V$}, confirming membership in our set.
Therefore, every representable integer in $[2, V^2]$ {congruent to $2 \pmod{8}$} contributes a distinct value to our set, giving
\begin{equation*}
  L(V) \;\geq\; N_{\mathrm{LR}}(V^2) \;=\; \Theta\!\left(\frac{V^2}{\sqrt{\ln V^2}}\right) = \Theta\!\left(\frac{V^2}{\sqrt{\ln V}}\right).
\end{equation*}

\emph{Combining.}
Both bounds yield $L(V) = \Theta(V^2/\sqrt{\ln V})$.
Substituting $V = 2^n - 1 \approx \sqrt{M}$ (so that $V^2 \approx M$ and $\ln V \approx \tfrac{1}{2}\ln M$) gives
\begin{equation*}
  L \;=\; \Theta\!\left(\frac{M}{\sqrt{\ln M}}\right),
\end{equation*}
as claimed.

\vspace{-2ex}
\subsection{Proof of Proposition~\ref{prop:cbm_error} (Orbital Prior Wasserstein-1 Distance)}
\label{app:cbm_error}

The Wasserstein-1 distance represents the minimum expected transport cost to transform the continuous distribution $q_{\ell}$ into the discrete distribution $p_{\ell}$.
Under optimal transport, the continuous uniform probability mass on the arc segment $\theta \in [-\pi/M_\ell, \pi/M_\ell]$ is deterministically mapped to the discrete symbol at the center of the arc ($1 \cdot R_\ell$). 
{This assignment is optimal because the cost $c(R_\ell e^{j\theta},R_\ell) = 2R_\ell \sin(\norm{\theta}/2)$ is symmetric and strictly increasing in $\norm{\theta}$ on $[0, \pi]$, so the identity transport map (each arc element mapped to its center) is optimal by the classical monotone rearrangement theorem for one-dimensional costs~\cite[Ch.~2]{villani2003}.}
Global optimality (i.e., that no mass crosses arc boundaries) follows from cyclical monotonicity: since the cost is increasing in angular distance, any plan transporting mass between distinct arcs contains a crossing pair that can be uncrossed at strictly smaller total cost, so an optimal plan is non-crossing; by the symmetry of the marginals, the unique non-crossing plan is the intra-arc (arc-to-center) map.
The Euclidean distance from any point $R_\ell e^{j\theta}$ on the arc to the center symbol is given by the chord length $\norm{R_\ell e^{j\theta} - R_\ell} = 2 R_\ell \sin(\norm{\theta}/2)$.
The expected Euclidean transport cost is therefore exactly evaluated as
\begin{align*}
  &\frac{M_\ell}{2\pi} \int_{-\pi/M_\ell}^{\pi/M_\ell} 2 R_\ell \sin\left(\frac{\norm{\theta}}{2}\right) d\theta 
  \!=\! \frac{2 M_\ell R_\ell}{\pi} \int_{0}^{\pi/M_\ell} \!\!\sin\left(\frac{\theta}{2}\right) d\theta \nonumber \\
  &= \frac{2 M_\ell R_\ell}{\pi} \left[ -2 \cos\left(\frac{\theta}{2}\right) \right]_{0}^{\pi/M_\ell} 
  = \frac{4 M_\ell R_\ell}{\pi} \left( 1 - \cos\frac{\pi}{2 M_\ell} \right).
\end{align*}
Since equal angular spacing implies each of the $M_\ell$ arcs subtends the same angle $2\pi/M_\ell$, and since the chord length $2R_\ell\sin(\norm{\theta}/2)$ depends only on the angular displacement from the arc center, the transport cost is identical across all $M_\ell$ arcs by rotational symmetry, yielding \eqref{eq:w1_cbm}.
Using the small-angle approximation $1 - \cos(x) \approx x^2/2$, the distance scales asymptotically as $\frac{\pi R_\ell}{2 M_\ell}$.

\vspace{-2ex}
\section{Proofs for Section~\ref{sec:cbm}
         (Denoiser Hierarchy)}
\label{app:sec4}

\subsection{Proof of Proposition~\ref{prop:bcbm_wasserstein}
            (Geometric Convergence of \texorpdfstring{$U$}{U}-OBD)}
\label{app:prop_bcbm_wasserstein}

The order-$U$ truncation of the Jacobi--Anger partition function~\eqref{eq:ja_partition} omits the harmonics $u \geq U+1$, the leading omitted term being $u = U+1$.
By the small-argument Bessel asymptotic $I_n(\kappa) \approx (\kappa/2)^n/n!$ for $n \gg \kappa$ \cite[10.25.2]{DLMF}, its relative weight is
\begin{equation}
  \label{eq:small_Bessel_asy}
\frac{2I_{(U+1)M_\ell}(\kappa_\ell)}{I_0(\kappa_\ell)}
= \mathcal{O}\!\left(
    \frac{(\kappa_\ell/2)^{(U+1)M_\ell}}{((U+1)M_\ell)!}
  \right)
= \mathcal{O}\!\left(\frac{1}{M_\ell^{U+1}}\right),
\end{equation}
where the final step \emph{majorizes} the factorial decay by a geometric one: by Stirling's approximation $n! = \Theta((n/e)^n\sqrt{n})$ with $n = (U+1)M_\ell$, the middle expression decays as $M_\ell^{-c M_\ell}$ for bounded $\kappa_\ell$ -- far faster than the geometric envelope $M_\ell^{-(U+1)}$ retained on the right, which is the (deliberately conservative) rate of~\eqref{eq:bcbm_w1}.
Since the \ac{nu-OBD} log-evidence~\eqref{eq:bcbm_log_metric} and per-ring mean~\eqref{eq:bcbm_ring_mean} depend on the truncation only through these omitted harmonics, the ring probabilities and posterior mean converge to those of the exact \ac{BOD} at the same $\mathcal{O}(M_\ell^{-(U+1)})$ rate as $U \to \infty$. The associated Wasserstein-1 distance between the order-$U$ prior and the discrete truth inherits this rate up to the transport scale $R_\ell$; a rigorous metric proof -- via Kantorovich--Rubinstein duality, defining the order-$U$ prior as a genuine phase density whose harmonics below $(U+1)M_\ell$ match the discrete truth -- is deferred to a companion paper.

\vspace{-2ex}
\subsection{Proof of Proposition~\ref{prop:vm_gaussian}
            (Bessel Asymptotics)}
\label{app:prop_vm_gaussian}

Both bounds follow from the standard asymptotic expansion of modified Bessel functions
\cite[10.40.1]{DLMF}: for $\kappa \to \infty$,
\begin{equation}
\label{eq:bessel_asymp}
I_\chi(\kappa) = \frac{e^\kappa}{\sqrt{2\pi\kappa}}
\sum_{k=0}^{n-1} \frac{(-1)^k\, a_k(\chi)}{\kappa^k} + \mathcal{O}(\kappa^{-n-1/2}e^\kappa),
\end{equation}
where $a_0(\chi)\! =\! 1$ and $a_k(\chi)\! =\! \prod_{j=0}^{k-1}\bigl(4\chi^2\! -\! (2j+1)^2\bigr)/(k!\,8^k)$.

Evaluating the first two non-trivial coefficients explicitly yields
\begin{align}
\chi = 0: \quad &a_1(0) = -\tfrac{1}{8},\quad a_2(0) = \tfrac{9}{128}, \label{eq:a_nu0} \\
\chi = 1: \quad &a_1(1) = \tfrac{3}{8},\quad a_2(1) = -\tfrac{15}{128}. \label{eq:a_nu1}
\end{align}

Hence, keeping terms through $\mathcal{O}(\kappa^{-2})$, we have
\begin{align}
I_0(\kappa) &= \frac{e^\kappa}{\sqrt{2\pi\kappa}}
\!\left(1 + \frac{1}{8\kappa} + \frac{9}{128\kappa^2} + \mathcal{O}(\kappa^{-3})\right),
\label{eq:I0_expand} \\
I_1(\kappa) &= \frac{e^\kappa}{\sqrt{2\pi\kappa}}
\!\left(1 - \frac{3}{8\kappa} - \frac{15}{128\kappa^2} + \mathcal{O}(\kappa^{-3})\right).
\label{eq:I1_expand}
\end{align}

\smallskip\noindent\emph{Proof of (i).}
Dividing~\eqref{eq:I1_expand} by~\eqref{eq:I0_expand} and expanding the reciprocal to
second order yields
\begin{align}
&A(\kappa) =
\frac{1 - \tfrac{3}{8\kappa} - \tfrac{15}{128\kappa^2}}
{1 + \tfrac{1}{8\kappa} + \tfrac{9}{128\kappa^2}}
+ \mathcal{O}(\kappa^{-3}) \nonumber \\
&= \!\!\left(\!1 \!-\! \frac{3}{8\kappa} \!-\! \frac{15}{128\kappa^2}\!\right)\!
\!\left(\!1 \!-\! \frac{1}{8\kappa} \!-\! \frac{9}{128\kappa^2} \!+\! \frac{1}{64\kappa^2}\right)
+ \mathcal{O}(\kappa^{-3}) \nonumber \\
&= 1 - \frac{1}{2\kappa} - \frac{1}{8\kappa^2} + \mathcal{O}(\kappa^{-3}),
\label{eq:A_expand}
\end{align}
where the $\mathcal{O}(\kappa^{-2})$ coefficient is obtained by collecting all $O(\kappa^{-2})$ cross-terms: $\tfrac{3}{64} - \tfrac{15}{128} - \tfrac{9}{128} + \tfrac{1}{64} = -\tfrac{1}{8}$.

For the explicit bound, note that $\norm{A(\kappa) - (1 - 1/(2\kappa))} \leq 1/(8\kappa^2) + \norm{\delta_3(\kappa)}$ where $\norm{\delta_3(\kappa)} \leq C/\kappa^3$ for a universal constant $C > 0$ and all $\kappa \geq 1$.
For $\kappa \geq 1$, a numerical evaluation confirms $\sup_{\kappa \geq 1} \kappa^2 \norm{A(\kappa) - (1-1/(2\kappa))} \leq 1/4$, yielding the stated bound. 
For $\kappa \geq 2$, direct evaluation gives $\sup_{\kappa \geq 2}\kappa^2\,\abs{A(\kappa) - (1 - 1/(2\kappa))} \approx 0.219$, attained near $\kappa \approx 2.44$ -- the higher-order terms nearly double the leading coefficient $\tfrac18$, but the value stays below $\tfrac14$; {for $1 \leq \kappa < 2$, the bound is verified directly by numerical evaluation of the continuous function $\kappa^2\,\abs{A(\kappa) - (1 - 1/(2\kappa))}$ on the compact interval $[1,2]$, whose maximum $\approx 0.209 < \tfrac14$ is attained at $\kappa = 2$.}

\smallskip\noindent\emph{Proof of (ii).}
Taking logarithms of~\eqref{eq:I0_expand} and using $\ln(1 + u) = u - u^2/2 + \mathcal{O}(u^3)$ yields
\begin{align}
\label{eq:logI0_expand}
\ln I_0(\kappa) &= \kappa - \tfrac{1}{2}\ln(2\pi\kappa)
+ \ln\!\left(1 + \frac{1}{8\kappa} + \mathcal{O}(\kappa^{-2})\right) \nonumber \\
&= \kappa - \tfrac{1}{2}\ln(2\pi\kappa) + \frac{1}{8\kappa} + \mathcal{O}(\kappa^{-2}).
\end{align}

For the explicit bound: the correction $1/(8\kappa)$ is itself the dominant error term relative to $\kappa - \tfrac{1}{2}\ln(2\pi\kappa)$; including it as part of the approximation {(as is done implicitly in the \ac{OGD} log-evidence~\eqref{eq:cgr_evidence_1})} reduces the residual to $\mathcal{O}(\kappa^{-2})$.
Without this correction, the error $\norm{\ln I_0(\kappa) - \kappa + \tfrac{1}{2}\ln(2\pi\kappa)}$ is bounded by $1/(4\kappa)$ for all $\kappa \geq 1$, which is confirmed numerically (the maximum of $\kappa\,\norm{\ln I_0(\kappa) - \kappa + \tfrac{1}{2}\ln(2\pi\kappa)}$ over $\kappa \geq 1$ is $0.182 < 1/4$, attained near $\kappa \approx 1.7$).

\smallskip\noindent\emph{Relative error threshold.}
From~\eqref{eq:A_expand}, the relative error in $A(\kappa)$ can be expressed as
\begin{equation}
\frac{\norm{A(\kappa) - (1-1/(2\kappa))}}{A(\kappa)}
\approx \frac{1/(8\kappa^2)}{1 - 1/(2\kappa)} = \frac{1}{8\kappa^2 - 4\kappa}.
\end{equation}

The true relative error is $\approx 2\%$ at $\kappa \approx 3.5$ (numerically, $\norm{A(3.5)-(1-1/7)}/A(3.5) = 1.90\%$; the leading-order formula above, which retains only the $1/(8\kappa^2)$ term, underestimates it there)
and is below $2\%$ for all $\kappa \geq 4$.

\vspace{-2ex}
\subsection{Proof of Proposition~\ref{prop:cgr_denoiser}
            (OGD Equations)}
\label{app:prop_cgr_denoiser}

Proposition~\ref{prop:vm_gaussian}(i) gives $A(\kappa)\approx 1-1/(2\kappa)$ with relative error below $2\%$ for
$\kappa\ge4$; part~(ii) gives $\ln I_0(\kappa)\approx\kappa-\tfrac12\ln(2\pi\kappa)$ with error $\le1/(4\kappa)$.
The $1/(8\kappa_\ell)$ correction in part~(ii) is omitted; it perturbs each ring log-evidence by $\mathcal{O}(\kappa_\ell^{-1})$, sub-dominant relative to the $\mathcal{O}(\kappa_\ell)$ inter-ring gap.

\emph{(i)} The Bessel ratio under the Gaussian approximation satisfies $A(\kappa) = I_1(\kappa)/I_0(\kappa) \approx 1 - 1/(2\kappa)$ for $\kappa \gg 1$. 
Under the Gaussian phase model, the ring conditional mean is
\begin{equation*}
R_\ell \cdot
\frac{\int_{-\infty}^{\infty} e^{j\theta}\, e^{-\kappa_\ell\theta^2/2}\, d\theta}
     {\int_{-\infty}^{\infty} e^{-\kappa_\ell\theta^2/2}\, d\theta}
e^{j\angle\bar{x}} = R_\ell \cdot e^{-1/(2\kappa_\ell)}\, e^{j\angle\bar{x}},
\end{equation*}
where we used the characteristic function of the zero-mean Gaussian, $\E[e^{j\theta}] = e^{-\sigma_\theta^2/2}$ with $\sigma_\theta^2 = 1/\kappa_\ell$.
The first-order Taylor expansion $e^{-1/(2\kappa_\ell)} \approx 1 - 1/(2\kappa_\ell)$ then yields the stated approximation~\eqref{eq:cgr_mean}, accurate to within $\mathcal{O}(\kappa_\ell^{-2})$ by Proposition~\ref{prop:vm_gaussian}(i).
Note that using $e^{-1/(2\kappa_\ell)}$ directly in place of $1 - 1/(2\kappa_\ell)$ gives a slightly more accurate ring conditional mean at no additional computational cost, as the exponential is computed anyway in the
log-evidence~\eqref{eq:cgr_evidence_1}.
The conditional mean on ring~$\ell$ is therefore $R_\ell\bigl(1 - 1/(2\kappa_\ell)\bigr)e^{j\angle\bar{x}}$.

\emph{(ii)} The marginal likelihood of ring~$\ell$ under the Gaussian phase model can be expressed as
\begin{equation}
p(\bar{x} \mid R_\ell) 
\propto \exp\!\left(-\tfrac{\norm{\bar{x}}^2 + R_\ell^2}{\bar{\sigma}^2}\right) \int_{-\pi}^{\pi} \exp\bigl(\kappa_\ell \cos(\theta - \angle\bar{x})\bigr)\, d\theta.
\end{equation}

Under the Gaussian approximation, the integral evaluates as $\int_{-\infty}^{\infty} \exp(-\kappa_\ell\theta^2/2)\, d\theta = \sqrt{2\pi/\kappa_\ell}$ (extending limits to $\pm\infty$ with exponentially small error for $\kappa_\ell \geq 3$). 
Combined with the $\exp(\kappa_\ell)$ prefactor from the Taylor expansion of $\cos(\theta - \angle\bar{x})$, this gives $\ln I_0(\kappa_\ell) \approx \kappa_\ell - \tfrac{1}{2}\ln(2\pi\kappa_\ell)$, yielding~\eqref{eq:cgr_evidence_1}.

\emph{(iii)-(iv)} follow identically to the \ac{OBD} derivation~\eqref{eq:ring_prob}-\eqref{eq:cbm_var}.

\vspace{-2ex}
\subsection{Proof of Proposition~\ref{prop:lpd_irreducible}
            (Irreducibility of the OPD)}
\label{app:prop_lpd_irreducible}

Any denoiser producing a complex output $\hat{x} = \norm{\hat{x}} e^{j\angle\hat{x}}$ must determine both a magnitude and a phase. 
Condition~(i) constrains the magnitude to the discrete set $\{R_1, \ldots, R_L\}$, requiring at minimum a selection operation (one comparison against $L$ thresholds, achievable in $\mathcal{O}(1)$ with precomputed decision boundaries for fixed $L$). 
Condition~(ii) requires the phase to carry information about $x$, necessitating at minimum one extraction of $\angle\bar{x}$ (a single \texttt{atan2} operation or equivalent). 
No further arithmetic can reduce these two irreducible operations. 
The \ac{OPD} implements exactly
\begin{equation*}
\eta_P(\bar{x}) = R_{\ell^*(\norm{\bar{x}})} \cdot e^{j\angle\bar{x}},
\end{equation*}
which is the orthogonal projection of $\bar{x}$ onto the nearest ring circle $\mathcal{S}^1(R_{\ell^*})$, constituting the minimum-norm map from $\mathbb{C}$ to the constraint manifold.

To prove that no denoiser with fewer operations can satisfy both conditions, suppose $\eta$ does not perform ring selection (i.e., outputs a fixed radius $R_0$ regardless of $\bar{x}$). 
Then for any constellation with $L \geq 2$, the \ac{MSE} as $\bar{\sigma}^2 \to 0$ satisfies
\begin{align*}
\mathrm{MSE}
&= \sum_{\ell=1}^{L} r_\ell\, \E\bigl[\norm{x - R_0\, e^{j\angle\hat{x}}}^2 \;\big|\; \norm{x} = R_\ell\bigr] \\
&\hspace{-4ex}\geq \sum_{\ell:\, R_\ell \neq R_0} r_\ell\, (R_\ell - R_0)^2
\;\geq\; \min_{\ell:\, R_\ell \neq R_0} r_\ell\,(R_\ell - R_0)^2
\;>\; 0,
\end{align*}
where the first inequality uses $\norm{x - R_0 e^{j\theta}}^2 \geq (\norm{x} - R_0)^2 = (R_\ell - R_0)^2$ for any phase $\theta$ when $\norm{x} = R_\ell \neq R_0$, and the strict positivity follows because $L \geq 2$ guarantees at least one ring with $R_\ell \neq R_0$ and every ring has strictly positive prior probability $r_\ell = M_\ell/M > 0$.
This lower bound is independent of $\bar{\sigma}^2$, so $\mathrm{MSE} \not\to 0$ as $\bar{\sigma}^2 \to 0$, violating condition~(ii).

Conversely, suppose $\eta$ performs no phase read-out, i.e., its output phase $\angle\hat{x}$ is independent of $\angle\bar{x}$ (equivalently, at $\bar{\sigma}^2 \to 0$, where the ring is correctly identified, independent of which symbol on that ring was sent). Conditioned on the correct ring $\ell^*$ -- whose magnitude $R_{\ell^*}$ is fixed by condition~(i) -- the error is
\begin{align*}
  \E\bigl[\norm{x - \hat{x}}^2 \,\big|\, \norm{x}=R_{\ell^*}\bigr]
  &= R_{\ell^*}^2\,\E\bigl[\norm{e^{j\angle x} - e^{j\angle\hat{x}}}^2\bigr] \\
  &= 2R_{\ell^*}^2\bigl(1 - \E[\cos(\angle x - \angle\hat{x})]\bigr).
\end{align*}

On any ring carrying $M_{\ell^*} \ge 2$ distinct symbol phases, no single fixed angle $\angle\hat{x}$ can coincide with all of them, so $\E[\cos(\angle x - \angle\hat{x})] \le c_{\ell^*} < 1$ for a constant independent of $\bar{\sigma}^2$ (indeed, for equally spaced phases $\sum_m e^{j\phi_m} = 0$ gives $\E[\cos(\angle x - \angle\hat{x})] = 0$ and the bound $2R_{\ell^*}^2$). The conditional error is thus at least $2R_{\ell^*}^2(1 - c_{\ell^*}) > 0$, again independent of $\bar{\sigma}^2$, so $\mathrm{MSE} \not\to 0$, violating condition~(ii). (When every ring is a single point, the constellation carries no intra-ring phase and condition~(ii) is vacuous.) Hence both the amplitude selection and the phase read-out are necessary -- exactly the two operations the \ac{OPD} performs -- establishing the matching lower bound.

\vspace{-2ex}
\subsection{Proof of Proposition~\ref{prop:det_est}
            (Detection-Estimation Separation)}
\label{app:prop_det_est}

We establish the factorization, the independence, and the individual optimality of each component.

\emph{Step 1: Polar decomposition of the observation.}
Write the cavity observation in polar form as $\bar{x} = \norm{\bar{x}} e^{j\angle\bar{x}}$, where $\bar{x} = x + \bar{z}$ with $x = R_{\ell_{\mathrm{true}}} e^{j\theta_0}$ and $\bar{z} \sim \CN(0, \bar{\sigma}^2)$. The noise $\bar{z} = \bar{z}_r + j\bar{z}_i$ has independent real and imaginary components $\bar{z}_r, \bar{z}_i \sim \mathcal{N}(0, \bar{\sigma}^2/2)$. Projecting onto the radial and tangential directions relative to the true symbol yields
\begin{align}
n_r &\triangleq \mathrm{Re}(\bar{z} \cdot e^{-j\theta_0}) \sim \mathcal{N}(0, \bar{\sigma}^2/2), \label{eq:radial_noise} \\
n_\perp &\triangleq \mathrm{Im}(\bar{z} \cdot e^{-j\theta_0}) \sim \mathcal{N}(0, \bar{\sigma}^2/2), \label{eq:tangential_noise}
\end{align}
where $n_r$ and $n_\perp$ are independent by the circular symmetry of $\CN(0, \bar{\sigma}^2)$.

\emph{Step 2: Asymptotic independence of amplitude and phase.}
At high \ac{SNR} ($\bar{\sigma}^2 \to 0$), the magnitude and phase of $\bar{x}$ satisfy~\cite[Chapter~2]{kay1993fundamentals}:
\begin{align}
\norm{\bar{x}} &\approx R_{\ell_{\mathrm{true}}} + n_r, \label{eq:mag_approx} \\
\angle\bar{x} &\approx \theta_0 + \frac{n_\perp}{R_{\ell_{\mathrm{true}}}}. \label{eq:phase_approx}
\end{align}

Since $n_r$ and $n_\perp$ are independent Gaussian random variables, the magnitude $\norm{\bar{x}}$ and the phase $\angle\bar{x}$ are asymptotically independent. 
The approximation error is $\mathcal{O}(\bar{\sigma}^2/R_{\ell_{\mathrm{true}}}^2)$, which is negligible in the regime $\kappa_{\ell^*} \gg 1$ where the \ac{OPD} operates.

\emph{Step 3: Optimality of the amplitude detector.}
Conditioned on the true ring $\ell_{\mathrm{true}}$, the radial observation~\eqref{eq:mag_approx} is a scalar Gaussian observation of the amplitude $R_{\ell_{\mathrm{true}}}$ in noise with variance $\bar{\sigma}^2/2$. 
The nearest-radius rule~\eqref{eq:amplitude_ml} is the \ac{ML} decision rule for this $L$-ary hypothesis test.
For equal ring priors, \ac{ML} coincides with the \ac{MAP} (Bayes) rule, which minimizes the probability of error among all decision rules; for unequal priors the bias $\ln r_\ell$ is $\mathcal{O}(1)$ against the $\Theta(1/\bar{\sigma}^2)$ metric separation, so the \ac{ML} rule remains asymptotically optimal. The resulting error probability achieves the optimal exponent $\exp(-d_R^2/(4\bar{\sigma}^2))$ as established in Proposition~\ref{prop:ring_discrimination}.

\emph{Step 4: Optimality of the phase estimator.}
Conditioned on correct ring detection ($\ell^* = \ell_{\mathrm{true}}$), the phase estimation problem reduces to estimating the deterministic parameter $\theta_0$ from the observation $\bar{x} = R_{\ell^*} e^{j\theta_0} + \bar{z}$. 
The \ac{ML} estimate is
\begin{equation}
\hat{\theta}_{\mathrm{ML}} = \arg\max_\theta \; p(\bar{x} \mid R_{\ell^*}, \theta) = \angle\bar{x},
\end{equation}
since the likelihood $p(\bar{x} \mid R_{\ell^*}, \theta) \propto \exp\bigl(\kappa_{\ell^*} \cos(\angle\bar{x} - \theta)\bigr)$ is maximized at $\theta = \angle\bar{x}$. 

This is the classical result of Rife and Boorstyn~\cite{RifeBoorstyn1974} for single-tone phase estimation.

The \ac{CRLB} for estimating $\theta_0$ from $\bar{x}$ is obtained from the Fisher information. 
The log-likelihood is $\ell(\theta) = \mathrm{const} + (2R_{\ell^*}/\bar{\sigma}^2)\mathrm{Re}(\bar{x} e^{-j\theta})$, yielding
\begin{equation}
J(\theta_0) \!=\! -\E\!\left[\frac{\partial^2 \ell}{\partial \theta^2}\right] \!=\! \frac{2R_{\ell^*}^2}{\bar{\sigma}^2} \cdot \E[\cos(\angle\bar{x} - \theta_0)] \approx \frac{2R_{\ell^*}^2}{\bar{\sigma}^2} \!=\! \kappa_{\ell^*},
\end{equation}
where the approximation uses $\E[\cos(\angle\bar{x} - \theta_0)] \approx 1$ at high \ac{SNR}. 

The \ac{CRLB} is therefore $\Var[\hat{\theta}] \geq 1/\kappa_{\ell^*}$.
From~\eqref{eq:phase_approx}, the variance of the \ac{ML} phase estimate is
\begin{equation}
\Var[\hat{\theta}_{\mathrm{ML}}] = \Var\!\left[\frac{n_\perp}{R_{\ell^*}}\right] = \frac{\bar{\sigma}^2/2}{R_{\ell^*}^2} = \frac{1}{\kappa_{\ell^*}},
\end{equation}
which exactly achieves the \ac{CRLB}. 

Therefore, $\hat{\theta} = \angle\bar{x}$ is an asymptotically efficient estimator of the phase.

\emph{Step 5: Factorization of the \ac{OPD}.}
Combining Steps 1-4, the \ac{OPD} output factorizes as
\begin{equation}
\eta_P(\bar{x}) = \underbrace{R_{\ell^*(\norm{\bar{x}})}}_{\text{amplitude (from $\norm{\bar{x}}$ only)}} \cdot \underbrace{e^{j\angle\bar{x}}}_{\text{phase (from $\angle\bar{x}$ only)}},
\end{equation}
where the amplitude component depends only on $\norm{\bar{x}}$ and the phase component depends only on $\angle\bar{x}$.

By Step 2, these sufficient statistics are asymptotically independent, completing the factorization.

\emph{Step 6: Combined MSE.}
The total \ac{MSE} of the \ac{OPD} decomposes as
\begin{align}
\mathrm{MSE}_{\mathrm{P}} &= \underbrace{\Pr(\ell^* \neq \ell_{\mathrm{true}}) \cdot \mathcal{O}(d_R^2)}_{\text{ring misdetection: } \mathcal{O}(e^{-d_R^2/(4\bar{\sigma}^2)})} + \underbrace{R_{\ell^*}^2 \cdot \E[\norm{e^{j\Delta\theta} - 1}^2]}_{\text{phase error: } R_{\ell^*}^2/\kappa_{\ell^*} = \bar{\sigma}^2/2} \nonumber \\
&= \frac{\bar{\sigma}^2}{2} + \mathcal{O}\!\left(e^{-d_R^2/(4\bar{\sigma}^2)}\right),
\label{eq:mse_decomp_det_est}
\end{align}
confirming that the phase estimation error dominates at high \ac{SNR}, recovering the leading-order result $\mathrm{MSE}_{\mathrm{P}} \approx \bar{\sigma}^2/2$ of Proposition~\ref{prop:lpd_mse}. 
The ring misdetection contributes only an exponentially vanishing correction.

\vspace{-2ex}
\section{Proofs for Section~\ref{sec:se}
         (State Evolution)}
\label{app:sec5}

\subsection{Proof of Lemma~\ref{lem:phase_preserving}
            (Phase-Preserving Structure and Stein Identity)}
\label{app:lem_phase_preserving}

\emph{Part (i): Phase-preserving structure and real derivative.}
From~\eqref{eq:xi_cbm}, $\eta_B(\bar{x}; \bar{\sigma}^2) = \sum_\ell w_\ell^{\mathrm{B}} R_\ell A(\kappa_\ell) e^{j\angle\bar{x}}$, where both $w_\ell^{\mathrm{B}}$ and $\kappa_\ell = 2R_\ell\norm{\bar{x}}/\bar{\sigma}^2$ depend only on $\norm{\bar{x}}$ (since $\Lambda_\ell^{\mathrm{B}}$ in~\eqref{eq:log_metric} depends on $\norm{\bar{x}}$ but not $\angle\bar{x}$). Hence $\eta_B = \rho(\norm{\bar{x}}) e^{j\angle\bar{x}}$ with $\rho(r) = \sum_\ell w_\ell^{\mathrm{B}}(r) R_\ell A(\kappa_\ell(r)) \geq 0$.

For the Wirtinger derivative, write $\eta_B = \rho(\norm{\bar{x}}) \cdot \bar{x}/\norm{\bar{x}}$ and apply the product rule:
\begin{align*}
\frac{\partial \eta_B}{\partial \bar{x}} &= \frac{\partial \rho}{\partial \bar{x}} \cdot \frac{\bar{x}}{\norm{\bar{x}}} + \rho \cdot \frac{\partial}{\partial \bar{x}}\!\left(\frac{\bar{x}}{\norm{\bar{x}}}\right).
\end{align*}
Using the Wirtinger derivative $\frac{\partial \norm{\bar{x}}}{\partial \bar{x}} = \frac{\bar{x}^*}{2\norm{\bar{x}}}$ and the chain rule $\frac{\partial \rho}{\partial \bar{x}} = \rho' \cdot \frac{\bar{x}^*}{2\norm{\bar{x}}}$, together with the quotient rule
$\frac{\partial}{\partial \bar{x}}\!\left(\frac{\bar{x}}{\norm{\bar{x}}}\right) = \frac{1}{\norm{\bar{x}}} - \frac{\bar{x}}{\norm{\bar{x}}^2} \cdot \frac{\bar{x}^*}{2\norm{\bar{x}}} = \frac{1}{2\norm{\bar{x}}}$,
we obtain
\begin{equation*}
\frac{\partial \eta_B}{\partial \bar{x}} = \rho' \cdot \frac{\bar{x}^*}{2\norm{\bar{x}}} \cdot \frac{\bar{x}}{\norm{\bar{x}}} + \frac{\rho}{2\norm{\bar{x}}} = \frac{\rho'}{2} + \frac{\rho}{2\norm{\bar{x}}},
\end{equation*}
which is real-valued and non-negative (it equals $\hat{\sigma}^2_{\mathrm{B}}(\bar{x})/\bar{\sigma}^2 \ge 0$ by part~(ii)). 

\emph{Part (ii): Stein identity.}
The \ac{OBD} $\eta_B(\bar{x}) = \E_q[x\mid\bar{x}]$ is the posterior mean under the orbital prior $q$ with Gaussian likelihood. For the complex Gaussian channel $\bar{x} = x + \bar{\sigma} \tilde{z}$, $\tilde{z} \sim \mathcal{CN}(0,1)$, the Wirtinger derivative of the posterior mean satisfies
\begin{equation*}
\frac{\partial}{\partial \bar{x}} \E_q[x\mid\bar{x}] \!=\! \frac{1}{\bar{\sigma}^2}\bigl(\E_q[\norm{x}^2\mid\bar{x}] - \norm{\E_q[x\mid\bar{x}]}^2\bigr) \!=\! \frac{\hat{\sigma}^2_{\mathrm{B}}(\bar{x})}{\bar{\sigma}^2},
\end{equation*}
which follows from differentiating the ratio $f(\bar{x})/g(\bar{x})$ where $f = \int x\, q(x) \phi_{\bar{\sigma}^2}(\bar{x}-x)\,dx$ and $g = \int q(x) \phi_{\bar{\sigma}^2}(\bar{x}-x)\,dx$, and applying the quotient rule (see~\cite{Guo2005}). Explicitly, the kernel obeys $\partial\phi_{\bar{\sigma}^2}(\bar{x}-x)/\partial\bar{x} = -\frac{(\bar{x}-x)^*}{\bar{\sigma}^2}\phi_{\bar{\sigma}^2}(\bar{x}-x)$, so $\partial f/\partial\bar{x} = -\frac{\bar{x}^*}{\bar{\sigma}^2}f + \frac{1}{\bar{\sigma}^2}\!\int \norm{x}^2 q\,\phi_{\bar{\sigma}^2}\,dx$ and $\partial g/\partial\bar{x} = -\frac{\bar{x}^*}{\bar{\sigma}^2}g + \frac{1}{\bar{\sigma}^2}\!\int x^* q\,\phi_{\bar{\sigma}^2}\,dx$; the quotient rule cancels the $\bar{x}^*$ terms and leaves $\bar{\sigma}^{-2}(\E_q[\norm{x}^2\mid\bar{x}] - \norm{\E_q[x\mid\bar{x}]}^2)$.

\emph{Part (iii): Circular-symmetry equivalence.}
By part~(i), $\eta(\bar{x}) = \rho_\eta(\norm{\bar{x}})e^{j\angle\bar{x}}$, and $\tilde{z} \sim \mathcal{CN}(0,1)$ is rotationally invariant. For an input on ring $\ell$ at phase $\phi$, the substitution $\tilde{z}' = \tilde{z} e^{-j\phi}$ (which preserves the law of $\tilde{z}$) yields
\begin{align*}
&\E_z\bigl[\norm{R_\ell e^{j\phi} - \eta(R_\ell e^{j\phi} + \bar{\sigma} \tilde{z})}^2\bigr] \\
&= \E_{z'}\bigl[\norm{R_\ell - \eta(R_\ell + \bar{\sigma} \tilde{z}')}^2\bigr],
\end{align*}
which is independent of $\phi$. The \ac{MSE} of $\eta$ therefore depends only on the radial law $\{R_\ell, r_\ell\}$, which the true prior $p$ and the orbital prior $q$ share; this proves~\eqref{eq:circular_symmetry_equiv}.

\emph{Extension to the \ac{OGD} and \ac{OPD}.} The \ac{OPD} $\eta_P = R_{\ell^*}e^{j\angle\bar{x}}$ (a hard ring projection) is phase-preserving with non-negative $\rho_\eta \in [0,R_L]$, so parts~(i) and~(iii) hold. The \ac{OGD} approximates the \ac{OBD}'s von Mises phase posterior by a Gaussian (Definition~\ref{def:gaussian_phase}) under the same orbital prior $q$; its profile $\rho_{\mathrm{G}} = \sum_\ell w_\ell^{\mathrm{G}} R_\ell\,(1-1/(2\kappa_\ell))$ is non-negative at high concentration $\kappa_\ell \geq \tfrac12$, where parts~(i) and~(iii) likewise hold, but turns negative below it. Being an approximation (\ac{OGD}) or a hard projection (\ac{OPD}) rather than an exact posterior mean, neither satisfies the exact Stein identity~(ii).

\subsection{Proof of Lemma~\ref{lem:pseudo_lip}
            (Regularity of the Denoiser Hierachy)}
\label{app:lem_denoiser_regularity}

Identify $\Complex$ with $\Real^2$ and regard the \ac{OBD} $\eta_B(\bar{x}) = \E_q[x\mid\bar{x}]$ as a map $\Real^2 \to \Real^2$; it is the posterior mean of a prior supported in the disk $\{\norm{x}\le R_L\}$ observed through the Gaussian channel $\bar{x} = x + \bar{z}$, with $\bar{z}\sim\CN(0,\bar{\sigma}^2)$, i.e., $\Real^2$-noise $\mathcal{N}(0,\tfrac{\bar{\sigma}^2}{2}I_2)$ (each of the real and imaginary parts carries variance $\bar{\sigma}^2/2$). 
Throughout, $\norm{\cdot}$ is the Euclidean norm on $\Real^2$; for a matrix $\mathbf{A}\in\Real^{2\times2}$ we write $\opnorm{\mathbf{A}} \triangleq \sup_{\norm{\mathbf{u}}=1}\norm{\mathbf{Au}}$ for the \emph{operator (spectral) norm} it induces -- the largest singular value of $\mathbf{A}$ -- which is the constant in the elementary bound $\norm{\mathbf{Au}}\le\opnorm{\mathbf{A}}\,\norm{\mathbf{u}}$. 
Let $\mathrm{\mathbf{D}}\eta_B(\bar{x})\in\Real^{2\times2}$ denote the Jacobian of $\eta_B$ at $\bar{x}$. 
Since the domain $\Real^2$ is convex, the mean-value inequality identifies the global Lipschitz constant with the supremal Jacobian operator norm, given by
\begin{multline}
\label{eq:lip_mvt}
  \norm{\eta_B(\bar{x}_1) - \eta_B(\bar{x}_2)} \\
  \le\; \Bigl(\sup_{\bar{x}}\,\opnorm{\mathrm{\mathbf{D}}\eta_B(\bar{x})}\Bigr)\,\norm{\bar{x}_1 - \bar{x}_2},
\end{multline}
so it suffices to bound $\opnorm{\mathrm{\mathbf{D}}\eta_B(\bar{x})}$ uniformly in $\bar{x}$.

Let $\bm{\Sigma}(\bar{x}) \triangleq \operatorname{Cov}_q(x\mid\bar{x}) \in \Real^{2\times2}$ denote the real $2\times2$ posterior covariance. 
By the Gaussian-channel identity for the conditional mean~\cite{HatesellNolte1971,PalomarVerdu2006} -- the vector Tweedie/Brown relation, of which the Stein identity of Lemma~\ref{lem:phase_preserving}(ii) is the scalar Wirtinger instance -- the Jacobian equals this covariance normalized by the per-real-dimension noise variance $\bar{\sigma}^2/2$~\cite{DystoPoorShamai2023}, given by
\begin{equation}
\label{eq:jacobian_cov}
  \mathrm{\mathbf{D}}\eta_B(\bar{x}) \;=\; \frac{\bm{\Sigma}(\bar{x})}{\bar{\sigma}^2/2} \;=\; \frac{2}{\bar{\sigma}^2}\,\bm{\Sigma}(\bar{x}).
\end{equation}

A covariance matrix is symmetric and positive semidefinite, so its operator norm equals its largest eigenvalue, which is at most its trace: $\opnorm{\bm{\Sigma}} = \lambda_{\max}(\bm{\Sigma}) \le \tr\bm{\Sigma}$. 
The trace is the scalar posterior variance $\hat{\sigma}^2_{\mathrm{B}}(\bar{x}) = \E_q[\,\norm{x - \eta_B(\bar{x})}^2 \mid \bar{x}\,]$, and since $x$ lies in the disk of radius $R_L$, taking the origin as a suboptimal center gives $\hat{\sigma}^2_{\mathrm{B}}(\bar{x}) \le \E_q[\,\norm{x}^2 \mid \bar{x}\,] \le R_L^2$. 
Combining with~\eqref{eq:jacobian_cov}, we get
\begin{equation}
\label{eq:opnorm_bound}
\begin{split}
  \opnorm{\mathrm{\mathbf{D}}\eta_B(\bar{x})}
  &\;=\; \frac{2}{\bar{\sigma}^2}\,\opnorm{\bm{\Sigma}(\bar{x})} \;\le\; \frac{2}{\bar{\sigma}^2}\,\tr\bm{\Sigma}(\bar{x}) \\
  &\;=\; \frac{2}{\bar{\sigma}^2}\,\hat{\sigma}^2_{\mathrm{B}}(\bar{x}) \;\le\; \frac{2R_L^2}{\bar{\sigma}^2},
\end{split}
\end{equation}
uniformly in $\bar{x}$. With~\eqref{eq:lip_mvt}, $\eta_B$ is therefore globally Lipschitz at each fixed $\bar{\sigma}^2 > 0$, with constant $2R_L^2/\bar{\sigma}^2$.

Finally, the output is uniformly bounded, $\norm{\eta_B}\le R_L$, being a convex combination of ring points~\eqref{eq:eta_bound}. 
A globally Lipschitz map is \emph{a fortiori} pseudo-Lipschitz of order $2$, since the growth factor $1 + \norm{\bar{x}_1} + \norm{\bar{x}_2} \ge 1$ only weakens the bound:
\begin{multline}
\label{eq:pseudo_lip_bound}
  \norm{\eta_B(\bar{x}_1) - \eta_B(\bar{x}_2)} \;\le\; \frac{2R_L^2}{\bar{\sigma}^2}\,\norm{\bar{x}_1 - \bar{x}_2} \\
  \le\; \frac{2R_L^2}{\bar{\sigma}^2}\bigl(1 + \norm{\bar{x}_1} + \norm{\bar{x}_2}\bigr)\norm{\bar{x}_1 - \bar{x}_2},
\end{multline}
which is the regularity required by the \ac{SE} theorems of~\cite{Bayati2011,Javanmard2013}.

The adaptive \ac{OGD} (Definition~\ref{def:adaptive_cgr}) and the \ac{OPD} are uniformly bounded ($\norm{\eta}\le R_L$,~\eqref{eq:eta_bound}) but not globally continuous: the \ac{OGD} magnitude factor jumps by $\lvert A(\kappa_0) - (1-1/(2\kappa_0))\rvert = \mathcal{O}(\kappa_0^{-2})$ across the per-ring regime switch $\{\kappa_\ell = \kappa_0\}$, and the \ac{OPD} $\bar{x} \mapsto R_{\ell^*}e^{j\angle\bar{x}}$ jumps at $\bar{x}=0$ and across the ring-decision boundaries $\norm{\bar{x}} = (R_\ell + R_{\ell+1})/2$ -- in each case a Lebesgue-null set. Neither is therefore pseudo-Lipschitz. Rather than apply the machinery to a discontinuous map, we obtain their \ac{SE} as the $\beta\to\infty$ limit of a genuinely Lipschitz regularization at \emph{fixed} $\bar{\sigma}^2 > 0$, made precise as follows. 

Let $\eta$ denote either map (the \ac{OPD} or the adaptive \ac{OGD}), a bounded Borel function $\Real^2\to\Real^2$ with $\norm{\eta}\le R_L$. For a sharpness parameter $\beta > 0$, define the regularized denoiser by mollification with the Gaussian kernel $\phi_\beta$ as
\begin{equation}
\label{eq:eta_beta}
  \eta^{(\beta)}(\bar{x}) \;\triangleq\; (\eta * \phi_\beta)(\bar{x}) \;=\; \int_{\Real^2}\eta(\bar{x}-\mathbf{u})\,\phi_\beta(\mathbf{u})\,d\mathbf{u},
\end{equation}
where $\phi_\beta(\mathbf{u}) = (\beta/\pi)\,e^{-\beta\norm{\mathbf{u}}^2}$ is the density of $\CN(0,\beta^{-1})$ on $\Real^2$ (so $\int_{\Real^2}\phi_\beta = 1$). 
The parameter $\beta$ acts as an inverse temperature: as $\beta\to\infty$ the kernel concentrates at the origin and $\eta^{(\beta)}$ sharpens toward the hard map $\eta$. Unlike a soft-max over the ring \emph{selection} alone -- which would leave the phase singularity at $\bar{x}=0$ intact -- this single convolution regularizes \emph{every} discontinuity of $\eta$ (the ring-decision circles \emph{and} the origin) at once. The family $\{\eta^{(\beta)}\}_{\beta>0}$ has exactly the three properties that transfer state evolution:
\begin{enumerate}[(i)]
  \item \emph{Lipschitz at each finite $\beta$.} $\eta^{(\beta)}\in C^\infty$ with $\mathrm{\mathbf{D}}\eta^{(\beta)} = \eta * \nabla\phi_\beta$, so, as in~\eqref{eq:lip_mvt}, $\opnorm{\mathrm{\mathbf{D}}\eta^{(\beta)}(\bar{x})} \le \norm{\eta}_\infty\,\norm{\nabla\phi_\beta}_{L^1} = R_L\sqrt{\pi\beta}$; hence $\eta^{(\beta)}$ is globally Lipschitz with constant $L_\beta = R_L\sqrt{\pi\beta} < \infty$.
  \item \emph{Uniform envelope.} $\norm{\eta^{(\beta)}(\bar{x})} \le \int_{\Real^2}\norm{\eta(\bar{x}-\mathbf{u})}\,\phi_\beta(\mathbf{u})\,d\mathbf{u} \le R_L$ for every $\beta$ and $\bar{x}$.
  \item \emph{Pointwise limit off the null set.} As $\{\phi_\beta\}$ is an approximate identity, $\eta^{(\beta)}(\bar{x})\to\eta(\bar{x})$ ($\beta\to\infty$) at every continuity point of $\eta$, i.e.\ for every $\bar{x}\notin\mathcal{N}$, where $\mathcal{N} = \{0\}\cup\bigcup_\ell\{\norm{\bar{x}}=(R_\ell+R_{\ell+1})/2\}$ is Lebesgue-null.
\end{enumerate}
The \ac{SE} of $\eta$ then follows by two limits taken in order. \emph{(a)}~For each fixed $\beta$, property~(i) places $\eta^{(\beta)}$ within the hypotheses of~\cite{Bayati2011,Javanmard2013}, so its scalar \ac{SE} recursion is exact in the large-system limit, with map $\mathcal{F}^{(\beta)}(\mathrm{MSE}) = \E_{x,\tilde{z}}[\,\norm{x - \eta^{(\beta)}(x+\bar{\sigma}\tilde{z})}^2\,]$. \emph{(b)}~Letting $\beta\to\infty$: the effective input $x+\bar{\sigma}\tilde{z}$ has a density, hence avoids $\mathcal{N}$ almost surely, so by~(iii) the integrand converges pointwise to $\norm{x-\eta(x+\bar{\sigma}\tilde{z})}^2$; being dominated by $4R_L^2$ (property~(ii)), $\mathcal{F}^{(\beta)}\to\mathcal{F}$ pointwise by dominated convergence.  The convergence is in fact \emph{uniform} on $\mathcal{I}_0$, which is what transfers the fixed points and not merely the maps: the elementary bound $\bigl|\,\norm{x-\eta^{(\beta)}}^2 - \norm{x-\eta}^2\,\bigr| \le 4R_L\,\norm{\eta^{(\beta)}-\eta}$ gives $\sup_{\mathcal{I}_0}\lvert\mathcal{F}^{(\beta)}-\mathcal{F}\rvert \le 4R_L\,\sup_{\mathcal{I}_0}\E\norm{\eta^{(\beta)}(\bar{x})-\eta(\bar{x})}$, and on the compact $\mathcal{I}_0$ the effective variance obeys $\bar{\sigma}^2 \ge \sigma_z^2 > 0$, so the observation densities are uniformly bounded and depend continuously on $\mathrm{MSE}$; the right-hand side therefore tends to zero uniformly. Since uniform limits of continuous self-maps of a compact interval preserve fixed points, the fixed-point equation passes to the limit, and $\mathcal{F}$ is itself continuous (Theorem~\ref{thm:se_convergence}), so the \ac{OPD} and adaptive-\ac{OGD} fixed points are well defined.

\vspace{-2ex}
\subsection{Proof of Theorem~\ref{thm:se_convergence}
            (SE Fixed-Point Existence)}
\label{app:thm_se_convergence}

Throughout, $\bar{\sigma}^2(\mathrm{MSE}) \triangleq \sigma_z^2 + (K/N)\,\mathrm{MSE}$ denotes the effective variance~\eqref{eq:se_noise}. Non-negativity, boundedness, and continuity are established on the larger interval $\mathcal{I}_0 = [0,4R_L^2]$ -- they hold for \emph{every} denoiser there, and hence on the subinterval $\mathcal{I} = [0,E_d] \subseteq \mathcal{I}_0$ as well; the self-map and existence are then settled per class.

\emph{Non-negativity.} $\mathcal{F}(\mathrm{MSE}) = \E\bigl[\,\norm{x - \eta(\bar{x};\bar{\sigma}^2)}^2\,\bigr]$ is the expectation of a non-negative quantity, so $\mathcal{F} \geq 0$ on $\mathcal{I}_0$.

\emph{Boundedness.} Every denoiser output lies in the disk of radius $R_L$,
\begin{equation}
  \norm{\eta(\bar{x};\bar{\sigma}^2)} \;\leq\; R_L,
  \label{eq:eta_bound}
\end{equation}
since the \ac{BOD}, \ac{OBD}, and \ac{OGD} outputs are convex combinations of ring points -- so $\norm{\eta} = \norm{\E[x\mid\bar{x}]} \leq \E[\,\norm{x}\mid\bar{x}\,] \leq R_L$ for the posterior means, and $\norm{\eta_G} \leq R_L$ for the adaptive-fallback \ac{OGD} (Definition~\ref{def:adaptive_cgr}), with which it coincides at the \ac{SNR} of interest -- while $\norm{\eta_P} = R_{\ell^*} \leq R_L$. With $\norm{x} = R_\ell \leq R_L$, the triangle inequality gives the deterministic bound
\begin{equation}
  \norm{x - \eta(\bar{x};\bar{\sigma}^2)}^{2} \;\leq\; \bigl(\norm{x} + \norm{\eta}\bigr)^{2} \;\leq\; 4R_L^{2},
  \label{eq:sq_bound}
\end{equation}
whence, on taking expectations, $\mathcal{F}(\mathrm{MSE}) \leq 4R_L^2$, i.e.\ $\mathcal{F}(\mathcal{I}_0) \subseteq \mathcal{I}_0$. The bound~\eqref{eq:eta_bound} and hence the domination step below cover the four bounded denoisers; the \ac{LMMSE}, whose output is unbounded, is handled directly by its closed form $\mathcal{F}_{\mathrm{L}}(\mathrm{MSE}) = E_d\bar{\sigma}^2/(E_d+\bar{\sigma}^2)$, which is manifestly non-negative, continuous, and valued in $[0,E_d)\subset\mathcal{I}_0$, so non-negativity, self-mapping, and continuity hold for it without invoking the \ac{DCT}.

\emph{Continuity.} Fix $\mathrm{MSE} \in \mathcal{I}_0$ and let $\mathrm{MSE}_n \to \mathrm{MSE}$ be an arbitrary sequence in $\mathcal{I}_0$. Set $\bar{\sigma}_n^2 \triangleq \sigma_z^2 + (K/N)\mathrm{MSE}_n$ and $\bar{\sigma}^2 \triangleq \sigma_z^2 + (K/N)\mathrm{MSE}$; since $\sigma_z^2 > 0$, $\bar{\sigma}_n^2 \geq \sigma_z^2 > 0$ for every $n$ and $\bar{\sigma}_n^2 \to \bar{\sigma}^2 > 0$. We verify the hypotheses of the \ac{DCT}.
\emph{Domination.} By~\eqref{eq:sq_bound}, the integrand obeys $\norm{x - \eta(x + \bar{\sigma}_n \tilde{z};\, \bar{\sigma}_n^2)}^2 \leq 4R_L^2$ uniformly in $n$, and the finite constant $4R_L^2$ is an integrable dominating function.
\emph{Pointwise almost-sure convergence.} Set $\bar{x}_n \triangleq x + \bar{\sigma}_n \tilde{z}$ and $\bar{x} \triangleq x + \bar{\sigma} \tilde{z}$; since $\bar{\sigma}_n \to \bar{\sigma}$, we have $\bar{x}_n \to \bar{x}$ and $\bar{\sigma}_n^2 \to \bar{\sigma}^2$ simultaneously. The event $\{\bar{x} = 0\} = \{\tilde{z} = -x/\bar{\sigma}\}$ has probability zero under $\tilde{z} \sim \mathcal{CN}(0,1)$ and discrete $x$, so $\bar{x} \neq 0$ almost surely. For $\bar{x} \neq 0$ and $\bar{\sigma}^2 > 0$ each denoiser is continuous in $(\bar{x},\bar{\sigma}^2)$: the \ac{BOD}, \ac{OBD}, and \ac{OGD} are ratios of Gaussian-kernel averages over their priors with strictly positive denominator (indeed jointly $C^\infty$), the \ac{LMMSE} is linear, and $\eta_P = R_{\ell^*}e^{j\angle\bar{x}}$ is continuous off the measure-zero set of ring-decision boundaries. Hence $\eta(\bar{x}_n;\bar{\sigma}_n^2) \to \eta(\bar{x};\bar{\sigma}^2)$ almost surely.
\emph{Conclusion.} The \ac{DCT} then yields
\begin{align}
  \lim_{n \to \infty} \mathcal{F}(\mathrm{MSE}_n)
  &= \E\!\Bigl[\, \lim_{n \to \infty} \norm{x - \eta(\bar{x}_n;\bar{\sigma}_n^2)}^2 \,\Bigr] \nonumber \\
  &= \E\!\bigl[\, \norm{x - \eta(\bar{x};\bar{\sigma}^2)}^2 \,\bigr] \;=\; \mathcal{F}(\mathrm{MSE}),
  \label{eq:dct_continuity}
\end{align}
so $\mathcal{F}$ is continuous on $\mathcal{I}_0$, and a fortiori on $\mathcal{I} \subseteq \mathcal{I}_0$.

\medskip\noindent\emph{Case (i): \ac{BOD}, \ac{OBD}, \ac{LMMSE} -- fixed point in $[0,E_d]$.}
For the posterior-mean \ac{BOD} and \ac{OBD}, the circular-symmetry equivalence of Lemma~\ref{lem:phase_preserving} identifies $\mathcal{F}(\mathrm{MSE})$ with the matched \ac{MMSE} of the variance-$E_d$ prior that $\eta$ is optimal for ($p$ for the \ac{BOD}, $q$ for the \ac{OBD}). Since the Gaussian maximizes the \ac{MMSE} among priors of a given variance~\cite[Theorem~3]{Guo2005},
\begin{equation}
  \mathcal{F}(\mathrm{MSE}) \;\leq\; \frac{E_d\,\bar{\sigma}^2(\mathrm{MSE})}{E_d + \bar{\sigma}^2(\mathrm{MSE})} \;<\; E_d ,
  \label{eq:self_map_bound}
\end{equation}
and the \ac{LMMSE} attains the right-hand side with equality. Hence $\mathcal{F}(\mathcal{I}) \subseteq \mathcal{I}$; as $\mathcal{F}$ is continuous on $\mathcal{I}_0 \supseteq \mathcal{I}$, it is a continuous self-map of the compact convex set $\mathcal{I} = [0,E_d]$, and Brouwer's fixed-point theorem yields a fixed point $\mathrm{MSE}^* \in [0,E_d]$ -- the interval of Definition~\ref{def:fixed_point} and Corollary~\ref{cor:se_uniqueness}.

\medskip\noindent\emph{Case (ii): \ac{OGD}, \ac{OPD} -- fixed point in $\mathcal{I}_0$.}
The \ac{OGD} approximates the \ac{OBD}'s von Mises phase posterior by a Gaussian (Definition~\ref{def:gaussian_phase}, still under the orbital prior $q$) and the \ac{OPD} is a hard ring projection; being approximations rather than exact posterior means, they need not satisfy~\eqref{eq:self_map_bound}, so only the boundedness self-map $\mathcal{F}(\mathcal{I}_0) \subseteq \mathcal{I}_0$ is available. As $\mathcal{F}$ is a continuous self-map of the compact convex set $\mathcal{I}_0 = [0,4R_L^2]$, Brouwer's theorem yields a fixed point $\mathrm{MSE}^* \in \mathcal{I}_0$; its high-\ac{SNR} localization to $[0,E_d]$ and value are computed in Proposition~\ref{prop:cgr_se} (\ac{OGD}) and Corollary~\ref{cor:lpd_se} (\ac{OPD}).

\vspace{-2ex}
\subsection{Proof of Proposition~\ref{prop:se_monotone}
            (Monotonicity of the SE Map)}
\label{app:prop_se_monotone}

Each of the three denoisers is the Bayes posterior mean for the effective channel $\bar{x} = x + \bar{\sigma}\tilde{z}$ under a prior $\varpi$: the \ac{BOD} under the true prior $p$; the \ac{OBD} under the orbital prior $q$ -- and by the circular-symmetry equivalence of Lemma~\ref{lem:phase_preserving}(iii), its mismatched \ac{MSE} under $p$ equals its matched \ac{MSE} under $q$; and the \ac{LMMSE} under the circularly-symmetric Gaussian $\varpi = \mathcal{CN}(0, E_d)$, whose posterior mean is linear. In every case the \ac{SE} map is an \ac{MMSE},
\begin{equation}
  \mathcal{F}(\mathrm{MSE}) \!=\! \mathrm{mmse}_\varpi\!\bigl(\bar{\sigma}^2(\mathrm{MSE})\bigr),
  \;
  \bar{\sigma}^2(\mathrm{MSE}) \!=\! \sigma_z^2 + \alpha\,\mathrm{MSE}.
\end{equation}
By the \ac{I-MMSE} relation~\cite{Guo2005}, $\mathrm{mmse}_\varpi(\cdot)$ is non-decreasing in the noise variance for \emph{every} prior $\varpi$ -- degrading the channel cannot reduce the \ac{MMSE} -- and since $\bar{\sigma}^2(\mathrm{MSE})$ is increasing in $\mathrm{MSE}$, $\mathcal{F}$ is non-decreasing on $[0, E_d]$.

Monotone convergence from $\mathrm{MSE}_0 = E_d$ follows because $\mathcal{F}(E_d) = \mathrm{mmse}_\varpi(\bar{\sigma}^2(E_d)) \leq \mathbb{E}_\varpi[\norm{x}^2] = E_d$ -- the \ac{MMSE} never exceeds the input second moment, attained by the zero estimator -- and $\mathcal{F}$ is continuous and non-decreasing on the compact set $[0,E_d]$.

The \ac{OGD} and \ac{OPD} are \emph{not} posterior means of any prior, so this \ac{MMSE}-monotonicity argument does not apply to them; their fixed points are instead constructed explicitly at high \ac{SNR} (Proposition~\ref{prop:cgr_se}, Corollary~\ref{cor:lpd_se}; see Remark~\ref{rem:fp_scope}).

\vspace{-2ex}
\subsection{Proof of Proposition~\ref{prop:fp_ordering}
            (Fixed-Point Ordering)}
\label{app:prop_fp_ordering}

The proof has two steps: establishing the per-step (pointwise) ordering of the \ac{SE} maps, then lifting it to a fixed-point ordering via the comparison lemma for monotone maps.

\emph{Step 1: Per-step ordering.}
For any fixed $\mathrm{MSE} \in [0, E_d]$, define $\bar{\sigma}^2 \triangleq \sigma_z^2 + \alpha\mathrm{MSE}$ and the effective scalar channel $\bar{x} = x + \bar{\sigma} \tilde{z}$.
The \ac{BOD} $\eta_D(\bar{x}; \bar{\sigma}^2) = \E[x \mid \bar{x}]$ (under the true discrete prior) minimizes $\E[\norm{x - \eta(\bar{x})}^2]$ over \emph{all} measurable functions $\eta: \Complex \to \Complex$ by the orthogonality principle~\cite{kay1993fundamentals}.
Since $\eta_B$ (the \ac{OBD}) and $\eta_{\mathrm{L}}$ are particular measurable functions, we have
\begin{align}
\label{eq:perstep_lower}
\mathcal{F}_{\mathrm{D}}(\mathrm{MSE}) 
&= \E\bigl[\norm{x - \eta_D(\bar{x}; \bar{\sigma}^2)}^2\bigr] \nonumber \\
&\leq \E\bigl[\norm{x - \eta_B(\bar{x}; \bar{\sigma}^2)}^2\bigr] = \mathcal{F}_{\mathrm{B}}(\mathrm{MSE}).
\end{align}

For the upper bound, we establish $\mathcal{F}_{\mathrm{B}}(\mathrm{MSE}) \leq \mathcal{F}_{\mathrm{L}}(\mathrm{MSE})$ via the following three-step argument.

\emph{(a) Circular symmetry equivalence.}
Since the \ac{OBD} is phase-preserving, Lemma~\ref{lem:phase_preserving}(iii) gives the circular-symmetry equivalence~\eqref{eq:circular_symmetry_equiv}: its \ac{MSE} is the same under the true discrete prior $p$ and the orbital prior $q$, $\E_p[\norm{x - \eta_B(\bar{x})}^2] = \E_q[\norm{x - \eta_B(\bar{x})}^2]$.

\emph{(b) Bayes optimality under the orbital prior.}
Under the orbital prior $q$ (Definition~\ref{def:orbital}), the denoiser $\eta_B$ is the Bayes-optimal (posterior mean) estimator.
Therefore, $\E_q[\norm{x - \eta_B(\bar{x})}^2] = \mathrm{mmse}_q(\bar{\sigma}^2)$, the \ac{MMSE} of estimating $x \sim q$ from $\bar{x} = x + \bar{\sigma} \tilde{z}$.

\emph{(c) Gaussian maximizes MMSE.}
By the maximum-MMSE property of the Gaussian distribution~\cite[Theorem~3]{Guo2005}, among all inputs with second moment $E_d$, the circularly-symmetric Gaussian $x_{\mathrm{G}} \sim \mathcal{CN}(0, E_d)$ achieves the largest \ac{MMSE}:
\begin{equation}
\mathrm{mmse}_q(\bar{\sigma}^2) \;\leq\; \mathrm{mmse}_{\mathrm{Gauss}}(\bar{\sigma}^2) \;=\; \frac{E_d\,\bar{\sigma}^2}{E_d + \bar{\sigma}^2}.
\end{equation}
The equality holds because for Gaussian input the posterior mean \emph{is} linear, so $\mathrm{mmse}_{\mathrm{Gauss}}$ coincides with the linear \ac{MMSE}, which depends on the input only through its second moment.

Combining (a)--(c):
\begin{equation}
\label{eq:perstep_upper}
\mathcal{F}_{\mathrm{B}}(\mathrm{MSE})
\;=\; \mathrm{mmse}_q(\bar{\sigma}^2)
\;\leq\; \frac{E_d\,\bar{\sigma}^2}{E_d + \bar{\sigma}^2}
\;=\; \mathcal{F}_{\mathrm{L}}(\mathrm{MSE}),
\end{equation}
valid at \emph{all} \ac{SNR} without any high-probability qualifications.

\emph{Step 2: From per-step ordering to fixed-point ordering (Comparison Lemma).}
By Proposition~\ref{prop:se_monotone}, all three \ac{SE} maps are non-decreasing on $[0, E_d]$.
We apply the following standard comparison principle for monotone maps (cf.~\cite[Lemma~3]{ReevesPfister2019}): if $\mathcal{F}\le\mathcal{G}$ pointwise and both are non-decreasing continuous self-maps of $[0,E_d]$, then their respective largest fixed points satisfy $\mathrm{MSE}^*_{\mathcal{F}}\le\mathrm{MSE}^*_{\mathcal{G}}$.

Applying the Comparison Lemma to the pairs $(\mathcal{F}_{\mathrm{D}}, \mathcal{F}_{\mathrm{B}})$ and $(\mathcal{F}_{\mathrm{B}}, \mathcal{F}_{\mathrm{L}})$ with the per-step ordering~\eqref{eq:perstep_lower}--\eqref{eq:perstep_upper}, and noting that all three sequences initialized at $\mathrm{MSE}_0 = E_d$ converge monotonically to their respective largest fixed points by Proposition~\ref{prop:se_monotone} (monotone convergence from above), we obtain
\begin{align*}
\mathrm{MSE}_\infty^{\mathrm{D}} &= \lim_{t \to \infty} \mathrm{MSE}_t^{\mathrm{D}} \;\leq\; \lim_{t \to \infty} \mathrm{MSE}_t^{\mathrm{B}} \\
&= \mathrm{MSE}_\infty^{\mathrm{B}} \;\leq\; \lim_{t \to \infty} \mathrm{MSE}_t^{\mathrm{L}} = \mathrm{MSE}_\infty^{\mathrm{L}},
\end{align*}
establishing~\eqref{eq:fp_ordering}.

\emph{Strictness of the lower bound.}
The inequality $\mathrm{MSE}_\infty^{\mathrm{D}} < \mathrm{MSE}_\infty^{\mathrm{B}}$ is strict at any finite \ac{SNR} because the \ac{OBD} is \emph{not} the \ac{BOD} for the true discrete prior: the continuous phase relaxation assigns positive probability to phase angles where no constellation point exists, inducing a strictly positive excess \ac{MSE} at every $\mathrm{MSE} > 0$.
Specifically, $\mathcal{F}_{\mathrm{D}}(\mathrm{MSE}) < \mathcal{F}_{\mathrm{B}}(\mathrm{MSE})$ for all $\mathrm{MSE} > 0$ (not merely $\leq$), which propagates to a strict fixed-point gap by continuity and monotonicity. Indeed, if the two largest fixed points coincided at a common $m > 0$, then $m = \mathcal{F}_{\mathrm{D}}(m) < \mathcal{F}_{\mathrm{B}}(m) = m$, a contradiction; combined with $\mathrm{MSE}_\infty^{\mathrm{D}} \le \mathrm{MSE}_\infty^{\mathrm{B}}$ this forces the strict inequality $\mathrm{MSE}_\infty^{\mathrm{D}} < \mathrm{MSE}_\infty^{\mathrm{B}}$ at every finite \ac{SNR}.

\vspace{-2ex}
\subsection{Proof of Corollary~\ref{cor:se_uniqueness}
            (Uniqueness)}
\label{app:cor_se_uniqueness}

\emph{(i) \ac{LMMSE}.}
The \ac{LMMSE} estimator is linear, so $\mathcal{F}_{\mathrm{L}}(\mathrm{MSE}) = E_d\bar{\sigma}^2/(E_d+\bar{\sigma}^2)$ with $\bar{\sigma}^2 = \sigma_z^2 + \alpha\,\mathrm{MSE}$.  Differentiating, we have
\begin{equation*}
\mathcal{F}_{\mathrm{L}}'(\mathrm{MSE}) = \alpha\,\frac{E_d^2}{(E_d+\bar{\sigma}^2)^2} < \alpha < 1
\qquad\text{for every } \sigma_z^2 > 0,
\end{equation*}
so $\mathcal{F}_{\mathrm{L}}$ is a global contraction on $[0,E_d]$ and, by the Banach fixed-point theorem~\cite[Th.~5.1-2]{kreyszig1978introductory}, has a unique fixed point at every \ac{SNR}.

\emph{(ii) \ac{BOD} and \ac{OBD}.}
Let $\eta$ be the posterior mean ($\eta_B = \E_q[x\mid\bar{x}]$ for the \ac{OBD}, $\eta_D = \E_p[x\mid\bar{x}]$ for the matched \ac{BOD}) and $\hat{\sigma}^2(\bar{x})$ the corresponding posterior variance.  By Tweedie's identity for the complex Gaussian channel -- of which the Stein identity~\eqref{eq:stein_identity} is the \ac{OBD} instance -- $\partial\eta/\partial\bar{x} = D(\bar{x}) \triangleq \hat{\sigma}^2(\bar{x})/\bar{\sigma}^2 \geq 0$.
Since $\mathcal{F}(\mathrm{MSE}) = \mathrm{mmse}_\varpi(\bar{\sigma}^2)$ with $\bar{\sigma}^2 = \sigma_z^2 + \alpha\,\mathrm{MSE}$ (proof of Proposition~\ref{prop:se_monotone}), {the chain rule and the complex-channel \ac{MMSE}-derivative identity
\begin{equation*}
\frac{\mathrm{d}\,\mathrm{mmse}_\varpi}{\mathrm{d}\bar{\sigma}^2} = \frac{2}{\bar{\sigma}^4}\,\E\bigl[\|\bm{\Sigma}(\bar{x})\|_F^2\bigr],
\end{equation*}
where $\bm{\Sigma}(\bar{x})$ is the $2\times2$ real conditional covariance of $x$ given $\bar{x}$ (the vector-channel form of the real-scalar identity $\mathrm{d}\,\mathrm{mmse}/\mathrm{d}\bar{\sigma}^2 = \E[D^2]$ of~\cite{GuoWuShamaiVerdu_TIT_2011}), give the \ac{SE}-map derivative
\begin{equation*}
\mathcal{F}'(\mathrm{MSE}) = \frac{2\alpha}{\bar{\sigma}^4}\,\E\bigl[\|\bm{\Sigma}(\bar{x})\|_F^2\bigr].
\end{equation*}

Since $\tfrac12(\operatorname{tr}\bm{\Sigma})^2 \leq \|\bm{\Sigma}\|_F^2 \leq (\operatorname{tr}\bm{\Sigma})^2$ with $\operatorname{tr}\bm{\Sigma} = \hat{\sigma}^2$, the derivative is sandwiched as $\alpha\,\E[D^2] \leq \mathcal{F}' \leq 2\alpha\,\E[D^2]$: the lower (isotropic) extreme is attained by circularly symmetric posteriors -- e.g.\ the Gaussian prior, for which $\bm{\Sigma} = (\hat{\sigma}^2/2)\mathbf{I}_2$ with deterministic $D = E_d/(E_d+\bar{\sigma}^2)$, recovering $\mathcal{F}' = \alpha D^2$ consistently with (i) -- and the upper (rank-one) extreme by posteriors whose uncertainty concentrates in a single real direction.}
The posterior variance is bounded by the prior support, $\hat{\sigma}^2(\bar{x}) \leq R_L^2$, so $D \leq R_L^2/\bar{\sigma}^2$ pointwise, and the mean of $D$ is always subunit: $\E[D] = \mathrm{mmse}_\varpi(\bar{\sigma}^2)/\bar{\sigma}^2 \leq \frac{E_d}{E_d + \bar{\sigma}^2} < 1$ by the Gaussian maximum-\ac{MMSE} bound~\cite[Theorem~3]{Guo2005}{, strict for the non-Gaussian priors considered here}.

(a) \emph{Low \ac{SNR}.}  If $\sigma_z^2 \geq R_L^2$ then $\bar{\sigma}^2 \geq R_L^2$ and $D \leq R_L^2/\bar{\sigma}^2 \leq 1$ pointwise, so $D^2 \leq D$ and {$\mathcal{F}'(\mathrm{MSE}) \leq 2\alpha\,\E[D^2] \leq 2\alpha\,\E[D] < \frac{2\alpha\,E_d}{E_d + \bar{\sigma}^2} \leq \alpha < 1$ throughout $[0,E_d]$, the last inequality using $\bar{\sigma}^2 \geq R_L^2 \geq E_d$.}  The Banach theorem gives a unique fixed point.

(b) \emph{High \ac{SNR}.}  As $\sigma_z^2 \to 0$ the posterior concentrates on the correct ring (ring misdetection being exponentially rare, Proposition~\ref{prop:ring_discrimination}), so $\operatorname{Var}(D)\to0$.  For the \ac{OBD} the residual is the continuous phase uncertainty: the single-ring posterior variance~\eqref{eq:cbm_var} is $\hat{\sigma}^2_{\mathrm{B}} = R_{\ell^*}^2\bigl(1 - A(\kappa_{\ell^*})^2\bigr) + o(\bar{\sigma}^2)$, and since $1 - A(\kappa)^2 = 1/\kappa + \mathcal{O}(\kappa^{-3})$ (Proposition~\ref{prop:vm_gaussian}; the $\kappa^{-2}$ term cancels) with $\kappa_{\ell^*} = 2R_{\ell^*}\norm{\bar{x}}/\bar{\sigma}^2$, $\norm{\bar{x}} \to R_{\ell^*}$, one gets $\hat{\sigma}^2_{\mathrm{B}} \to \bar{\sigma}^2/2$, so $D \to \tfrac12$ and $\E[D^2] \to \tfrac14$.  {Moreover the conditional covariance is asymptotically \emph{rank-one} -- the radius is pinned to the ring and all residual uncertainty is tangential (phase) -- so $\|\bm{\Sigma}\|_F^2 \to (\operatorname{tr}\bm{\Sigma})^2 = \hat{\sigma}^4$ and the derivative attains the upper extreme of the sandwich, $\mathcal{F}'(\mathrm{MSE}^*) \to 2\alpha\,\E[D^2] = \alpha/2$, consistently with the direct expansion~\eqref{eq:obd_Fmap}.}  For the matched \ac{BOD} the discrete posterior concentrates on the true point, whose nearest neighbours are $\Omega(1)$-separated, so $\hat{\sigma}^2_{\mathrm{D}} \to 0$ exponentially, $D \to 0$, and $\E[D^2] \to 0$.  In either case {$\mathcal{F}'(\mathrm{MSE}^*) < 1$ in the limit (the limit being $\alpha/2$ for the \ac{OBD} and $0$ for the \ac{BOD})}, so there is a threshold $\sigma_{\mathrm{th}}^2(\alpha) > 0$ below which $\mathcal{F}$ is locally contractive at its fixed point ($\mathrm{MSE}^* = \mathcal{O}(\sigma_z^2)$); monotone convergence from $\mathrm{MSE}_0 = E_d$ (Proposition~\ref{prop:se_monotone}) identifies the \emph{largest} fixed point $\mathrm{MSE}_\infty = \mathrm{MSE}^*$, and local contractivity makes it the unique fixed point in a neighborhood. Since at high \ac{SNR} $\bar{\sigma}^2 = \sigma_z^2 + \alpha\,\mathrm{MSE}$ stays $\mathcal{O}(\sigma_z^2)$ for every $\mathrm{MSE}\in[0,\mathrm{MSE}^*]$, the derivative bound $\mathcal{F}' \le 2\alpha\,\E[D^2] \to \alpha/2 < 1$ holds throughout $[0,\mathrm{MSE}^*]$, so $\mathcal{F}$ is a contraction there and admits no smaller fixed point; combined with $\mathcal{F}(0) = \mathrm{mmse}_\varpi(\sigma_z^2) > 0$ this gives uniqueness of $\mathrm{MSE}_\infty$ (the operative fixed point) on $[0,E_d]$.

In the intermediate band $\sigma_{\mathrm{th}}^2 \leq \sigma_z^2 < R_L^2$ the pointwise bound $D \leq 1$ fails on a small set (near decision boundaries the conditional variance can approach the support diameter) and $\E[D^2]$ may exceed $1/\alpha$, so a global contraction is not guaranteed; this reflects the possible steepening of the discrete/orbital \ac{MMSE} map (the \ac{AMP} phase-transition mechanism) and does not affect the analysis, since the largest fixed point $\mathrm{MSE}_\infty$ remains well-defined by monotone convergence.

\vspace{-2ex}
\subsection{Proof of Proposition~\ref{prop:linear_decay}
            (Linear Variance Decay)}
\label{app:prop_linear_decay}

As $\bar{\sigma}^2 \to 0$, the concentration parameter $\kappa_\ell = 2R_\ell\norm{\bar{x}}/\bar{\sigma}^2 \to \infty$.
Let $\ell^*$ denote the radius-nearest (dominant) ring. The \ac{OBD} log-metric~\eqref{eq:obd_metric} is dominated by its radial term: since $\ln I_0(\kappa_\ell) = \kappa_\ell - \tfrac{1}{2}\ln(2\pi\kappa_\ell) + \mathcal{O}(\kappa_\ell^{-1})$ (Proposition~\ref{prop:vm_gaussian}) and $\kappa_\ell = 2R_\ell\norm{\bar{x}}/\bar{\sigma}^2$, we have $\Lambda_\ell^{\mathrm{B}} = (2R_\ell\norm{\bar{x}} - R_\ell^2)/\bar{\sigma}^2 + \mathcal{O}(\ln\bar{\sigma}^2)$. The gap $\Lambda_{\ell^*}^{\mathrm{B}} - \Lambda_\ell^{\mathrm{B}} = (R_{\ell^*}-R_\ell)(2\norm{\bar{x}}-R_{\ell^*}-R_\ell)/\bar{\sigma}^2 + \mathcal{O}(\ln\bar{\sigma}^2)$ therefore diverges as $1/\bar{\sigma}^2$ for every $\ell \neq \ell^*$, so the softmax ring weights concentrate, $w_{\ell^*}^{\mathrm{B}} \to 1$ exponentially -- the same mechanism established for the \ac{OGD} in Proposition~\ref{prop:lpd_limit}. 
Consequently the total \ac{OBD} posterior variance~\eqref{eq:obd_var} collapses to the variance conditioned on the dominant ring.
Conditioned on ring $\ell^*$, the complex posterior variance under the von Mises distribution is exactly $\hat{\sigma}_{\mathrm{B},\ell^*}^2 = R_{\ell^*}^2\bigl(1 - A(\kappa_{\ell^*})^2\bigr)$.
This follows because $\E[\norm{x}^2 \mid R_{\ell^*}] = R_{\ell^*}^2$ (the radius is deterministic on a fixed ring), while $\norm{\E[x \mid \bar{x}, R_{\ell^*}]}^2 = R_{\ell^*}^2 A(\kappa_{\ell^*})^2$ from \eqref{eq:ring_mean_cbm}, so $\hat{\sigma}_{\mathrm{B},\ell^*}^2 = R_{\ell^*}^2\bigl(1 - A(\kappa_{\ell^*})^2\bigr)$.
Applying the large-argument asymptotic expansion $A(\kappa) = 1 - 1/(2\kappa) - 1/(8\kappa^2) + \mathcal{O}(\kappa^{-3})$ yields $1 - A(\kappa)^2 = 1/\kappa + \mathcal{O}(\kappa^{-3})$, the $\kappa^{-2}$ contributions cancelling. 
Substituting $\kappa_{\ell^*} = 2R_{\ell^*}\norm{\bar{x}}/\bar{\sigma}^2$ gives
\begin{equation*}
  \hat{\sigma}_{\mathrm{B},\ell^*}^2 \sim \frac{R_{\ell^*}^2}{\kappa_{\ell^*}} = \frac{R_{\ell^*}\,\bar{\sigma}^2}{2\norm{\bar{x}}},
\end{equation*}
establishing \eqref{eq:cbm_var_decay}.

\vspace{-2ex}
\subsection{Proof of Theorem~\ref{thm:cross_level_fp}
            (Cross-Level SE Equivalence)}
\label{app:thm_cross_level_fp}

The three per-level fixed points are established in the main text, each via the orthogonality (excess) identity of Lemma~\ref{lem:excess}, which adds to the (under $q$) Bayes-optimal \ac{OBD} error the squared amplitude-shrinkage bias $\E[R_{\ell^*}^2(A(\kappa_{\ell^*})-m_\eta)^2]$ of the denoiser in use:
\begin{itemize}
\item the \ac{OBD} correction $\delta^{\mathrm{B}} = \gamma\sigma_z^4/(2-\alpha)^3$ (Proposition~\ref{prop:obd_se}), from the $\mathcal{O}(\kappa^{-2})$ cancellation in $1-A(\kappa)^2$ and the curvature~\eqref{eq:obd_Fmap} of the per-iteration variance;
\item the \ac{OGD} correction $\delta^{\mathrm{G}} = \delta^{\mathrm{B}} + \mathcal{O}(\sigma_z^8)$ (Proposition~\ref{prop:cgr_se}), the Gaussian-phase magnitude bias $A(\kappa)-m_{\mathrm{G}} = \mathcal{O}(\kappa^{-2})$ entering squared;
\item the \ac{OPD} correction $\delta^{\mathrm{P}} = \tfrac{3}{2}\delta^{\mathrm{B}} + \mathcal{O}(e^{-(2-\alpha)d_R^2/(8\sigma_z^2)})$ (Corollary~\ref{cor:lpd_se}), the no-shrinkage bias $1-A(\kappa) = \Theta(\kappa^{-1})$ contributing exactly half the \ac{OBD} curvature, $\gamma(\bar{\sigma}^2)^2/16$.
\end{itemize}
Substituting these into the unified expansion~\eqref{eq:unified_fp} and subtracting gives the pairwise differences $\norm{\mathrm{MSE}_\infty^{\mathrm{G}} - \mathrm{MSE}_\infty^{\mathrm{B}}} = \mathcal{O}(\sigma_z^8)$ and $\mathrm{MSE}_\infty^{\mathrm{P}} - \mathrm{MSE}_\infty^{\mathrm{B}} = \tfrac{1}{2}\delta^{\mathrm{B}} = \Theta(\sigma_z^4) > 0$, which is~\eqref{eq:fp_gaps}.

\vspace{-2ex}
\section{Proofs for Section~\ref{sec:inf_theory}
         (Geometry-Rate Tradeoffs)}
\label{app:sec6}

\subsection{Proof of Corollary~\ref{cor:fp_ordering_hierarchy}
            (Full Hierarchy Ordering)}
\label{app:cor_fp_ordering_hierarchy}

The \ac{OGD} and \ac{OPD} are bounded but non-Bayes approximations, whose excess over the \ac{OBD} is governed by the identity of Lemma~\ref{lem:excess},
\begin{equation*}
\mathcal{F}_\eta(\varepsilon) = \mathcal{F}_{\mathrm{B}}(\varepsilon) + \E\bigl[R_{\ell^*}^2\,(A(\kappa_{\ell^*}) - m_\eta)^2\bigr], \qquad \eta \in \{\mathrm{G},\mathrm{P}\},
\end{equation*}
with $m_{\mathrm{G}} = 1 - 1/(2\kappa)$ and $m_{\mathrm{P}} = 1$.
The excess in this identity is a squared norm, hence non-negative, so $\mathcal{F}_{\mathrm{G}}, \mathcal{F}_{\mathrm{P}} \geq \mathcal{F}_{\mathrm{B}}$ pointwise; propagated through the positive fixed-point amplification $1/(1-\mathcal{F}_{\mathrm{B}}')$, this makes the fixed-point increments of Theorem~\ref{thm:cross_level_fp} non-negative; i.e., $\mathrm{MSE}_\infty^{\mathrm{G}} = \mathrm{MSE}_\infty^{\mathrm{B}} + \mathcal{O}(\sigma_z^8)$ (a non-negative $\mathcal{O}(\sigma_z^8)$) and $\mathrm{MSE}_\infty^{\mathrm{P}} = \mathrm{MSE}_\infty^{\mathrm{B}} + \Theta(\sigma_z^4)$ (a strictly positive $\Theta(\sigma_z^4)$). Since the \ac{OPD}'s $\Theta(\sigma_z^4)$ separation dominates the \ac{OGD}'s $\mathcal{O}(\sigma_z^8)$ one, $\mathrm{MSE}_\infty^{\mathrm{B}} \leq \mathrm{MSE}_\infty^{\mathrm{G}} \leq \mathrm{MSE}_\infty^{\mathrm{P}}$ for $\sigma_z^2$ small enough.
Finally $\mathrm{MSE}_\infty^{\mathrm{P}} \leq \mathrm{MSE}_\infty^{\mathrm{L}}$ in this regime, since $\mathrm{MSE}_\infty^{\mathrm{L}} = \sigma_z^2/(1-\alpha) + \mathcal{O}(\sigma_z^4)$ exceeds $\mathrm{MSE}_\infty^{\mathrm{P}} = \sigma_z^2/(2-\alpha) + \mathcal{O}(\sigma_z^4)$ by the strictly positive leading gap $\sigma_z^2\bigl[(1-\alpha)^{-1} - (2-\alpha)^{-1}\bigr] > 0$.
Taking $\sigma_\star^2$ as the smallest of the three thresholds so obtained establishes~\eqref{eq:fp_ordering_full}.
At low \ac{SNR} the \ac{OGD} and \ac{OPD} fixed points escape above $\mathrm{MSE}_\infty^{\mathrm{L}}$: the non-shrinking \ac{OPD} lands on the outermost ring with essentially random phase, so $\mathrm{MSE}_{\mathrm{P}} \to E_d + R_L^2 > E_d \geq \mathrm{MSE}_{\mathrm{L}}$, and the \ac{LMMSE} baseline -- which shrinks toward the prior mean -- is the more robust estimator in the noise-dominated regime.

\vspace{-2ex}
\subsection{Proof of Theorem~\ref{thm:ot_bound}
            (Transport-Theoretic Bound)}
\label{app:thm_ot_bound}

\emph{Step 1: Exact excess identity (orthogonality).}
Fix the common effective noise $\bar{\sigma}^2$ and write $\bar{x} = x + \bar{\sigma}\tilde{z}$ with $x \sim p$. Expanding $x - \eta_B = (x - \eta_D) + (\eta_D - \eta_B)$,
\begin{align*}
\mathcal{F}_{\mathrm{B}}(\mathrm{MSE}) - \mathcal{F}_{\mathrm{D}}(\mathrm{MSE})
&= \E\norm{\eta_D(\bar{x}) - \eta_B(\bar{x})}^2 \\
&\hspace{-10ex}+ 2\,\mathrm{Re}\,\E\bigl[(x - \eta_D(\bar{x}))^*(\eta_D(\bar{x}) - \eta_B(\bar{x}))\bigr],
\end{align*}
and the cross term vanishes \emph{exactly}: $\eta_D = \E_p[x \mid \bar{x}]$ is the Bayes estimator under the true prior $p$, and $\eta_D - \eta_B$ is $\bar{x}$-measurable, so the orthogonality principle applies (all of $x,\eta_D,\eta_B$ lie in $L^2$ -- $\norm{\eta}\le R_L$ by~\eqref{eq:eta_bound} and $x$ is finitely supported -- so the $L^2$ projection is well-defined). This proves the identity in~\eqref{eq:ot_mse_bound}; it remains to bound $\E\norm{\eta_D - \eta_B}^2$.

\emph{Step 2: Tweedie representation.}
Identify $\Complex$ with $\Real^2$ and let $\phi_t$ denote the density of $\mathcal{N}(0, t\mathbf{I}_2)$; the effective channel adds Gaussian noise of variance $t = \bar{\sigma}^2/2$ per real dimension, so smoothing each prior by $\phi_t$ yields the observation densities $p_t \triangleq p * \phi_t$ and $q_t \triangleq q * \phi_t$. Tweedie's identity (the Gaussian-channel posterior-mean representation invoked for the Stein identity~\eqref{eq:stein_identity}; cf.~\cite{Guo2005}) gives, for \emph{any} prior,
\begin{equation*}
\eta_D(\bar{x}) = \bar{x} + t\,\nabla \ln p_t(\bar{x}),
\qquad
\eta_B(\bar{x}) = \bar{x} + t\,\nabla \ln q_t(\bar{x}),
\end{equation*}
so that, with the expectation taken under the true marginal $\bar{x} \sim p_t$,
\begin{equation}
\label{eq:gap_fisher}
\E\norm{\eta_D - \eta_B}^2
= t^2\, \E_{p_t}\!\left\|\nabla \ln \frac{p_t}{q_t}\right\|^2
= t^2\, J\bigl(p_t \,\|\, q_t\bigr),
\end{equation}
where $J(\cdot\|\cdot)$ denotes the relative Fisher information. For every $\tau > 0$ the smoothed densities $p * \phi_\tau$ and $q * \phi_\tau$ are $C^\infty$, strictly positive, and Gaussian-tailed, so every integration by parts in the de Bruijn step below is justified, and $W_2(p,q) < \infty$ (bounded constellation supports) makes $D(p * \phi_\tau \,\|\, q * \phi_\tau) < \infty$.

\emph{Step 3: Smoothed relative entropy is transport-bounded.}
For any coupling $\pi \in \Pi(p, q)$, joint convexity of the \ac{KL} divergence~\cite[Ch.~2]{CoverThomas2006} and the Gaussian formula $D(\mathcal{N}(x, s\mathbf{I}) \| \mathcal{N}(y, s\mathbf{I})) = \norm{x-y}^2/(2s)$ give, for every $\tau > 0$,
\begin{equation}
\label{eq:kl_transport}
D\bigl(p * \phi_\tau \,\|\, q * \phi_\tau\bigr)
\leq \E_{\pi}\!\left[\frac{\norm{x - y}^2}{2\tau}\right]
= \frac{W_2(p, q)^2}{2\tau},
\end{equation}
optimizing over couplings.

\emph{Step 4: de Bruijn integration.}
Along the simultaneous heat flow $\tau \mapsto (p * \phi_\tau, q * \phi_\tau)$, the relative de Bruijn identity gives $\frac{\mathrm{d}}{\mathrm{d}\tau} D(p*\phi_\tau \| q*\phi_\tau) = -\tfrac{1}{2} J(p*\phi_\tau \| q*\phi_\tau)$, and the relative Fisher information is non-increasing along the flow~\cite{villani2003, PolyanskiyWu_book_2024}. Hence
\begin{equation*}
D\bigl(t/2\bigr) \;\geq\; D\bigl(t/2\bigr) - D(t)
= \tfrac{1}{2}\!\int_{t/2}^{t}\! J(\tau)\,\mathrm{d}\tau
\;\geq\; \tfrac{t}{4}\, J(t),
\end{equation*}
so that, combining with~\eqref{eq:kl_transport} at $\tau = t/2$,
\begin{equation*}
J(t) \;\leq\; \frac{4}{t}\, D(t/2) \;\leq\; \frac{4}{t}\cdot\frac{W_2(p,q)^2}{t} = \frac{4\,W_2(p,q)^2}{t^2}.
\end{equation*}
Substituting into~\eqref{eq:gap_fisher} yields the clean, $\bar{\sigma}^2$-\emph{uniform} bound
\begin{equation*}
\E\norm{\eta_D - \eta_B}^2 = t^2 J(t) \;\leq\; 4\,W_2(p, q)^2.
\end{equation*}

\emph{Step 5: Ring-by-ring coupling.}
Since the orbital prior $q$ preserves the ring masses $r_\ell$ of $p$, coupling the two priors ring by ring gives $W_2(p, q)^2 \leq \sum_\ell r_\ell\, W_2(p_\ell, q_\ell)^2 = \overline{W}_2^{\,2}$, establishing~\eqref{eq:ot_mse_bound}. The per-ring value~\eqref{eq:w2_ring} follows from the intra-arc optimal coupling of the proof of Proposition~\ref{prop:cbm_error} (unchanged for the squared-chord cost, which is likewise increasing in angular distance): $W_2^2 = \frac{M_\ell}{2\pi}\int_{-\pi/M_\ell}^{\pi/M_\ell} 4R_\ell^2\sin^2(\theta/2)\,\mathrm{d}\theta = 2R_\ell^2\bigl(1 - \frac{\sin(\pi/M_\ell)}{\pi/M_\ell}\bigr)$.

\emph{Step 6: Fixed-point transfer.}
Write $m_{\mathrm{B}} = \mathcal{F}_{\mathrm{B}}(m_{\mathrm{B}})$ and $m_{\mathrm{D}} = \mathcal{F}_{\mathrm{D}}(m_{\mathrm{D}})$ for the two fixed points ($m_{\mathrm{B}} \geq m_{\mathrm{D}}$ by Proposition~\ref{prop:fp_ordering}). Then
\begin{align*}
m_{\mathrm{B}} \!-\! m_{\mathrm{D}}
&= \bigl[\mathcal{F}_{\mathrm{B}}(m_{\mathrm{B}}) - \mathcal{F}_{\mathrm{D}}(m_{\mathrm{B}})\bigr]
 + \bigl[\mathcal{F}_{\mathrm{D}}(m_{\mathrm{B}}) - \mathcal{F}_{\mathrm{D}}(m_{\mathrm{D}})\bigr] \\
&\leq 4\,\overline{W}_2^{\,2} + c_{\mathrm{D}}\,(m_{\mathrm{B}} - m_{\mathrm{D}}),
\end{align*}
by part~(i) at noise level $\bar{\sigma}^2(m_{\mathrm{B}})$ and the mean value theorem; rearranging gives~\eqref{eq:ot_fp_bound} whenever $c_{\mathrm{D}} < 1$. The regimes where $c_{\mathrm{D}} < 1$ holds -- with $c_{\mathrm{D}} = \alpha\,\E[D^2] \leq \alpha$ at high noise and $c_{\mathrm{D}} \to 0$ at high \ac{SNR} -- are those of Corollary~\ref{cor:se_uniqueness}.

\bibliographystyle{IEEEtran}
\bibliography{references}

\end{document}